\documentclass[11pt,a4paper,hidelinks]{book}
\usepackage{IEEEtrantools}

\usepackage[a4paper, centering,
		     top=3.5cm, bottom=3cm,
		     left=2cm, right=2cm,
		     bindingoffset=1.5cm]{geometry}
\usepackage{setspace}
\usepackage[explicit]{titlesec}
\usepackage{lmodern}
\usepackage{lipsum}

\newlength\chapnumb
\titleformat{\chapter}[block]{\normalfont}{}{0pt}
{\parbox[b]{\chapnumb}{\fontsize{120}{144}\selectfont\thechapter}
 \parbox[b]{\dimexpr\textwidth-\chapnumb\relax}{\raggedleft\hfill{\fontsize{35}{42}\selectfont#1}\\
    									  \rule{\dimexpr\textwidth-\chapnumb\relax}{0.4pt}}}
\titleformat{name=\chapter,numberless}[block]{\normalfont}{}{0pt}
{\parbox[b]{\chapnumb}{\mbox{}}
 \parbox[b]{\dimexpr\textwidth-\chapnumb\relax}{\raggedleft\hfill{\fontsize{35}{42}\selectfont#1}\\
									  \rule{\dimexpr\textwidth-\chapnumb\relax}{0.4pt}}}

\usepackage{fancyhdr}
\makeatletter
\def\cleardoublepage{
	\clearpage
	\if@twoside
		\ifodd
			\c@page
		\else
			\hbox{}
			\vspace*{\fill}
			\vspace{\fill}
			\thispagestyle{empty}
			\newpage
 			\if@twocolumn
				\hbox{}
				\newpage
			\fi
		\fi
	\fi
}
\makeatother

\usepackage[utf8]{inputenc}

\usepackage[cmex10]{amsmath}
\usepackage{amsthm}
\usepackage{amsfonts}
\usepackage{amssymb}
\usepackage{thmtools}

\usepackage{cases}

\usepackage{graphicx}
\usepackage[caption=false,font=footnotesize]{subfig}
\usepackage{epstopdf}
\graphicspath{ {./pct_background/}
		       {./pct_ccphd/}
		       {./pct_frontmatter/}
		       {./pct_distSOfiltering/}
		       {./pct_distMOtracking/}
		       {./pct_centralMOtracking/} }

\usepackage{algpseudocode}
\usepackage{algorithmicx}

\usepackage{color}
\usepackage{colortbl}
\usepackage{multirow}
\usepackage{longtable}
\usepackage{enumitem}

\usepackage{hyperref}
\usepackage{cite}
\usepackage[tight,undotted,nohints]{minitoc}
\newcommand{\spminitoc}{\vspace{-25pt}
					 \singlespacing
					 \minitoc
					 \onehalfspacing
}

\usepackage[normalem]{ulem}

\theoremstyle{definition}
\newtheorem{defi}{Definition}
\newtheorem{rem}{Remark}

\theoremstyle{plain}
\newtheorem{pro}{Proposition}
\newtheorem{lem}{Lemma}
\newtheorem{thm}{Theorem}

\newcommand{\rbb}{\mathbb{R}}
\newcommand{\nbb}{\mathbb{N}}
\newcommand{\ybb}{\mathbb{Y}}
\newcommand{\bbb}{\mathbb{B}}
\newcommand{\nbsp}{\bbb}
\newcommand{\ncal}{\mathcal{N}}
\newcommand{\lb}[1]{\mathbf{#1}}
\newcommand{\lbf}[3][]{\boldsymbol{#2}\ifx\relax#1\relax\else_{#1}\fi\ifx\relax#3\relax\else\!\left( #3 \right)\fi}
\newcommand{\rfs}[4][]{\ifx\relax{\pi}\relax\else#1{\pi}\fi\ifx\relax#2\relax\else_{#2}\fi\ifx\relax#3\relax\else^{#3}\fi\ifx\relax#4\relax\else\!\left( #4 \right)\fi}
\newcommand{\lmb}[2][]{\boldsymbol{\pi}\ifx\relax#1\relax\else^{#1}\fi\ifx\relax#2\relax\else\!\left( #2 \right)\fi}
\newcommand{\clmb}[4][]{\ifx\relax\boldsymbol{\pi}\relax\else#1{\boldsymbol{\pi}}\fi\ifx\relax#2\relax\else_{#2}\fi\ifx\relax#3\relax\else^{#3}\fi\ifx\relax#4\relax\else\!\left( #4 \right)\fi}
\newcommand{\cd}[2][]{\rho\ifx\relax#1\relax\else_{#1}\fi\ifx\relax#2\relax\else\!\left( #2 \right)\fi}
\newcommand{\cdd}[3][]{\rho\ifx\relax#1\relax\else_{#1}\fi\ifx\relax#2\relax\else^{#2}\fi\ifx\relax#3\relax\else\!\left( #3 \right)\fi}
\newcommand{\ex}[2][]{r\ifx\relax#1\relax\else_{#1}\fi\ifx\relax#2\relax\else^{\left( #2 \right)}\fi}
\newcommand{\nex}[2][]{q\ifx\relax#1\relax\else_{#1}\fi\ifx\relax#2\relax\else^{\left( #2 \right)}\fi}
\newcommand{\p}[3][]{p\ifx\relax#1\relax\else_{#1}\fi\ifx\relax#2\relax\else^{\left( #2 \right)}\fi\ifx\relax#3\relax\else\!\left( #3 \right)\fi}
\newcommand{\np}[3][]{p\ifx\relax#1\relax\else_{#1}\fi\ifx\relax#2\relax\else^{#2}\fi\ifx\relax#3\relax\else\!\left( #3 \right)\fi}
\newcommand{\dli}[1]{\Delta\!\left( #1 \right)}
\newcommand{\w}[2][]{w\ifx\relax#1\relax\else_{#1}\fi\!\left( #2 \right)}
\newcommand{\lbs}[1]{\mathcal{L}\ifx\relax#1\relax\else\!\left( #1 \right)\fi}
\newcommand{\inc}[2]{1_{#1}\!\left( #2 \right)}
\newcommand{\lbb}{\mathbb{L}}
\newcommand{\lbsp}[1][]{\lbb\ifx\relax#1\relax\else_{#1}\fi}
\newcommand{\xbb}{\mathbb{X}}
\newcommand{\stsp}{\xbb}
\newcommand{\fs}[2][]{\mathcal{F}\ifx\relax#1\relax\else_{#1}\fi\!\left( #2 \right)}
\newcommand{\bfbox}[1]{\framebox[1.1\width]{\textbf{#1}}}
\newcommand{\intf}[4][]{\ifx\relax{d}\relax\else#1{d}\fi\ifx\relax#2\relax\else_{#2}\fi\ifx\relax#3\relax\else^{#3}\fi\ifx\relax#4\relax\else\!\left( #4 \right)\fi}
\newcommand{\lpdf}[4][]{\ifx\relax{s}\relax\else#1{s}\fi\ifx\relax#2\relax\else_{#2}\fi\ifx\relax#3\relax\else^{#3}\fi\ifx\relax#4\relax\else\!\left( #4 \right)\fi}
\newcommand{\be}{\begin{equation}}
\newcommand{\ee}{\end{equation}}
\newcommand{\ben}{\begin{equation*}}
\newcommand{\een}{\end{equation*}}
\newcommand{\ba}{\begin{array}}
\newcommand{\ea}{\end{array}}
\newcommand{\bie}{\begin{IEEEeqnarray}{rCl}}
\newcommand{\eie}{\end{IEEEeqnarray}}

\newcommand{\mb}{\mathbf}

\usepackage{eso-pic}
\usepackage{transparent}
\newcommand\BackgroundPic{
	\put(0,0){
		\parbox[b][\paperheight]{\paperwidth}{
		\vfill
		{\transparent{0.03}\hspace{10cm}\includegraphics[height=\paperheight,keepaspectratio]{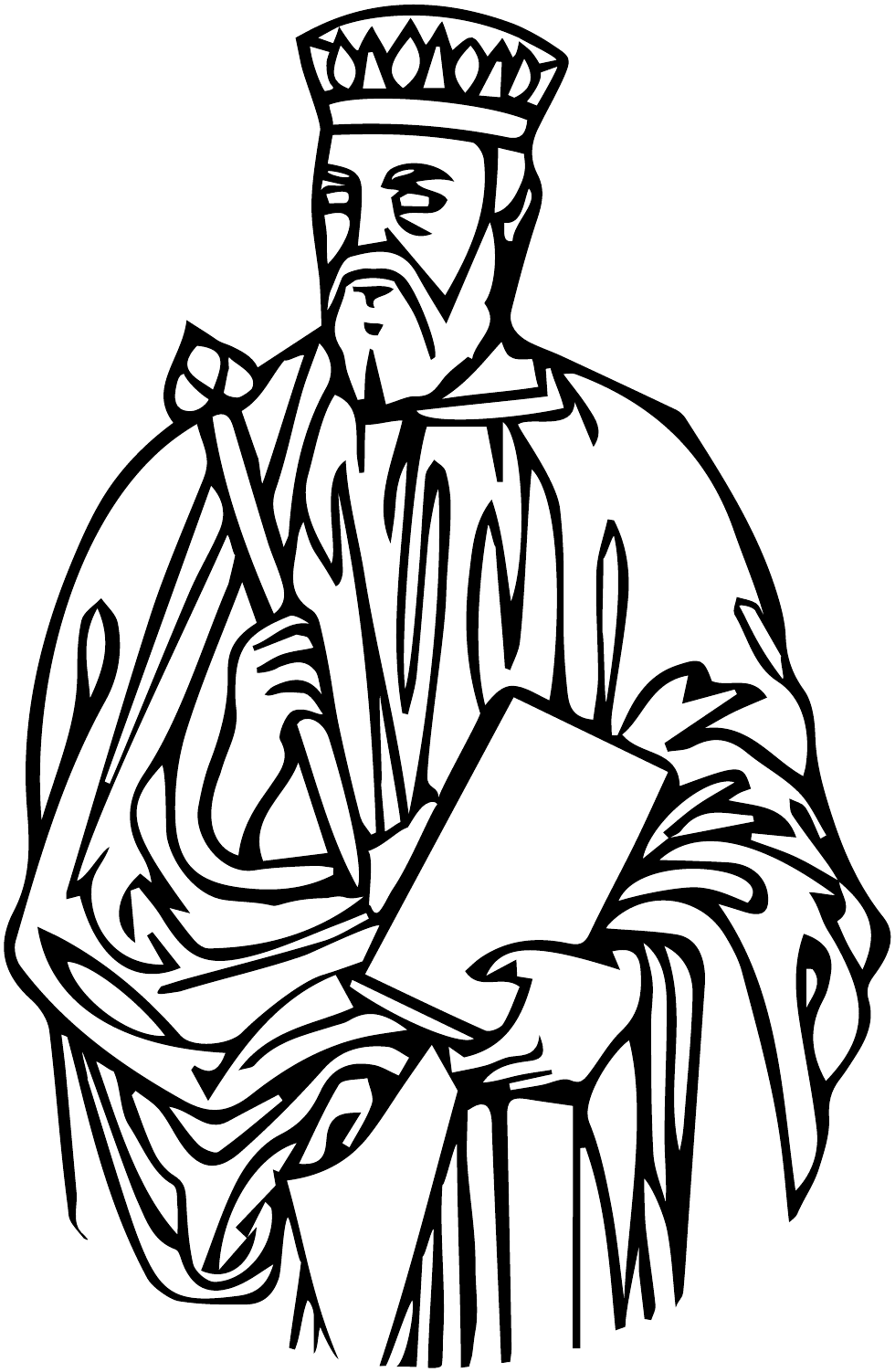}}
		\vfill
		}
	}
}

\begin{document}
\doparttoc
\dominitoc
%\dominilof
%\dominilot

% FRONTMATTER
\frontmatter
\pagestyle{empty}
\fancyfoot[C]{}

\newcommand{\HRule}{\rule{\linewidth}{0.5mm}}
\renewcommand{\arraystretch}{1.3}
\begin{titlepage}
\AddToShipoutPicture*{\BackgroundPic}
\begin{center}
	\includegraphics[width=0.25\textwidth]{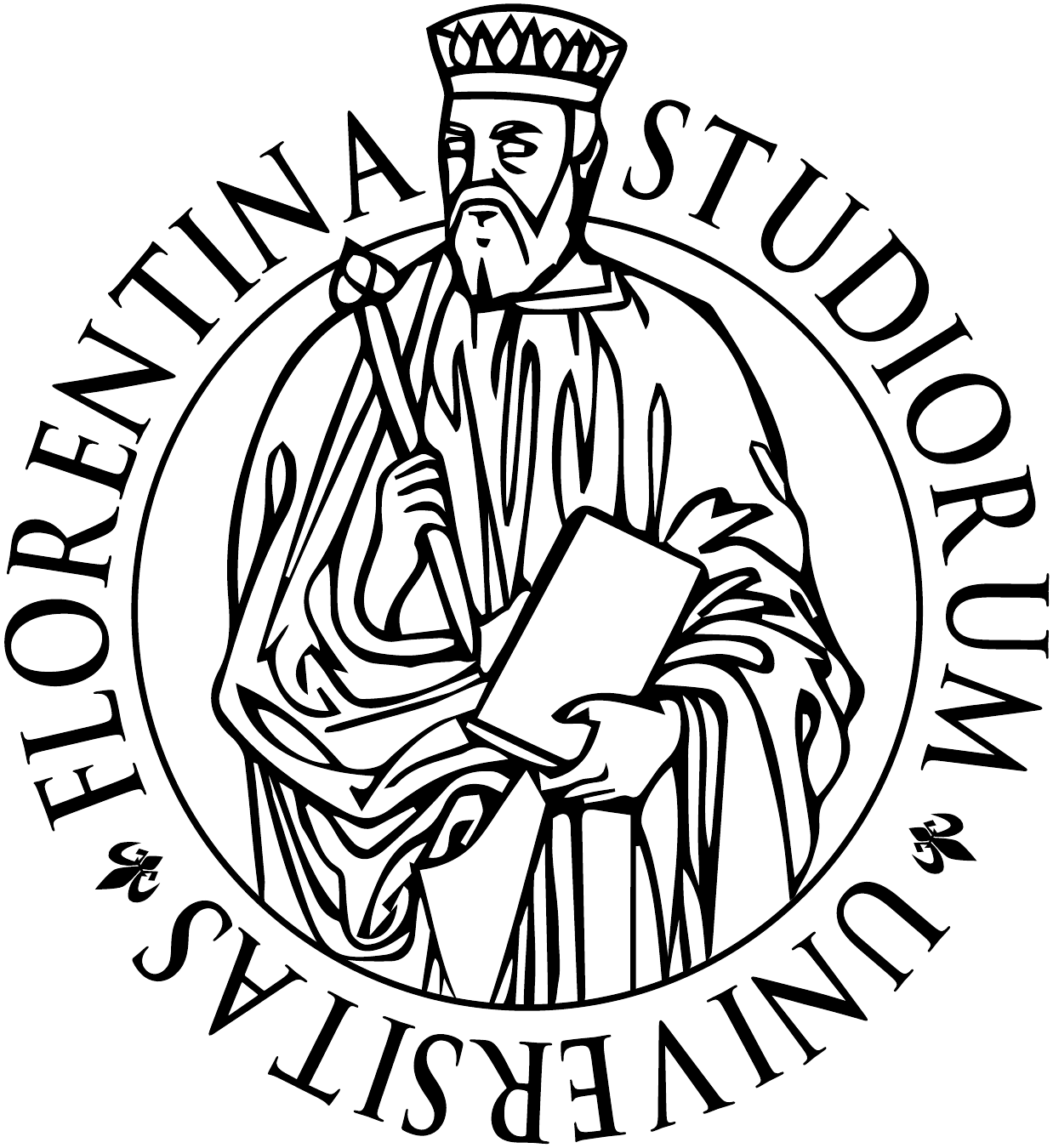}\\
	{\Large
		\textsc{Università degli Studi di Firenze}\\[1cm]
		\textsc{Dottorato in Informatica, Sistemi e Telecomunicazioni}\\\vspace{-0.25em}
		\textsc{Indirizzo: Ingegneria Informatica e dell'Automazione}\\\vspace{-0.25em}
		\textsc{Ciclo XXVII}\\\vspace{-0.25em}
		\textsc{Coordinatore: Prof. Luigi Chisci}\\[1cm]
	}
	\HRule~\\
	{\Large
		\textbf{Distributed multi-object tracking\\over sensor networks:\\a random finite set approach}\\
	}
	\HRule~\\[1cm]
	{\large
		\textsc{Dipartimento di Ingegneria dell'Informazione}\\
		\textsc{Settore Scientifico Disciplinare ING-INF/04}\\[1.5cm]
	}
\end{center}
\begin{minipage}[t][][t]{0.33\textwidth}
	\begin{flushleft}
		\textsc{Author:}\\[0.5cm]
		Claudio Fantacci
	\end{flushleft}
\end{minipage}
\begin{minipage}[t][][t]{0.33\textwidth}
	\begin{center}
		\textsc{Supervisors:}\\[0.5cm]
			Prof. Luigi Chisci\\[1cm]
			Prof. Giorgio Battistelli
	\end{center}
\end{minipage}
\begin{minipage}[t][][t]{0.33\textwidth}
	\begin{flushright}
		\textsc{Coordinator:}\\[0.5cm]
			Prof. Luigi Chisci\\[1cm]
	\end{flushright}
\end{minipage}
\vfill
\begin{center}
{\large
	Years 2012/2014
}
\end{center}
\end{titlepage}

% TOC
\mainmatter
\pagestyle{empty}
\pagenumbering{Roman}

\tableofcontents

\listoftheorems[ignoreall,show={thm,cor,pro,lem}]

\listoffigures
\mtcaddchapter

\listoftables
\mtcaddchapter

\newpage
\mtcaddchapter
\singlespacing
\chapter*{List of Acronyms}
\begin{longtable}{p{.30\textwidth} >{\raggedright\arraybackslash}p{.70\textwidth}} 
BF						& Bayes Filter\\
BMF					& Belief Mass Function\\
CBMM					& Centralized Bayesian Multiple-Model\\
CDF					& Cumulative Distribution Function\\
CI					& Covariance Intersection\\
CM$\delta$GLMB		& Consensus Marginalized $\delta$-Generalized Labelled Multi-Bernoulli\\
CGM-CPHD				& Consensus Gaussian Mixture-Cardinalized Probability Hypothesis Density\\
CLCP					& Consensus on Likelihoods and Priors\\
CLMB					& Consensus Labelled Multi-Bernoulli\\
COM					& Communication\\
CP						& Consensus on Posteriors\\
CPHD					& Cardinalized Probability Hypothesis Density\\
CT						& Coordinated Turn\\
$\delta$-GLMB			& $\delta$-Generalized Labeled Multi-Bernoulli\\
DGPB$_1$				& Distributed First Order Generalized Pseudo-Bayesian\\
DIMM					& Distributed Interacting Multiple Model\\
DIMM-CL				& Distributed Interacting Multiple Model with Consensus on Likelihoods\\
DMOF					& Distributed Multi-Object Filtering\\
DMOT					& Distributed Multi-Object Tracking\\
DOA					& Direction Of Arrival\\
DSOF					& Distributed Single-Object Filtering\\
DWNA					& Discrete White Noise Acceleration\\
EKF						& Extended Kalman Filter\\
EMD					& Exponential Mixture Density\\
FISST					& FInite Set STatistics\\
GCI						& Generalized Covariance Intersection\\
GGM-CPHD				& Global Gaussian Mixture - Cardinalized Probability Hypothesis Density\\
GLMB					& Generalized Labeled Multi-Bernoulli\\
GM						& Gaussian Mixture\\
GM-CM$\delta$GLMB	& Gaussian Mixture-Consensus Marginalized $\delta$-Generalized Labelled Multi-Bernoulli\\
GM-CLMB				& Gaussian Mixture-Consensus Labelled Multi-Bernoulli\\
GM-$\delta$GLMB		& Gaussian Mixture-$\delta$-Generalized Labelled Multi-Bernoulli\\
GM-LMB				& Gaussian Mixture-Labelled Multi-Bernoulli\\
GM-M$\delta$GLMB	& Gaussian Mixture-Marginalized $\delta$-Generalized Labelled Multi-Bernoulli\\
GPB$_{1}$				& First Order Generalized Pseudo-Bayesian\\
I.I.D. / i.i.d. / iid			& Independent and Identically Distributed\\
IMM					& Interacting Multiple Model\\
JPDA					& Joint Probabilistic Data Association\\
KF						& Kalman Filter\\
KLA					& Kullback-Leibler Average\\
KLD					& Kullback-Leibler Divergence\\
LMB					& Labeled Multi-Bernoulli\\
M$\delta$-GLMB		& Marginalized $\delta$-Generalized Labelled Multi-Bernoulli\\
MHT					& Multiple Hypotheses Tracking\\
MM						& Multiple Model\\
MOF					& Multi-Object Filtering\\
MOT					& Multi-Object Tracking\\
MSKF					& Multi-Sensor Kalman Filter\\
NCV					& Nearly-Constant Velocity\\
NWGM					& Normalized Weighted Geometric Mean\\
OSPA					& Optimal SubPattern Assignment\\
OT						& Object Tracking\\ 
PF						& Particle Filter\\
PDF					& Probability Density Function\\
PHD					& Probability Hypothesis Density\\
PMF					& Probability Mass Function\\
PRMSE					& Position Root Mean Square Error\\
RFS						& Random Finite Set	\\
SEN					& Sensor\\
SMC					& Sequential Monte Carlo\\
SOF					& Single-Object Filtering\\
TOA					& Time Of Arrival\\
UKF					& Unscented Kalman Filter\\
UT						& Unscented Transform
\end{longtable}

\onehalfspacing

\newpage
\singlespacing
\chapter*{List of Symbols}
\begin{longtable}{p{.30\textwidth} p{.70\textwidth}} 
$\mathcal{G}$														& Directed graph\\
$\mathcal{N}$														& Set of nodes\\
$\mathcal{A}$														& Set of arcs\\
$\mathcal{S}$														& Set of sensors\\
$\mathcal{C}$														& Set of communication nodes\\
$\mathcal{N}^{j}$													& Set of in-neighbours (including $j$ itself)\\
$\operatorname{col}\!\left( \, \cdot^{i} \right)_{i \in \mathcal{I}}$	& Stacking operator\\
$\operatorname{diag}\!\left( \, \cdot^{i} \right)_{i \in \mathcal{I}}$	& Square diagonal matrix operator\\
$\left< f, g\right> \triangleq \displaystyle\int f(x)\,g(x) dx$			& Inner product operator for vector valued functions\\
$\top$																& Transpose operator\\
$\mathbb{R}$														& Real number space\\
$\mathbb{R}^{n}$													& $n$ dimensional Euclidean space\\
$\mathbb{N}$														& Natural number set\\
$k$																	& Time index\\
$x_{k}$																& State vector\\
$\mathbb{X}$														& State space\\
$n_{x}$																& Dimension of the state vector\\
$f_{k}(x)$															& State transition function\\
$w_{k}$															& Process noise\\
$\varphi\!\left( k | \zeta \right)$													& Markov transition density\\
$y_{k}$																& Measurement vector\\
$\mathbb{Y}$														& Measurement space\\
$n_{y}$																& Dimension of the measurement vector\\
$h_{k}(x)$															& Measurement function\\
$v_{k}$																& Measurement noise\\
$g_{k}(y|x)$															& Likelihood function\\
$x_{1:k}$															& State history\\
$y_{1:k}$															& Measurement history\\
$p_{k}(x)$															& Posterior/Filtered probability density function\\
$p_{k|k-1}(x)$														& Prior/Predicted probability density function\\
$p_{0}(x)$															& Initial probability density function\\
$g^{i}_{k}\left( y | x \right)$														& Likelihood function of sensor $i$\\
$A_{k-1}$															& $n_{x} \times n_{x}$ state transition matrix\\
$C_{k}$																& $n_{y} \times n_{x}$ observation matrix\\
$Q_{k}$															& Process noise covariance matrix\\
$R_{k}$																& Measurement noise covariance matrix\\
$\mathcal{N}\!\left( \cdot; \, \cdot, \cdot \right)$					& Gaussian probability density function\\
$\hat{x}_{k}$														& Updated mean (state vector)\\
$P_{k}$																& Updated covariance matrix (associated to $\hat{x}_{k}$)\\
$\hat{x}_{k|k-1}$													& Predicted mean (state vector)\\
$P_{k|k-1}$															& Predicted covariance matrix (associated to $\hat{x}_{k|k-1}$)\\
$e_{k}$																& Innovation\\
$S_{k}$																& Innovation covariance matrix\\
$K_{k}$																& Kalman gain\\
$\alpha_{\sigma}$, $\beta_{\sigma}$, $\kappa_{\sigma}$							& Unscented transformation weights parameters\\
$c$, $w_{m}$, $w_{c}$, $W_{c}$											& Unscented transformation weights\\
$X$																	& Finite-valued-set\\
$\displaystyle h^{X} \triangleq \prod_{x \in X}h(x)$					& Multi-object exponential\\
$\delta _{X}(\cdot)$												& Generalized Kronecker delta\\
$1_{X}(\cdot)$														& Generalized indicator function\\
$| \cdot |$															& Cardinality operator\\
$\mathcal{F}(X)$													& Space of finite subsets of $X$\\
$\mathcal{F}_{n}(\mathbb{X})$										& Space of finite subsets of $X$ with exactly $n$ elements\\
$f(X)$, $\pi(X)$														& Multi-object densities\\
$\rho(n)$															& Cardinality probability mass function\\
$\beta(X)$															& Belief mass function\\
$d(x)$																& Probability hypothesis density function\\
$E[ \,\cdot\, ]$														& Expectation operator\\
$\bar{n}$, $D$														& Expected number of objects\\
$s(x)$																& Location density (normalized $d(x)$)\\
$r$																	& Bernoulli existence probability\\
$q \triangleq 1 - r$																	& Bernoulli non-existence probability\\
$X_{k}$																& Multi-object random finite set\\
$Y_{k}$																& Random finite set of the observations\\
$Y_{1:k}$															& Observation history random finite set\\
$Y^{i}_{k}$															& Random finite set of the observations of sensor $i$\\
$B_{k}$																& Random finite set of new-born objects\\
$P_{S,k}$															& Survival probability\\
$\mathcal{C}$, $\mathcal{C}_{k}$																& Random finite sets of clutter\\
$\mathcal{C}^{i}$, $\mathcal{C}^{i}_{k}$															& Random finite sets of clutter of sensor $i$\\
$P_{D,k}$															& Detection probability\\
$\pi_{k}(X)$															& Posterior/Filtered multi-object density\\
$\pi_{k|k-1}(X)$														& Prior/Predicted multi-object density\\
$\varphi_{k|k-1}(X|Z)$												& Markov multi-object transition density\\
$g_{k}(Y|X)$														& Multi-object likelihood function\\
$\rho_{k|k-1}(n)$													& Predicted cardinality probability mass function\\
$d_{k|k-1}(x)$														& Predicted probability hypothesis density function\\
$\rho_{k}(n)$														& Updated cardinality probability mass function\\
$d_{k}(x)$															& Updated probability hypothesis density function\\
$p_{b,k}(\cdot)$													& Birth cardinality probability mass function\\
$\rho_{S, k|k-1}\!\left( \cdot \right)$								& Cardinality probability mass function of survived objects\\
$d_{b,k}(\cdot)$													& Probability hypothesis density function of new-born objects\\
$\mathcal{G}_{k}^{0}\!\left( \cdot, \cdot, \cdot  \right)$, 
$\mathcal{G}_{Y_{k}}\!\left( \cdot \right)$							& Cardinalized probability hypothesis density function generalized likelihood functions\\
$\ell$																& Label\\
$\mathbf{X}$														& Labeled finite-valued-set\\
$\mathbf{x}$														& Labeled state vector\\
$\mathcal{L}\!\left( \mathbf{X} \right)$								& Label projection operator\\
$\Delta\!\left( \mathbf{X} \right)$									& Distinct label indicator\\
$\mathbb{L}_{k}$, $\mathbb{B}$									& Label set for objects born at time $k$\\
$\mathbb{L}_{1:k-1}$, $\mathbb{L}_{-}$							& Label set for objects up to time $k-1$\\
$\mathbb{L}_{1:k}$, $\mathbb{L}$									& Label set for objects up to time $k$\\
$\xi$																& Association history\\
$\Xi$																& Association history set\\
$\boldsymbol{f}(\mathbf{X})$, $\boldsymbol{\pi}(\mathbf{X})$	& Labeled multi-object densities\\
$w^{\left( c \right)}$, 
$w^{\left( c \right)}\!\left( \mathbf{X} \right)$						& Generalized labeled multi-Bernoulli weights indexed with $c \in \mathbb{C}$\\
$p^{\left( c \right)}$												& Generalized labeled multi-Bernoulli location probability density function indexed with $c \in \mathbb{C}$\\
$(I, \xi) \in \mathcal{F}\!\left( \mathbb{L} \right) \times \Xi$			& Labeled multi-object hypothesis\\
$w^{\left( I, \xi \right)}\!\left( \mathbf{X} \right)$					& $\delta$-Generalized labeled multi-Bernoulli weight of hypothesis $( I, \xi )$\\
$p^{\left( I \right)}$													& $\delta$-Generalized labeled multi-Bernoulli location probability density function with association history $\xi$\\
$r^{\left( \ell \right)}$												& Bernoulli existence probability associated to the object with label $\ell$\\
$q^{\left( \ell \right)}$												& Bernoulli non-existence probability associated to the object with label $\ell$\\
$\boldsymbol{\varphi}_{k|k-1}(X|Z)$								& Markov labeled multi-object transition density\\
$\left(p \oplus q \right)(x) \triangleq
\dfrac{p(x) \, q(x)}{\left< p, q \right>}$								& Information fusion operator\\
$\left( \alpha \odot p \right)(x) \triangleq
\dfrac{\left[ p(x) \right]^{\alpha}}{\left< p^{\alpha}, 1 \right>}$		& Information weighting operator\\
$\overline{p}(x)$, $\overline{q}$, $\overline{f}$, $\dots$			& Weighted Kullback-Leibler average\\
$\left( q, \Omega \right)$											& Information vector and (inverse covariance) matrix\\
$\omega^{i,j}$														& Consensus weight of node $i$ relative to $j$\\
$\omega_{l}^{i,j}$													& Consensus weight of node $i$ relative to $j$ at the consensus step $l$\\
$\Pi$																& Consensus matrix\\
$l$																	& Consensus step\\
$L$																	& Maximum number of consensus steps\\
$\rho^{i}$															& Likelihood scalar weight of node $i$\\
$b^{i}$																& Estimate of the fraction $| \mathcal{S}/\mathcal{N}|$\\
$\delta \Omega_{k}^{i} \triangleq
\left( C_{k}^{i} \right)^\top \left( R_{k}^{i} \right)^{-1} C_{k}^{i}$		& Information matrix gain\\
$\delta q_{k}^{i} \triangleq
\left( C_{k}^{i} \right)^\top \left( R_{k}^{i} \right)^{-1} y_{k}^{i}$		& Information vector gain\\
$\mathcal{P}_{c}$													& Sets of probability density functions over a continuous state space\\
$\mathcal{P}_{d}$													& Sets of probability mass functions over a discrete state space\\
$p_{jt}$															& Markov transition probability from mode $t$ to $j$\\
$\mu^{j}_{k}$														& Filtered modal probability of mode $j$\\
$\mu^{j}_{k|k-1}$													& Predicted modal probability of mode $j$\\
$\mu^{j|t}_{k}$														& Filtered modal probability of mode $j$ conditioned to mode $t$\\
$\alpha_{i}$														& Gaussian mixture weight of the component $i$\\
$N_{G}$															& Number of components of a Gaussian mixture\\
$\alpha_{ij}$														& Fused Gaussian mixture weight relative to components $i$ and $j$\\
$T_{s}$																& Sampling interval\\
$\beta$															& Fused Gaussian mixture normalizing constant\\
$p_{x}$, $p_{y}$													& Object planar position coordinates\\
$\dot{p}_{x}$, $\dot{p}_{y}$										& Object planar velocity coordinates\\
$\lambda_{c}$														& Poisson clutter rate\\
$N_{mc}$															& Number of Monte Carlo trials\\
$N_{max}$															& Maximum number of Gaussian components\\
$\gamma_{m}$														& Merging threshold\\
$\gamma_{t}$														& Truncation threshold\\
$\sigma_{TOA}$													& Standard deviation of time of arrival sensor\\
$\sigma_{DOA}$													& Standard deviation of direction of arrival sensor\\
$\boldsymbol{\pi}_{k}(\mathbf{X})$								& Posterior/Filtered labeled multi-object density\\
$\boldsymbol{\pi}_{k|k-1}(\mathbf{X})$								& Prior/Predicted labeled multi-object density\\
$\boldsymbol{f}_{B}\!\left( \mathbf{X} \right)$						& Labeled multi-object birth density\\
$w_{B}\left( \mathbf{X} \right)$										& Labeled multi-object birth weight\\
$p_{B}(x,\ell)$														& Labeled multi-object birth location probability density function\\
$r_{B}^{(\ell)}$														& Labeled multi-Bernoulli newborn object weight\\
$p^{(\ell)}_{B}(x)$													& Labeled multi-Bernoulli newborn object location probability density function\\
$w_{S}^{(\xi)}\left( \mathbf{X} \right)$								& Labeled multi-object survival object weight with association history $\xi$\\
$p_{S}(x,\ell)$														& Labeled multi-object survival object location probability density function with association history $\xi$\\
$\theta$															& New association map\\
$\Theta(I)$															& New association map set corresponding to the label subset $I$\\
$w_{k}^{\left(I,\xi\right)}$											& $\delta$-Generalized labeled multi-Bernoulli posterior/filtered weight of hypothesis $( I, \xi )$\\
$p_{k}^{(\xi)}(x,\ell)$												& $\delta$-Generalized labeled multi-Bernoulli posterior/filtered location probability density function with association history $\xi$\\
$w_{k|k-1}^{\left(I,\xi\right)}$										& $\delta$-Generalized labeled multi-Bernoulli prior/predicted weight of hypothesis $( I, \xi )$\\
$p_{k|k-1}^{(\xi)}(x,\ell)$											& $\delta$-Generalized labeled multi-Bernoulli prior/predicted location probability density function with association history $\xi$ and label $\ell$\\
$w_{k}^{\left(I\right)}$												& Marginalized $\delta$-Generalized labeled multi-Bernoulli posterior/filtered weight of the set $I$\\
$p_{k}^{(I)}(x,\ell)$													& Marginalized $\delta$-Generalized labeled multi-Bernoulli posterior/filtered location probability density function of the set $I$ and label $\ell$\\
$w_{k|k-1}^{\left(I\right)}$											& Marginalized $\delta$-Generalized labeled multi-Bernoulli prior/predicted weight of the set $I$\\
$p_{k|k-1}^{(I)}(x,\ell)$												& Marginalized $\delta$-Generalized labeled multi-Bernoulli prior/predicted location probability density function of the set $I$ and label $\ell$\\
$r_{S}^{(\ell)}$														& Labeled multi-Bernoulli survival object weight\\
$p^{(\ell)}_{S}(x)$													& Labeled multi-Bernoulli survival object location probability density function\\
$r^{(\ell)}_{k}$														& Labeled multi-Bernoulli posterior/filtered existence probability of the object with label $\ell$\\
$p^{(\ell)}_{k}(x)$													& Labeled multi-Bernoulli posterior/filtered location probability density function of the object with label $\ell$\\
$r^{(\ell)}_{k|k-1}$													& Labeled multi-Bernoulli prior/predicted existence probability of the object with label $\ell$\\
$p^{(\ell)}_{k|k-1}(x)$												& Labeled multi-Bernoulli prior/predicted location probability density function of the object with label $\ell$\\
$\left< f, g\right> \triangleq \displaystyle\int f(X)\,g(X) \delta X$	& Inner product operator for finite-valued-set functions\\
$\left< \mathbf{f}, \mathbf{g}\right> \triangleq
\displaystyle\int \mathbf{f}(\mathbf{X}) 
\,\mathbf{g}(\mathbf{X}) \delta\mathbf{X}$						& Inner product operator for labeled finite-valued-set functions
\end{longtable}

\onehalfspacing

% MAINMATTER
\mainmatter
\pagestyle{fancy}
\fancyhead[LE]{\slshape \leftmark}
\fancyhead[RE]{}
\fancyhead[LO]{}
\fancyhead[RO]{\slshape \rightmark}
\fancyfoot[C]{\thepage}
\renewcommand{\headrulewidth}{0.4pt}
\renewcommand{\footrulewidth}{0pt}
\pagenumbering{arabic}

% FORWARD
\chapter*{Acknowledgment}
\addcontentsline{toc}{chapter}{Acknowledgment}
First and foremost, I would sincerely like to thank my supervisors Prof. Luigi Chisci and Prof. Giorgio Battistelli for their constant and relentless support and guidance during my 3-year doctorate.
Their invaluable teaching and enthusiasm for research have made my Ph.D. education a very challenging yet rewarding experience.
I would also like to thank Dr. Alfonso Farina, Prof. Ba-Ngu Vo and Prof. Ba-Tuong Vo.
It has been a sincere pleasure to have the opportunity to work with such passionate, stimulating and friendly people.
Our collaboration, knowledge sharing and feedback have been of great help and truly appreciated.
Last but not least, I would like to thank all my friends and my family who supported, helped and inspired me during my studies.

% FORWARD
\chapter*{Foreword}
\addcontentsline{toc}{chapter}{Foreword}
Statistics, mathematics and computer science have always been the favourite subjects in my academic career.
The Ph.D. in automation and computer science engineering brought me to address challenging problems involving such disciplines.
In particular, multi-object filtering concerns the joint detection and estimation of an unknown and possibly time-varying number of objects, along with their dynamic states, given a sequence of observation sets.
Further, its distributed formulation also considers how to efficiently address such a problem over a heterogeneous sensor network in a fully distributed, scalable and computationally efficient way.
Distributed multi-object filtering is strongly linked with statistics and mathematics for modeling and tackling the main issues in an elegant and rigorous way, while computer science is fundamental for implementing and testing the resulting algorithms.
This topic poses significant challenges and is indeed an interesting area of research which has fascinated me during the whole Ph.D. period.

This thesis is the result of the research work carried out at the University of Florence (Florence, Italy) during the years 2012-2014, of a scientific collaboration for the biennium 2012-2013 with Selex ES (former SELEX SI, Rome, Italy) and of 6 months spent as a visiting Ph.D. scholar at the Curtin University of Technology (Perth, Australia) during the period January-July 2014.

% ABSTRACT
\chapter*{Abstract}
\addcontentsline{toc}{chapter}{Abstract}
The aim of the present dissertation is to address distributed tracking over a network of heterogeneous and geographically dispersed nodes (or agents) with sensing, communication and processing capabilities.
Tracking is carried out in the Bayesian framework and its extension to a distributed context is made possible via an information-theoretic approach to data fusion which exploits \textit{consensus} algorithms and the notion of \textit{Kullback–Leibler Average} (KLA) of the Probability Density Functions (PDFs) to be fused.

The first step toward distributed tracking considers a single moving object.
Consensus takes place in each agent for spreading information over the network so that each node can track the object.
To achieve such a goal, consensus is carried out on the local single-object posterior distribution, which is the result of local data processing, in the Bayesian setting, exploiting the last available measurement about the object.
Such an approach is called \textit{Consensus on Posteriors} (CP).
The first contribution of the present work \cite{cpcl} is an improvement to the CP algorithm, namely \textit{Parallel Consensus on Likelihoods and Priors} (CLCP).
The idea is to carry out, in parallel, a separate consensus for the novel information (likelihoods) and one for the prior information (priors). This parallel procedure is conceived to avoid underweighting the novel information during the fusion steps. The outcomes of the two consensuses are then combined to provide the fused posterior density.
Furthermore, the case of a single highly-maneuvering object is addressed.
To this end, the object is modeled as a jump Markovian system and the \textit{multiple model} (MM) filtering approach is adopted for local estimation. Thus, the consensus algorithms needs to be re-designed to cope with this new scenario.
The second contribution \cite{batchifan2014} has been to devise two novel consensus MM filters to be used for tracking a maneuvering object.
The novel consensus-based MM filters are based on the \textit{First Order Generalized Pseudo-Bayesian} (GPB$_{1}$) and \textit{Interacting Multiple Model} (IMM) filters.

The next step is in the direction of distributed estimation of multiple moving objects.
In order to model, in a rigorous and elegant way, a possibly time-varying number of objects present in a given area of interest, the \textit{Random Finite Set} (RFS) formulation is adopted since it provides the notion of \textit{probability density for multi-object states} that allows to directly extend existing tools in distributed estimation to multi-object tracking.
The multi-object Bayes filter proposed by Mahler is a theoretically grounded solution to recursive Bayesian tracking based on RFSs.
However, the multi-object Bayes recursion, unlike the single-object counterpart, is affected by combinatorial complexity and is, therefore, computationally infeasible except for very small-scale problems involving few objects and/or measurements.
For this reason, the computationally tractable \textit{Probability Hypothesis Density} (PHD) and \textit{Cardinalized PHD} (CPHD) filtering approaches will be used as a first endeavour to distributed multi-object filtering.
The third contribution \cite{ccphd} is the generalisation of the single-object KLA to the RFS framework, which is the theoretical fundamental step for developing a novel consensus algorithm based on CPHD filtering, namely the \textit{Consensus CPHD} (CCPHD).
Each tracking agent locally updates multi-object CPHD, i.e. the cardinality distribution and the PHD, exploiting the multi-object dynamics and the available local measurements, exchanges such information with communicating agents and then carries out a fusion step to combine the information from all neighboring agents.

The last theoretical step of the present dissertation is toward distributed filtering with the further requirement of unique object identities.
To this end the labeled RFS framework is adopted as it provides a tractable approach to the multi-object Bayesian recursion.
The $\delta$-GLMB filter is an exact closed-form solution to the multi-object Bayes recursion which jointly yields state and label (or trajectory) estimates in the presence of clutter, misdetections and association uncertainty.
Due to the presence of explicit data associations in the $\delta$-GLMB filter, the number of components in the posterior grows without bound in time.
The fourth contribution of this thesis is an efficient approximation of the $\delta$-GLMB filter \cite{mdglmbf}, namely \textit{Marginalized $\delta$-GLMB} (M$\delta$-GLMB), which preserves key summary statistics (i.e. both the PHD and cardinality distribution) of the full labeled posterior.
This approximation also facilitates efficient multi-sensor tracking with detection-based measurements.
Simulation results are presented to verify the proposed approach.
Finally, distributed labeled multi-object tracking over sensor networks is taken into account.
The last contribution \cite{fanvovo2015} is a further generalization of the KLA to the labeled RFS framework, which enables the development of two novel consensus tracking filters, namely the \textit{Consensus Marginalized $\delta$-Generalized Labeled Multi-Bernoulli} (CM-$\delta$GLMB) and the \textit{Consensus Labeled Multi-Bernoulli} (CLMB) tracking filters.
The proposed algorithms provide a fully distributed, scalable and computationally efficient solution for multi-object tracking.

Simulation experiments on challenging single-object or multi-object tracking scenarios confirm the effectiveness of the proposed contributions.

% INTRODUCTION
\chapter{Introduction}
\label{chap:intro}
Recent advances in wireless sensor technology have led to the development of large networks consisting of radio-interconnected nodes (or agents) with sensing, communication and processing capabilities.
Such a net-centric technology enables the building of a more complete picture of the environment, by combining information from individual nodes (usually with limited observability) in a way that is \textit{scalable} (w.r.t. the number of nodes), \textit{flexible} and \textit{reliable} (i.e. \textit{robust} to failures). Getting these benefits calls for architectures in which individual agents can operate without knowledge of the information flow in the network.
Thus, taking into account the above-mentioned considerations, \textit{Multi-Object Tracking} (MOT) in sensor networks requires redesigning the architecture and algorithms to address the following issues:
\begin{itemize}
	\item lack of a central fusion node;
	\item scalable processing with respect to the network size;
	\item each node operates without knowledge of the network topology;
	\item each node operates without knowledge of the dependence between its own information and the information received from other nodes.
\end{itemize}

To combine limited information (usually due to low observability) from individual nodes, a suitable \textit{information fusion} procedure is required to reconstruct, from the node information, the state of the objects present in the surrounding environment. {T}he scalability requirement, the lack of a fusion center and knowledge on the network topology call for the adoption of a \textit{consensus} approach to achieve a collective fusion over the network by iterating local fusion steps among neighboring nodes \cite{Olfati,Xiao,Calafiore,cp}. In addition, due to the possible data incest problem in the presence of network loops that can causes \textit{double counting} of information, robust (but suboptimal) fusion rules, such as the \textit{Chernoff fusion} rule \cite{info,mori1} (that includes \textit{Covariance Intersection} (CI) \cite{juluhl1997,julier2008} and its generalization \cite{mah2000}) are required.

The focus of the present dissertation is \textit{distributed estimation}, from the single object to the more challenging multiple object case.

In the context of \textit{Distribute Single-Object Filtering} (DSOF), standard or Extended or Unscented Kalman filters are adopted as local estimators, the consensus involves a single Gaussian component per node, characterized by either the estimate-covariance or the information pair.
Whenever multiple models are adopted for better describing the motion of the object in the tracking scenario, multiple Gaussian components per node arise and consensus has to be extended to this multicomponent setting.
Clearly the presence of different Gaussian components related to different motion models of the same object or to different objects imply different issues and corresponding solution approaches that will be separately addressed.
In this single-object setting, the main contributions in the present work are:
\begin{enumerate}[label=\textsc{\roman*}.]
	\item the development of a novel consensus algorithm, namely \textit{Parallel Consensus on Likelihoods and Priors} (CLCP), that carries out, in parallel, a separate consensus for the novel information (likelihoods) and one for the prior information (priors);
	\item two novel consensus MM filters to be used for tracking a maneuvering object, namely \textit{Distributed First Order Generalized Pseudo-Bayesian} (DGPB$_{1}$) and \textit{Distributed Interacting Multiple Model} (DIMM) filters.
\end{enumerate}

Furthermore, \textit{Distribute Multi-Object Filtering} (DMOF) is taken into account.
To model a possibly time-varying number of objects present in a given area of interest in the presence of detection uncertainty and clutter, the \textit{Random Finite Set} (RFS) approach is adopted.
The RFS formulation provides the useful concept of \textit{probability density} for \textit{multi-object states} that allows to directly extend existing tools in distributed estimation to multi-object tracking.
Such a concept is not available in the MHT and JPDA approaches \cite{reid,far1985v1,far1985v2,book0,book,BlPo}.
However, the multi-object Bayes recursion, unlike the single-object counterpart, is affected by combinatorial complexity and is, therefore, computationally infeasible except for very small-scale problems involving very few objects and/or measurements.
For this reason, the computationally tractable \textit{Probability Hypothesis Density} (PHD) and \textit{Cardinalized PHD} (CPHD) filtering approaches will be used to address DMOF.
It is recalled that the CPHD filter propagates in time the discrete distribution of the number of objects, called \textit{cardinality distribution}, and the spatial distribution in the state space of such objects, represented by the PHD (or intensity function).
It is worth to point out that there have been several interesting contributions \cite{mah2000,chmoch1990,cljumhri2010,unjuclri2010,unclju2011} on multi-object fusion.
More specifically, \cite{chmoch1990} addressed the problem of optimal fusion in the case of known correlations while \cite{mah2000,cljumhri2010,unjuclri2010,unclju2011} concentrated on robust fusion for the practically more relevant case of unknown correlations. In particular, \cite{mah2000} first generalized CI in the context of multi-object fusion.
Subsequently, \cite{cljumhri2010} specialized the \textit{Generalized Covariance Intersection} (GCI) of \cite{mah2000} to specific forms of the multi-object densities providing, in particular, GCI fusion of cardinality distributions and PHD functions.
In \cite{unjuclri2010}, a Monte Carlo (particle) realization is proposed for the GCI fusion of PHD functions.
The two key contributions in this thesis work are:
\begin{enumerate}[label=\textsc{\roman*}.]
	\item the generalisation of the single-object KLA to the RFS framework;
	\item a novel \textit{consensus CPHD} (CCPHD) filter, based on a Gaussian Mixture (GM) implementation.
\end{enumerate}

\textit{Multi-object tracking} (MOT) involves the on-line estimation of an unknown and time-varying number of objects and their individual trajectories from sensor data \cite{BlPo, book0, mahler}.
The key challenges in multi-object tracking also include \textit{data association uncertainty}.
Numerous multi-object tracking algorithms have been developed in the literature and most of these fall under the three major paradigms of: \textit{Multiple Hypotheses Tracking} (MHT) \cite{reid,BlPo}, \textit{Joint Probabilistic Data Association} (JPDA) \cite{book0} and \textit{Random Finite Set} (RFS) filtering \cite{mahler}.
The proposed solutions are based on the recently introduced concept of labeled RFS that enables the estimation of multi-object trajectories in a principled manner \cite{vovo1}.
In addition, labeled RFS-based trackers do not suffer from the so-called ``\textit{spooky effect}'' \cite{spooky} that degrades performance in the presence of low detection probability like in the multi-object filters \cite{vo-vo-cantoni,ccphd,emd}.
Labeled RFS conjugate priors \cite{vovo1} have led to the development of a tractable analytic multi-object tracking solution called the \textit{$\delta$-Generalized Labeled Multi-Bernoulli} ($\delta$-GLMB) filter \cite{vovo2}.
The computational complexity of the $\delta$-GLMB filter is mainly due to the presence of explicit data associations.
For certain applications such as tracking with multiple sensors, partially observable measurements or decentralized estimation, the application of a $\delta$-GLMB filter may not be possible due to limited computational resources.
Thus, cheaper approximations to the $\delta$-GLMB filter are of practical significance in MOT.
Core contribution of the present work is a new approximation of the $\delta$-GLMB filter.
The result is based on the approximation proposed in \cite{papi2014} where it was shown that the more general \textit{Generalized Labeled Multi-Bernoulli} (GLMB) distribution can be used to construct a principled approximation of an arbitrary labeled RFS density that matches the PHD and the cardinality distribution.
The resulting filter is referred to as \textit{Marginalized $\delta$-GLMB} (M$\delta$-GLMB) since it can be interpreted as a marginalization over the data associations.
The proposed filter is, therefore, computationally cheaper than the $\delta$-GLMB filter while preserving key summary statistics of the multi-object posterior.
Importantly, the M$\delta$-GLMB filter facilitates tractable multi-sensor multi-object tracking.
Unlike PHD/CPHD and multi-Bernoulli based filters, the proposed approximation accommodates statistical dependence between objects.
An alternative derivation of the \textit{Labeled Multi-Bernoulli} (LMB) filter \cite{lmbf} based on the newly proposed M$\delta$-GLMB filter is presented.

Finally, \textit{Distributed MOT} (DMOT) is taken into account.
The proposed solutions are based on the above-mentioned labeled RFS framework that has led to the development of the $\delta$-GLMB tracking filter \cite{vovo2}.
However, it is not known if this filter is amenable to DMOT.
Nonetheless, the M$\delta$-GLMB and the LMB filters are two efficient approximations of the $\delta$-GLMB filter that
\begin{itemize}
	\item have an appealing mathematical formulation that facilitates an efficient and tractable closed-form fusion rule for DMOT;
	\item preserve key summary statistics of the full multi-object posterior.
\end{itemize}
In this setting, the main contributions in the present work are:
\begin{enumerate}[label=\textsc{\roman*}.]
	\item the development of the first distributed multi-object tracking algorithms based on the labeled RFS framework, generalizing the approach of \cite{ccphd} from moment-based filtering to tracking with labels;
	\item the development of \textit{Consensus Marginalized $\delta$-Generalized Labeled Multi-Bernoulli} (CM$\delta$-GLMB) and \textit{Consensus Labeled Multi-Bernoulli} (CLMB) tracking filters.
\end{enumerate}

Simulation experiments on challenging tracking scenarios confirm the effectiveness of the proposed contributions.

~\newline\noindent The rest of the thesis is organized as follows.
\subsection*{Chapter \ref{chap:back} - Background}
This chapter introduces notation, provides the necessary background on recursive Bayesian estimation, Random Finite Sets (RFSs), Bayesian multi-object filtering, distributed estimation and the network model.

\subsection*{Chapter \ref{chap:dsof} - Distributed single-object filtering}
This chapter provides novel contributions on distributed nonlinear filtering with applications to nonlinear single-object tracking.
In particular: \textsc{i}) a \textit{Parallel Consensus on Likelihoods and Priors} (CLCP) filter is proposed to improve performance with respect to existing consensus approaches for distributed nonlinear estimation; \textsc{ii}) a consensus-based multiple model filter for jump Markovian systems is presented and applied to tracking of a highly-maneuvering objcet.

\subsection*{Chapter \ref{chap:dmof} - Distributed multi-object filtering}
This chapter introduces consensus multi-object information fusion according to an information-theoretic interpretation in terms of Kullback-Leibler averaging of multi-object distributions.
Moreover, the \textit{Consensus Cardinalized Probability Hypothesis Density} (CCPHD) filter is presented and its performance is evaluated via simulation experiments.

\subsection*{Chapter \ref{chap:mot} - Centralized multi-object tracking}
In this chapter, two possible approximations of the \textit{$\delta$-Generalized Labeled Multi-Bernoulli} ($\delta$-GLMB) density are presented, namely \textsc{i}) the \textit{Marginalized $\delta$-Generalized Labeled Multi-Bernoulli} (M$\delta$-GLMB) and \textsc{ii}) the \textit{Labeled Multi-Bernoulli} (LMB).
Such densities will allow to develop a new centralized tracker and to establish a new theoretical connection to previous work proposed in the literature.
Performance of the new centralized tracker is evaluated via simulation experiments.

\subsection*{Chapter \ref{chap:dmot} - Distributed multi-object tracking}
This chapter introduces the information fusion rules for \textit{Marginalized $\delta$-Generalized Labeled Multi-Bernoulli} (M$\delta$-GLMB) and \textit{Labeled Multi-Bernoulli} (LMB) densities.
An information-theoretic interpretation of such fusions, in terms of Kullback-Leibler averaging of labeled multi-object densities, is also established.
Furthermore, the \textit{Consensus M$\delta$-GLMB} (CM$\delta$-GLMB) and \textit{Consensus LMB} (CLMB) tracking filters are presented as two new labeled distributed multi-object trackers.
Finally, the effectiveness of the proposed trackers is discussed via simulation experiments on realistic distributed multi-object tracking scenarios.

\subsection*{Chapter \ref{chap:conclusion} - Conclusions and future work}
The thesis ends with concluding remarks and perspectives for future work.

% BACKGROUND
\chapter{Background}
\label{chap:back}
\spminitoc
\section{Network model}
\label{sec:net}
Recent advances in wireless sensor technology has led to the development of large networks consisting of radio-interconnected nodes (or agents) with sensing, communication and processing capabilities. Such a net-centric technology enables the building of a more complete picture of the environment, by combining information from individual nodes (usually with limited observability) in a way that is scalable (w.r.t. the number of nodes), flexible and reliable (i.e. robust to failures). Getting these benefits calls for architectures in which individual agents can operate without knowledge of the information flow in the network. Thus, taking into account the above-mentioned considerations, \textit{Object Tracking} (OT) in sensor networks requires redesigning the architecture and algorithms to address the following issues:
\begin{itemize}
	\item lack of a central fusion node;
	\item scalable processing with respect to the network size;
	\item each node operates without knowledge of the network topology;
	\item each node operates without knowledge of the dependence between its own information and the information received from other nodes.
\end{itemize}

The network considered in this work (depicted in Fig. \ref{fig:netmodelSC}) consists of two types of heterogeneous and geographically dispersed nodes (or agents): \textit{communication} (COM) nodes have only processing and communication capabilities, i.e. they can process local data as well as exchange data with the neighboring nodes, while \textit{sensor} (SEN) nodes have also sensing capabilities, i.e. they can sense data from the environment.
Notice that, since COM nodes do not provide any additional information, their presence is needed only to improve network connectivity.

\begin{figure}[h!]
	\centering
	\includegraphics[width=\textwidth]{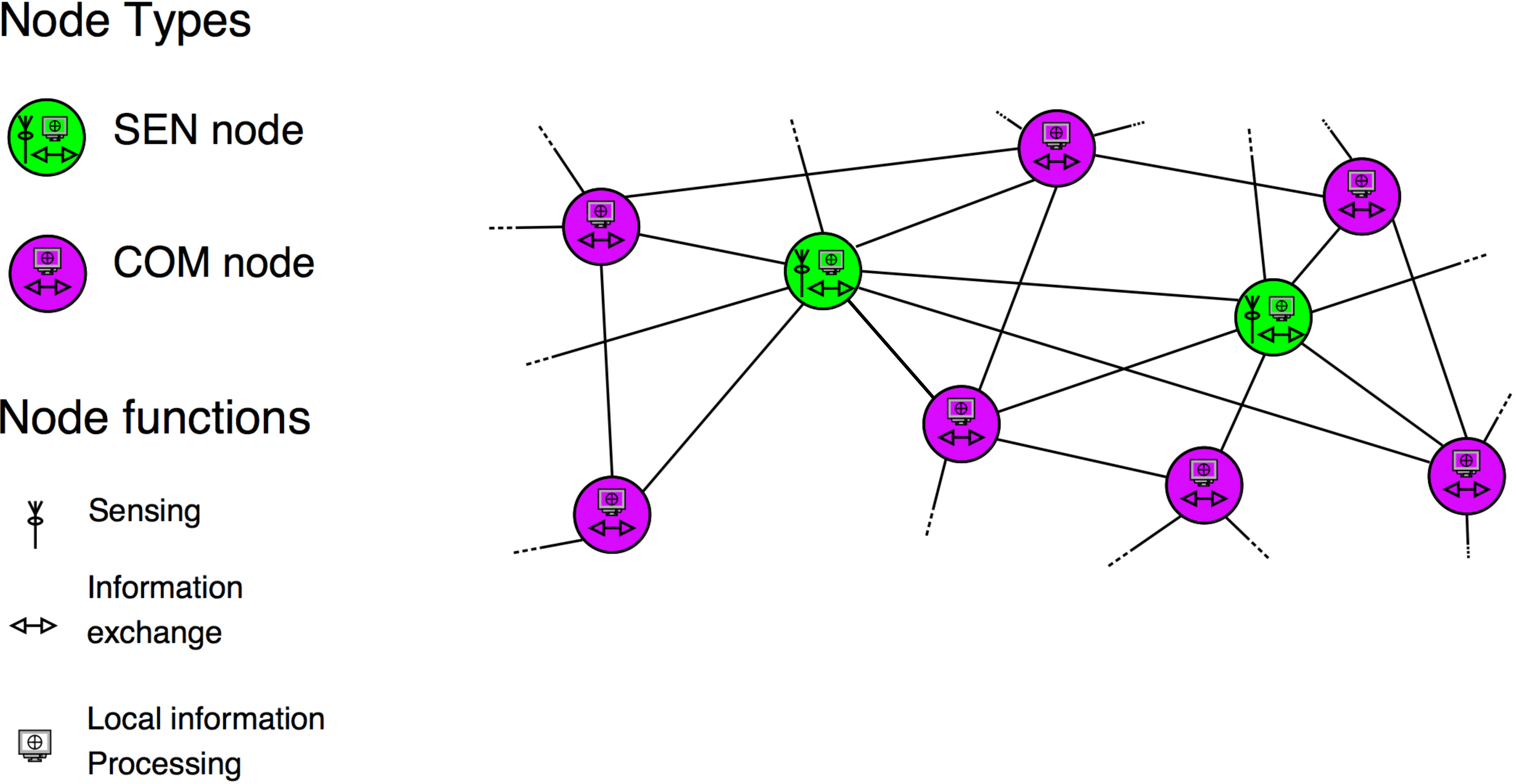}
	\caption{Network model}
	\label{fig:netmodelSC}
\end{figure}

From a mathematical viewpoint, the network is described by a directed graph $\mathcal{G} = \left( \mathcal{N},\mathcal{A}\right)$ where $\mathcal{N} = \mathcal{S} \cup \mathcal{C}$ is the set of nodes, $\mathcal{S}$ is the set of sensor and $\mathcal{C}$ the set of communication nodes, and $\mathcal{A} \subseteq \mathcal{N}\times \mathcal{N}$ is the set of arcs, representing links (or connections). In particular, $(i,j)\in \mathcal{A}$ if node $j$ can receive data from node $i$.
For each node $j\in \mathcal{N}$, $\mathcal{N}^{j}\triangleq \left\{ i\in \mathcal{N}:(i,j)\in \mathcal{A}\right\} $ denotes the set of in-neighbours (including $j$ itself), i.e. the set of nodes from which node $j$ can receive data.

Each node performs local computation, exchanges data with the neighbors and gathers measurements of kinematic variables (e.g., angles, distances, Doppler shifts, etc.) relative to objects present in the surrounding environment (or surveillance area). 
The focus of this thesis will be the development of networked estimation algorithms that are scalable with respect to network size, and to allow each node to operate without knowledge of the dependence between its own information and the information from other nodes.

\section{Recursive Bayesian estimation}
\label{sec:bayesest}
The main interest of the present dissertation is \textit{estimation}, which refers to inferring the values of a set of unknown variables from information provided by a set of noisy measurements whose values depend on such unknown variables.
Estimation theory dates back to the work of Gauss \cite{gauss} on determining the orbit of celestial bodies from their observations.
These studies led to the technique known as \textit{Least Squares}.
Over centuries, many other techniques have been proposed in the field of estimation theory\cite{fis1912,kol1950,stra1960,vantree2004,jaz2007,andmoo2012}, e.g., the \textit{Maximum Likelihood}, the \textit{Maximum a Posteriori} and the \textit{Minimum Mean Square Error} estimation.
The Bayesian approach models the quantities to be estimated as random variables characterized by \textit{Probability Density Functions} (PDFs), and provides an improved estimation of such quantities by conditioning the PDFs on the available noisy measurements.
Hereinafter, we refer to the Bayesian approach as to \textit{recursive Bayesian estimation} (or \textit{Bayesian filtering}), a renowned and well-established probabilistic approach for recursively propagating, in a principled way via a two-step procedure, a PDF of a given time-dependent variable of interest.
The first key concept of the present work is, indeed, Bayesian filtering.
The propagated PDF will be used to describe, in a probabilistic way, the behaviour of a moving object.
In the following, a summary of the \textit{Bayes Filter} (BF) is given, as well as a review of a well known closed-form solution of it, the \textit{Kalman Filter} (KF) \cite{kal1960,kalbuc1961} obtained in the linear Gaussian case.

\subsection{Notation}
\label{ssec:bayesnotation}
The following notation is adopted throughout the thesis:
$\operatorname{col}\!\left( \, \cdot^{i} \right)_{i \in \mathcal{I}}$, where $I$ is a finite set, denotes the vector/matrix obtained by stacking the arguments on top of each other;
$\operatorname{diag}\!\left( \, \cdot^{i} \right)_{i \in \mathcal{I}}$, where $I$ is a finite set, denotes the square diagonal matrix obtained by placing the arguments in the $(i, i)$-th position of the main diagonal; the standard inner product notation is denoted as
\be
	\left\langle f,g\right\rangle \triangleq \int f(x)\,g(x)dx \, ;
\label{eq:innerproduct}
\ee
vectors are represented by lowercase letters, e.g. $x$, $\mathbf{x}$; spaces are represented by blackboard bold letters e.g. $\xbb$, $\ybb$, $\lbsp$, etc.
The superscript $\top$ stems for the transpose operator.

\subsection{Bayes filter}
\label{ssec:bayes}
Consider a discrete-time state-space representation for modelling a dynamical system.
At each time $k \in \nbb$, such a system is characterized by a \textit{state vector} $x_{k} \in \xbb \subseteq \rbb^{n_{x}}$, where $n_{x}$ is the dimension of the state vector.
The state evolves according to the following discrete-time stochastic model:
\be
	x_{k} = f_{k-1}\!\left( x_{k-1}, w_{k-1} \right) \, ,
\label{eq:dtmodel}
\ee
where $f_{k-1}$ is a, possibly nonlinear, function; $w_{k-1}$ is the process noise modeling uncertainties and disturbances in the object motion model. The time evolution (\ref{eq:dtmodel}) is equivalently represented by a Markov transition density
\be
	\varphi_{k|k-1}\!\left( x | \zeta \right) \, ,
\label{eq:mtmodel}
\ee
which is the PDF associated to the transition from the state $\zeta = x_{k-1}$ to the new state $x = x_{k}$.

Likewise, at each time $k$, the dynamical system described with state vector $x_{k}$ can be observed via a noisy \textit{measurement vector} $y_{k} \in \ybb \subseteq \rbb^{n_{y}}$, where $n_{y}$ is the dimension of the observation vector.
The measurement process can be modelled by the measurement equation
\be
	y_{k} = h_{k}\!\left( x_{k}, v_{k} \right) \, ,
\label{eq:mmodel}
\ee
which provides an indirect observation of the state $x_{k}$ affected by the measurement noise $v_{k}$. The modeling of the measurement vector is equivalently represented by the \textit{likelihood function}
\be
	g_{k}\!\left( y | x \right) \, ,
\label{eq:like}
\ee
which is the PDF associated to the generation of the measurement  vector $y = y_{k}$ from the dynamical system with state $x = x_{k}$.

The aim of recursive state estimation (or filtering) is to sequentially estimate over time $x_{k}$ given the measurement history $y_{1:k} \triangleq \left\{ y_{1}, \dots, y_{k} \right\}$. It is assumed that the PDF associated to $y_{1:k}$ given the state history $x_{1:k} \triangleq \left\{ x_{1}, \dots, x_{k} \right\}$ is
\be
	g_{1:k}\!\left( y_{1:k} | x_{1:k} \right) = \prod_{\kappa = 1}^{k} g_{\kappa}(y_{\kappa} | x_{\kappa}) \, ,
\label{eq:indlike}
\ee
i.e. \textit{the measurements $y_{1:k}$ are conditionally independent on the states $x_{1:k}$}.
In the Bayesian framework, the entity of interest is the \textit{posterior density} $\p[k]{}{x}$ that contains all the information about the state vector $x_{k}$ given all the measurements up to time $k$. Such a PDF can be recursively propagated in time resorting to the well know Chapman-Kolmogorov equation and the Bayes' rule \cite{holee1964}
\bie
	\p[k|k-1]{}{x} & = & \int \varphi_{k|k-1}\!\left( x | \zeta \right) \p[k-1]{}{\zeta} d \zeta \, ,\label{eq:chapkol}\\
	\p[k]{}{x} & = & \dfrac{g_{k}\!\left( y_{k} | x \right) \, \p[k|k-1]{}{x}}{\displaystyle \int g_{k}\!\left( y_{k} | \zeta \right) \, \p[k|k-1]{}{\zeta} d \zeta} \, ,\label{eq:bayesrule}
\eie
given an \textit{initial density} $\p[0]{}{\cdot}$.
The PDF $\p[k|k-1]{}{\cdot}$ is referred to as the \textit{predicted density}, while $\p[k]{}{\cdot}$ is the \textit{filtered density}.

Let us consider a multi-sensor \textit{centralized} setting in which a sensor network $\left( \mathcal{N},\mathcal{A}\right)$ conveys all the measurements to a central fusion node. Assuming that the measurements taken by the sensors are independent, the Bayesian filtering recursion can be naturally extended as follows:
\bie
	\p[k|k-1]{}{x} & = & \int \varphi_{k|k-1}\!\left( x | \zeta \right) \p[k-1]{}{\zeta} d \zeta \, ,\label{eq:mschapkol}\\
	\p[k]{}{x} & = & \dfrac{\displaystyle \prod_{i \in \mathcal{N}} g^{i}_{k}\!\left( y^{i}_{k} | x \right) \, \p[k|k-1]{}{x}}{\displaystyle \int \prod_{i \in \mathcal{N}} g^{i}_{k}\left( y^{i}_{k} | \zeta \right) \, \p[k|k-1]{}{\zeta} d \zeta} \, .\label{eq:msbayesrule}
\eie

\subsection{Kalman Filter}
\label{ssec:kalman}
The KF \cite{kal1960,kalbuc1961,holee1964} is a closed-form solution of (\ref{eq:chapkol})-(\ref{eq:bayesrule}) in the linear Gaussian case. That is, suppose that (\ref{eq:dtmodel}) and (\ref{eq:mmodel}) are linear transformations of the state with additive Gaussian white noise, i.e.
\bie
	x_{k} & = & A_{k-1} x_{k-1} + w_{k-1} \, ,\label{eq:lindtmodel}\\
	y_{k} & = & C_{k} x_{k} + v_{k} \, ,\label{eq:linmmodel}
\eie
where $A_{k-1}$ is the $n_{x} \times n_{x}$ state transition matrix, $C_{k}$ is the $n_{y} \times n_{x}$ observation matrix, $w_{k-1}$ and $v_{k}$ are mutually independent zero-mean white Gaussian noises with covariances $Q_{k-1}$ and $R_{k}$, respectively.
Thus, the Markov transition density and the likelihood functions are
\bie
	\varphi_{k|k-1}\!\left( x | \zeta \right) & = & \ncal\!\left( x; \, A_{k-1} \zeta, Q_{k-1} \right) \, ,\label{eq:linmtmodel}\\
	g_{k}\!\left( y | x \right) & = & \ncal\!\left( y; \, C_{k} x, R_{k} \right) \, ,\label{eq:linlike}
\eie
where
\be
	\ncal\!\left( x; \, m, P \right) \triangleq \left| 2 \pi P \right|^{-\frac{1}{2}} e^{-\frac{1}{2}\left( x - m \right)^{\top} P^{-1} \left( x - m \right)} \,
\label{eq:gausspdf}
\ee
is a Gaussian PDF.
Finally, suppose that the prior density
\be
	\p[k-1]{}{x} = \ncal\!\left( x; \, \hat{x}_{k-1}, P_{k-1} \right)
\label{eq:linprior}
\ee
is Gaussian with mean $\hat{x}_{k-1}$ and covariance $P_{k-1}$.
Solving (\ref{eq:chapkol}), the predicted density turns out to be
\be
	\p[k|k-1]{}{x} = \ncal\!\left( x; \, \hat{x}_{k|k-1}, P_{k|k-1} \right) \,,
\label{eq:linprediction}
\ee
a Gaussian PDF with mean $\hat{x}_{k|k-1}$ and covariance $P_{k|k-1}$.
Moreover, solving (\ref{eq:bayesrule}), the posterior density (or \textit{updated density}), turns out to be
\be
	\p[k]{}{x} = \ncal\!\left( x; \, \hat{x}_{k}, P_{k} \right) \, ,
\label{eq:linupdate}
\ee
i.e. a Gaussian PDF with mean $\hat{x}_{k}$ and covariance $P_{k}$.
\begin{rem}
	If the posterior distributions are in the same family as the prior probability distribution, the prior and posterior are called \textit{conjugate distributions}, and the prior is called a \textit{conjugate prior} for the likelihood function. The Gaussian distribution is a conjugate prior.
\end{rem}
The KF recursion for computing both predicted and updated pairs $\left( \hat{x}_{k-1}, P_{k-1} \right)$ and $\left( \hat{x}_{k}, P_{k} \right)$ is reported in Table \ref{alg:kf}.
\begin{table}[!h]
	\caption{The Kalman Filter (KF)}
	\label{alg:kf}
	\hrulefill\hrule
	\begin{algorithmic}[0]
		\For{$k = 1, 2, \dots$}\vspace{0.5em}
			\State \bfbox{\textsc{Prediction}}
			\State $\hat{x}_{k|k-1} = A_{k-1} \hat{x}_{k-1}$\Comment{Predicted mean}
			\State $P_{k|k-1} = A_{k-1} P_{k-1} A^{\top}_{k-1} + Q_{k-1}$\Comment{Predicted covariance matrix}\vspace{0.5em}
			\State \bfbox{\textsc{Correction}}
			\State $e_{k} = y_{k} - C_{k}\hat x_{k|k-1}$\Comment{Innovation}
			\State $S_{k} = R_{k} + C_{k}P_{k|k-1}C_{k}^{\top}$\Comment{Innovation covariance matrix}
			\State $K_{k} = P_{k|k-1}C_{k}^{\top}S_{k}^{-1}$\Comment{Kalman gain}
			\State $\hat x_{k} = \hat x_{k|k-1} + K_{k}e_{k}$\Comment{Updated mean}
			\State $P_{k} = P_{k|k-1} - K_{k}S_{k}K_{k}^{\top}$\Comment{Updated covariance matrix}\vspace{0.5em}
		\EndFor
	\end{algorithmic}
	\hrule\hrulefill
\end{table}

In the centralized setting the network $\left( \mathcal{N},\mathcal{A}\right)$ conveys all the measurements
\bie
	y^{i}_{k} & = & C^{i}_{k} x_{k} + v^{i}_{k} \, ,\\
	v^{i}_{k} & \sim & \ncal\!\left( 0, R^{i}_{k} \right) \, ,
\eie
$i \in \ncal$, to a fusion center in order to evaluate (\ref{eq:msbayesrule}).
The result amounts to stack all the information from all nodes $i \in \ncal$ as follows
\bie
	y_{k} & = & \operatorname{col}\!\left( y^{i}_{k} \right)_{i \in \ncal}\\
	C_{k} & = & \operatorname{col}\!\left( C^{i}_{k} \right)_{i \in \ncal}\\
	R_{k} & = & \operatorname{diag}\!\left( R^{i}_{k} \right)_{i \in \ncal}
\eie
and then to perform the same steps of the KF. A summary of the \textit{Multi-Sensor KF} (MSKF) is reported in Table \ref{alg:mskf}.
\begin{table}[!h]
	\caption{The Multi-Sensor Kalman Filter (MSKF)}
	\label{alg:mskf}
	\hrulefill\hrule
	\begin{algorithmic}[0]
		\For{$k = 1, 2, \dots$}\vspace{0.5em}
			\State \bfbox{\textsc{Prediction}}
			\State $\hat{x}_{k|k-1} = A_{k-1} \hat{x}_{k-1}$\Comment{Predicted mean}
			\State $P_{k|k-1} = A_{k-1} P_{k-1} A^{\top}_{k-1} + Q_{k-1}$\Comment{Predicted covariance matrix}\vspace{0.5em}
			\State \bfbox{\textsc{Stacking}}
			\State $y_{k} = \operatorname{col}\!\left( y^{i}_{k} \right)_{i \in \ncal}$
			\State $C_{k} = \operatorname{col}\!\left( C^{i}_{k} \right)_{i \in \ncal}$
			\State $R_{k} = \operatorname{diag}\!\left( R^{i}_{k} \right)_{i \in \ncal}$\vspace{0.5em}
			\State \bfbox{\textsc{Correction}}
			\State $e_{k} = y_{k} - C_{k}\hat x_{k|k-1}$\Comment{Innovation}
			\State $S_{k} = R_{k} + C_{k}P_{k|k-1}C_{k}^{\top}$\Comment{Innovation covariance matrix}
			\State $K_{k} = P_{k|k-1}C_{k}^{\top}S_{k}^{-1}$\Comment{Kalman gain}
			\State $\hat x_{k} = \hat x_{k|k-1} + K_{k}e_{k}$\Comment{Updated mean}
			\State $P_{k} = P_{k|k-1} - K_{k}S_{k}K_{k}^{\top}$\Comment{Updated covariance matrix}\vspace{0.5em}
		\EndFor
	\end{algorithmic}
	\hrule\hrulefill
\end{table}

The KF has the advantage of being Bayesian optimal, but is not directly applicable to nonlinear state-space models.
Two well known approximations have proven to be effective in situations where one or both the equations (\ref{eq:dtmodel}) and (\ref{eq:mmodel}) are nonlinear: $\textsc{i})$ Extended KF (EKF) \cite{EKF} and $\textsc{ii})$ Unscented KF (EKF) \cite{juluhl1997}.
The EKF is a first order approximation of the Kalman filter based on local linearization.
The UKF uses the sampling principles of the \textit{Unscented Transform} (UT) \cite{juluhldw1995} to propagate the first and second order moments of the predicted and updated densities.

\subsection{The Extended Kalman Filter}
\label{ssec:ekf}
This subsection presents the EKF which is basically an extension of the linear KF whenever one or both the equations (\ref{eq:dtmodel}) and (\ref{eq:mmodel}) are nonlinear transformations of the state with additive Gaussian white noise \cite{EKF}, i.e.
\bie
	x_{k} & = & f_{k-1}\!\left( x_{k-1} \right) + w_{k-1} \, ,\label{eq:nonlindtmodel}\\
	y_{k} & = & h_{k}\!\left( x_{k} \right) + v_{k} \, .\label{eq:nonlinmmodel}
\eie

The prediction equations of the EKF are of the same form as the KF, with the transition matrix $A_{k-1}$ of (\ref{eq:lindtmodel}) evaluated via linearization about the updated mean $\hat{x}_{k-1}$, i.e.
\be
	A_{k-1} = \left. \dfrac{\partial f_{k-1}(\cdot)}{\partial x} \right|_{x = \hat{x}_{k-1}} \, .\label{eq:jacobiana}
\ee
The correction equations of the EKF are also of the same form as the KF, with the observation matrix $C_{k}$ of (\ref{eq:linmmodel}) evaluated via linearization about the predicted mean $\hat{x}_{k|k-1}$, i.e.
\be
	C_{k} = \left. \dfrac{\partial h_{k}(\cdot)}{\partial x} \right|_{x = \hat{x}_{k|k - 1}} \, .\label{eq:jacobianc}
\ee
The EKF recursion for computing both predicted and updated pairs $\left( \hat{x}_{k-1}, P_{k-1} \right)$ and $\left( \hat{x}_{k}, P_{k} \right)$ is reported in Table \ref{alg:ekf}.
\begin{table}[!h]
	\caption{The Extended Kalman Filter (EKF)}
	\label{alg:ekf}
	\hrulefill\hrule
	\begin{algorithmic}[0]
		\For{$k = 1, 2, \dots$}\vspace{0.5em}
			\State \bfbox{\textsc{Prediction}}
			\State $A_{k-1} = \left. \dfrac{\partial f_{k-1}(\cdot)}{\partial x} \right|_{x = \hat{x}_{k-1}}$\Comment{Linearization about the updated mean}
			\State $\hat{x}_{k|k-1} = A_{k-1} \hat{x}_{k-1}$\Comment{Predicted mean}
			\State $P_{k|k-1} = A_{k-1} P_{k-1} A^{\top}_{k-1} + Q_{k-1}$\Comment{Predicted covariance matrix}\vspace{0.5em}
			\State \bfbox{\textsc{Correction}}
			\State $C_{k} = \left. \dfrac{\partial h_{k}(\cdot)}{\partial x} \right|_{x = \hat{x}_{k|k - 1}}$\Comment{Linearization about the predicted mean}
			\State $e_{k} = y_{k} - C_{k}\hat x_{k|k-1}$\Comment{Innovation}
			\State $S_{k} = R_{k} + C_{k}P_{k|k-1}C_{k}^{\top}$\Comment{Innovation covariance matrix}
			\State $K_{k} = P_{k|k-1}C_{k}^{\top}S_{k}^{-1}$\Comment{Kalman gain}
			\State $\hat x_{k} = \hat x_{k|k-1} + K_{k}e_{k}$\Comment{Updated mean}
			\State $P_{k} = P_{k|k-1} - K_{k}S_{k}K_{k}^{\top}$\Comment{Updated covariance matrix}\vspace{0.5em}
		\EndFor
	\end{algorithmic}
	\hrule\hrulefill
\end{table}

\subsection{The Unscented Kalman Filter}
\label{ssec:ukf}
The UKF is based on the \textit{Unscented Transform} (UT), a derivative-free technique capable of providing a more accurate statistical characterization of a random variable undergoing a nonlinear transformation \cite{juluhl1997}.
In particular, the UT is a deterministic technique suited to provide an approximation of the mean and covariance matrix of a given random variable subjected to a nonlinear transformation via a minimal set of its samples.
Let us consider the mean $m$ and associated covariance matrix $P$ of a generic random variable along with a nonlinear transformation function $g\!\left(\cdot\right)$, the UT proceeds as follows:
\begin{itemize}
	\item generates $2n_{x}+1$ samples $X \in \mathbb{R}^{n_{x}\times\left(2n_{x}+1\right)}$, the so called $\sigma$-points, starting from the mean $m$ with deviation given by the matrix square root $\Sigma$ of $P$;
	\item propagates the $\sigma$-points through the nonlinear transformation function $g\!\left(\cdot\right)$ resulting in $G \in \mathbb{R}^{n_{x}\times\left(2n_{x}+1\right)}$;
	\item calculates the new transformed mean $m^{\prime}$ and associated covariance matrix $P_{gg}$ as well as the cross-covariance matrix $P_{xg}$ of the initial and transformed $\sigma$-points.
\end{itemize}
The pseudo-code of the UT is reported in Table \ref{alg:ut}.
\begin{table}[h!]
	\caption{The Unscented Transformation (UT)}
	\label{alg:ut}
	\hrulefill\hrule
	\begin{algorithmic}[0]
		\Procedure{UT}{$m$, $P$, $g$}
			\State $c$, $w_{m}$, $W_{c}$ = UTW$(\alpha_{\sigma}, \beta_{\sigma}, \kappa_{\sigma})$\Comment{Weights are calculated exploiting UTW in Table \ref{alg:weights}}
			\State $\Sigma = \sqrt{P}$
			\State  $X = \Big[ m \dots m \Big] + \sqrt{c} \, \Big[ \underline{0}, \Sigma, -\Sigma \Big]$\Comment{$\underline{0}$ is a zero column vector}
			\State $G = g\!\left( X \right)$\Comment{$g\!\left( \cdot \right)$ is applied to each column of $X$}
			\State $m^{\prime} = G w_{m}$
			\State $P_{gg} = G W_{c} G^{\top}$
			\State $P_{yg} = G W_{c} G^{\top}$
			\State \textbf{Return} $m^{\prime}$, $P_{gg}$, $P_{xg}$
		\EndProcedure
	\end{algorithmic}
	\hrule\hrulefill
	\caption{Unscented Transformation Weights (UTW)}
	\label{alg:weights}
	\hrulefill\hrule
	\begin{algorithmic}[0]
		\Procedure{UTW}{$\alpha_{\sigma}$, $\beta_{\sigma}$, $\kappa_{\sigma}$}
			\State $\varsigma = \alpha_{\sigma}^{2}(n_{x} + \kappa_{\sigma})-n_{x}$
			\State $w^{(0)}_{m} = \varsigma\left(n_{x} + \varsigma\right)^{-1}$
			\State $w^{(0)}_{c} = \varsigma\left(n_{x} + \varsigma\right)^{-1} + (1 - \alpha_{\sigma}^{2} + \beta_{\sigma})$
			\State $w^{(1, \dots, 2n_{x})}_{m}, w^{(1, \dots, 2n_{x})}_{c} = \left[2(n_{x} + \varsigma)\right]^{-1}$
			\State $w_{m} = \left[ w_{m}^{\left( 0 \right)}, \dots, w_{m}^{\left( 2n_{x} \right)}\right]^{\top}$
			\State $w_{c} = \left[ w_{c}^{\left( 0 \right)}, \dots, w_{c}^{\left( 2n_{x} \right)}\right]^{\top}$
			\State $W_{c} = \left( I - \left[ w_{m} \dots w_{m} \right] \right) \operatorname{diag}\!\left( w_{c}^{(0)} \dots w_{c}^{(2n)} \right) \left( I - \left[ w_{m} \dots w_{m} \right] \right)^{\top}$
			\State $c = \alpha_{\sigma}^{2}(n_{x} + \kappa_{\sigma})$
			\State \textbf{Return} $c$, $w_{m}$, $W_{c}$
		\EndProcedure
	\end{algorithmic}
	\hrule\hrulefill
\end{table}

Given three parameters $\alpha_{\sigma}$, $\beta_{\sigma}$ and $\kappa_{\sigma}$, the weights $c$, $w_{m}$ and $W_{c}$ are calculated exploiting the algorithm in Table \ref{alg:weights}.
Moment matching properties and performance improvements are discussed in \cite{juluhl1997,wanmer2001} by resorting to specific values of $\alpha_{\sigma}$, $\beta_{\sigma}$ and $\kappa_{\sigma}$.
It is of common practice to set these three parameters as constants, thus computing the weights once at the beginning of the estimation process.

The UT can be applied in the KF recursion allowing to obtain a nonlinear recursive estimator known as UKF \cite{juluhl1997}.
The pseudo-code of the UKF is shown in Table \ref{alg:ukf}.
\begin{table}[h!]
	\caption{The Unscented Kalman Filter (UKF)}
	\label{alg:ukf}
	\hrulefill\hrule
	\begin{algorithmic}[0]
		\For{$k = 1, 2, \dots$}\vspace{0.5em}
			\State \bfbox{\textsc{Prediction}}
			\State $\hat{x}_{k|k-1}, P_{k|k-1} =$ UT$\left( \hat{x}_{k-1|k-1}, P_{k-1|k-1}, f(\cdot) \right)$
			\State $P_{k|k-1} = P_{k|k-1} + Q$\vspace{0.5em}
			\State \bfbox{\textsc{Correction}}
			\State $\hat{y}_{k|k-1}, S_{k}, C_{k} =$ UT$\left( \hat{x}_{k|k-1}, P_{k|k-1}, h(\cdot) \right)$
			\State $S_{k} = S_{k} + R$
			\State $\hat{x}_{k} = \hat{x}_{k|k-1} + C_{k}S_{k}^{-1} \left( y_{k} - \hat{y}_{k|k-1} \right)$
			\State $P_{k} = P_{k|k-1} - C_{k}S_{k}^{-1}C_{k}^{\top}$
		\EndFor
	\end{algorithmic}
	\hrule\hrulefill
\end{table}

The main advantages of the UKF approach are the following:
\begin{itemize}
	\item it does not require the calculation of the Jacobians (\ref{eq:jacobiana}) and (\ref{eq:jacobianc}). The UKF algorithm is, therefore, very suitable for highly nonlinear problems and represents a good trade-off between accuracy and numerical efficiency;
	\item being derivative free, it can cope with functions with jumps and discontinuities;
	\item it is capable of capturing higher order moments of nonlinear transformations \cite{juluhl1997}.
\end{itemize}
Due to the above mentioned benefits, the UKF is herewith adopted as the nonlinear recursive estimator.

The KF represents the basic tool for recursively estimating the state, in a Bayesian framework, of a moving object, i.e. to perform \textit{Single-Object Filtering} (SOF) \cite{far1985v1,far1985v2,book0,book,bar-shalom,BlPo}.
A natural evolution of SOF is \textit{Multi-Object Tracking} (MOT), which involves the on-line estimation of an unknown and (possibly) time-varying number of objects and their individual trajectories from sensor data \cite{BlPo,book0,mahler}.
The key challenges in multi-object tracking include \textit{detection uncertainty}, \textit{clutter}, and \textit{data association uncertainty}.
Numerous multi-object tracking algorithms have been developed in the literature and most of these fall under three major paradigms: \textit{Multiple Hypotheses Tracking} (MHT) \cite{reid,BlPo}; \textit{Joint Probabilistic Data Association} (JPDA) \cite{book0}; and \textit{Random Finite Set} (RFS) filtering \cite{mahler}.
In this thesis, the focus is on the RFS formulation since it provides the concept of \textit{probability density} for \textit{multi-object state} that allows to directly extend the single-object Bayesian recursion.
Such a concept is not available in the MHT and JPDA approaches \cite{reid,far1985v1,far1985v2,book0,book,BlPo}.

\section{Random finite set approach}
\label{sec:rfs}
The second key concept of the present dissertation is the RFS \cite{goomahngu1997,mahler,mah2004,mah2013} approach.
In the case where the need is to recursively estimate the state of a possibly time varying number of multiple dynamical systems, RFSs allow to generalize standard Bayesian filtering to a unified framework.
In particular, states and observations will be modelled as RFSs where not only the single state and observation are random, but also their number (set cardinality).
The purpose of this section is to cover the aspects of the RFS approach that will be useful for the subsequent chapters.
Overviews on RFSs and further advanced topics concerning point process theory, stochastic geometry and measure theory can be found in \cite{mat1975,stokenmec1995,goomahngu1997,mahler1,VSD05,mahler}.

\begin{rem}
	It is worth pointing out that in multi-object scenarios there is a subtle difference between \textit{filtering} and \textit{tracking}. In particular, the first refers to estimating the state of a possibly time-varying number of objects without, however, uniquely identifying them, i.e. after having estimated a multi-object density a decision-making operation is needed to extract the objects themselves.
	On the other hand, the term ``tracking'' refers to jointly estimating a possibly time-varying number of objects and to uniquely mark them over time so that no decision-making operation has to be carried out and object trajectories are well defined. Finally, it is clear that in single-object scenarios the terms \textit{filtering} and \textit{tracking} can be used interchangeably.
\end{rem}

\subsection{Notation}
\label{ssec:notation}
Throughout the thesis, finite sets are represented by uppercase letters, e.g. $X$, $\mathbf{X}$.
The following multi-object exponential notation is used
\be
	h^{X} \triangleq \prod_{x \in X}h(x) \, ,
\label{eq:setexp}
\ee
where $h$ is a real-valued function, with $h^{\varnothing} = 1$ by convention \cite{mahler}.
The following generalized Kronecker delta \cite{vovo1,vovo2} is also adopted
\be
	\delta _{Y}(X)\triangleq \left\{ 
						\begin{array}{l}
							1,\text{ if }X=Y \\ 
							0,\text{ otherwise}
						\end{array} \right. ,
\ee
along with the inclusion function, a generalization of the indicator function, defined as
\be
	1_{Y}(X)\triangleq \left\{ 
					\begin{array}{l}
						1,\text{ if }X\subseteq Y \\ 
						0,\text{ otherwise}
						\end{array} \right. .
\ee
The shortand notation $1_{Y}(x)$ is used in place of $1_{Y}(\{x\})$ whenever $X$ = $\{x\}$.
The cardinality (number of elements) of the finite set $X$ is denoted by $| X |$.
The following PDF notation will be also used.
\bie
	\operatorname{Poisson}_{\left[ \lambda \right]}\!\left( n \right) & = & \dfrac{e^{-\lambda}\lambda^{n}}{{n!}} \, , \lambda \in \nbb \, , n \in \nbb \, ,\label{eq:abbpoisspdf}\\
	\operatorname{Uniform}_{\left[ a, b \right]}\!\left( n \right) & = & \left\{ \ba{ll}
															\dfrac{1}{b - a} \, , & n \in \left[ a, b \right]\\
															0 \, , & n \notin \left[ a, b \right]
														\ea \right. \, , a \in \rbb \, , b \in \rbb \, , a < b \, . \label{eq:abbuniformpdf}
\eie

\subsection{Random finite sets}
\label{ssec:rfss}
In a typical multiple object scenario, the number of objects varies with time due to their appearance and disappearance.
The sensor observations are affected by misdetection (e.g., occlusions, low radar cross section, etc.) and false alarms (e.g., observations from the environment, clutter, etc.). This is further compounded by association uncertainty, i.e. it is not known which object generated which measurement. The objective of multi-object filtering is to jointly estimate over time the number of objects and their states from the observation history.

In this thesis we adopt the RFS formulation, as it provides the concept of \textit{probability density of the multi-object state} that allows us to directly generalize (single-object) estimation to the multi-object case. Indeed, from an estimation viewpoint, the multi-object system state is naturally represented as a finite set \cite{VVPS10}. More concisely, suppose that at time $k$, there are $N_{k}$ objects with states $x_{k,1},\ldots ,x_{k,N_{k}}$, each taking values in a state space $\mathbb{X}\subseteq \mathbb{R}^{n_{x}}$, i.e. the \textit{multi-object state} at time $k$ is the finite set
\be
	X_{k}=\{x_{k,1},\ldots ,x_{k,N_{k}}\}\subset \mathbb{X}.
\ee
Since the multi-object state is a finite set, the concept of RFS is required to model, in a probabilistic way, its uncertainty.

An RFS $X$ on a space $\mathbb{X}$ is a random variable taking values in $\mathcal{F}(\mathbb{X})$, the space of finite subsets of $\mathbb{X}$. The notation $\mathcal{F}_{n}(\mathbb{X})$ will be also used to refer to the space of finite subsets of $\mathbb{X}$ with exactly $n$ elements.
\begin{defi}
	An RFS $X$ is a random variable that takes values as (unordered) finite sets, i.e. a finite-set-valued random variable.
\end{defi}
At the fundamental level, like any other random variable, an RFS is described by its probability distribution or probability density.
\begin{rem}
	What distinguishes an RFS from a random vector is that:
	$\textsc{i})$ the number of points is random;
	$\textsc{ii})$ the points themselves are random and unordered.
\end{rem}
The space $\mathcal{F}(\mathbb{X})$ does not inherit the usual Euclidean notion of integration and density.
In this thesis, we use the \textit{FInite Set STatistics} (FISST) notion of integration/density to characterize RFSs \cite{mahler1, mahler}.

From a probabilistic viewpoint, an RFS $X$ is completely characterized by its \textit{multi-object} \textit{density} $f(X)$.
In fact, given $f(X)$, the cardinality \textit{Probability Mass Function} (PMF) $\cd{n}$ that $X$ have $n \geq 0$ elements and the joint conditional PDFs $f(x_1,x_2,\dots,x_n|n)$ over $\mathbb{X}^n$ given that $X$ have $n$ elements,
can be obtained as follows:
\begin{IEEEeqnarray}{rCl}
	\cd{n} & = & \dfrac{1}{n!}~ \displaystyle{\int_{\mathbb{X}^n}}~ f \left( \left\{ x_1, \dots, x_n\right\} \right)~ dx_1 \cdots dx_n\vspace{0.5em}\\
	f\!\left( x_1, \dots, x_n | n \right)  & = & \dfrac{1}{n! \, \cd{n}}~ f \left( \left\{ x_1, \dots, x_n\right\} \right)
\end{IEEEeqnarray}
\begin{rem}
	The multi-object density $f(X)$ is nothing but the multi-object counterpart of the state PDF in the single-object case.
\end{rem}
In order to measure probability over subsets of $\mathbb{X}$ or compute expectations of random set variables, it is convenient to introduce the following definition of \textit{set integral} for a generic real-valued function $g(X)$ (not necessarily a multi-object density) of an RFS variable
$X$:
\begin{IEEEeqnarray}{rCl}
	\displaystyle{\int_{\xbb}} g\!\left( X \right) \delta X & \triangleq & \displaystyle{\sum_{n=0}^\infty} \, \dfrac{1}{n!} \displaystyle{\int_{\xbb^{n}}} g\!\left( \left\{ x_1, \dots, x_n \right\} \right) dx_{1} \cdots dx_{n}\label{eq:unlsetint}\\
	& = & g\!\left( \varnothing \right) + \displaystyle{\int_{\xbb}} g\!\left( \left\{ x \right\} \right) dx + \dfrac{1}{2}\displaystyle{\int_{\xbb^{2}}} g\!\left( \left\{ x_1, x_2 \right\} \right) dx_1 dx_2 + \cdots \IEEEnonumber
\end{IEEEeqnarray}
In particular,
\be
	\beta\!\left( X \right) \triangleq \operatorname{Prob}\!\left( X \subset \xbb \right) = \displaystyle{\int_{\xbb}} f\!\left( X \right) \delta X
\ee
measures the probability that the RFS $X$ is included in the subset $\xbb$ of $\rbb^{n_{x}}$.
The function $\beta\!\left( X \right)$ is also known as the \textit{Belief-Mass Function} (BMF) \cite{mahler}.
\begin{rem}
	The BMF $\beta\!\left( X \right)$ is nothing but the multi-object counterpart of the state \textit{Cumulative Distribution Function} (CDF) in the single-object case.
\end{rem}
It is also easy to see that, thanks to the set integral definition (\ref{eq:unlsetint}), the multi-object density, like the single-object state PDF, satisfies
the trivial normalization constraint
\be
	\int_{\mathbb{X}} f\!\left( X \right) \delta X = 1.
\ee

It is worth pointing out that the multi-object density, while completely characterizing an RFS, involves a combinatorial complexity; hence simpler, though incomplete, characterizations are usually adopted in order to keep the \textit{Multi-Object Filtering} (MOF) problem computationally tractable.
In this respect, the \textit{first-order moment} of the multi-object density, better known as \textit{Probability Hypothesis Density} (PHD) or \textit{intensity function}, has been found to be a very successful characterization \cite{mah2004,mahler1,mahler}.
In order to define the PHD function, let us introduce the number of elements of the RFS $X$ which is given by
\be
	N_{X} = \displaystyle{\int_{\xbb}} \phi_{X}\!\left( x \right) d x \, ,
\ee
where
\be
	\phi_{X}\!\left( x \right) \triangleq \sum_{\xi \in X} \delta_{x}\!\left( \xi \right) \, .
\ee
We would like to define the PHD function $\intf{}{}{x}$ of $X$ over the state space $\xbb$ so that the expected number of elements of $X$ in $\xbb$ is obtained by integrating $\intf{}{}{\cdot}$ over $\xbb$, i.e.
\be
	\operatorname{E}\!\left[ N_{X} \right] = \displaystyle \int_{\xbb} \intf{}{}{x} d x \, .
\label{eq:objnum}
\ee
Since
\bie
	\operatorname{E}\!\left[ N_{X} \right] & = & \displaystyle \int N_{X} f\!\left( X \right) \delta X\\
	& = & \displaystyle \int \left[ \int_{\xbb} \phi_{X}\!\left( x \right) d x \right] f\!\left( X \right) \delta X\\
	& = & \displaystyle \int_{\xbb} \left[ \int \phi_{X}\!\left( x \right) f\!\left( X \right) \delta X \right] d x \, ,
\label{eq:d1}
\eie
comparing (\ref{eq:objnum}) with (\ref{eq:d1}), it turns out that
\be
	d\!\left( x \right) \triangleq \operatorname{E}\!\left[ \phi_{X}\!\left( x \right) \right] = \int \phi_{X}\!\left( x \right) f\!\left( X \right) \delta X
\label{eq:def-d}
\ee
Without loss of generality, the PHD function can be expressed in the form
\be
	\intf{}{}{x} = \overline{n} \, s\!\left( x \right)
\label{eq:PHDform}
\ee
where
\be
	\overline{n} = \operatorname{E}\!\left[ n \right] = \operatorname{E}\!\left[ n\!\left( \xbb \right) \right] = \displaystyle \sum_{n = 0} ^{\infty} n \cd{n}
\label{target-num}
\ee
is the expected number of objects and $s\!\left( \cdot \right)$ is a single-object PDF, called \textit{location density}, such that
\be
	\int_{\xbb} s\!\left( x \right) d x = 1 \, .
\ee
It is worth to highlight that, in general, the PHD function $\intf{}{}{\cdot}$ and the cardinality PMF $\cd{\cdot}$ do not completely characterize the multi-object distribution.
However, for specific RFSs defined in the next section, the characterization is complete.

\subsection{Common classes of RFS}
\label{ssec:classrfs}
A review of the common RFS densities is provided \cite{mahler} hereafter.

\subsubsection{Poisson RFS}
\label{ssec:poissrfs}
A \textit{Poisson} RFS $X$ on $\xbb$ is uniquely characterized by its intensity function $\intf{}{}{\cdot}$.
The Poisson RFSs have the unique property that the distribution of the cardinality of $X$ is Poisson with mean
\be
	D = \int_{\xbb} \intf{}{}{x} d x \, ,
\label{eq:poissmean}
\ee
and for a given cardinality the elements of $X$ are i.i.d. with probability density
\be
	\lpdf{}{}{x} = \dfrac{\intf{}{}{x}}{D} \, .
\label{eq:poisspdf}
\ee
The probability density of $X$ can be written as
\be
	\rfs{}{}{X} = e^{-D} \prod_{x \in X} \intf{}{}{x}
\label{eq:poissrfspdf}
\ee
The Poisson RFS is traditionally described as characterizing no spatial interaction or complete spatial randomness in the following sense.
For any collection of disjoint subsets $B_{i} \in \xbb$, $i \in \nbb$, it can be shown that the count functions $\left| X \cap B_{i} \right|$ are independent random variables which are Poisson distributed with mean
\be
	D_{B_{i}} = \int_{B_{i}} \intf{}{}{x} d x \, ,
\ee
and for a given number of points occurring in $B_{i}$ the individual points are i.i.d. according to
\be
	\dfrac{\intf{}{}{x} \inc{B_{i}}{x}}{N_{B_{i}}} \, .
\ee
The procedure in Table \ref{alg:samplepoissrfs} illustrates how to generate a sample from a Poisson RFS.

\begin{table}[!h]
\caption{Sampling a Poisson RFS}
\label{alg:samplepoissrfs}
\renewcommand{\arraystretch}{1.3}
\hrulefill\hrule
\begin{algorithmic}[0]
	\State $X = \varnothing$
	\State Sample $n \sim \operatorname{Poisson}_{\left[ D \right]}$
	\For{$i =1, \dots, n$}
		\State Sample $x_{i} \sim \lpdf{}{}{\cdot}$
		\State $X = X \cup \left\{ x \right\}$
	\EndFor
\end{algorithmic}
\hrule\hrule
\end{table}

\subsubsection{Independent identically distributed cluster RFS}
\label{ssec:iidcrfs}
An i.i.d. cluster RFS $X$ on $\xbb$ is uniquely characterized by its cardinality distribution $\cd{\cdot}$ and matching intensity function $\intf{}{}{\cdot}$.
The cardinality distribution must satisfy
\be
	D = \sum_{n = 0}^{\infty} n \cd{n} = \int_{\xbb} \intf{}{}{x} d x\, ,
\ee
but can otherwise be arbitrary, and for a given cardinality the elements of $X$ are i.i.d. with probability density
\be
	\lpdf{}{}{x} = \dfrac{\intf{}{}{x}}{D}
\ee
The probability density of an i.i.d. cluster RFS can be written as
\be
	\rfs{}{}{X} = \left| X \right|! \, \cd{\left| X \right|} \prod_{x \in X} \lpdf{}{}{x}
\label{eq:iidcrfspdf}
\ee
Note that an i.i.d. cluster RFS essentially captures the spatial randomness of the Poisson RFS without the restriction of a Poisson cardinality distribution.
The procedure in Table \ref{alg:sampleiidcrfs} illustrates how a sample from an i.i.d. RFS is generated.

\begin{table}[!h]
\caption{Sampling an i.i.d. RFS}
\label{alg:sampleiidcrfs}
\renewcommand{\arraystretch}{1.3}
\hrulefill\hrule
\begin{algorithmic}[0]
	\State $X = \varnothing$
	\State Sample $n \sim \cd{\cdot}$
	\For{$i =1, \dots, n$}
		\State Sample $x_{i} \sim \lpdf{}{}{\cdot}$
		\State $X = X \cup \left\{ x \right\}$
	\EndFor
\end{algorithmic}
\hrule\hrule
\end{table}

\subsubsection{Bernoulli RFS}
\label{ssec:bernrfs}
A Bernoulli RFS $X$ on $\xbb$ has probability $\nex{} \triangleq 1 - \ex{}$ of being empty, and probability $\ex{}$ of being a singleton whose only element is distributed according to a probability density $\p{}{}$ defined on $\xbb$.
The cardinality distribution of a Bernoulli RFS is thus a Bernoulli distribution with parameter $\ex{}$.
A Bernoulli RFS is completely described by the parameter pair $\left( \ex{}, \p{}{\cdot} \right)$.

\subsubsection{Multi-Bernoulli RFS}
\label{ssec:mbernrfs}
A multi-Bernoulli RFS $X$ on $\xbb$ is a union of a fixed number of independent Bernoulli RFSs $X^{\left( i \right)}$ with \textit{existence probability} $\ex{i} \in \left( 0, 1 \right)$ and probability density $\p{i}{\cdot}$ defined on $\xbb$ for $i = 1, \dots , I$, i.e.
\be
	X = \bigcup_{i = 1}^{I} X^{\left( i \right)} \, .
\ee
It follows that the mean cardinality of a multi-Bernoulli RFS is
\be
	D = \sum_{i = 1}^{I} \ex{i}
\ee
A multi-Bernoulli RFS is thus completely described by the corresponding multi-Bernoulli parameter set $\left\{ \left( \ex{i}, \p{i}{\cdot} \right) \right\}_{i = 1}^{I}$.
Its probability density is
\be
	\rfs{}{}{X} = \prod_{j = 1}^{I} \left( 1 - \ex{j} \right) \sum_{1 \le i_{1} \neq \dots \neq i_{\left| X \right|} \le I} \prod_{j = 1}^{\left| X \right|} \dfrac{\ex{i_{j}} \, \p{i_{j}}{x_{j}}}{1 - \ex{i_{j}}}
\label{eq:mbernrfspdf}
\ee
For convenience, probability densities of the form (\ref{eq:mbernrfspdf}) are abbreviated by the form $\rfs{}{}{X} = \left \{ \left( \ex{i}, \p{i}{} \right)\right\}_{i = 1}^{I}$.
A multi-Bernoulli RFS jointly characterizes unions of non-interacting points with less than unity probability of occurrence and arbitrary spatial distributions.
The procedure in Table \ref{alg:samplembrfs} illustrates how a sample from a multi-Bernoulli RFS is generated.

\begin{table}[!h]
\caption{Sampling a multi-Bernoulli RFS}
\label{alg:samplembrfs}
\renewcommand{\arraystretch}{1.3}
\hrulefill\hrule
\begin{algorithmic}[0]
	\State $X = \varnothing$
	\For{$i = 1, \dots, I$}
		\State Sample $u \sim \operatorname{Uniform}_{\left[0, 1 \right]}$
		\If{$u \leq \ex{i}$}
			\State Sample $x \sim \p{i}{\cdot}$
			\State $X = X \cup \left\{ x \right\}$
		\EndIf
	\EndFor
\end{algorithmic}
\hrule\hrule
\end{table}

\subsection{Bayesian multi-object filtering}
\label{ssec:RFSBayesFiltering}
Let us now introduce the basic ingredients of the MOF problem \cite{mahler,mah2004,mah2013,mah2014book}, i.e. the object RFS $X_{k} \subset \xbb$ at  time $k$ and the RFS $Y_{k}^{i}$ of measurements gathered by node $i \in \mathcal{N}$ at time $k$.
It is also convenient to define the overall measurement RFSs at time $k$,
\be
	Y_{k} \triangleq \bigcup_{i \in \mathcal{N}} Y_{k}^{i} \, ,
\ee
and up to time $k$,
\be
	Y_{1:k} \triangleq \bigcup_{\kappa=1}^{k} Y_{\kappa} \, .
\ee
In the random set framework, the Bayesian approach to MOF consists, therefore, of recursively estimating the object set $X_{k}$ conditioned to the observations $Y_{1:k}$.
The object set is assumed to evolve according to  a multi-object dynamics
\be
X_{k} = \Phi_{k-1}\!\left( X_{k-1} \right) \cup B_{k-1}
\label{eq:MTD}
\ee
where $B_{k-1}$ is the RFS of \textit{new-born} objects at time $k-1$ and
\bie
	\Phi_{k-1}\!\left( X \right) & = & \displaystyle \bigcup_{x \in X} \phi_{k-1}\!\left( x \right) \, ,\label{eq:MTD1}\\
	\phi_{k-1}\!\left( x \right) & = & \left\{ \ba{ll}
											\left\{  x^{\prime} \right\} \, , & \mbox{with survival probability } P_{S,k-1} \\
											\varnothing \, , & \mbox{otherwise } \ea \right.\label{eq:MTD2}
\eie
Notice that according to (\ref{eq:MTD})-(\ref{eq:MTD2}) each object in the set $X_{k}$ is either a new-born object from the set $B_{k-1}$ or an object survived from $X_{k-1}$, with probability $P_{S,k-1}$, and whose state vector $x^{\prime}$ has evolved according to the single-object dynamics (\ref{eq:dtmodel}) with $x = x_{k-1}$ and $x^{\prime} = x_{k}$.
In a similar way, observations are assumed to be generated, at each node $i \in \mathcal{N}$, according to the measurement model
\be
 	Y_{k}^{i} = \Psi_{k}^{i}\!\left( X_{k} \right) \cup \mathcal{C}_{k}^{i}
\label{eq:MM}
\ee
where $\mathcal{C}_{k}^{i}$ is the \textit{clutter} RFS (i.e. the set of measurements not due to objects) at time $k$ and node $i$, and
\bie
	\Psi_{k}^{i}\!\left( X \right) & = & \displaystyle \bigcup_{x \in X} \psi_{k}^{i}\!\left( x \right) \, ,\label{eq:MM1} \\
	\psi_{k}^{i}\!\left( x \right) & = & \left\{ \ba{ll}
												\left\{  y_{k} \right\}\, , & \mbox{with detection probability } P_{D,k}\\
												\varnothing \, , & \mbox{otherwise } \ea \right. \label{eq:MM2}
\eie
Notice that according to (\ref{eq:MM})-(\ref{eq:MM2}) each measurement in the set $Y_{k}^{i}$ is either a false one from the clutter set $\mathcal{C}_{k}^{i}$ or is related to an object in $X_{k}$, with probability $P_{D,k}$, according to the single-sensor measurement equation
(\ref{eq:mmodel}).
Hence it is clear how one of the major benefits of the random set approach is to directly include in the MOF problem formulation (\ref{eq:MTD})-(\ref{eq:MM2}) fundamental practical issues such as object birth and death, false alarms (clutter) and missed detections.
Following a Bayesian approach, the aim is to propagate in time the \textit{posterior} (\textit{filtered}) \textit{multi-object densities} $\rfs{k}{}{X}$ of the object set
$X_{k}$ given $Y_{1:k}$, as well as the \textit{prior} (\textit{predicted}) ones $\rfs{k|k-1}{}{X}$ of $X_{k}$ given $Y_{1:k-1}$.
Exploiting random set theory, it has been found \cite{mahler} that the multi-object densities follow recursions that are conceptually analogous to the well-known
ones of Bayesian nonlinear filtering, i.e.
\bie
	\rfs{k|k-1}{}{X} & = & \int \varphi_{k|k-1}\!\left( X | Z \right) \rfs{k-1}{}{Z} \delta Z \, , \label{eq:MTunlBayesPred}\\
	\rfs{k}{}{X} & = & \dfrac{g_{k}\!\left( Y_{k} | X \right) \rfs{k|k-1}{}{X}}{\displaystyle \int g_{k}\!\left( Y_{k} | Z \right) \rfs{k|k-1}{}{Z} \delta Z} \, ,\label{eq:MTunlBayesUpdate}
\eie
where $\varphi_{k|k-1}(\,\cdot \,|\,\cdot \,)$ is the \textit{multi-object transition density} to time $k$, $g_{k}(\,\cdot\,|\,\cdot \,)$ is the \textit{multi-object likelihood function} at time $k$ \cite{mahler1,mahler2,mahler}.
\begin{rem}
	Despite this conceptual resemblance, however, the multi-object Bayes recursions (\ref{eq:MTunlBayesPred})-(\ref{eq:MTunlBayesUpdate}), unlike their single-object counterparts, are affected by combinatorial complexity and are, therefore, computationally infeasible except for very small-scale MOF problems involving few objects and/or measurements.
\end{rem}
For this reason, the computationally tractable PHD and \textit{Cardinalized PHD} (CPHD) filtering approaches will be followed and briefly reviewed.

\subsection{CPHD filtering}
\label{sec:cphd}
The CPHD filter \cite{mahler2,vo-vo-cantoni} propagates in time the cardinality PMFs $\cd[k|k-1]{n}$ and $\cd[k]{n}$ as well as the PHD functions $\intf{k|k-1}{}{x}$ and $\intf{k}{}{x}$ of $X_{k}$ given $Y_{1:k-1}$ and, respectively, $Y_{1:k}$ assuming that the clutter RFS, the predicted and filtered RFSs are i.i.d. cluster processes (see (\ref{eq:iidcrfspdf})).
The resulting CPHD recursions (prediction and correction) are as follows
\begin{IEEEeqnarray}{l}
	\mbox{\bfbox{Prediction}}\IEEEnonumber\\
	\ba{l}
		\cd[k|k-1]{n} = \displaystyle \sum_{j=0}^{n} \, p_{b}\!\left( n-j \right) \, \rho_{S, k|k-1}\!\left( j \right)\\
		\intf{k|k-1}{}{x} = \intf{b,k}{}{x} + \displaystyle \int \varphi_{k|k-1}\!\left( x | \zeta \right) \, P_{S,k-1}\!\left( \zeta \right) \intf{k-1}{}{\zeta} d\zeta 
	\ea\label{eq:cphdpred}\vspace{0.5em}\\
	\mbox{\bfbox{Correction}}\IEEEnonumber\\
\ba{l}
	\cd[k]{n} = \dfrac{\mathcal{G}_{k}^{0}\!\left( \intf{k|k-1}{}{\cdot}, Y_{k}, n \right) \, \cd[k|k-1]{n}}{\displaystyle \sum_{\eta = 0}^{\infty} \, \mathcal{G}_{k}^{0}\!\left( \intf{k|k-1}{}{\cdot}, Y_{k}, \eta \right) ~\cd[k|k-1]{\eta}}\\
	\intf{k}{}{x} = \mathcal{G}_{Y_{k}}\!\left( x \right) \, \intf{k|k-1}{}{x}                           
\ea\label{eq:cphdcor}
\end{IEEEeqnarray}
where: $p_{b,k}(\cdot)$ is the assumed birth cardinality PMF; $\rho_{S, k|k-1}\!\left( \cdot \right)$ is the cardinality PMF of survived objects, given by
\be
	\rho_{S, k|k-1}\!\left( j \right) = \displaystyle \sum_{t = j}^{\infty} \, \left( \ba{c} t \\ j \ea \right) P_{S,k}^{j} \, (1 - P_{S,k})^{h-j} \, \cd[k-1]{t} \, ;
\ee
$d_{b,k}(\cdot)$ is the assumed PHD function of new-born objects; $\varphi_{k|k-1}(x | \xi)$ is the state transition PDF (\ref{eq:mtmodel}) associated to the single object dynamics; the generalized likelihood functions $\mathcal{G}_{k}^{0}\!\left( \cdot, \cdot, \cdot  \right)$ and $\mathcal{G}_{Y_{k}}\!\left( \cdot \right)$ have cumbersome expressions which can be found in \cite{vo-vo-cantoni}.

Enforcing Poisson cardinality distributions for the clutter and object sets, the update of the cardinality PMF is no longer needed and the CPHD filter reduces to the PHD filter \cite{mahler1,vo-ma}.
The CPHD filter provides enhanced robustness with respect to misdetections, clutter as well as various uncertainties on the multi-object model (i.e. detection and survival probabilities, clutter and birth distributions) at the price of an increased computational load: a CPHD recursion has $O(m^3 n_{max})$ complexity compared to $O(m n_{max})$ of PHD, $m$ being the number of measurements and $n_{max}$ the maximum number of objects. 
There are essentially two ways of implementing PHD/CPHD filters, namely the \textit{Particle Filter} (PF)  or \textit{Sequential Monte Carlo} (SMC) \cite{VSD05} and the \textit{Gaussian Mixture} (GM) \cite{vo-ma,vo-vo-cantoni} approaches, the latter being cheaper in terms of computation and memory requirements and, hence, by far preferable for distributed implementation on a sensor network wherein nodes have limited processing and communication capabilities.

\subsection{Labeled RFSs}
\label{ssec:lrfss}
The aim of MOT is not only to recursively estimate the state of the objects and their time varying number, but also to keep track of their trajectories. To this end, the notion of label is introduced in the RFS approach \cite{vovo1,vovo2} so that each estimated object can be uniquely identified and its track reconstructed.

Let $\mathcal{L}: \mathbb{X}\mathcal{\times }\mathbb{L}\rightarrow \mathbb{L}$ be the projection $\mathcal{L}((x,\ell ))=\ell $, where $\mathbb{X}\subseteq \mathbb{R}^{n_{x}}$, $\mathbb{L}=\mathcal{\{\alpha }_{i}:i\in \mathbb{N\}}$ and $\mathcal{\alpha }_{i}$'s are distinct.
To estimate the trajectories of the objects they need to be uniquely identified by a (unobserved) label drawn from a discrete countable space $\mathbb{L}=\left\{ \alpha_{i}: i\in \mathbb{N} \right\}$, where the $\alpha_{i}$’s are distinct.
To incorporate object identity, a label $\ell \in \mathbb{L}$ is appended to the state $x$ of each object and the multi-object state is regarded as a finite set on $\mathbb{X} \times \mathbb{L}$, i.e. $\lb{x} \triangleq \left( x, \ell \right) \in \mathbb{X} \times \mathbb{L}$.
However, this idea alone is not enough since $\mathbb{L}$ is discrete and it is possible (with non-zero probability) that multiple objects have the same identity, i.e. be marked with the same label. This problem can be alleviated using a special case of RFS called \textit{labeled} RFS \cite{vovo1,vovo2}, which is, in essence, a marked RFS with distinct labels.
Then, a finite subset $\mathbf{X}$ of $\mathbb{X}\mathcal{\times }\mathbb{L}$ has distinct labels if and only if $\mathbf{X}$ and its labels $\mathcal{L}(\mathbf{X})=\{\mathcal{L}(\mathbf{x}):\mathbf{x} \in \mathbf{X}\}$ have the same cardinality, i.e. $|\mathbf{X}|=|\mathcal{L}(\mathbf{X})|$. The function $\Delta (\mathbf{X})\triangleq $ $\delta _{|\mathbf{X}|}(|\mathcal{L}(\mathbf{X})|)$ is called the \textit{distinct label indicator} and assumes value $1$ if and only if the condition $|\mathbf{X}|=|\mathcal{L}(\mathbf{X})|$ holds.
\begin{defi}
	A labeled RFS with state space $\stsp$ and (discrete) label space $\lbsp$ is an RFS on $\stsp \times \lbsp$ such that each realization $\lb{X}$ has distinct labels, i.e.
	\be
		\left| \lb{X} \right| = \left| \lbs{X} \right| \, .
	\ee
\end{defi}

The unlabeled version of a labeled RFS with density $\lmb{}$ is distributed according to the \textit{marginal density} \cite{vovo1}
\be
	\pi\!\left( \left\{ x_{1}, \dots, x_{n} \right\} \right) = \sum_{\left( \ell_{1}, \dots, \ell_{n} \right) \in \lbsp^{n}} \lmb{\left\{ \left( x_{1}, \ell_{1} \right), \dots, \left( x_{n}, \ell_{n} \right) \right\}} \, .
\ee
The unlabeled version of a labeled RFS is its projection from $\stsp \times \lbsp$ into $\stsp$, and is obtained by simply discarding the labels. The cardinality distribution of a labeled RFS is the same as its unlabeled version \cite{vovo1}.

Hereinafter, symbols for labeled states and their distributions are bolded to distinguish them from unlabeled ones, e.g. $\mathbf{x}$, $\mathbf{X}$, $\boldsymbol{\pi }$, etc.

\subsection{Common classes of labeled RFS}
\label{ssec:classlrfs}
A review of the common labeled RFS densities is provided \cite{vovo1}.

\subsubsection{Labeled Poisson RFS}
\label{ssec:lpoissrfs}
A labeled Poisson RFS $\lb{X}$ with state space $\stsp$ and label space $\lbsp = \left\{ \alpha_{i} : i \in \nbb \right\}$, is a Poisson RFS $X$ on $\stsp$ with intensity $\intf{}{}{\cdot}$, tagged with labels from $\lbsp$.
A sample from such labeled Poisson RFS can be generated by the procedure reported in Table \ref{alg:samplelpoissrfs}.
\begin{table}[!h]
\caption{Sampling a labeled Poisson RFS}
\label{alg:samplelpoissrfs}
\renewcommand{\arraystretch}{1.3}
\hrulefill\hrule
\begin{algorithmic}[0]
	\State $\lb{X} = \varnothing$
	\State Sample $n \sim \operatorname{Poisson}_{\left[ D \right]}$
	\For{$i = 1, \dots, n$}
		\State Sample $x \sim \intf{}{}{\cdot}/D$
		\State $\lb{X} = \lb{X} \cup \left\{ \left( x, \alpha_{i} \right) \right\}$
	\EndFor
\end{algorithmic}
\hrule\hrule
\end{table}

The probability density of an LMB RFS is given by 
\be
	\lmb{\lb{X}} = \delta_{\lbsp\!\left( \left| \lb{X} \right| \right)}\!\left( \lbs{\lb{X}} \right) \, \operatorname{Poisson}_{\left[ D \right]}\!\left( \left| \lb{X} \right| \right) \, \prod_{\lb{x} \in \lb{X}} \dfrac{\intf{}{}{x}}{D} \, ,\label{eq:lpoisspdf}
\ee
where $\lbsp\!\left( n \right) = \left\{ \alpha_{i} \in \lbsp \right\}_{i=1}^{n}$; $\delta_{\lbsp\!\left( n \right)}\!\left( \left\{ \ell_{1}, \dots, \ell_{n} \right\} \right)$ serves to check whether or not labels are distinct.

\begin{rem}
	By a tracking point of view, each label of a labeled Poisson RFS $\lb{X}$ cannot directly refer to an object since all the corresponding states are sampled from the same intensity $\intf{}{}{x}$, i.e. there is not a 1-to-1 mapping between labels and location PDFs.
	Thus, it is not clear how to properly use such a labeled distribution in a MOT problem.
\end{rem}

\subsubsection{Labeled independent identically distributed cluster RFS}
\label{ssec:liidcrfs}
In the same fashion of the unlabeled RFSs, the labeled Poisson RFS can be generalized to the labeled i.i.d. cluster RFS by removing the Poisson assumption on the cardinality and specifying an arbitrary cardinality distribution.
A sample from such a labeled i.i.d. RFS can be generated by the procedure reported in Table \ref{alg:sampleliidcrfs}.
\begin{table}[!h]
\caption{Sampling a labeled i.i.d. cluster RFS}
\label{alg:sampleliidcrfs}
\renewcommand{\arraystretch}{1.3}
\hrulefill\hrule
\begin{algorithmic}[0]
	\State $\lb{X} = \varnothing$
	\State Sample $n \sim \cd{\cdot}$
	\For{$i = 1, \dots, n$}
		\State Sample $x \sim \intf{}{}{\cdot}/D$
		\State $\lb{X} = \lb{X} \cup \left\{ \left( x, \alpha_{i} \right) \right\}$
	\EndFor
\end{algorithmic}
\hrule\hrule
\end{table}

The probability density of an LMB RFS is given by 
\be
	\lmb{\lb{X}} = \delta_{\lbsp\!\left( \left| \lb{X} \right| \right)}\!\left( \lbs{\lb{X}} \right) \, \cd{\left| \lb{X} \right|} \, \prod_{\lb{x} \in \lb{X}} \dfrac{\intf{}{}{x}}{D} \, ,\label{eq:liidpdf}
\ee

The same consideration drawn for the labeled Poisson RFS is inherited by the i.i.d. cluster RFS.

\subsubsection{Generalized Labeled Multi-Bernoulli RFS}
\label{ssec:glmbrfs}
A \textit{Generalized Labeled Multi-Bernoulli} (GLMB) RFS \cite{vovo1} is a labeled RFS with state space $\stsp$ and (discrete) label space $\lbsp$ distributed according to
\be
	\lmb{\lb{X}} = \dli{\lb{X}} \, \sum_{c \in \mathbb{C}} w^{\left( c \right)}\!\left( \lbs{\lb{X}} \right) \left[ \p{c}{} \right]^{\lb{X}}
\label{eq:glmbpdf}
\ee
where: $\mathbb{C}$ is a discrete index set; $w^{\left( c \right)}\!\left( L \right)$ and $\p{c}{}$ satisfy the normalization constraints:
\bie
	\sum_{L \subseteq \lbsp} \sum_{c \in \mathbb{C}} w^{\left( c \right)}\!\left( L \right) & = & 1 \, ,\label{eq:vinc:1}\\
	\int \p{c}{x, \ell} d x & = & 1 \, .
\eie
A GLMB can be interpreted as a mixture of $\left| \mathbb{C} \right|$ multi-object exponentials
\be
	w^{\left( c \right)}\!\left( \lbs{\lb{X}} \right) \left[ \p{c}{} \right]^{\lb{X}} \, .
\ee
Each term in the mixture (\ref{eq:glmbpdf}) consists of the product of two factors:
\begin{enumerate}
	\item a weight $w^{\left( c \right)}\!\left( \lbs{\lb{X}} \right)$ that only depends on the labels $\lbs{\lb{X}}$ of the multi-object state $\lb{X}$;
	\item a multi-object exponential $\left[ \p{c}{} \right]^{\lb{X}}$ that depends on the entire multi-object state.
\end{enumerate}
The cardinality distribution of a GLMB is given by
\be
	\cd{n} = \sum_{L \in \fs[n]{\lbsp}} \sum_{c \in \mathbb{C}} w^{\left( c \right)}\!\left( L \right) \, ,\label{eq:lrfscardinality}
\ee
from which, in fact, summing over all possible $n$ implies (\ref{eq:vinc:1}).
The PHD of the unlabeled version of a GLMB is given by
\be
	\intf{}{}{x} = \sum_{c \in \mathbb{C}} \sum_{\ell \in \lbsp} \p{c}{x, \ell} \sum_{L \subseteq \lbsp} \inc{L}{\ell} w^{\left( c \right)}\!\left( L \right) \, .
\label{eq:glmbphd}
\ee

\begin{rem}
	The GLMB RFS distribution has the appealing feature of having a 1-to-1 mapping between labels and location PDFs provided by the single-object densities $\p{c}{x, \ell}$, for each term of the multi-object exponential mixture indexed with $c$.
	However, in the context of MOT, it is not clear how to exploit (in particular how to implement) such a distribution \cite{vovo1,vovo2}.
\end{rem}

\subsubsection{$\delta$-Generalized Labeled Multi-Bernoulli RFS}
\label{ssec:dglmbrfs}
A $\delta$-Generalized Labeled Multi-Bernoulli ($\delta$-GLMB) RFS with state space $\stsp$ and (discrete) label space $\lbsp$ is a special case of a GLMB with
\bie
	\mathbb{C} & = & \fs{\lbsp} \times \Xi \, ,\\
	w^{\left( c \right)}\!\left( L \right) & = & w^{\left( I, \xi \right)}\!\left( L \right) = w^{\left( I, \xi \right)} \delta_{I}\!\left( L \right) \, ,\\
	\p{c}{\cdot} & = & \p{I, \xi}{\cdot} = \p{\xi}{\cdot} \, ,
\eie
i.e. is distributed according to
\bie
	\lmb{\lb{X}} & = & \dli{\lb{X}} \sum_{\left( I, \xi \right) \in \fs{\lbsp} \times \Xi} w^{\left( I, \xi \right)} \delta_{I}\!\left( \lbs{\lb{X}} \right) \left[ \p{\xi}{} \right]^{\lb{X}}\label{eq:dglmbpdf}\\
	& = & \dli{\lb{X}} \sum_{I \in \fs{\lbsp}} \delta_{I}\!\left( \lbs{\lb{X}} \right) \sum_{\xi \in \Xi} w^{\left( I, \xi \right)} \left[ \p{\xi}{} \right]^{\lb{X}} \, ,\label{eq:dglmbpdf2}
\eie
where $\Xi$ is a discrete space.
In MOT a $\delta$-GLMB can be used (and implemented) to represent the multi-object densities over time.
In particular:
\begin{itemize}
	\item each finite set $I \in \fs{\lbsp}$ represents a different configurations of labels whose cardinality provides the number of objects;
	\item each $\xi \in \Xi$ represents a history of association maps, e.g. $\xi = \left( \theta_{1}, \dots, \theta_{k} \right)$, where an association map at time $1 \le j \le k$ is a function $\theta_{j}$ which maps track labels at time $j$ to a measurement at time $j$ with the constraint that a track can generate at most one measurement, and a measurement can be assigned to at most one track;
	\item the pair $\left( I, \xi \right)$ is called \textit{hypothesis}.
\end{itemize}
To clarify the notion of $\delta$-GLMB, let us consider the following two examples.
\begin{description}
	\item[E1:] Suppose the following two possibilities
		\begin{enumerate}
			\item 0.4 chance of 1 object with label $\ell_{1}$, i.e. $I_{1} = \left\{ \ell_{1} \right\}$, and density $\p{}{x, \ell_{1}}$;
			\item 0.6 chance of 2 objects with, respectively, labels $\ell_{1}$ and $\ell_{2}$, i.e. $I_{2} = \left\{ \ell_{1}, \ell_{2} \right\}$, and densities $\p{}{x, \ell_{1}}$ and $\p{}{x, \ell_{2}}$.
		\end{enumerate}
		Then, the $\delta$-GLMB representation is:
		\be
			\lmb{\lb{X}} = 0.4 \, \delta_{\left\{ \ell_{1} \right\}}\!\left( \mathcal{L}\!\left( \lb{X} \right) \right) \, \p{}{x, \ell_{1}} +
					     0.6 \, \delta_{\left\{ \ell_{1}, \ell_{2} \right\}}\!\left( \mathcal{L}\!\left( \lb{X} \right) \right) \, \p{}{x, \ell_{1}} \, \p{}{x, \ell_{2}} \, .
		\label{eq:pdfex1}
		\ee
		Notice that in this example there are no association histories, i.e. $\Xi = \varnothing$. Thus (\ref{eq:pdfex1}) has exactly two hypotheses $\left( I_{1}, \varnothing \right)$ and $\left( I_{2}, \varnothing \right)$.
	\item[E2:] Let us now consider the previous example with $\Xi \neq \varnothing$:
		\begin{enumerate}
			\item 0.4 chance of 1 object with label $\ell_{1}$, i.e. $I_{1} = \left\{ \ell_{1} \right\}$, association histories $\xi_{1}$ and $\xi_{2}$, i.e. there are two hypotheses $\left( I_{1}, \xi_{1} \right)$ and $\left( I_{1}, \xi_{2} \right)$, weights $w^{\left( I_{1}, \xi_{1} \right)} = 0.3$ and $w^{\left( I_{1}, \xi_{2} \right)} = 0.1$, densities $\p{\xi_{1}}{x, \ell_{1}}$ and $\p{\xi_{2}}{x, \ell_{1}}$;
			\item 0.6 chance of 2 objects with, respectively, label $\ell_{1}$ and $\ell_{2}$, i.e. $I_{2} = \left\{ \ell_{1}, \ell_{2} \right\}$, association histories $\xi_{1}$ and $\xi_{2}$, i.e. there are two hypotheses $\left( I_{2}, \xi_{1} \right)$ and $\left( I_{2}, \xi_{2} \right)$, weights $w^{\left( I_{2}, \xi_{1} \right)} = 0.4$ and $w^{\left( I_{2}, \xi_{2} \right)} = 0.2$, densities $\p{\xi_{1}}{x, \ell_{1}}$, $\p{\xi_{1}}{x, \ell_{2}}$, $\p{\xi_{2}}{x, \ell_{1}}$ and $\p{\xi_{2}}{x, \ell_{2}}$.
		\end{enumerate}
		Then, the $\delta$-GLMB representation is:
		\bie
			\lmb{\lb{X}} & = & \delta_{\left\{ \ell_{1} \right\}}\!\left( \mathcal{L}\!\left( \lb{X} \right) \right) \,
						\left[ 0.3 \, \p{\xi_{1}}{x, \ell_{1}} + 0.1 \, \p{\xi_{2}}{x, \ell_{1}} \right] + \IEEEnonumber\\
					  && \delta_{\left\{ \ell_{1}, \ell_{2} \right\}}\!\left( \mathcal{L}\!\left( \lb{X} \right) \right) \,
					     	\left[ 0.4 \, \p{\xi_{1}}{x, \ell_{1}} \, \p{\xi_{1}}{x, \ell_{2}} + 0.2 \, \p{\xi_{2}}{x, \ell_{1}} \, \p{\xi_{2}}{x, \ell_{2}} \right]\, .
		\label{eq:pdfex2}
		\eie
		Thus (\ref{eq:pdfex2}) has exactly four hypotheses $\left( I_{1}, \xi_{1} \right)$, $\left( I_{1}, \xi_{2} \right)$, $\left( I_{2}, \xi_{1} \right)$ and $\left( I_{2}, \xi_{2} \right)$.
\end{description}
The weight $w^{\left( I, \xi \right)}$ represents the probability of hypothesis $\left( I, \xi \right)$ and $\p{\xi}{x, \ell}$ is the probability density of the kinematic state of track $\ell$ for the association map history $\xi$.

\subsubsection{Labeled Multi-Bernoulli RFS}
\label{ssec:lmbrfs}
A \textit{labeled multi-Bernoulli} (LMB) RFS $\mathbf{X}$ with state space $\mathbb{X}$, label space $\mathbb{L}$ and (finite) parameter set $\{(r^{(\ell)}, p^{(\ell)} ): \ell \in \mathbb{L} \}$, is a multi-Bernoulli RFS on $\mathbb{X}$ augmented with labels corresponding to the successful (non-empty) Bernoulli components.
In particular, $r^{(\ell)}$ is the \textit{existence probability} and $p^{(\ell)}$ is the PDF on the state space $\mathbb{X}$ of the \textit{Bernoulli component} $(r^{(\ell)}, p^{(\ell)})$ with unique label $\ell \in \mathbb{L}$.
The procedure in Table \ref{alg:samplelmbrfs} illustrates how a sample from a labeled multi-Bernoulli RFS is generated.

\begin{table}[!h]
\caption{Sampling a labeled multi-Bernoulli RFS}
\label{alg:samplelmbrfs}
\renewcommand{\arraystretch}{1.3}
\hrulefill\hrule
\begin{algorithmic}[0]
	\State $\lb{X} = \varnothing$
	\For{$\ell \in \mathbb{L}$}
		\State Sample $u \sim \operatorname{Uniform}_{\left[0, 1 \right]}$
		\If{$u \leq \ex{\ell}$}
			\State Sample $x \sim \p{\ell}{\cdot}$
			\State $\lb{X} = \lb{X} \cup \left\{ \left( x, \ell \right) \right\}$
		\EndIf
	\EndFor
\end{algorithmic}
\hrule\hrule
\end{table}

The probability density of an LMB RFS is given by 
\begin{equation}
	\boldsymbol{\pi}(\mathbf{X})=\Delta (\mathbf{X}) \, w(\mathcal{L}(\mathbf{X})) \, p^{\mathbf{X}}  \label{eq:lmbpdf}
\end{equation}
where

\begin{IEEEeqnarray}{rCl}
w(L) & = & \prod\limits_{\ell \in L} 1_{\mathbb{L}}(\ell) \, r^{(\ell )} \prod\limits_{\ell \in \mathbb{L} \backslash L}\left( 1-r^{(\ell)}\right) \, ,\label{eq:lmb:w}\\
\p{\ell}{x} & \triangleq & p(x,\ell ) \, .\label{eq:lmb:pdf}
\end{IEEEeqnarray}

\noindent For convenience, the shorthand notation $\boldsymbol{\pi} = \left \{ \left(r^{(\ell)}, p^{(\ell)} \right ) \right \}_{\ell \in \mathbb{L}}$ will be adopted for the density of an LMB RFS.
The LMB is also a special case of the GLMB \cite{vovo1,vovo2,lmbf} having $\mathbb{C}$ with a single element (thus the superscript is simply avoided) and
\bie
	\p{c}{x, \ell} & = & \p{}{x, \ell} = \p{\ell}{x} \, ,\\
	w^{\left( c \right)}\!\left( L \right) & = & w\!\left( L \right) \, .
\eie

\subsection{Bayesian multi-object tracking}
\label{ssec:RFSBayesTracking}
The labeled RFS paradigm \cite{vovo1,vovo2}, along with the mathematical tools provided by FISST \cite{mahler}, allows to formalize in a rigorous and elegant way the multi-object Bayesian recursion for MOT.

For MOT, the object label is an ordered pair of integers $\ell = \left( k, i \right)$, where $k$ is the \textit{time of birth} and $i\in \mathbb{N}$ is a unique index to distinguish objects born at the same time.
The label space for objects born at time $k$ is $\lbsp[k] = \left\{ k \right\} \times \nbb$.
An object born at time $k$ has state $\lb{x} \in \stsp \times \lbsp[k]$.
Hence, the label space for objects at time $k$ (including those born prior to $k$), denoted as $\lbsp[0:k]$, is constructed recursively by $\lbsp[0:k] = \lbsp[0:k-1] \cup \lbsp[k]$ (note that $\lbsp[0:k-1]$ and $\lbsp[k]$ are disjoint).
A multi-object state $\lb{X}$ at time $k$ is a finite subset of $\stsp \times \lbsp[0:k]$.

Suppose that, at time $k$, there are $N_{k}$ objects with states $\lb{x}_{k,1}, \dots, \lb{x}_{k, N_{k}}$, each taking values in the (labeled) state space $\stsp \times \lbsp[0:k]$, and $M_{k}$ measurements $y_{k,1}, \dots, y_{k, M_{k}}$ each taking values in an observation space $\ybb$. The \textit{multi-object state} and \textit{multi-object observation}, at time $k$, \cite{mahler1, mahler} are, respectively, the finite sets 
\bie
	\lb{X}_{k} & = & \left\{ \lb{x}_{k,1}, \dots, \lb{x}_{k,N_{k}} \right\} \, ,\\
	Y_{k} & = & \left\{ y_{k,1}, \dots,y_{k,M_{k}} \right\} \, .
\eie

Let $\clmb{k}{}{\cdot}$ denote the \textit{multi-object filtering density} at time $k$, and $\clmb{k|k-1}{}{\cdot}$ the \textit{multi-object prediction density} to time $k$ (formally $\clmb{k}{}{\cdot}$ and $\clmb{k|k-1}{}{\cdot}$ should be written respectively as $\clmb{k}{}{\cdot | Y_{0}, \dots, Y_{k-1}, Y_{k}}$, and $\clmb{k|k-1}{}{\cdot | Y_{0}, \dots, Y_{k-1}}$, but for simplicity the dependence on past measurements is omitted).
Then, the \textit{multi-object Bayes recursion} propagates $\clmb{k}{}{}$ in time \cite{mahler1,mahler} according to the following update and prediction
\bie
	\clmb{k|k-1}{}{\lb{X}} & = & \int \boldsymbol{\varphi}_{k|k-1}\!\left( \lb{X} | \lb{Z} \right) \clmb{k-1}{}{\lb{Z}} \delta \lb{Z} \, ,  \label{eq:MTBayesPred}\\
	\clmb{k}{}{\lb{X}} & = & \dfrac{g_{k}\!\left( Y_{k} | \lb{X} \right) \clmb{k|k-1}{}{\lb{X}}}{\displaystyle \int g_{k}\!\left( Y_{k} | \lb{Z} \right) \clmb{k|k-1}{}{\lb{Z}} \delta \lb{Z}} \, ,\label{eq:MTBayesUpdate}
\eie
where $\boldsymbol{\varphi}_{k|k-1}(\,\cdot \,|\,\cdot \,)$ is the \textit{labeled multi-object transition density} to time $k$, $g_{k}(\,\cdot\,|\,\cdot \,)$ is the \textit{multi-object likelihood function} at time $k$, and the integral is a \textit{set integral} defined, for any function 
$\boldsymbol{f} : \mathcal{F}\!\left( \mathbb{X} \times \mathbb{L} \right) \rightarrow \rbb$, by
\begin{equation}
	\int \boldsymbol{f}\!\left( \lb{X} \right) \delta \lb{X} = \sum_{n = 0}^{\infty} \frac{1}{n!} \sum_{\left( \ell _{1}, \dots, \ell_{n} \right) \in \lbsp} \int_{\lb{X}^{n}} \boldsymbol{f}\!\left( \left\{ \left( x_{1}, \ell_{1} \right), \dots, \left( x_{n}, \ell_{n} \right) \right\} \right) d x_{1} \cdots dx_{n} \, .\label{eq:setint}
\end{equation}
The multi-object posterior density captures all information on the number of objects, and their states \cite{mahler}.
The multi-object likelihood function encapsulates the underlying models for detections and false alarms while the multi-object transition density embeds the underlying models of motion, birth and death \cite{mahler}.

\begin{rem}
	Eqs. (\ref{eq:MTBayesPred})-(\ref{eq:MTBayesUpdate}) represent the multi-object counterpart of (\ref{eq:chapkol})-(\ref{eq:bayesrule}).
\end{rem}

Let us now consider a multi-sensor \textit{centralized} setting in which a sensor network $\left( \mathcal{N},\mathcal{A}\right)$ conveys all the measurement to a central fusion node. Assuming that the measurements taken by the sensors are independent, the multi-object Bayesian filtering recursion can be naturally extended as follows:
\bie
	\clmb{k|k-1}{}{\lb{X}} & = & \int \boldsymbol{\varphi}_{k|k-1}\!\left( \lb{X} | \lb{Z} \right) \clmb{k-1}{}{\lb{Z}} \delta \lb{Z} \, ,  \label{eq:msMTBayesPred}\\
	\clmb{k}{}{\lb{X}} & = & \dfrac{\displaystyle \prod_{i \in \ncal} g^{i}_{k}\!\left( Y^{i}_{k} | \lb{X} \right) \clmb{k|k-1}{}{\lb{X}}}{\displaystyle \int \displaystyle \prod_{i \in \ncal} g^{i}_{k}\!\left( Y^{i}_{k} | \lb{Z} \right) \clmb{k|k-1}{}{\lb{Z}} \delta \lb{Z}} \, ,\label{eq:msMTBayesUpdate}
\eie

\begin{rem}
	Eqs. (\ref{eq:msMTBayesPred})-(\ref{eq:msMTBayesUpdate}) represents the multi-object counterpart of (\ref{eq:mschapkol})-(\ref{eq:msbayesrule}) which has been made possible thanks to the concept of RFS densities.
\end{rem}

In the work \cite{vovo1}, the full labeled multi-object Bayesian recursion is derived for the general class of GLMB densities, while in \cite{vovo2} the analytical implementation of the multi-object Bayesian recursion for the $\delta$-GLMB density is provided. Moreover, in the work \cite{lmbf}, a multi-object Bayesian recursion for the LMB density, which turns out not to be in a closed form, is discussed and presented.

For convenience, hereinafter, explicit reference to the time index $k$ for label sets will be omitted by denoting $\lbsp[-] \triangleq \lbsp[0:k-1]$, $\nbsp \triangleq \lbsp[k]$, $\lbsp \triangleq \lbsp[-] \cup \nbsp$.

\section{Distributed information fusion}
\label{sec:distinfofusion}
The third and fourth core concepts of this dissertation are the \textit{Kullback-Leibler Average} (KLA) and \textit{consensus}, which are two mathematical tools for distributing information over sensor networks.

To combine limited information from individual nodes, a suitable information fusion procedure is required to reconstruct, from the information of the various local nodes, the state of the objects present in the surrounding environment.
The scalability requirement, the lack of a fusion center and knowledge on the network topology (see section \ref{sec:net}) dictate the adoption of a consensus approach to achieve a collective fusion over the network by iterating local fusion steps among neighboring nodes \cite{Olfati,Xiao,Calafiore,cp}.
In addition, due to the data incest problem in the presence of network loops that causes \textit{double counting} of information, robust (but suboptimal) fusion rules, such as the \textit{Chernoff fusion} rule \cite{info,mori1} (that includes \textit{Covariance Intersection} \cite{juluhl1997,julier2008} and its generalization \cite{mah2000}) are required.

In this chapter, the KLA and consensus will be separately dealt with, respectively, as a robust suboptimal fusion rule and as a technique to spread information, in a scalable way, throughout sensor networks.

\subsection{Notation}
\label{ssec:difnotation}

Given PDFs $p,$ $q$ and a scalar $\alpha >0$, the information fusion $\oplus$ and weighting operators $\odot $ \cite{ccphd, cp, cpcl} are defined as follows 
\begin{IEEEeqnarray}{rCl}
	\left(p \oplus q \right)(x) & \triangleq & \dfrac{p(x) \, q(x)}{\left< p, q \right>} \label{oplus}\\
	\left( \alpha \odot p \right)(x) & \triangleq & \dfrac{\left[ p(x) \right]^{\alpha}}{\left< p^{\alpha}, 1 \right>} \label{odot}.
\end{IEEEeqnarray}
It can be checked that the fusion and weighting operators satisfy the following properties: 
\begin{IEEEeqnarray*}{ll}
	\mbox{\textsc{p.a}\quad} & (p \oplus q) \oplus h = p \oplus (q \oplus h) = p \oplus q \oplus h \\
	\mbox{\textsc{p.b}\quad} & p \oplus q = q \oplus p \\
	\mbox{\textsc{p.c}\quad} & (\alpha \, \beta) \odot p = \alpha \odot (\beta \odot p) \\
	\mbox{\textsc{p.d}\quad} & 1 \odot p = p \\
	\mbox{\textsc{p.e}\quad} & \alpha \odot (p \oplus q) = (\alpha \odot p) \oplus (\alpha \odot q) \\
	\mbox{\textsc{p.f}\quad} & (\alpha + \beta) \odot p = (\alpha \odot p) \oplus (\beta \odot q)
\end{IEEEeqnarray*}
for any PDFs $p$ and $q$, positive scalars $\alpha $ and $%
\beta $.

\subsection{Kullback-Leibler average of PDFs}
\label{ssec:klapdf}

A key ingredient for networked estimation is the capability to fuse in a consistent way PDFs of the quantity to be estimated provided by different nodes.
In this respect, a sensible information-theoretic definition of fusion among PDFs is the \textit{Kullback-Leibler Average} (KLA) \cite{batchi2011,cp} relying on the \textit{Kullback-Leibler Divergence} (KLD).
Given PDFs $\left\{ p^{i}\!\left( \cdot \right) \right\}_{i \in \ncal}$ and relative weights $\omega^{i} > 0$ such that $\sum_{i \in \ncal} \omega^{i} = 1$, their weighted KLA $\overline{p}\!\left( \cdot \right)$ is defined as
\be
	\overline{p}(\cdot) = \arg \inf_{p\!\left( \cdot \right)} \displaystyle{\sum_{i \in \ncal}} \omega^{i} D_{KL} \left( p \parallel p^{i} \right)
\label{eq:kla}
\ee
where
\be
	D_{KL}\!\left( p \parallel p^{i} \right) = \displaystyle{\int} p\!\left( x \right) \operatorname{log}\left(\dfrac{p\!\left( x \right)}{p^{i}\!\left( x \right)}\right) dx
\ee
denotes the KLD between the PDFs $p\left( \cdot \right)$ and $p^{i}\!\left( \cdot \right)$.
In \cite{cp} it is shown that the weighted KLA in (\ref{eq:kla}) coincides with the \textit{normalized weighted geometric mean} (NWGM) of the PDFs, i.e.
\be
	\overline{p}(x) = \dfrac{\displaystyle\prod_{i \in \ncal} \left[ p^{i}\!\left( x \right) \right]^{\omega^{i}}}
                                {\displaystyle \int \prod_{i \in \ncal} \left[ p^{i}\!\left( x \right) \right]^{\omega^{i}} d x}
                                \, \triangleq \, \displaystyle{\bigoplus_{i \in \ncal}} \, \left( \omega^{i} \odot p^{i}\!\left( x \right) \right)
\label{eq:geomean}
\ee
where the latter equality follows from the properties of the operators $\odot$ and $\oplus$.
Note that in the unweighted KLA $\omega^{i} = 1 / \left| \mathcal{N} \right|$, i.e.
\begin{equation}
	\overline{p}\left( x \right) = \displaystyle{\bigoplus_{i \in \ncal}} \,\left( \dfrac{1}{\left| \ncal \right|} \odot p^{i} \right) \! \left( x \right) \, .
\label{eq:uwgeomean}
\end{equation}
If all PDFs are Gaussian, i.e. $p^{i}\!\left( \cdot \right) = \ncal\!\left( \, \cdot \, ; \, m^{i}, P^{i} \right)$, $\overline{p}\!\left( \cdot \right)$ in (\ref{eq:geomean}) turns out to be Gaussian \cite{cp}, i.e. $\overline{p}\!\left( \cdot \right) = \ncal\!\left( \, \cdot \, ; \, \overline{m}, \overline{P} \right)$.
In particular, defining the \textit{information} (inverse covariance) \textit{matrix}
\be
	\Omega \triangleq P^{-1}
\ee
and information vector
\be
	q \triangleq P^{-1} m
\ee
associated to the Gaussian mean $m$ and covariance $P$, one has
\begin{IEEEeqnarray}{rCl}
	\overline{\Omega} & = & \displaystyle{\sum_{i \in \mathcal{N}}}~\omega^{i} \Omega^{i} \, ,\label{eq:ci1}\\
	\overline{q} & = & \displaystyle{\sum_{i \in \mathcal{N}}}~\omega^{i} q^{i} \, . \label{eq:ci2}
\end{IEEEeqnarray}
which corresponds to the well known \textit{Covariance Intersection} (CI) fusion rule \cite{juluhl1997}.
Hence the KLA of Gaussian PDFs is a Gaussian PDF whose information pair $(\overline{\Omega},\overline{q})$ is obtained by the weighted arithmetic mean of the information pairs $( \Omega^{i},q^{i} )$ of the averaged Gaussian PDFs.

\subsection{Consensus algorithms}
\label{ssec:consensus}
In recent years, \textit{consensus algorithms} have emerged as a powerful tool for distributed computation over networks \cite{Olfati,Xiao} and have been widely used in distributed parameter/state estimation algorithms \cite{OlfatiInnovation,Tomlin,Sayed,Chiuso,Calafiore,Stankovic,batchi2011,Farina,li-jia,clike1,ccphd,cp,clike2,cpcl,batchifan2014,fanbatchi2012,batchifan2013eusipco,batchifan2013irs,batchifan2013,ccjpda}. In its basic form, a consensus algorithm can be seen as a technique for distributed averaging over a network; each agent aims to compute the collective average of a given quantity by iterative regional averages, where the terms ``\textit{collective}'' and ``\textit{regional}'' mean ``\textit{over all network nodes}'' and, respectively, ``\textit{over neighboring nodes only}''. 
Consensus can be exploited to develop scalable and reliable distributed fusion techniques.

To this end, let us briefly introduce a prototypal average consensus problem.
Let node $i \in \mathcal{N}$ be provided with an estimate ${\hat{\theta}}^{i}$ of a given quantity of interest ${\theta}$.
The objective is to develop an algorithm that computes in a distributed way, in each node, the average
\be
	\bar{\theta} = \dfrac{1}{\left| \ncal \right|} \displaystyle{\sum_{i \in \ncal}} \hat{\theta}^{i} \, .
\label{eq:distavg}
\ee
To this end, let ${\hat{\theta}}^{i}_{0}  = {\hat{\theta}}^{i}$, then a simple consensus algorithm takes the following iterative form:
\be
	\hat{\theta}^{i}_{l+1} = \displaystyle{\sum_{j \in \ncal^{i}}} \omega^{i, j} \, \hat{\theta}^{j}_{l} \, , \qquad \forall i \in \ncal
\label{eq:consensus}
\ee
where the \textit{consensus weights} must satisfy the conditions
\bie
	\displaystyle{\sum_{j \in \mathcal{N}^{i}}} \omega^{i,j} & = & 1 \, , \qquad \forall i \in \mathcal{N}\\
	\omega^{i,j} & \geq & 0 \, , \qquad \forall i,j \in \mathcal{N} \, .
\label{eq:consensusconstraints}
\eie
Notice from (\ref{eq:consensus})-(\ref{eq:consensusconstraints}) that at a given consensus step the estimate in any node is computed as a convex combination of the estimates of the neighbors at the previous consensus step. In other words, the iteration (\ref{eq:consensus}) is simply a regional average computed in node $i$, the objective of consensus being convergence of such regional averages to the collective average (\ref{eq:distavg}). Important convergence properties, depending on the consensus weights, can be found in \cite{Olfati,Xiao}. For instance, let us denote by $\Pi$ the consensus matrix whose generic $(i,j)$-element coincides with the consensus weight $\omega^{i,j}$ (if $j \notin \mathcal N^{i}$ then $\omega^{i,j}$ is taken as $0$).
Then, if the consensus matrix $\Pi$ is primitive and doubly stochastic\footnote{A non-negative square matrix $\Pi$ is doubly stochastic if all its rows and columns sum up to $1$. Further, it is primitive if there exists an integer $m$ such that all the elements of $\Pi^m$ are strictly positive.}, the consensus algorithm (\ref{eq:consensuspdf}) asymptotically yields the average (\ref{eq:distavg}) in that
\be
 	\lim_{l \rightarrow \infty} \hat{\theta}^{i} = \bar{\theta} \, , \qquad \forall i \in \ncal \, .
\ee

A necessary condition for the matrix $\Pi$ to be primitive is that the graph $\mathcal{G}$ associated with the sensor network be strongly connected \cite{Calafiore}. Moreover, in the case of an undirected graph $\mathcal{G}$, a possible choice ensuring convergence to the collective average is given by the so-called \textit{Metropolis weights} \cite{Xiao,Calafiore}.
\bie
	\omega^{i, j} & = & \frac{1}{1 + \operatorname{max}\!\left\{ \left| \ncal^{i} \right| , \left| \ncal^{j} \right| \right\}} \, , \qquad i \in \ncal \, , j \in \ncal^{i} \, , i \ne j \\
	\omega^{i, i} & = & 1 - \sum_{j \in \ncal^{i} , \, j \ne i} \omega^{i,j} \, .
\eie

\subsection{Consensus on posteriors}
\label{ssec:cpdfs}
The idea is to exploit consensus to reach a \textit{collective agreement} (over the entire network), by having each node iteratively updating and passing its local information to neighbouring nodes. Such repeated local operations provide a mechanism for propagating information throughout the whole network.
The idea is to perform, at each time instant and in each node of the network $i \in \ncal$, a local correction step to find the local posterior PDF followed by consensus on such posterior PDFs to determine the collective unweighted KLA of the posterior densities $p_{k}^{i}$.

Suppose that at time $k$, each agent $i$ starts with the posterior $p_{k}^{i}$ as the initial iterate $p_{k,0}^{i}$, and computes the $l$-th consensus iterate by 
\begin{equation}
	p_{k, l}^{i}=\displaystyle{\bigoplus_{j\in \mathcal{N}^{i}}}\,\left( \omega^{i,j}\odot p_{k, l-1}^{j}\right)
	\label{eq:consensuspdf}
\end{equation}
where $\omega ^{i,j}\geq 0$, satisfying $\sum_{j\in \mathcal{N}^{i}}~\omega^{i,j}=1$, are the consensus weights relating agent $i$ to nodes $j\in \mathcal{N}^{i}$. Then, using the properties of the operators $\oplus $ and $\odot $, it can be shown that \cite{cp}
\begin{equation}
	p_{k,l}^{i}=\displaystyle{\bigoplus_{j\in \mathcal{N}}}\,\left( \omega_{l}^{i,j}\odot p_{k}^{j}\right)
\end{equation}
where $\omega _{l}^{i,j}$ is the $(i,j)$-th entry of $\Pi ^{n}$, with $\Pi$ the (square) consensus matrix with $(i,j)$-th entry $\omega ^{i,j}1_{\mathcal{N}^{i}}(j)$ (it is understood that $p_{k}^{j}$ is omitted from the fusion whenever $\omega _{l}^{i,j}=0$). More importantly, it was shown in \cite{Olfati, Xiao} that if the consensus matrix $\Pi $ is primitive, (i.e. non-negative, and there exists an integer $m$ such that $\Pi ^{m}$ is positive) and doubly stochastic (all rows and columns sum to 1), then for any $i,j\in \mathcal{N}$, one has
\be
	\lim_{l \rightarrow \infty} \omega_{l}^{i, j} = \dfrac{1}{\left| \ncal \right|}.
\label{eq:consweightlimit}
\ee
In other words, at time $k$, if the consensus matrix is primitive then the consensus iterate of each node in the network ``tends'' to the collective unweighted KLA (\ref{eq:uwgeomean}) of the posterior densities \cite{Calafiore,cp}.

%To summarize, each node $i \in \mathcal{N}$, at a given time $k$, will have to perform the algorithm of Table \ref{alg:sodistest}.
%\begin{table}[!h]
%	\caption{Distributed SOT via Consensus (CSOF) pseudo-code}
%	\label{alg:sodistest}
%	\renewcommand{\arraystretch}{1.3}
%	\hrulefill
%	\hrule
%	\begin{algorithmic}[0]
%		\Procedure{CSOF}{\textsc{Node} $k$, \textsc{Time} $t$}
%			\State \textsc{Local Prediction} \Comment{See eq. (\ref{eq:chapkol})}
%			\State \textsc{Local Update} \Comment{See eq. (\ref{eq:bayesrule})}
%			\State
%			\For{$n = 1, \dots, N$}
%				\State \textsc{Information Exchange}
%				\State \textsc{Fusion} \Comment{See eq. (\ref{eq:consensuspdf})}
%			\EndFor
%			\State
%		\EndProcedure
%	\end{algorithmic}
%	\hrule\hrule
%\end{table}
%
%\subsection{Consensus on posteriors}
The \textit{Consensus on Posteriors} (CP) approach to \textit{Distributed SOF} (DSOF) is summarized by the algorithm of Table \ref{alg:cp} to be carried out at each sampling interval $k$ in each node $i \in \mathcal{N}$.
In the linear Gaussian case, the CP algorithm involves a Kalman filter for local prediction and correction steps and CI for the consensus steps.
In this case, it has been been proved \cite{batchi2011,cp} that, under suitable assumptions, the CP algorithm guarantees mean-squared bounded estimation error in all network nodes for any number $L \geq 1$ of consensus iterations.
In the nonlinear and/or non Gaussian case, the PDFs in the various steps can be approximated as Gaussian and the corresponding means-covariances updated via, e.g., EKF \cite{EKF} or UKF \cite{juluhl2004}.
\begin{rem}
Whenever the number of consensus steps $L$ tends to infinity, consensus on posteriors is unable to recover the solution of the Bayesian filtering problem.
In fact, in the CP approach, the novel information undergoing consensus combined with the prior information, is unavoidably underweighted.
\end{rem}
\begin{table}[!h]
\caption{Consensus on Posteriors (CP)}
\label{alg:cp}
\centering
\hrulefill\hrule
\begin{algorithmic}[0]
	\Procedure{CP}{\textsc{Node} $i$, \textsc{Time} $k$}
		\State \bfbox{Prediction}
		\State $\np[k|k-1]{i}{x_{k}} = \displaystyle \int \varphi_{k|k-1}\!\left( x_{k} | x_{k-1} \right) \np[k-1]{i}{x_{k-1}} d x_{k-1}$\vspace{0.5em}
		\State \bfbox{Local Correction}
		\State $\np[k]{i}{x_{k} | y^{i}_{k}} = \left\{ \ba{ll}
								\left( g_{k}^{i}\!\left( y^{i}_{k} | \cdot \right) \oplus \np[k|k-1]{i}{\cdot} \right) \left( x \right) \, , & i \in \mathcal{S}\\
								\np[k|k-1]{i}{x_{k}} \, , & i \in \mathcal{C}
							\ea \right.$\vspace{0.5em}
		\State \bfbox{Consensus}
		\State $\np[k,0]{i}{x_{k}} = \np[k]{i}{x_{k} | y^{i}_{k}}$
		\For{$l = 1, \dots, L$}
			\State $\np[k, l]{i}{x_{k}} = \displaystyle{\bigoplus_{j \in \ncal^{i}}} \left( \omega^{i,j} \odot \np[k,l-1]{j}{x} \right)$
		\EndFor
		\State $\np[k]{i}{x_{k}} = \np[k, L]{i}{x_{k}}$
	\EndProcedure
\end{algorithmic}
\hrule\hrule
\end{table}

%The heart of present dissertation is to study, develop and present new distributed MOT algorithms whose structure is a natural extension of the one in Table \ref{alg:sodistest}.
%To this end, the \textit{four core concepts} introduced in sections \ref{sec:bayesest}, \ref{sec:rfs} and \ref{sec:distinfofusion}, i.e. the BF, the RFS approach, the KLA and the Consensus, will be exploited.

% DISTRIBUTED SINGLE-OBJECT FILTERING
\chapter{Distributed single-object filtering}
\label{chap:dsof}
\spminitoc
This chapter presents two applications of distributed information fusion for single-object filtering \cite{cpcl,fanbatchi2012,batchifan2014}.
The first application presents an improvement of the CP approach, described in subsection \ref{ssec:cpdfs}, and is applied to track a single non-maneuvering object.
The second application takes into account the possibility of the object of being highly maneuvering.
Thus, a multiple-model filtering approach is used to devise a consensus-based algorithm capable of tracking a single highly-maneuvering object.
The effectiveness of the proposed algorithms is demonstrated via simulation experiments on realistic scenarios.

\section{Consensus-based distributed filtering}
The approach proposed in this section is based on the idea of carrying out, in parallel, a separate consensus for the novel information (likelihoods) and one for the prior information (priors).
This parallel procedure is conceived as an improvement of the CP approach to avoid underweighting the novel information during the fusion steps.
The outcomes of the two consensuses are then combined to provide the fused posterior density.

\subsection{Parallel consensus on likelihoods and priors}
The proposed parallel \textit{Consensus on Likelihoods and Priors} (CLCP) approach to DSOF is summarized by the algorithm of Table \ref{CLCP-Algorithm} to be carried out at each sampling interval $k$ in each node $i \in \mathcal{N}$.
\begin{table}[h!]
\caption{Consensus on Likelihoods and Priors (CLCP) pseudo-code}
\label{CLCP-Algorithm}
\centering
\hrulefill\hrule
\begin{algorithmic}[0]
	\Procedure{CLCP}{\textsc{Node} $i$, \textsc{Time} $k$}
		\State \bfbox{Prediction}
		\State $\np[k|k-1]{i}{x} = \displaystyle \int \varphi_{k|k-1}\!\left( x | \zeta \right) \np[k-1]{i}{x} d \zeta$\vspace{0.5em}
		\State \bfbox{Consensus}
		\State $g_{k, 0}^{i}\!\left( x \right) = \left\{ \ba{ll}
											g_{k}^{i}\!\left( y^{i}_{k} | x \right) \, , & i \in \mathcal{S} \\
											1 \, , & i \in \mathcal{C}
									\ea \right.$
 		\State $\np[k|k-1, 0]{i}{x} = \np[k|k-1]{i}{x}$\vspace{0.5em}
		\For{${l}=1, \dots, L$ }
			\State $g_{k, {l}}^{i}\!\left( x \right) = \displaystyle{\bigoplus_{j \in \mathcal{N}^{i}}} \left( \omega^{i,j} \odot g_{k,{l}-1}^{j}\!\left( x \right)  \right)$
			\State $\np[k|k-1, {l}]{i}{x} = \displaystyle{\bigoplus_{j \in \mathcal{N}^{i}}} \left( \omega^{i,j} \odot \np[k|k-1, {l}-1]{j}{x}  \right)$
		\EndFor\vspace{0.5em}
		\State \bfbox{Correction}
		\State $\np[k]{i}{x} = \np[k|k-1, L]{i}{x} \oplus \left( \rho_{k}^{i} \odot g_{k, L}^{i}\left( x \right) \right)$
	\EndProcedure
\end{algorithmic}
\hrule\hrule
\end{table}
Notice that in the correction step, a suitable positive weight $\rho_{k}^{i}\ge1$ is applied to the outcome of the consensus on likelihoods, in order to possibly counteract the underweighting of novel information. To elaborate more on this issue, observe that each local posterior PDF $p_{k}^{i}(x)$ resulting from application of the CLCP algorithm
turns out to be equal to
\be
	p_{k}^{i}(x) = \left ( \bigoplus_{j \in \mathcal{N}} \left( \omega_{L}^{i,j} \odot p^{j}_{k|k-1}(x) \right) \right ) \oplus \left ( \bigoplus_{j \in \mathcal{S}} \left( \rho_{k}^{i} \omega_{L}^{i,j} \odot g_{k}^{i}\!\left( y^{i}_{k} | x \right) \right) \right ) \, .
\ee
As it can be seen,  $p_{k}^{i}(x)$ is composed of two parts: \textsc{i}) a weighted geometric mean of the priors $p^{j}_{k|k-1}(x)$; \textsc{ii}) a weighted combination of the likelihoods $g_{k}^{i}\!\left( y^{i}_{k} | x \right)$. The fact that in the first part the weights
$\omega_{L}^{i,j} $ sum up to one ensures that no double counting of the common information contained in the priors can occur \cite{julier2008,BaJuAg}. 
As for the second part,
care must be taken in the choice of the scalar weights $\rho_{k}^{i}$. For instance, in order to avoid overweighting some of the likelihoods, it is important that, for any pair $i,j$, one has $\rho_{k}^{i} \omega_{L}^{i,j} \le 1 $. In this way, each component of independent information is replaced with a conservative approximation, thus ensuring the conservativeness of the overall fusion rule \cite{BaJuAg}.
On the other hand, the product $\rho_{k}^{i} \omega_{L}^{i,j}$ should not be too small in order to avoid excessive underweighting.
Based on these considerations and recalling the consensus property (\ref{eq:consweightlimit}), different strategies for choosing  $\rho_{k}^{i}$ can be devised.

A reasonable choice would amount to letting 
\be
	\rho_{k}^{i} = \displaystyle \min_{j \in \mathcal{S}} \left( \dfrac{1}{\omega_{L}^{i,j}} \right)
\ee
whenever at least one of the weights $\omega_{L}^{i,j}$ is different from zero. Such a choice would assign to $\rho_{k}^{i}$ the closest value to the true number of agents $\left| \ncal \right|$ of the network.
This strategy is easily applicable when $L=1$ but, unfortunately, for $L>1$ requires that each node of the network can compute $\Pi^L$.
In most settings, this is not possible and alternative solutions must be adopted.

For example, when $L\!\gg\!1$ and the matrix $\Pi$ is primitive and doubly stochastic, one has that $\omega_{L}^{i,j} \approx 1/|\mathcal N|$ for any $i,j$. Then, in this case, one can let
\be
	\rho_{k}^{i} = | \mathcal{N} | \, .
\ee
Such a choice has the appealing feature of giving rise to a distributed algorithm converging to the centralized one as $L$ tends to infinity. However, when only a moderate number of consensus steps is performed, it leads to an overweighting of some likelihood components. Furthermore, the number of nodes $| \mathcal{N} |$ might be unknown, in particular for a time-varying network. Notice that when such a choice is adopted,
due to the properties of the fusion $\oplus$ and weighting $\odot$ operators, the CLCP algorithm can be shown to be mathematically equivalent (in the sense
that they would generate the same $p_{k}^{i}(x)$) to the optimal distributed protocol of \cite{OrCo}.\\
\indent An alternative solution is to exploit consensus so as to compute, in a distributed way, a normalization factor 
to improve the filter performance while preserving consistency of each local filter.
For example, an estimate of the fraction $| \mathcal{S} | / | \mathcal{N} |$ of sensor nodes
in the network can be computed via the consensus algorithm\\
\be
	b_{k,{l}}^{i} = \displaystyle{\sum_{j \in \mathcal{N}^{i}}}~ \omega^{i,j} ~b_{k,{l}-1}^j, \qquad {l} = 1,\ldots,L
\label{eq:b}
\ee
with the initialization $b_{k,0}^{i} = 1$ if $i \in \mathcal S$, and $b_{k,0}^{i} = 0$ otherwise.
Then, the choice
\begin{equation}
 \rho^{i}_{k} = \left \{ \begin{array}{ll}
 					1 \, ,					& \mbox{ if } b_{k,L}^{i} = 0 \\
					\dfrac{1}{b_{k,L}^{i}} \, ,	& \mbox{ otherwise }
				\end{array}\right.
\label{eq:omega}
\end{equation}
has the desirable property of ensuring that $\rho_{k}^{i} \omega_{L}^{i,j} \le 1$ for any $i,j$. Another positive feature of (\ref{eq:omega}) is that no a priori knowledge of the network is required.

\subsection{Approximate CLCP}
It must be pointed out that unfortunately the CLCP recursion of Table \ref{CLCP-Algorithm} does not admit an exact analytical solution except for the
linear Gaussian case.
Henceforth it will be assumed that the noises are Gaussian, i.e. $w_{k} \sim \mathcal{N}\!\left( 0, Q_{k} \right)$
and $v^{i}_{k} \sim \mathcal{N}\!\left( 0, R_{k}^{i} \right)$, but that the system, i.e. $f_{k}(\cdot)$ and/or $h_{k}^{i}(\cdot)$, might be nonlinear.
The main issue is how to deal with nonlinear sensors in the consensus on likelihoods as for the other tasks (i.e. consensus on priors, correction, prediction) it is well known how to handle nonlinearities exploiting, e.g., EKF \cite{EKF} or UKF \cite{juluhl2004} or particle filters \cite{PF}.
Under the Gaussian assumption, the local likelihoods take the form
\be
	g_{k}^{i}(x) = g_{k}^{i}\!\left( y^{i}_{k} | x \right) = \mathcal{N}\!\left( y_{k}^{i} - h_{k}^{i}(x); 0, R_{k}^{i} \right)
\ee
Whenever sensor $i$ is linear, i.e. $h_{k}^{i}\left( x \right) = C_{k}^{i} x$,
\be
	g_{k}^{i}(x) \propto e^{- \frac{1}{2} \left( x^\top \delta \Omega_{k}^{i} x - 2 x^\top \delta q_{k}^{i} \right)}
\ee
where
\be
\ba{c}
 \delta \Omega_{k}^{i} \triangleq \left( C_{k}^{i} \right)^\top \left( R_{k}^{i} \right)^{-1} C_{k}^{i} \, ,\\
 \delta q_{k}^{i} \triangleq \left( C_{k}^{i} \right)^\top \left( R_{k}^{i} \right)^{-1} y_{k}^{i} \, ,
\ea
\label{CL}
\ee
Then it is clear that multiplying, or exponentiating by suitable weights, the likelihoods $g_{k}^{i}(\cdot)$ is equivalent to adding, or multiplying by such weights, the corresponding $\delta\Omega_{k}^{i}$ and $\delta q_{k}^{i}$ defined in (\ref{CL}).
Then, in the case of linear Gaussian sensors, the consensus on likelihoods reduces to a consensus on the information pairs defined in (\ref{CL}).
For nonlinear sensors, a sensible approach seems therefore to approximate the nonlinear measurement function $h_{k}^{i}(\cdot)$ by a linear affine one, i.e.
\be
 h_{k}^{i} \left( x \right) \cong C_{k}^{i} \left( x - \hat{x}_{k|k-1}^{i} \right) + h^{i}\!\left(\hat{x}_{k|k-1}\right)
\label{APP}
\ee
and then replace the measurement equation (\ref{eq:mmodel}) with $\overline{y}_{k}^{i} = C_{k}^{i} x_{k} + v_{k}^{i}$ for a suitably defined pseudo-measurement
\be
\overline{y}_{k}^{i} ~=~ y_{k}^{i} - \hat{y}_{k|k-1}^{i} + C_{k}^{i} \hat{x}_{k|k-1}^{i}
\label{PM}
\ee
The sensor linearization (\ref{APP}) can be carried out for example by following the EKF paradigm.
As an alternative, exploiting the unscented transform and the unscented information filter \cite{UIF}, from $\hat{x}_{k|k-1}^{i}$ and the relative covariance $P_{k|k-1}^{i}$, one can get the
$\sigma$-points $\hat{x}^{i,j}_{k|k-1}$ and relative weights $\omega_j$ for $j=0,1,\dots,2n$ ($n = \operatorname{dim}\!\left( x \right)$), and thus compute
\bie
	\hat{y}_{k|k-1}^{i} & = & \displaystyle{\sum_{j=0}^{2n}}~ \omega_j~ \hat{y}_{k|k-1}^{i,j},~ \hat{y}_{k|k-1}^{i,j} = h_{k}^{i} \left( \hat{x}_{k|k-1}^{i,j} \right) \\
	P^{yx,i}_{k|k-1} & = & \displaystyle{\sum_{j=0}^{2n}}~ \omega_j ~ \left( \hat{y}_{k|k-1}^{i} - \hat{y}_{k|k-1}^{i,j} \right) \left( \hat{x}_{k|k-1}^{i} - \hat{x}_{k|k-1}^{i,j} \right)^\top \\
	C_{k}^{i} & = & P_{k|k-1}^{yx,i}~ \left( P_{k|k-1}^{i} \right)^{-1}
\eie
Summarizing the above derivations, the proposed consensus-based distributed nonlinear filter is detailed in Table \ref{alg:aclcp}.

The reason for performing in parallel consensus on priors and likelihoods is the counteraction of underweighting of the novel information contained in the likelihoods by using the weighting factor $\rho_{k}^{i}$.
Clearly, such a choice requires transmitting, at each consensus step, approximately twice the number of floating-point data with respect to other algorithms like CP.
Notice that, depending on the packet size, this need not necessarily increase the packet rate.

\begin{table}[h!]
\caption{Analytical implementation of Consensus on Likelihoods and Priors (CLCP)}
\label{alg:aclcp}
\centering
\hrulefill\hrule
\begin{algorithmic}[0]
\Procedure{CLCP}{\textsc{Node} $i$, \textsc{Time} $t$}
	\State \bfbox{Consensus}
	\State If $i \in \mathcal{S}$: $\left\{ \ba{c}
								\delta \Omega_{k,0}^{i} = \left( C_{k}^{i} \right)^\top \left( R_{k}^{i} \right)^{-1} C_{k}^{i}\\
								\delta q_{k,0}^{i} = \left( C_{k}^{i} \right)^\top \left( R_{k}^{i} \right)^{-1} \overline{y}_{k}^{i}
							     \ea \right.$
	\State If $i \in \mathcal{C}$: $\left\{ \ba{c}
								\delta \Omega_{k,0}^{i} = 0\\
								\delta q_{k,0}^{i} = 0
								\ea \right.$
	\State $\Omega_{k|k-1,0}^{i} = \left( P_{k|k-1}^{i} \right)^{-1}$
	\State $q_{k|k-1,0}^{i} = \left( P_{k|k-1}^{i} \right)^{-1} \hat{x}_{k|k-1}^{i}$
	\For{${l}=1, \dots, L$}
		\State \bfbox{Likelihood}
		\State $\delta \Omega_{k,{l}}^{i} = \displaystyle{\sum_{j \in \mathcal{N}^{i}}} \omega^{i,j} \delta \Omega_{k,{l}-1}^j$
		\State $\delta q_{k,{l}}^{i} = \displaystyle{\sum_{j \in \mathcal{N}^{i}}} \omega^{i,j} \delta q_{k,{l}-1}^j$
		\State \bfbox{Prior}
		\State $\Omega_{k|k-1,{l}}^{i} = \displaystyle{\sum_{j \in \mathcal{N}^{i}}} \omega^{i,j} \Omega_{k|k-1,{l}-1}^j$
		\State $q_{k|k-1,{l}}^{i} =\displaystyle{\sum_{j \in \mathcal{N}^{i}}} \omega^{i,j} q_{k|k-1,{l}-1}^j$
	\EndFor\vspace{0.5em}
	\State \bfbox{Correction}
	\State $\Omega_{k}^{i} = \Omega_{k|k-1,L}^{i} ~+~ \rho_{k}^{i} ~\delta \Omega_{k,L}^{i}$ 
	\State $q_{k}^{i} = q_{k|k-1,L}^{i} ~+ ~\rho_{k}^{i} ~\delta q_{k,L}^{i}$\vspace{0.5em}
	\State \bfbox{Prediction}
	\State $P_{k}^{i} = \left( \Omega_{k}^{i} \right)^{-1}$, $\hat{x}_{k}^{i} = \left( \Omega_{k}^{i} \right)^{-1} q_{k}^{i}$
	\State \textbf{from} $\hat{x}_{k}^{i}, P_{k}^{i}$ \textbf{compute} $\hat{x}_{k+1|k}^{i}, P_{k+1|k}^{i}$ \textbf{via UKF}
\EndProcedure
\end{algorithmic}
\hrule\hrule
\end{table}

\subsection{A tracking case-study}
To evaluate performance of the proposed approach, the networked object tracking scenario of fig. \ref{fig:objectnetwork} has been simulated.
\begin{figure}[h!]
	\centering
	\includegraphics[width=0.85\textwidth]{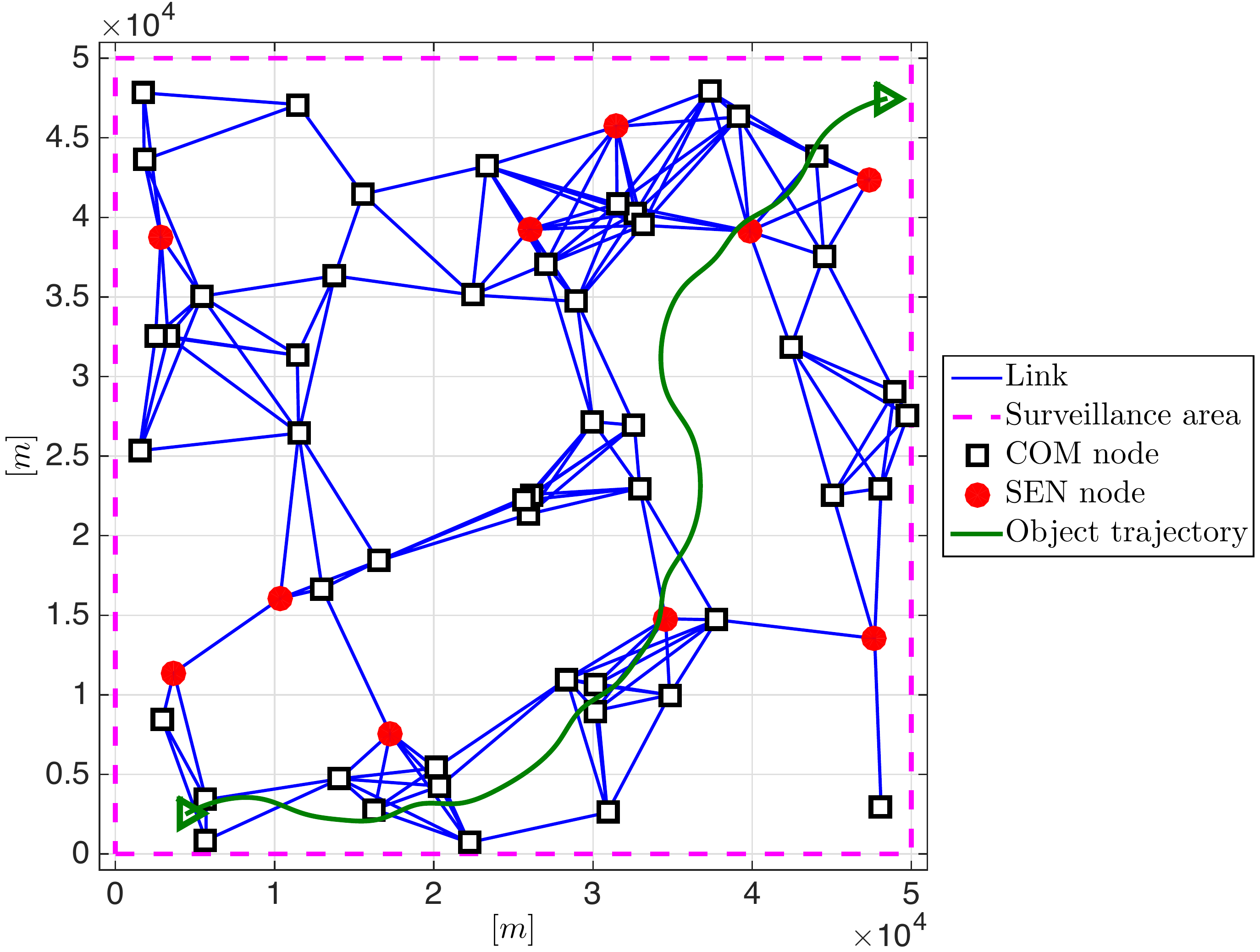}
	\caption{Network and object trajectory.}
	\label{fig:objectnetwork}
\end{figure}
The network consists of $50$ COM nodes and $10$ SEN measuring the \textit{Time Of Arrival} (TOA), i.e. the object-to-sensor distance.
A linear white-noise acceleration model \cite[p. 270]{bar-shalom}, with $4$-dimensional state consisting of positions and velocities along the coordinate axes, is assumed for the object motion; with this choice of the state, sensors are clearly nonlinear.
Standard deviation of the measurement noise has been set to $100 \, [m]$, the sampling interval to $5 \, [s]$, the duration of the simulation to $1500 \, [s]$, i.e. $300$ samples.
In the above scenario, $1000$ Monte Carlo trials with independent measurement noise realizations have been performed. The \textit{Position Root Mean Square Error} (PRMSE) averaged over time, Monte Carlo trials and nodes is reported in Table \ref{tab:results} for different values of $L$ (number of consensus steps) and different consensus-based distributed nonlinear filters.
In all simulations, Metropolis weights  \cite{Calafiore,Xiao}  have been adopted for consensus.
For each filter and each $L$ two values are reported: the top one refers to the PRMSE computed over whole the simulation horizon, whereas the bottom one refers to the PRMSE computed after the transient period.
Three approaches are compared in the UKF variants: the CLCP, CP (Consensus on Posteriors) \cite{batchi2011, cp} and CL (Consensus on Likelihoods) wherein consensus is carried out on the novel information only\footnote{In order to have a fair comparison, the CL algorithm is implemented as in Table \ref{alg:aclcp} with the difference that no consensus on the priors is performed.}. For the CLCP, three different choices for the weights $\rho^{i}_{k}$ are considered. From Table \ref{tab:results}, it can be seen that the weight choice $\rho^{i}_{k} = | \mathcal N|$ provides the best performance only when $L$ is sufficiently large, while for few consensus steps the other choices are preferable. It is worth pointing out that the CL approach, which becomes optimal as $L \rightarrow \infty$, guarantees a bounded estimation error only when the number $L$ of consensus steps is sufficiently high (in Table \ref{tab:results} ``$-$'' indicates a divergent PRMSE).
\begin{table}[h!]
	\setlength\arrayrulewidth{0.5pt}\arrayrulecolor{black} 
	\setlength\doublerulesep{0.5pt}\doublerulesepcolor{black} 
	\caption{Performance comparison}
	\label{tab:results}
	\centering
	\scalebox{1}{
	\begin{tabular}{c|c||c||c|c|c|}
		\multicolumn{6}{c}{~}\\
		\cline{2-6} &\multicolumn{1}{>{\columncolor[gray]{.95}}c||}{CP} &  \multicolumn{1}{>{\columncolor[gray]{.95}}c||}{CL} & \multicolumn{3}{>{\columncolor[gray]{.95}}c|}{CLCP}\\
		\hline
		Choice of $\rho^{i}_{k}$ &  & $\left| \mathcal{N} \right|$ & (\ref{eq:omega}) & $\left| \mathcal{N} \right|$ & $\min_{j \in \mathcal{S}} (1/\omega_{L}^{i,j})$ \\
		\hline
		\multicolumn{1}{>{\columncolor[gray]{.95}}c|}{} & $1733$ & $-$ & $1909$ & $-$ & $1645$ \\ \cline{2-6}
		\multicolumn{1}{>{\columncolor[gray]{.95}}c|}{\multirow{-2}{*}{\textit{$L = 1$}}} & $513$ & $-$ & $407$ & $-$ & $332$ \\
		\hline
		\multicolumn{1}{>{\columncolor[gray]{.95}}c|}{} & $1299$ & $-$ & $1182$ & $1621$ & $821$ \\ \cline{2-6}
		\multicolumn{1}{>{\columncolor[gray]{.95}}c|}{\multirow{-2}{*}{\textit{$L = 2$}}} & $423$ & $-$ & $257$ & $545$ & $226$ \\
		\hline
		\multicolumn{1}{>{\columncolor[gray]{.95}}c|}{} & $1041$ & $-$ & $491$ & $736$ & $480$ \\ \cline{2-6}
		\multicolumn{1}{>{\columncolor[gray]{.95}}c|}{\multirow{-2}{*}{\textit{$L = 3$}}} & $389$ & $-$ & $231$ & $129$ & $198$ \\
		\hline
		\multicolumn{1}{>{\columncolor[gray]{.95}}c|}{} & $936$ & $4269$ & $391$ & $455$ & $355$ \\ \cline{2-6}
		\multicolumn{1}{>{\columncolor[gray]{.95}}c|}{\multirow{-2}{*}{\textit{$L = 4$}}} & $370$ & $3238$ & $219$ & $118$ & $182$ \\
		\hline
		\multicolumn{1}{>{\columncolor[gray]{.95}}c|}{} & $868$ & $977$ & $341$ & $283$ & $299$ \\ \cline{2-6}
		\multicolumn{1}{>{\columncolor[gray]{.95}}c|}{\multirow{-2}{*}{\textit{$L = 5$}}} & $358$ & $134$ & $210$ & $111$ & $170$ \\
		\hline
		\hline
		\multicolumn{1}{>{\columncolor[gray]{.95}}c|}{} & $681$ & $222$ & $230$ & $138$ & $173$ \\ \cline{2-6}
		\multicolumn{1}{>{\columncolor[gray]{.95}}c|}{\multirow{-2}{*}{\textit{$L = 15$}}} & $318$ & $94$ & $182$ & $85$ & $122$ \\
		\hline
	\end{tabular} }
\end{table}

\section{Consensus-based distributed multiple-model filtering}
The section addresses distributed state estimation of jump Markovian systems and its application to tracking of a maneuvering object by means of a network of heterogeneous sensors and communication nodes.
It is well known that a single-model Kalman-like filter is ineffective for tracking a highly maneuvering object and that \textit{Multiple Model} (MM) filters \cite{ackfu1970,blobar1988,bar-shalom,liji2005} are by far superior for this purpose.
The contributions provide novel consensus MM filters to be used for tracking maneuvering objects with sensor networks.
Two novel consensus-based MM filters are presented. Simulation experiments in a tracking case-study, involving a strongly maneuvering object and a sensor network characterized by weak connectivity, demonstrate the superiority of the proposed filters with respect to existing solutions.

\subsection{Notation}
The following new notations are adopted throughout the next sections. 
$\mathcal{P}_{c}$ and $\mathcal{P}_{d}$ denote the sets of PDFs over the \textit{continuous} state space $\rbb^{n}$ and, respectively, of PMFs over the \textit{discrete} state space $\mathcal{R}$, i.e.
\bie
	\mathcal{P}_{c} & = & \left\{ p\!\left( \cdot \right) : \rbb^{n} \rightarrow \rbb ~\left|~ \displaystyle{\int_{\rbb^{n}}} p\!\left( x \right) dx = 1 \mbox{ and } p\!\left( x \right) \geq 0, \forall x \in \rbb^{n} \right. \right\} \, ,\\
	\mathcal{P}_{d} & = & \left\{ \mu = \operatorname{col}\!\left( \mu^{j} \right)_{j \in \mathcal{R}} \in \rbb^{|\mathcal{R}|} ~\left|~\displaystyle{\sum_{j \in \mathcal{R}}} \mu^{j} = 1 \mbox{ and } \mu^{j} \geq 0, \forall j \in \mathcal{R} \right. \right\}.
\eie

\subsection{Bayesian multiple-model filtering}
Let us now focus on state estimation for the \textit{jump Markovian system}
\bie
x_{k} & = & f\!\left( m_{k}, x_{k-1} \right) + w_{k-1} \label{eq:NLDJS} \\
y_{k} & = & h\!\left( m_{k}, x_{k} \right) + v_{k} \label{eq:NLDJS2}
\eie
where $m_{k} \in \mathcal{R} \triangleq \left\{ 1,2,\dots,r \right\}$ denotes the system \textit{mode} or discrete state; $w_{k}$ is the process noise with PDF $p_{w}\!\left( m_{k}, \cdot \right)$; $v_{k}$ is the measurement noise, independent of $w_{k}$, with PDF $p_{v}\!\left( m_{k}, \cdot \right)$.
It is assumed that the system can operate in $r$ possible modes, each mode $j \in \mathcal{R}$ being characterized by a mode-matched model with state-transition function $f\!\left( j, \cdot \right)$ and process noise PDF $p_{w}\!\left( j, \cdot \right)$, measurement function $h\!\left( j, \cdot \right)$ and measurement noise PDF $p_{v}\!\left( j, \cdot \right)$.
Further, mode transitions are modelled by means of a homogeneous Markov chain with suitable transition probabilities
\be
p_{jt} \triangleq \operatorname{prob}\!\left( m_{k}=j | m_{k-1} = t \right),~~~j, t \in \mathcal{R}
\label{eq:Markov}
\ee
whose possible time-dependence is omitted for notational simplicity.
%From now on, the event $m_{k} = m^{j}$ will be more simply denoted as $m^{j}_{k}$.
It is well known that a jump Markovian system (\ref{eq:NLDJS})-(\ref{eq:Markov}) can effectively model the motion of a maneuvering object
\cite[section 11.6]{ackfu1970,bar-shalom}. For instance a \textit{Nearly-Constant Velocity} (NCV) model can be used to describe straight-line object motion while
\textit{Coordinated Turn} (CT) models with different angular speeds can describe object maneuvers.
In alternative, models with different process noise covariances (small for straight-line motion and larger for object maneuvers) can be used with the same kinematic transition function $f(\cdot)$. 
For multimodal systems, the classical single-model filtering approach is clearly inadequate.
To achieve better state estimation performance, the MM filtering approach \cite[section 11.6]{bar-shalom} can be adopted. 
MM filters can provide, in principle, the Bayes-optimal solution for the jump Markov system state estimation problem by running in parallel $r$ mode-matched Bayesian filters (one for each mode) on the same input measurements. 
Each Bayesian filter provides, at each time $k$, the conditional PDFs $p_{k|k-1}^{j}(\cdot)$ and $p_{k}^{j}(\cdot)$ of the state vector $x_{k}$ given the observations $y^{k} \triangleq \left\{ y_{1}, \dots, y_{k} \right\}$ and the mode hypothesis $m_{k} = j$.
Further, for each hypothesis $m_{k} = j$, the conditional modal probability $\mu_{k|\tau}^{j} \triangleq \operatorname{prob}\!\left( m_{k} = j | y^{\tau} \right)$ is also updated.
In summary, Bayesian MM filtering amounts to representing and propagating in time information of the state vector of the jump Markovian system in terms of the $r$ mode-matched PDFs $p^{j}_{\tau}\!\left( \cdot \right), p^{j}_{k|\tau}\!\left( \cdot \right) \in \mathcal{P}_{c}$, $j \in \mathcal{R}$, along with the PMF $\mu_{k|\tau} = \operatorname{col}\!\left( \mu_{k|\tau}^{j} \right)_{j \in \mathcal{R}} \in \mathcal{P}_{d}$.
Then, from such distributions the overall state conditional PDF is obtained by means of the PDF mixture
\bie
	p_{k|k-1}\!\left( \cdot \right) & = & \displaystyle{\sum_{j \in \mathcal{R}}} \, \mu_{k|k-1}^{j} \, p^{j}_{k|k-1}\!\left( \cdot \right) \, = \, \mu_{k|k-1}^{\top} \, p_{k|k-1}\!\left( \cdot \right) \, ,\\
	p_{k}\!\left( \cdot \right) & = & \displaystyle{\sum_{j \in \mathcal{R}}} \, \mu_{k|k}^{j} \, p^{j}_{k}\!\left( \cdot \right) \, = \, \mu_{k|k}^{\top} \, p_{k}\!\left( \cdot \right) \, ,
\eie
where $p_{k|k-1}\!\left( \cdot \right) \triangleq \operatorname{col}\!\left( p_{k|k-1}^{j}\!\left( \cdot \right) \right)_{j \in \mathcal{R}}$ and $p_{k}\!\left( \cdot \right) \triangleq \operatorname{col}\!\left( p_{k}^{j}\!\left( \cdot \right) \right)_{j \in \mathcal{R}}$.\newline
The Bayesian MM filter is summarized in Table \ref{CBMMF}.

\begin{table}[H!]
\caption{Centralized Bayesian Multiple-Model (CBMM) filter pseudo-code}
\label{CBMMF}
\hrulefill\hrule
\begin{algorithmic}[0]
	\Procedure{CBMM(node $i$, time $k$)}{}
		\State \bfbox{Prediction}
		\For{mode $j \in \mathcal{R}$}
			\State $p^{j}_{k|k-1}\!\left( x \right) = \displaystyle{\int} \, p_{w}\!\left( j, x - f\!\left( j, \zeta \right) \right) \, p_{k-1}^{j}\!\left( \zeta \right) \,  d\zeta$
			\State $\mu^{j}_{k|k-1} = \displaystyle{\sum_{t \in \mathcal{R}}} \,  p_{jt}  \, \mu^{t}_{k}$
		\EndFor\vspace{0.5em}
		\State \bfbox{Correction}
		\For{mode $j \in \mathcal{R}$}
			\State $ p_{k}^j(x) = \dfrac{p_{v} \left( j, y_{k} -h \left(j,x \right) \right) \, p_{k|k-1}(x)}{\displaystyle{\int} p_{v} \left( j, y_{k} - h (j, \zeta) \right) \, p_{k|k-1}(\zeta) \, d\zeta}$
			\State ${g}_{k}^{j} = p\!\left( y_{k} | y^{k-1}, m_{k} = j \right)$
			\State $\mu_{k}^{j} = \dfrac{g_{k}^{j} \, \mu_{k|k-1}^{j}}{\displaystyle\sum_{t \in \mathcal{R}}\, g_{k}^{t} \, \mu_{k|k-1}^{t}}$
		\EndFor\vspace{0.5em}
		\State \bfbox{Mode fusion}
		\State $p_{k}(x) = \displaystyle{\sum_{j \in \mathcal{R}}} \,  \mu_{k}^j  \, p_{k}^j(x)$\vspace{0.5em}
		\State \bfbox{Mixing}
		\State $\forall j, t \in \mathcal{R}: \, \mu_{k}^{t|j} = \dfrac{p_{jt} \, \mu_{k}^{t}}{\displaystyle \sum_{\imath \in \mathcal{R}} \, p_{j\imath} \, \mu_{k}^{\imath}}$
		\State $\forall j \in \mathcal{R}: \, \overline{p}_{k}^{j}\!\left( x \right) = \displaystyle{\sum_{t \in \mathcal{R}} \mu_{k}^{t|j}} \, p_{k}^{t}\!\left( x \right)$\vspace{0.5em}
		\State \bfbox{Re-initialization}
		\State $\forall j \in \mathcal{R}: \, p^{j}_{k}\!\left( x \right) = \overline{p}^{j}_{k}\!\left( x \right)$
	\EndProcedure
\end{algorithmic}
\hrule\hrule
\end{table}

It is worth to highlight, however, that the filter in Table \ref{CBMMF} is not practically implementable even in the simplest linear Gaussian case, i.e. when
\textsc{i}) the functions $f\!\left( j,\cdot \right)$ and $h\!\left( j,\cdot \right)$ are linear for all $j \in \mathcal{R}$;
and \textsc{ii}) the PDFs $p_{0}^{j}\!\left( \cdot \right)$, $p_{w}\!\left( j,\cdot \right)$ and $p_{v}\!\left( j,\cdot \right)$ are Gaussian for all $j \in \mathcal{R}$.
In fact, both the mode fusion and the mixing steps (see Table \ref{CBMMF}) produce, just at the first time instant $k = 1$, a Gaussian mixture of $r$ components.  
When time advances, the number of Gaussian components grows exponentially with time $k$ as $r^{k}$.
To avoid this exponential growth, the two most commonly used MM filter algorithms, i.e. \textit{First Order Generalized Pseudo-Bayesian} (GPB$_{1}$) and \textit{Interacting Multiple Model} (IMM), adopt a different strategy to keep all mode-matched PDFs Gaussian throughout the recursions. 
In particular, the GPB$_1$ algorithm (see Table \ref{CGPB1}) performs at each time $k$, before the prediction step, a re-initialization of all mode-matched PDFs with a single Gaussian component having the same mean and covariance of the fused Gaussian mixture.
Conversely, the IMM algorithm (see Table \ref{CIMM}) carries out a mixing procedure to suitably re-initialize each mode-matched PDF with a different Gaussian component.
In the linear Gaussian case,  the correction and prediction steps of MM filters are carried out, in an exact way, by a bank of mode-matched Kalman filters. 
Further, the likelihoods $g_{k}^j$ needed to correct the modal probabilities are evaluated as
\be
	g_{k}^j = \dfrac{1}{\sqrt{\operatorname{det}\!\left( 2\omega S_{k}^{j}\right)}} \, e^{- \dfrac{1}{2} \left( e_{k}^j \right)^{\top} \left( S_{k}^{j} \right)^{-1} e_{k}^{j}} 
\ee
where $e_{k}^{j}$ and $S_{k}^{j}$ are the innovation and relative covariance provided by the KF matched to mode $j$.
Whenever the functions $f\!\left( j,\cdot \right)$ and/or $h\!\left( j,\cdot \right)$ are nonlinear, even correction and/or prediction destroy the Gaussian form of the mode-matched PDFs.
To preserve Gaussianity, as it is common practice in nonlinear filtering, the posterior PDF can be approximated as Gaussian by making use either of linearization of the system model around the current estimate \cite{EKF}, i.e. EKF, or of the unscented transform \cite{juluhl2004}, i.e. UKF.

\begin{table}[!h]
\caption{Centralized First-Order Generalized Pseudo-Bayesian (CGBP$_1$) filter pseudo-code}
\label{CGPB1}
\centering
\hrulefill\hrule
\begin{algorithmic}[0]
\Procedure{CGBP$_1$(node $i$, time $k$)}{}
	\State \bfbox{Prediction}
	\For{mode $j \in \mathcal{R}$}
		\State \textbf{from} $\hat{x}_{k-1}^{j}, P_{k-1}^{j}$ \textbf{compute} $\hat{x}_{k|k-1}^{j}, P_{k|k-1}^{j}$ \textbf{via EKF or UKF}
		\State $\mu^{j}_{k|k-1} = \displaystyle{\sum_{t \in \mathcal{R}}}\, p_{jt} \,\mu^{t}_{k}$
	\EndFor\vspace{0.5em}
	\State \bfbox{Correction}
	\For{mode $j \in \mathcal{R}$}
		\State \textbf{from} $\hat{x}_{k|k-1}^{j}, P_{k|k-1}^{j}$ \textbf{compute} $\hat{x}_{k}^{j}, P_{k}^{j}$, $e_{k}^{j}$,  $S_{k}^{j}$ \textbf{via EKF or UKF}
		\State $g_{k}^{j} =  \dfrac{1}{\sqrt{\operatorname{det}\!\left( 2\omega S_{k}^{j} \right)}} \, e^{-\frac{1}{2} \left( e_{k}^{j} \right)^\top \left( S_{k}^{j} \right)^{-1} e_{k}^{j}}$
		\State $\mu_{k}^{j} = \dfrac{g_{k}^{j} \, \mu_{k|k-1}^{j}}{\displaystyle \sum_{t \in \mathcal{R}}\, g_{k}^{t} \, \mu_{k|k-1}^{t}}$
	\EndFor\vspace{0.5em}
	\State \bfbox{Mode fusion}
	\State $\hat{x}_{k} = \displaystyle{\sum_{j \in \mathcal{R}}}\, \mu_{k}^{j} \,\hat{x}_{k}^{j}$
	\State $P_{k} = \displaystyle{\sum_{j \in \mathcal{R}}}\, \mu_{k}^{j} \,\left[ P_{k}^{j} + \left( \hat{x}_{k} - \hat{x}_{k}^{j} \right) \, \left(\hat{x}_{k} - \hat{x}_{k}^{j} \right)^\top \right]$
	\State $\forall j \in \mathcal{R} : \, \hat{x}_{k}^j = \hat{x}_{k}$, $P_{k}^j = P_{k}$
\EndProcedure
\end{algorithmic}
\hrule\hrule
\end{table}

\begin{table}[!h]
\caption{Centralized Interacting Multiple Model (CIMM) filter pseudo-code}
\label{CIMM}
\centering
\hrulefill\hrule
\begin{algorithmic}[0]
\Procedure{CIMM(node $i$, time $k$)}{}
	\State \bfbox{Prediction}
	\State As in the GPB$_1$ algorithm of Table \ref{CGPB1}\vspace{0.5em}
	\State \bfbox{Correction}
	\State As in the GPB$_1$ algorithm of Table \ref{CGPB1}\vspace{0.5em}
	\State \bfbox{Mode fusion}
	\State As in the GPB$_1$ algorithm of Table \ref{CGPB1}\vspace{0.5em}
	\State \bfbox{Mixing}
	\State $\forall j, t \in \mathcal{R} : \, \mu_{k}^{t|j} = \dfrac{p_{jt} \,\mu_{k}^{t}}{\displaystyle \sum_{\imath \in \mathcal{R}}\,p_{j\imath}\,\mu_{k}^{\imath}}$
	\For{mode $j \in \mathcal{R}$}
		\State $\overline{x}_{k}^j = \displaystyle{\sum_{t \in \mathcal{R}}}\, \mu^{t|j}_{k}\, \hat{x}_{k}^{t}$
		\State $\overline{P}_{k}^j = \displaystyle{\sum_{t \in \mathcal{R}}}\, \mu^{t|j}_{k}\, \left[ P_{k}^{t} + \left(\overline{x}_{k}^{j} - \hat{x}_{k}^{t} \right) \left(\overline{x}_{k}^{j} - \hat{x}_{k}^{t} \right)^{\top} \right]$
	\EndFor\vspace{0.5em}
	\State \bfbox{Re-initialization}
	\State $\forall j \in \mathcal{R} : \, \hat{x}_{k}^{j} = \overline{x}_{k}^{j}$, $P_{k}^{j} = \overline{P}_{k}^j$
\EndProcedure
\end{algorithmic}
\hrule\hrule
\end{table}
 
\subsection{Networked multiple model estimation via consensus}
A fundamental issue for networked estimation is to consistently fuse PDFs of the continuous quantity, or PMFs of the discrete quantity, to be estimated coming from different nodes.
A suitable information-theoretic fusion among continuous probability distributions  has been introduced in subsection \ref{ssec:klapdf}, which can be straightforwardly applied in this context. 

Let us now turn the attention to the case of discrete variables.
Given PMFs $\mu, \nu \in \mathcal{P}_d$ their KLD is defined as
\be
	D_{KL} \left( \mu \parallel \nu \right) \triangleq \displaystyle{\sum_{j \in \mathcal{R}}} \, \mu^j \, \operatorname{log}\!\left(\dfrac{\mu^j}{\nu^j}\right)
\ee  
According to (\ref{eq:kla}), the weighted KLA of the PMFs $\mu^{i} \in \mathcal{P}_c$, $i \in \mathcal{N}$, is defined as
\be
 \overline{\mu} \, = \, \arg \inf_{\mu \in \mathcal{P}_d} \, \displaystyle{\sum_{i \in \mathcal{N}}} \, \omega^{i} \,
 D_{KL} \left( \mu \parallel
 \mu^{i} \right).
\label{KLA-d}
\ee 
The following result holds.

\begin{thm}[KLA of PMFs]\label{thm:wkla}~\\
The weighted KLA in (\ref{KLA-d}) is given by
\be
 \overline{\mu}  = \operatorname{col}\!\left( \overline{\mu}^j \right)_{j\in \mathcal{R}} \, , \qquad
 \overline{\mu}^j = \dfrac{\displaystyle{\prod_{i \in \mathcal{N}}} \, \left( \mu^{i,j} \right)^{\omega^{i}}}{\displaystyle{\sum_{t \in \mathcal{R}} \, \displaystyle{\prod_{i \in \mathcal{N}}} \, \left( \mu^{i,t} \right)^{\omega^{i}}}} 
\label{thm:pmfkla}
\ee
\end{thm}

Analogously to the case of continuous probability distributions, fusion and weighting operators can be introduced for PMFs $\mu, \nu \in
\mathcal{P}_d$ and $\omega > 0$, as follows:
\be
\ba{rclc}
\mu \oplus \nu = \eta = \operatorname{col}\!\left( \eta^j \right)_{j \in \mathcal{R}}  & \mbox{with} & \eta^j = \dfrac{\mu^j \nu^j}{\displaystyle{\sum_{t \in \mathcal{R}}} \, \mu^{t} \nu^{t}} \triangleq \mu^j \oplus \nu^j, & \forall j \in \mathcal{R}\\
\omega \odot \mu = \eta = \operatorname{col}\!\left( \eta^j \right)_{j \in \mathcal{R}}  & \mbox{with} & \eta^j = \dfrac{\left( \mu^j \right)^\omega}{\displaystyle{\sum_{t \in \mathcal{R}}} \left( \mu^{t} \right)^\omega} \triangleq \omega \odot \mu^j, & \forall j \in \mathcal{R}
\ea 
\label{ops-d}
\ee
Thanks to the properties \textsc{p.a}-\textsc{p.f} of the operators $\oplus$ and $\odot$ (see subsection \ref{ssec:difnotation}), the KLA PMF in (\ref{KLA-d}) can be expressed as
\be
\overline{\mu} ~=~ \displaystyle{\bigoplus_{i \in \mathcal{N}}} \, \left( \omega^{i} \odot \mu^{i} \right)
\label{coll-fusion-d}
\ee
The collective fusion of PMFs (\ref{coll-fusion-d}) can be computed, in a distributed and scalable fashion, via consensus iterations
\be
\mu_{{l}}^{i} ~=~ \displaystyle{\bigoplus_{j \in \mathcal{N}^{i}}} \, \left( \omega^{i,j} \odot \mu_{{l}-1}^{j} \right)
\label{consensus-d}
\ee
initialized from $\mu^{i}_0= \mu^{i}$ and with consensus weights satisfying $\omega^{i,j} \geq 0$ and $\sum_{j \in \mathcal{N}^{i}}~ \omega^{i,j} = 1$.

\subsection{Distributed multiple-model algorithms}
\label{distributed-mm-alg}

Let us now address the problem of interest, i.e. DSOF for the jump Markovian system (\ref{eq:NLDJS})-(\ref{eq:NLDJS2}) over a sensor network $\left( \mathcal{N}, \mathcal{A} \right)$ of the type modeled in section \ref{sec:net}.
In this setting, measurements in (\ref{eq:NLDJS2}) are provided by sensor nodes, i.e.
\bie
	y_{k} & = & \operatorname{col}\!\left( y_{k}^{i} \right)_{i \in \mathcal{S}} \, , \\
	h(\cdot,\cdot) & = & \operatorname{col}\!\left( h^{i}(\cdot,\cdot) \right)_{i \in \mathcal{S}} \, ,\\
	v_{k} & = & \operatorname{col}\!\left( v_{k}^{i} \right)_{i \in \mathcal{S}} \, ,
\eie
where the sensor measurement noise $v_{k}^{i}$, $i \in \mathcal{S}$, is characterized by the PDF $p_{v^{i}}(\cdot,\cdot)$.
The main difficulty, in the distributed context, arises from the possible lack of complete observability from an individual node (e.g. a communication node or a sensor node measuring only angle or range or Doppler-frequency shift of a moving object)  which makes impossible for such a node to reliably estimate the system mode as well as the continuous state on the sole grounds of local information.
To get around this problem, consensus on both the discrete state PMF and either the mode-matched or the fused continuous state PDFs can be
exploited in order to spread information throughout the network and thus guarantee observability in each node, provided that the network is  strongly connected (i.e., for any pair of nodes $i$ and $j$ there exists a path from $i$ to $j$ and viceversa) and the jump Markovian system is collectively observable (i.e. observable from the whole set $\mathcal{S}$ of sensors).

Let us assume that at sampling time $k$, before processing the new measurements $y_{k}$, each node $i \in \mathcal{N}$ be provided with the prior mode-matched PDFs $p_{k|k-1}^{i}(\cdot) = col \left( p^{i,j}_{k|k-1}(\cdot) \right)_{j \in \mathcal{R}}$ along with the mode PMF
$\mu^{i}_{k|k-1} = col \left( \mu^{i,j}_{k|k-1} \right)_{j \in \mathcal{R}}$, which are the outcomes of the previous local and consensus computations up to time $k-1$.
Then, sensor nodes $i \in \mathcal{S}$ correct both mode-matched PDFs and mode PMF with the current local measurement $y_{k}^{i}$ while
communication nodes leave them unchanged; the PDFs and PMF obtained in this way are called local posteriors.
At this point, consensus is needed in each node $i \in \mathcal{N}$ to exchange local posteriors with the neighbours and regionally average them over the subnetwork $\mathcal{N}^{i}$ of in-neighbours.  
The more consensus iterations are performed, the faster will be convergence of the regional KLA to the collective KLA at the price of higher
communication cost and, consequently, higher energy consumption and lower network lifetime.
Two possible consensus strategies can be applied: 
\begin{description}
	\item[Consensus on the Fused PDF]- first carry out consensus on the mode PMFs, then fuse the local mode-matched posterior PDFs over $\mathcal{R}$ using the modal probabilities resulting from consensus, and finally carry out consensus on the fused PDF; 
	\item[Consensus on Mode-Matched PDFs]- carry out in parallel consensus on the mode PMFs and on the mode-matched PDFs and then perform  the mode fusion using the modal probabilities and mode-matched PDFs resulting from consensus.
\end{description}
Notice that the first approach is cheaper in terms of communication as the nodes need to exchange just a single fused PDF, instead of multiple ($r=|\mathcal{R}|$) mode-matched PDFs. Further, it seems the most reasonable choice to be used in combination with the GPB$_1$ approach, as only the fused PDF is needed in the re-initialization step.
Conversely, the latter approach is mandatory for the distributed IMM filter. 
Recall, in fact, that the IMM filter re-initializes each mode-matched PDF with a suitable mixture of such PDFs, totally disregarding the fused PDF.
Hence, to spread information about the continuous state through the network,  it is necessary to apply consensus to all mode-matched PDFs.

Summing up, two novel distributed MM filters are proposed: the \textit{Distributed GPB$_1$} (DGPB$_1$) algorithm of Table \ref{DGPB1} which adopts the \textit{Consensus on the Fused PDF} approach to limit data communication costs, and the \textit{Distributed IMM} (DIMM) algorithm of Table \ref{DIMM} which, conversely, adopts the \textit{Consensus on Mode-Matched PDFs} approach.
Since both DGPB$_1$ and DIMM filters, like their centralized counterparts, propagate Gaussian PDFs completely characterized by either
the estimate-covariance pair $\left( \hat{x},P\right)$ or the information pair $\left(q = P^{-1} \hat{x}, \Omega =
P^{-1} \right)$, the consensus iteration is simply carried as a weighted arithmetic average of the
information pairs associated to such PDFs, see (\ref{eq:ci1})-(\ref{eq:ci2}).

\begin{table}[H!]
\caption{Distributed GPB$_1$ (DGPB$_1$) filter pseudo-code}
\label{DGPB1}
\centering
\hrulefill\hrule
\begin{algorithmic}[0]
\Procedure{DGPB$_{1}$(node $i$, time $k$)}{}
	\State \bfbox{Prediction}
	\For{mode $j \in \mathcal{R}$}
		\State \textbf{given} $\hat{x}_{k-1}^{i,j}$ \textbf{and} $P_{k-1}^{i,j}$ \textbf{compute} $\hat{x}_{k|k-1}^{i,j}$ \textbf{and} $P_{k|k-1}^{i,j}$ \textbf{via EKF or UKF}
		\State $\mu^{i,j}_{k|k-1} = \displaystyle{\sum_{t \in \mathcal{R}}}\, p_{jt} \,\mu^{i,t}_{k-1}$
	\EndFor\vspace{0.5em}
	\State \bfbox{Correction}
	\For{mode $j \in \mathcal{R}$}
		\If{$i \in \mathcal{S}$}
			\State \textbf{from} $\hat{x}_{k|k-1}^{i,j}, P_{k|k-1}^{i,j}$ \textbf{compute} $\hat{x}_{k}^{i,j}, P_{k}^{i,j}$, $e_{k}^{i,j}$ $S_{k}^{i,j}$ \textbf{via EKF or UKF}
			\State $g_{k}^{i,j} = \left[ \operatorname{det}\!\left(2 \omega S_{k}^{i,j} \right) \right]^{-\frac{1}{2}} \, e^{-\frac{1}{2} \left( e_{k}^{i,j} \right)^\top \left( S_{k}^{i,j} \right)^{-1} e_{k}^{i,j}}$
			\State $\mu_{k,0}^{i,j} = g_{k}^{i,j} \mu_{k|k-1}^{i,j}\left[ \sum_{t \in \mathcal{R}}\,g_{k}^{i,t} \mu_{k|k-1}^{i,t}\right]^{-1}$
		\ElsIf{$i \in \mathcal{C}$}
			\State $\hat{x}_{k}^{i,j} = \hat{x}_{k|k-1}^{i,j}$, $P_{k}^{i,j} = P_{k|k-1}^{i,j}$, $\mu_{k,0}^{i,j} = \mu_{k|k-1}^{i,j}$
		\EndIf
	\EndFor\vspace{0.5em}
	\State \bfbox{Consensus on modal probabilities}
	\For{mode $j \in \mathcal{R}$}
		\For{${l}=1, \dots, L$}
			\State $\mu^{i,j}_{k,{l}} = \displaystyle{\bigoplus_{\imath \in \mathcal{N}^{i}}} \,\left[ \omega^{i,\imath} \odot \mu_{k,{l}-1}^{\imath,j} \right]$
		\EndFor
		\State $\mu_{k}^{i,j} = \mu_{k,L}^{i,j}$
	\EndFor\vspace{0.5em}
        \State \bfbox{Mode fusion}
	\State $\hat{x}_{k,0}^{i} = \displaystyle{\sum_{j \in \mathcal{R}}}\, \mu_{k}^{i,j} \,\hat{x}_{k}^{i,j}$
	\State $P_{k,0}^{i} = \displaystyle{\sum_{j \in \mathcal{R}}}\, \mu_{k}^{i,j} \,\left[ P_{k}^{i,j} + \left( \hat{x}_{k,0}^{i} - \hat{x}_{k}^{i,j} \right) \, \left(\hat{x}_{k,0}^{i} - \hat{x}_{k}^{i,j} \right)^\top \right]$\vspace{0.5em}
        \State \bfbox{Consensus on the fused PDF}
	\State $\Omega_{k,0}^{i} = \left( P_{k,0}^{i} \right)^{-1}$, $q_{k,0}^{i} = \Omega_{k,0}^{i} \hat{x}_{k,0}^{i}$
	\For{${l}=1, \dots, L$}
		\State $\Omega^{i}_{k,{l}} = \displaystyle{\sum_{\imath \in \mathcal{N}^{i}}} \,\omega^{i,\imath}  \,\Omega_{k,{l}-1}^{\imath}$
		\State $q^{i}_{k,{l}} = \displaystyle{\sum_{\imath \in \mathcal{N}^{i}}} \,\omega^{i,\imath}  \,q_{k,{l}-1}^{\imath}$
	\EndFor\vspace{0.5em}
	\State \bfbox{Re-initialization}
	\State $\forall j \in \mathcal{R} \, : \, \hat{x}_{k}^{i,j} = \left( \Omega^{i}_{k,L} \right)^{-1} q_{k,L}^{i}$, $P_{k}^{i,j} = \left( \Omega^{i}_{k,L} \right)^{-1}$
\EndProcedure
\end{algorithmic}
\hrule\hrule
\end{table}

\begin{table}[H!]
\caption{Distributed IMM (DIMM) filter pseudo-code}
\label{DIMM}
\centering
\hrulefill\hrule
\begin{algorithmic}[0]
\Procedure{DIMM(node $i$, time $k$)}{}
	\State \bfbox{Prediction}
	\State As in the DGPB$_1$ algorithm of Table \ref{DGPB1}\vspace{0.5em}
%	\For{mode $j \in \mathcal{R}$}
%		\State \textbf{given} $\hat{x}_{k-1}^{i,j}, P_{k-1}^{i,j}$ \textbf{compute} $\hat{x}_{k|k-1}^{i,j}, P_{k|k-1}^{i,j}$ \textbf{via EKF or UKF}
%		\State $\mu^{i,j}_{k|k-1} = \displaystyle{\sum_{t \in \mathcal{R}}}\, p_{jt} \,\mu^{i,t}_{k-1}$
%	\EndFor\vspace{0.5em}
	\State \bfbox{Correction}
	\For{mode $j \in \mathcal{R}$}
		\If{$i \in \mathcal{S}$}
			\State \textbf{from} $\hat{x}_{k|k-1}^{i,j}, P_{k|k-1}^{i,j}$ \textbf{compute} $\hat{x}_{k,0}^{i,j}, P_{k,0}^{i,j}$, $e_{k}^{i,j}$, $S_{k}^{i,j}$ \textbf{via EKF or UKF}
			\State $g_{k}^{i,j} = \left[ \operatorname{det}\!\left( 2 \omega S_{k}^{i,j} \right) \right]^{-\frac{1}{2}} \, e^{-\frac{1}{2} \left( e_{k}^{i,j} \right)^\top \left( S_{k}^{i,j} \right)^{-1} e_{k}^{i,j}}$
			\State $\mu_{k,0}^{i,j} = g_{k}^{i,j} \mu_{k|k-1}^{i,j}\left[\sum_{t \in \mathcal{R}}\,g_{k}^{i,t}\mu_{k|k-1}^{i,t}\right]^{-1}$
		\EndIf
		\If{$i \in \mathcal{C}$}
			\State $\hat{x}_{k,0}^{i,j} = \hat{x}_{k|k-1}^{i,j}$, $P_{k,0}^{i,j}= P_{k|k-1}^{i,j}$, $\mu_{k,0}^{i,j} = \mu_{k|k-1}^{i,j}$
		\EndIf
	\EndFor\vspace{0.5em}
	\State \bfbox{Parallel consensus on modal probabilities \& mode-matched PDFs}
	\For{mode $j \in \mathcal{R}$}
		\State $\Omega_{k,0}^{i,j} = \left( P_{k,0}^{i,j} \right)^{-1}$, $q_{k,0}^{i,j} = \Omega_{k,0}^{i,j} \, \hat{x}_{k,0}^{i,j}$
		\For{${l}=1, \dots, L$}
			\State $\mu^{i,j}_{k,{l}} = \displaystyle{\bigoplus_{\imath \in \mathcal{N}^{i}}} \,\left[ \omega^{i,\imath}  \odot \mu_{k,{l}-1}^{\imath,j} \right]$
			\State $\Omega^{i,j}_{k,{l}} = \displaystyle{\sum_{\imath \in \mathcal{N}^{i}}} \,\omega^{i,\imath}  \,\Omega_{k,{l}-1}^{\imath,j},\,q^{i,j}_{k,{l}}	 \, = \, \displaystyle{\sum_{\imath \in \mathcal{N}^{i}}} \,\omega^{i,\imath}  \,q_{k,{l}-1}^{\imath,j}$
		\EndFor
		\State $\mu_{k}^{i,j} = \mu_{k,L}^{i,j}$
	\EndFor\vspace{0.5em}
	\State \bfbox{Mode fusion}
	\State $\forall j \in \mathcal{R} \, : \, P_{k,L}^{i,j} = \left( \Omega_{k,L}^{i,j} \right)^{-1}$, $\hat{x}_{k,L}^{i,j} = P_{k,L}^{i,j} q_{k,L}^{i,j}$
	\State $\hat{x}_{k}^{i} = \displaystyle{\sum_{j \in \mathcal{R}}}\, \mu_{k}^{i,j} \,\hat{x}_{k,L}^{i,j}$
	\State $P_{k}^{i} = \displaystyle{\sum_{j \in \mathcal{R}}}\, \mu_{k}^{i,j} \,\left[ P_{k,L}^{i,j} + \left(\hat{x}_{k}^{i} - \hat{x}_{k,L}^{i,j} \right) \, \left(\hat{x}_{k}^{i} - \hat{x}_{k,L}^{i,j} \right)^\top \right]$\vspace{0.5em}
	\State \bfbox{Mixing}
	\For{mode $j \in \mathcal{R}$}
		\State $\forall t \in \mathcal{R} \, : \, \mu_{k}^{i,t|j} = p_{jt} \,\mu_{k}^{i,t} \left[ \sum_{\jmath \in \mathcal{R}} \, p_{j\jmath} \,\mu_{k}^{\jmath,h} \right]^{-1}$
		\State $\hat{x}_{k}^{i,j} = \displaystyle{\sum_{t \in \mathcal{R}}}\, \mu^{i,t|j}_{k}\, \hat{x}_{k,L}^{i,j}$
		\State $P_{k}^{i,j} = \displaystyle{\sum_{t \in \mathcal{R}}}\, \mu^{i,t|j}_{k}\,\left[  P_{k,L}^{i,t} + \left(\hat{x}_{k}^{i,t} - \hat{x}_{k,L}^{i,t} \right)  \left(\hat{x}_{k}^{i,t} - \hat{x}_{k,L}^{i,t} \right)^{\top} \right]$
	\EndFor
\EndProcedure
\end{algorithmic}
\hrule\hrule
\end{table}

\subsection{Connection with existing approach}
It is worth to point out the differences of the proposed DGPB$_1$ and DIMM algorithms of Tables \ref{DGPB1} and, respectively, \ref{DIMM} with respect
to the distributed IMM algorithm of reference \cite{li-jia}.
In the latter, each node carries out consensus on the newly acquired information and then updates the local prior with the outcome of such a consensus.
Conversely, in the DGPB$_1$ and DIMM algorithms, each node first updates its local prior with the local new information, and then consensus on the resulting local posteriors is carried out.
To be more specific, recall first that, following the Gaussian approximation paradigm and representing Gaussian PDFs with information pairs, the local correction
for each node $i$ and mode $j$, takes the form
\bie
	\Omega_{k}^{i,j} & = & \Omega_{k|k-1}^{i,j} + \delta \Omega_{k}^{i,j} \, ,\label{eq:infeq1}\\
	 q_{k}^{i,j} & = & q_{k|k-1}^{i,j} + \delta q_{k}^{i,j} \, , \label{eq:infeq2}
\eie
where $\delta \Omega_{k}^{i,j}, \delta q_{k}^{i,j}$ denote the innovation terms due to the new measurements $y_{k}^{i}$.
In particular, for a linear Gaussian sensor model characterized by measurement function $h^{i}(j,x) = C^{i,j} x$ and measurement noise PDF
$p_{v^{i}}(j,\cdot) = \mathcal{N}\!\left( \cdot; 0, R^{i,j} \right)$, the innovation terms in (\ref{eq:infeq1})-(\ref{eq:infeq2}) are given by
\bie
	\delta \Omega_{k}^{i,j} & = & \left( C^{i,j} \right)^\top \, \left( R^{i,j} \right)^{-1} \, C^{i,j} \, ,\\
	\delta q_{k}^{i,j} & = & \left( C^{i,j} \right)^\top \, \left( R^{i,j} \right)^{-1} \, y_{k}^{i} \, .
\eie
Further, for a nonlinear sensor $i$, approximate innovation terms $ \left( \delta \Omega_{k}^{i,j},  q_{k}^{i,j} \right)$
can be obtained making use either of the EKF or of the unscented transform as shown in \cite{li-jia}.

Now, the following observations can be made.
\begin{itemize}
	\item The algorithm in \cite{li-jia}, which will be indicated below with the acronym DIMM-CL (\textit{DIMM with Consensus on Likelihoods}), applies consensus to such innovation terms scaled by the number of nodes, i.e. to
$|\mathcal{N}| \,\delta \Omega_{k}^{i,j}$ and $|\mathcal{N}| \,\delta q_{k|k-1}^{i,j}$. Conversely, the DIMM algorithm of Table \ref{DIMM} first performs the correction (\ref{eq:infeq1})-(\ref{eq:infeq2}) and then applies consensus to the local posteriors $\left( \Omega_{k}^{i,j}, q_{k}^{i,j} \right)$, and similarly does the GPB$_1$ algorithm of Table \ref{DGPB1} with the fused information pairs $\left( \Omega_{k}^{i}, q_{k}^{i} \right)$.
	\item A further difference between DIMM and DIMM-CL concerns the update of the discrete state PMF: the former applies consensus to the posterior mode probabilities $\mu_{k}^{i,j}$, while the latter applies it to the mode likelihoods $g_{k}^{i,j}$.
\end{itemize}

A discussion on the relative merits/demerits of the consensus approaches adopted for DIMM and, respectively, DIMM-CL is in order.
The main positive feature of DIMM-CL is that it converges to the centralized IMM as $L \rightarrow \infty$. However, it requires a minimum number of consensus iterations
per sampling interval in order to avoid divergence of the estimation error. In fact, by looking at the outcome of CL
$$
 \delta \Omega_{k,L}^{i,j} \, = \, \displaystyle{\sum_{\imath \in \mathcal{N}}} \, \omega^{i,\imath}_L \, \delta \Omega_{k}^{\imath,j} \, ,
$$ 
it can be seen that a certain number of consensus
iterations is needed so that the local information has time to spread through the network and the system becomes 
observable or, at least detectable, from the subset of sensors  $\mathcal{S}_L^{i} = \left\{ j \in \mathcal{S}: \omega^{i,j}_L \neq 0 \right\}$.

The proposed approach does not suffer from such limitations thanks to the fact that the whole posterior PDFs are combined.
Indeed, as shown in \cite{batchi2011,cp}, the single-model consensus Kalman filter based on CP ensures stability, for any $L \geq 1$ under the assumptions of system observability from the whole network and strong network connectivity, in that it provides a mean-square bounded estimation error in each node of the network.
Although the stability results in \cite{batchi2011,cp} are not proved for jump Markov and/or nonlinear systems, the superiority of DIMM over DIMM-CL for tracking a maneuvering object
when only a limited number of consensus iterations is performed is confirmed by simulation experiments, as it will be seen in the next section.
With this respect, since the number of data transmissions, which primarily affect energy consumption, is clearly proportional to $L$, 
in many situations it is preferable for energy efficiency to perform only a few consensus steps, possibly a single one.
A further advantage of CP over CL is that the former does not require any prior knowledge of the network topology as well as on the number of nodes
in order to work properly.
As a by-product, the CI approach can also better cope with time-varying networks where, due to node join/leave and link disconnections, the network graph is changing in time.
On the other hand, the proposed approach does not approach the centralized filter as $L \rightarrow \infty$ since the adopted fusion rule follows
a cautious strategy so as to achieve robustness with respect to data incest \cite{cp}.

\subsection{A tracking case-study}

To assess performance of the proposed distributed multiple-model algorithms described in section \ref{distributed-mm-alg}, 
the 2D tracking scenario of Fig. \ref{fig:trajectory} is considered.
From Fig. \ref{fig:trajectory} notice that the object is highly maneuvering with speeds ranging in $[0,300] [m/s]$ and acceleration magnitudes
in $[0,1.5] ~g$, $g$ being the gravity acceleration.
The object state is denoted by $x = \left[ x,~\dot{x},~y,~\dot{y} \right]^{\top}$ where $(x,~y)$ and 
$( \dot{x},~\dot{y} )$ represent the object Cartesian position and, respectively, velocity components.  
The sampling interval is $T_{s} = 5 [s]$ and the total object navigation time is $640 [s]$, corresponding to $129$ sampling intervals. 
Specifically, $r=5$ different \textit{Coordinated-Turn} (CT) models \cite{bar-shalom,far1985v1,far1985v2} are used in the MM algorithms.
All models have the state dynamics
\be
x_{k+1} = \left[ \ba{cccc}
					1 & \frac{\sin(\omega T_s)}{\omega} & 0 & -\frac{1 - \cos(\omega T_s)}{\omega}\\
					0 & \cos(\omega T_s) & 0 & -\sin(\omega T_s)\\
					0 & \frac{1 - \cos(\omega T_s)}{\omega} & 1 & \frac{\sin(\omega T_s)}{\omega}\\
					0 & \sin(\omega T_s) & 0 & \cos(\omega T_s)
				\ea \right] x_{k} + w_{k} \, ,
\ee
\be
	\mathbf{Q} = \operatorname{var}\!\left( w_{k} \right) = \left[\ba{cccc}
																	\frac{1}{4}T_{s}^{4} & \frac{1}{2}T_{s}^{3} & 0 & 0 \\
																	\frac{1}{2}T_{s}^{3} & T_{s}^{2} & 0 & 0 \\
																	0 & 0 & \frac{1}{4}T_{s}^{4}  & \frac{1}{2}T_{s}^{3} \\
																	0 & 0 & \frac{1}{2}T_{s}^{3}& T_{s}^{2}
																\ea \right] \, \sigma_w^2
\ee
for five different constant angular speeds $\omega  \in \left\{ -1, -0.5, 0, 0.5, 1 \right\} \, [^{\circ}/s]$.
Notice, in particular, that, taking the limit for $\omega \rightarrow 0$, the model corresponding to $\omega=0$ is nothing but the well known
\textit{Discrete White Noise Acceleration} (DWNA) model \cite{far1985v1, bar-shalom}, 
The standard deviation of the process noise is taken as $\sigma_w = 0.1 \, [m/s^{2}]$ for the DWNA ($\omega=0$) model and $\sigma_w = 0.5 \, [m/s^{2}]$
for the other models.  
Jump probabilities (\ref{eq:Markov}), for the Markov chain, are chosen as follows:
\be
	\left[ p_{jk} \right]_{j,k = 1,\dots,5} = \left[ \ba{ccccc}
													0.95 & 0.05 & 0 & 0 & 0\\
													0.05 & 0.9 & 0.05 & 0 & 0\\
													0 & 0.05 & 0.9 & 0.05 & 0\\
													0 & 0 & 0.05 & 0.9 & 0.05\\
													0 & 0 & 0 & 0.05 & 0.95
												\ea \right]
\ee
All the local mode-matched filters exploit the UKF \cite{juluhl2004}.
All simulation results have been obtained by averaging over $300$ Monte Carlo trials. 
No prior information on the initial object position is assumed, i.e. in all filters the object state is initialized in the center of the surveillance region with zero velocity and associated covariance matrix equal to $\operatorname{diag}\!\left\{ 10^{8}, 10^{4}, 10^{8}, 10^{4} \right\}$.
\begin{figure}[h!]
	\centering
	\includegraphics[width=\columnwidth]{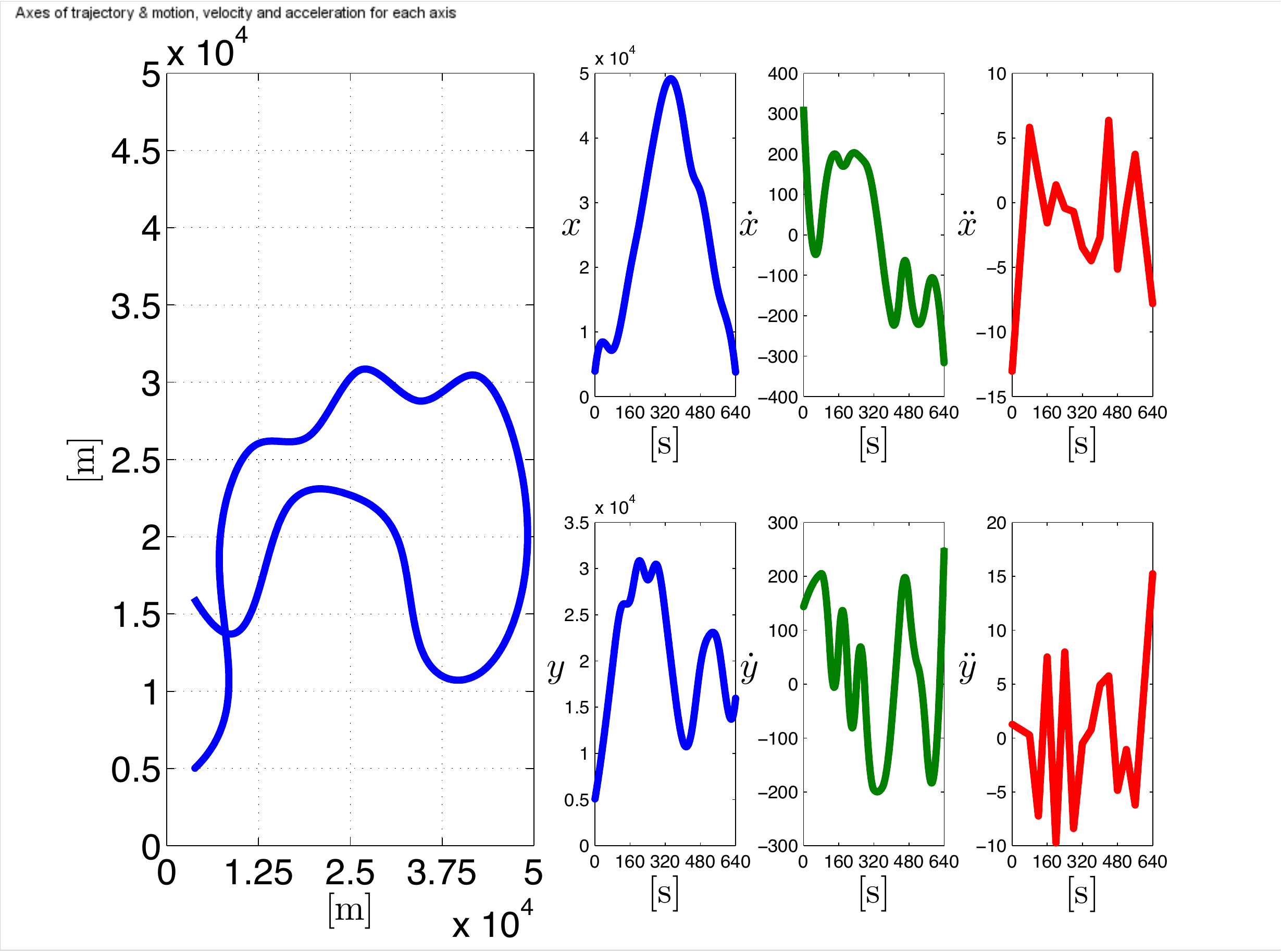}
	\caption{Object trajectory considered in the simulation experiments}
\label{fig:trajectory}
\end{figure}

A surveillance area of $50\times50 \, [km^2]$ is considered, wherein $4$ communication nodes, $4$ bearing-only sensors measuring the object \textit{Direction of Arrival} (DOA) and $4$ TOA are deployed as shown in Fig. \ref{fig:net-4c4t4d}. 
Notice that the presence of communication nodes improves energy efficiency in that it allows to reduce average distance among nodes and, hence, the required transmission power of each individual node, but also hampers information diffusion across the network.
This type of network represents, therefore, a valid benchmark to test the effectiveness of consensus state estimators.
The following measurement functions characterize the DOA and TOA sensors:
\be
	\ba{rcl}
		h^{i}(x) =\left\{ \ba{ll} \angle [ \left( x - x^{i} \right) + j \left( y - y^{i} \right)] \, , & \mbox{if $i$ is a DOA sensor} \\[0.5em]
               		                \sqrt{ \left( x - x^{i} \right)^2+ \left( y - y^{i} \right)^2} \, , & \mbox{if $i$ is a TOA sensor}
	\ea\right.
\ea
\ee
where $( x^{i},~y^{i} )$ represents the known position of sensor $i$ in Cartesian coordinates. 
The standard deviation of DOA and TOA measurement noises are taken respectively as $\sigma_{DOA} = 1 \, [\mbox{}^{\circ}]$ and $\sigma_{TOA} = 100 \, [m]$. 
The number of consensus steps used in the simulations ranges from $L = 1$ to $L=5$, reflecting a special attention towards energy efficiency.
\begin{figure}[h!]
	\centering
	\includegraphics[width=0.75\columnwidth]{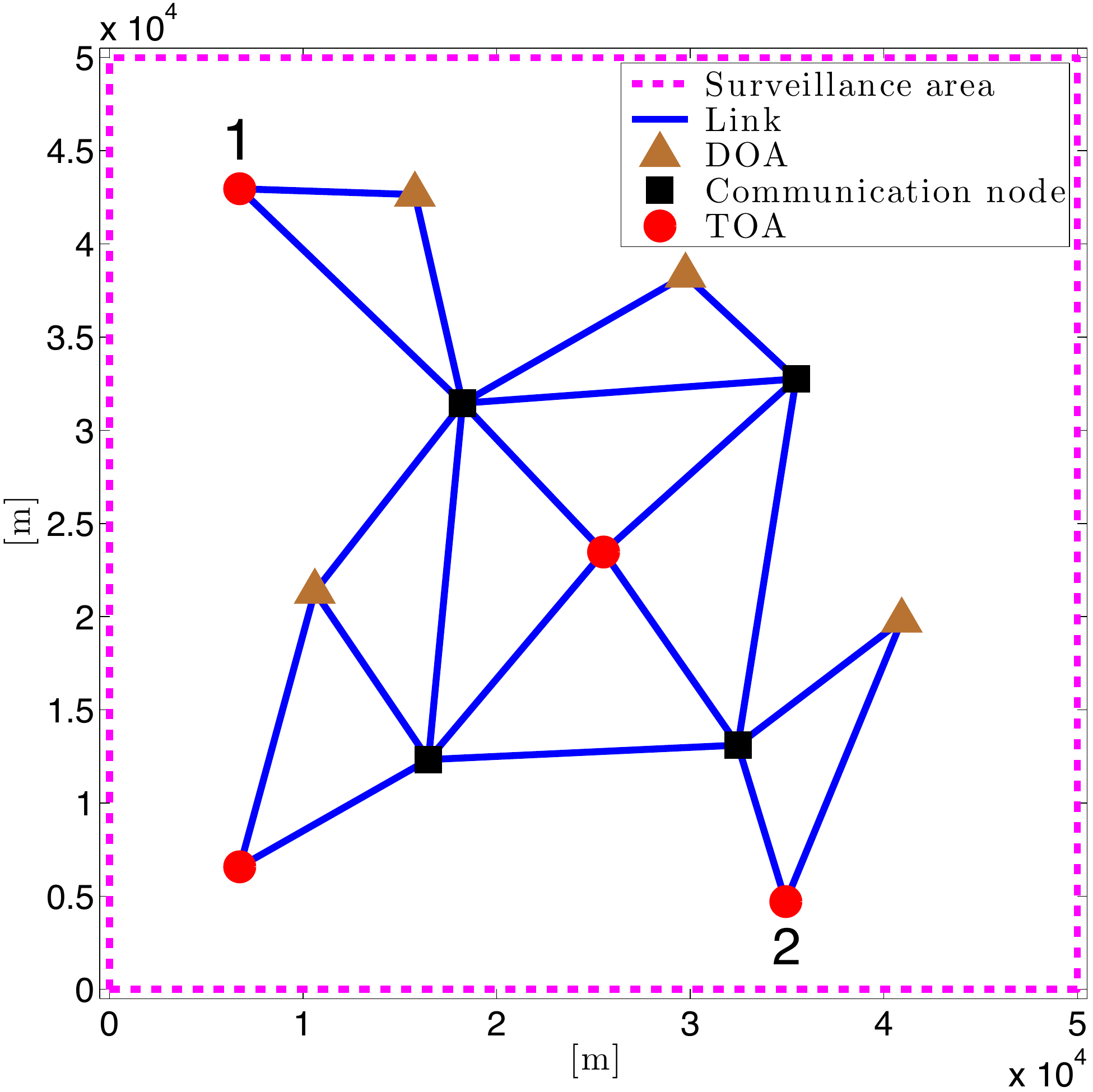}
	\caption{{Network with $4$ DOA, $4$ TOA sensors and $4$ communication nodes.}}
\label{fig:net-4c4t4d}
\end{figure}

The following MM algorithms have been compared: centralized GPB$_1$ and IMM, distributed GPB$_1$ and IMM (i.e. DGPB$_1$ and DIMM) as well as the DIMM-CL algorithm of reference \cite{li-jia}.
The performance of the various algorithms is measured by the PRMSE. Notice that averaging is carried out over time and Monte Carlo trials for the centralized algorithms, while further averaging over network nodes is applied for distributed algorithms.
Table \ref{tab:results2} summarizes performance of the various algorithms using the network of Fig. \ref{fig:net-4c4t4d}.
\begin{table}[h!]
	\setlength\arrayrulewidth{0.5pt}\arrayrulecolor{black} 
	\setlength\doublerulesep{0.5pt}\doublerulesepcolor{black} 
	\caption{Performance comparison for the sensor network of fig. \ref{fig:net-4c4t4d}}
	\label{tab:results2}
	\centering
	~\\
	\scalebox{1}{
	\begin{tabular}{c|c|c|}
		& \multicolumn{1}{>{\columncolor[gray]{.95}}c|}{GPB$_{1}$} & \multicolumn{1}{>{\columncolor[gray]{.95}}c|}{IMM}\\
		\hline
		\multicolumn{1}{>{\columncolor[gray]{.95}}c|}{PRMSE $[m]$} & $538$ & $446$\\
		\hline
	\end{tabular} }~\\\vspace{0.5em}
	\scalebox{1}{
	\begin{tabular}{c|c|c|}
		PRMSE $[m]$ & \multicolumn{1}{>{\columncolor[gray]{.95}}c|}{DGPB$_{1}$} & \multicolumn{1}{>{\columncolor[gray]{.95}}c|}{DIMM}\\
		\hline
		\multicolumn{1}{>{\columncolor[gray]{.95}}c|}{\textit{$L = 1$}} & $1603$ & $960$\\
		\hline
		\multicolumn{1}{>{\columncolor[gray]{.95}}c|}{\textit{$L = 2$}} & $1312$ & $781$\\
		\hline
		\multicolumn{1}{>{\columncolor[gray]{.95}}c|}{\textit{$L = 3$}} & $183$ & $711$\\
		\hline
		\multicolumn{1}{>{\columncolor[gray]{.95}}c|}{\textit{$L = 4$}} & $1103$ & $672$\\
		\hline
		\multicolumn{1}{>{\columncolor[gray]{.95}}c|}{\textit{$L = 5$}} & $1052$ & $648$\\
		\hline
	\end{tabular} }
\end{table}
Notice that the PRMSEs of DIMM-CL are not included in Table \ref{tab:results2} since for all the considered values of $L$ such an algorithm exhibited a divergent behaviour. 
{These results are actually consistent with the performance evaluation in \cite{li-jia} where, considering circular networks with $N \geq 4$ nodes, a minimum number of $L=10$ is required to track the object. Moreover, considering the network in Fig. \ref{fig:net-4c4t4d} without communication nodes, $L=80$ consensus steps are required for non-divergent behaviour of the DIMM-CL in \cite{li-jia}. Notice that such a number of consensus steps might still be too large for practical applications, which is also stated in \cite{li-jia}.}
As it can be seen from table \ref{tab:results2}, DIMM performs significantly better than DGPB$_{1}$ and, due to the presence of communication nodes, the distributed case exhibits significantly worse performance compared to the centralized case.
On the other hand it is evident that, even in presence of communication nodes (bringing no information about the object position) and of a highly maneuvering object, distributed MM algorithms still work satisfactorily by distributing the little available information throughout the network by means of consensus. 
In fact, despite these difficulties, the object can still be tracked with reasonable confidence. 
{Fig. \ref{fig:estimationPlot} shows, for $L=5$, the differences between two nodes that are placed almost at opposite positions in the network.} It is noticeable that the main difficulty in tracking by communication nodes arises when a turn is approaching, especially, in this trajectory, during the last turn. Clearly, one should use a higher number of consensus steps $L$ to allow data to reach two distant points in the network and, hence, to have similar estimated trajectory and mode probabilities.
\begin{figure}[h!]
	\centering
	\includegraphics[width=0.7\columnwidth]{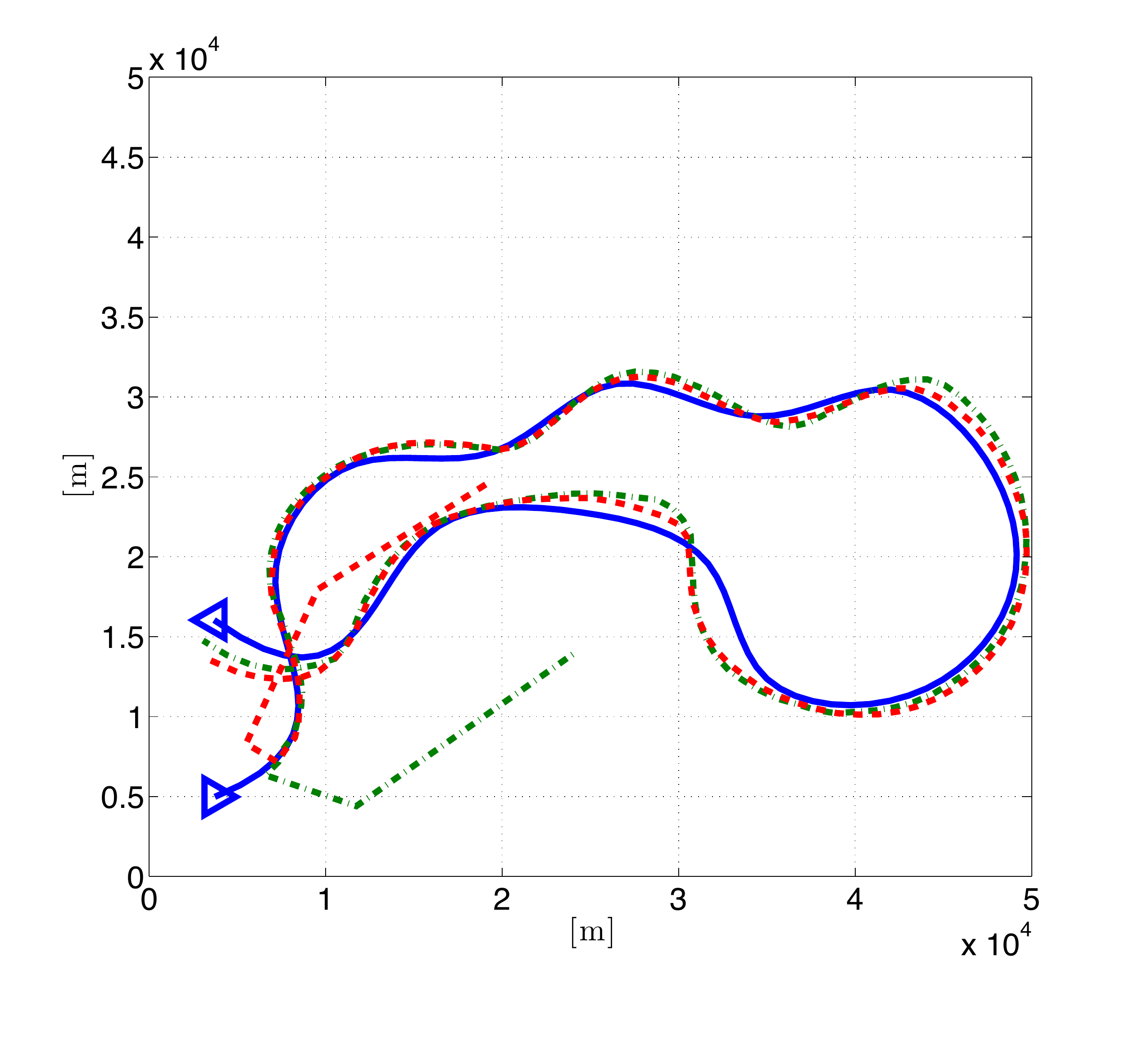}
	\caption{{Estimated object trajectories of node 1 (TOA, green dash-dotted line) and node 2 (TOA, red dashed line) in the sensor network of fig. \ref{fig:net-4c4t4d}. The real object trajectory is the blue continuous line.}}
\label{fig:estimationPlot}
\end{figure}

% DISTRIBUTED MULTI-OBJECT FILTERING
\chapter{Distributed multi-object filtering}
\label{chap:dmof}
\spminitoc
The focus of the next sections is on \textit{Distributed MOF} (DMOF) over the network of section \ref{sec:net} wherein each node (tracking agent) locally updates multi-object information exploiting the multi-object dynamics and the available local measurements, exchanges such information with communicating agents and then carries out a fusion step in order to combine the information from all neighboring agents.
Specifically, a CPHD filtering approach (see subsection \ref{sec:cphd}) to DMOF will be adopted \cite{ccphd}.
Hence, the agents will locally update and fuse the cardinality PMF $\cd{\cdot}$ and the location PDF $\lpdf{}{}{\cdot}$ (see (\ref{eq:iidcrfspdf})) that, for the sake of brevity, will also be referred to as the CPHD.

More formally, the DMOF problem over the network $\left( \mathcal{N}, \mathcal{A} \right)$ can be stated as follows.
Each node $i \in \mathcal{N}$ must estimate at each time $k \in \{ 1, 2, \dots \}$ the CPHD of the unknown multi-object set $X_{k}$ in (\ref{eq:MTD})-(\ref{eq:MM2}) given local measurements $Y_{\kappa}^{i}$ for all $\kappa \geq 1$ up to time $k$ and data received from all adjacent nodes $j \in \mathcal{N}^{i} \backslash \left\{ i \right\}$ so that the estimated pair $\left( \cdd[k]{i}{\cdot}, \lpdf{k}{i}{\cdot} \right)$ be as close as possible to the one that would be provided by a centralized CPHD filter simultaneously processing information from all nodes.

Clearly local CPHD filtering is the same as the centralized one of (\ref{eq:cphdpred})-(\ref{eq:cphdcor}), operating on the local measurement set $Y_{k}^{i}$ instead of $Y_{k}$. Conversely, CPHD fusion deserves special attention and will be tackled in the next sections.

\section{Multi-object information fusion via consensus}
\label{sec:moinfofusion}
Suppose that, in each node $i$ of the  sensor network, a multi-object density $f^{i}\!\left( X \right)$ is available
which has been computed on the basis of the information (collected locally or propagated from other nodes) which has become available to node $i$.
Aim of this section is to investigate whether it is possible to devise a suitable distributed algorithm guaranteeing that all the nodes of the network reach an agreement regarding the multi-object density of the unknown multi-object set.
An important feature the devised algorithm should enjoy is \textit{scalability}, i.e. the processing load of each node must be independent of the network size. For this reason, approaches based on the optimal (Bayes) fusion rule \cite{chmoch1990,mah2000} are ruled out. In fact, they would require, for each pair of adjacent nodes $(i,j)$, the knowledge of the multi-object density $f\!\left( X | I^{i} \cap I^{j} \right)$ conditioned to the common information $I^{i} \cap I^{j}$ and, in a practical network,  it is impossible to keep track of such a common information in a scalable way. Hence, some robust suboptimal fusion technique has to be devised.

\subsection{Multi-object Kullback-Leibler average}
\label{ssec:KLA}
The first important issue to be addressed is how to define the average of the local multi-object densities $f^{i}(X)$. 
To this end, taking into account the benefit provided by the RFS approach about defining multi-object densities, it is possible to extend the notion of KLA of single-object PDFs introduced in subsection \ref{ssec:klapdf} to multi-object ones.
Let us first define the notion of KLD to multi-object densities $f\!\left( X \right)$ and $g\!\left( X \right)$ by
\be
	D_{KL} \left( f \parallel g \right) \triangleq \displaystyle \int f\!\left( X \right) \, \log\!\left( \dfrac{f(X)}{g(X)}\right) \delta X
\label{KLD}
\ee
where the integral in (\ref{KLD}) must be interpreted as a set integral according to the definition (\ref{eq:unlsetint}).
Then, the weighted KLA $f_{KLA}\!\left( X \right)$ of the multi-object densities $f^{i}(X)$ is defined as follows
\be
	f_{KLA} = \arg \inf_{f} \displaystyle \sum_{i \in \ncal} \, \omega^{i} D_{KL}\!\left( f \parallel f^{i} \right).
\label{KLA}
\ee
with weights $\omega^{i}$ satisfying
\be
	\omega^{i} \geq 0 \, , \qquad \displaystyle \sum_{i \in \ncal} \omega^{i} = 1.
\label{weights}
\ee
Notice from (\ref{KLA}) that the weighted KLA of the agent densities is the one that minimizes the weighted sum of distances from such densities.
In particular, if $N$ is the number of agents and $\omega^{i}= 1/N$ for $i=1,\dots,N$, (\ref{KLA}) provides the (unweighted) KLA which averages the agent densities giving to all of them the same level of confidence.
An interesting interpretation of such a notion can be given recalling that, in Bayesian statistics, the KLD (\ref{KLD}) can be seen as the information gain achieved when moving from a prior $g(X) $ to a posterior $f(X)$. 
Thus, according to (\ref{KLA}), the average PDF is the one that minimizes the sum of the information gains from the initial multi-object densities.
Thus, this choice is coherent with the {\em Principle of Minimum Discrimination Information} (PMDI) according to which the probability density which best represents the current state of knowledge is the one which produces an information gain as small as possible (see \cite{Campbell,Akaike} for a discussion on such a principle and its relation with Gauss' principle and maximum likelihood estimation) or, in other words \cite{Jaynes}:
\begin{quotation}
\noindent``the probability assignment which most honestly describes what we know should be the most conservative assignment
in the sense that it does not permit one to draw any conclusions not warranted by the data''.
\end{quotation}
The adherence to the PMDI is important in order to counteract the so-called \textit{data incest} phenomenon, i.e. the unaware reuse of the same piece of information due to the presence of loops within the network. 

The following result holds.
\begin{thm}[Multi-object KLA]
\label{thm:KLA:GCI}
	~\\
	The weighted KLA defined in (\ref{KLA}) turns out to be given by
	\begin{equation}
		f_{KLA}  \left( X \right) = \dfrac{\displaystyle{\prod_{i \in \ncal}}\, \left[ f^{i} \left( X \right) \right]^{\omega^{i}}}{\displaystyle{\int} \displaystyle{\prod_{i \in \ncal}} \left[ f^{i} \left( X \right) \right]^{\omega^{i}} \delta X}
	\label{KLA:GCI}
	\end{equation}
\end{thm}

\subsection{Consensus-based multi-object filtering}
\label{ssec:consensusmof}
The identity (\ref{KLA:GCI}) with $\omega^{i} = 1/|\mathcal{N}|$ can be rewritten as
\begin{equation}\label{global:KLA}
f_{KLA}  \left( X \right) = \bigoplus_{i \in \mathcal N} \left ( \frac{1}{|\mathcal N|} \odot f^{i}( X ) \right ) \, .
\end{equation}
Then, the global (collective) KLA (\ref{global:KLA}), 
which would require all the local multi-object densities to be available, can be computed in a distributed and scalabale way by 
iterating  regional averages through the consensus algorithm
\begin{equation}
	{\hat{f}}^{i}_{{l}+1} (X) = \displaystyle{\bigoplus_{j \in \mathcal{N}^{i}}} \left ( \omega^{i,j} \odot {\hat{f}}^{j}_{{l}} (X) \right ) \, , \qquad \forall i \in \mathcal{N}
\label{consensus:RS}
\end{equation}
with ${\hat{f}}^{i}_{0} (X) = {{f}}^{i} (X)$.
In fact, thanks to properties \textsc{p.a} - \textsc{p.f}, it can be seen that\footnote{Notice that the scalar multiplication operator is defined only for strictly positive scalars.
However, in equation (\ref{consensus:RS2}) it is admitted that some of the scalar weights $\omega^{i,j}_{l}$ be zero. The understanding is that, whenever $\omega^{i,j}_{l}$ is equal to zero, the corresponding multi-object density  ${{f}}^{j} (X)$ is omitted from the addition. This can always be done, since for each $i \in \mathcal N$ and each ${l}$, at least one of the weights $\omega^{i,j}_{l}$ is strictly positive.}
\begin{equation}
	\hat{f}^{i}_{{l}}(X) = \displaystyle{\bigoplus_{j \in \mathcal{N}}} \left ( \omega^{i,j}_{l} \odot {{f}}^{j} (X) \right ) \, , \qquad\forall i \in \mathcal{N}
\label{consensus:RS2}
\end{equation}
where $\omega^{i,j}_{l}$ is defined as the element $(i,j)$ of the matrix $\Omega^{l}$. With this respect, recall that when the consensus weights $\omega^{i,j}$
are chosen so as to ensure that the matrix $\Omega$ is doubly stochastic, one has
\be
	\lim_{{l} \rightarrow + \infty} \omega^{i,j}_{l} = \frac{1}{|\mathcal N|} \, , \quad \forall i,j \in \mathcal N.
\ee
Hence, as the number of consensus steps increases, each local multi-object density ``tends'' to the global KLA (\ref{global:KLA}).

\subsection{Connection with existing approaches}
It is important to point out that the fusion rule (\ref{KLA:GCI}), which has been derived as KLA of the local multi-object densities, coincides with the Chernoff fusion \cite{info,mori1} known as \textit{Generalized Covariance Intersection} (GCI) for multi-object fusion, first proposed by Mahler \cite{mah2000}.
In particular, it is clear that ${f}_{KLA}\!(X)$ in (\ref{KLA:GCI}) is nothing but the NWGM of the agent multi-object densities $f^{i}\!\left( X \right)$; it is also called \textit{Exponential Mixture Density} (EMD) \cite{unjuclri2010,unclju2011}.
\begin{rem}
The name GCI stems from the fact that (\ref{KLA:GCI}) is the multi-object counterpart of the analogous fusion rule for (single-object) PDFs \cite{mah2000,hurley2002,jubauh2006,julier2008} which, in turn, is a generalization of CI originally conceived \cite{juluhl1997} for Gaussian PDFs.
\end{rem}
\noindent Remind that, given estimates $\hat{x}_{i}$ of the same quantity $x$ from multiple estimators with relative covariances $P_{i}$ and unknown correlations, their CI fusion is given by 
\bie
	P & = & \left( \displaystyle{\sum_{i \in \ncal}} \, \omega^{i} P_{i}^{-1} \right)^{-1}\IEEEnonumber\\\label{CI}\\
	\!\hat{x} & = & P \, \displaystyle{\sum_{i \in \ncal}} \, \omega^{i} P_{i}^{-1} \hat{x}_{i}\IEEEnonumber
\eie
The peculiarity of (\ref{CI}) is that, for any choice of the weights $\omega^{i}$ satisfying (\ref{weights}) and provided that all estimates are consistent in the sense that
\be
	E\!\left[ \left( x - \hat{x}_{i} \right)  \left( x - \hat{x}_{i} \right)^{\top} \right] \leq P_{i} \, , \qquad \forall i
\ee
then the fused estimate also turns out to be consistent, i.e.
\be
	E\!\left[ \left( x - \hat{x} \right)  \left( x - \hat{x} \right)^{\top} \right] \leq P \, .
\ee
It can easily be shown that, for normally distributed estimates, (\ref{CI}) is equivalent to
\be
	p (x) = \dfrac{\displaystyle{\prod_{i \in \ncal}} \left[ p^{i}(x) \right]^{\omega^{i}}}
        	              	  {\displaystyle{\int} \displaystyle{\prod_{i \in \ncal}} \left[ p^{i}(x) \right]^{\omega^{i}} dx}
\label{GCI}
\ee 
where $p^{i}(\cdot) \triangleq \mathcal{N}\!\left( \cdot; \hat{x}_{i}, P_{i} \right)$ is the Gaussian PDF with mean $\hat{x}_{i}$ and covariance $P_{i}$.
This suggested to use (\ref{GCI}) for arbitrary, possibly non Gaussian, PDFs. As a final remark, notice that the consistency property, which was the 
primary motivation that led to the development of the CI fusion rule, is in accordance with the PMDI discussed in Section \ref{ssec:KLA}, which indeed represents
one of the main positive features of the considered multi-object fusion rule.

\subsection{CPHD fusion}
\label{ssec:ccphd:gmfusion}
Whenever the object set is modelled as an i.i.d. cluster process, the agent multi-object densities to be fused take the form
\be
	f^{i} \left( X \right) = | X |! \,\, \cdd{i}{|X|} \, \displaystyle{\prod_{x \in X}} \lpdf{}{i}{x}
\label{agent-iid}
\ee
where $\left( \cdd{i}{n}, \lpdf{}{i}{x} \right)$ is the CPHD of agent $i$.
In \cite{cljumhri2010} it is shown that in this case the GCI fusion (\ref{KLA:GCI}) yields
\be
	\overline{f} \left( X \right) = | X |! \,\, \overline{\rho}\!\left( | X | \right) \, \displaystyle{\prod_{x \in X}} \, \overline{s} \left( x \right)
\label{fused-iid}
\ee
where
\bie
	\overline{s}(x) & = & \dfrac{\displaystyle{\prod_{i \in \ncal}} \left[ \lpdf{}{i}{x} \right]^{\omega^{i}}}{\displaystyle{\int} \displaystyle{\prod_{i \in \ncal}} \left[ \lpdf{}{i}{x} \right]^{\omega^{i}} dx} \, ,\label{loc} \\
	\overline{\rho}\!\left( n \right) & = & \dfrac{\displaystyle{\prod_{i \in \ncal}} \left[ \cdd{i}{n} \right]^{\omega^{i}} \left\{ \displaystyle{\int}\displaystyle{\prod_{i \in \ncal}} \left[ \lpdf{}{i}{x} \right]^{\omega^{i}} dx \right\}^n}{\displaystyle{\sum_{j=0}^\infty} \displaystyle{\prod_{i \in \ncal}} \left[ \cdd{i}{j} \right]^{\omega^{i}} \left\{ \displaystyle{\int}  \displaystyle{\prod_{i \in \ncal}} \left[ \lpdf{}{i}{x} \right]^{\omega^{i}} dx \right\}^{j}} \, .\label{card} 
\eie
In words, (\ref{fused-iid})-(\ref{card}) amount to state that the fusion of i.i.d. cluster processes provides an i.i.d. cluster process whose location density $\overline{s}(\cdot)$ is the weighted geometric mean of the agent location densities $\lpdf{}{i}{\cdot}$, while the fused cardinality $\overline{\rho}(\cdot)$ is obtained by the more complicated expression (\ref{card}) also involving the agent location PDFs besides the agent cardinality PMFs.

Please notice that, in principle, both the cardinality PMF and the location PDF are infinite-dimensional.
For implementation purposes, finite-dimensional parametrizations of both need to be adopted.
As far as the cardinality PMF $\cd{n}$ is concerned, it is enough to assume a sufficiently large maximum number of objects $n_{max}$ present in the scene and restrict $\cd{\cdot}$ to the finite subset of integers $\{ 0, 1, \dots, n_{max} \}$.
As for the location PDF, two finitely-parameterized representations based on the SMC or, respectively, GM approaches are most commonly adopted.
The SMC approach consists of representing location PDFs as linear combinations of delta Dirac functions, i.e.
\be
	\lpdf{}{}{x} = \displaystyle{\sum_{j=1}^{N_p}} w_{j} \, \delta \left( x - x_{j} \right)
\ee
Conversely, the GM approach expresses location PDFs as linear combinations of Gaussian components, i.e.
\be
	\lpdf{}{}{x} = \displaystyle{\sum_{j=1}^{N_G}} \alpha_{j} \, \mathcal{N}\!\left( x; \hat{x}_{j}, P_{j} \right)
\ee
Recent work \cite{unjuclri2010} has presented an implementation of the distributed multi-object fusion following the SMC approach.
For DMOF over a sensor network, typically characterized by limited processing power and energy resources of the individual nodes, it is of paramount importance to reduce as much as possible local (in-node) computations and inter-node data communication.
In this respect, the GM approach promises to be more parsimonious (usually the number of Gaussian components involved is orders of magnitude lower
than the number of particles required for a reasonable tracking performance) and hence preferable. 
For this reason, a GM implementation of the fusion (\ref{loc}) has been adopted in the present work which, further, exploits consensus in order to carry out the fusion in a fully distributed way.
For the sake of simplicity, let us consider only two agents, $a$ and $b$, with GM location densities
\be
	\lpdf{}{i}{x} = \displaystyle{\sum_{j=1}^{N_G^{i}}} \alpha_{j}^{i} \, \mathcal{N}\!\left( x; \hat{x}_{j}^{i}, P_{j}^{i} \right) \, , \qquad i = a, b
\label{2-agent-GMs}
\ee
A first natural question is whether the fused location PDF 
\be
	\overline{s}(x) = \dfrac{ \displaystyle \left[ \lpdf{}{a}{x} \right]^{\omega} \left[ \lpdf{}{b}{x} \right]^{1-\omega}  }
        		                            { \displaystyle \int \left[ \lpdf{}{a}{x} \right]^{\omega} \left[ \lpdf{}{b}{x} \right]^{1-\omega} dx }
\label{2-agent-fusion}
\ee
is also a GM.  
Notice that (\ref{2-agent-fusion}) involves exponentiation and multiplication of GMs. 
To this end, it is useful to draw the following observations concerning elementary operations on Gaussian components and mixtures.
\begin{itemize}
\item The power of a Gaussian component is a Gaussian component, more precisely
\be
	\left[ \alpha \, \mathcal{N}\!\left( x; \hat{x}, P \right) \right]^{\omega} = \alpha^{\omega} \, \beta\!\left( \omega, P \right) \, \mathcal{N}\!\left( x; \hat{x}, \dfrac{P}{\omega} \right)
\label{power}
\ee
where
\be
	\beta\!\left( \omega, P \right) \triangleq \dfrac{\left[ \operatorname{det} \! \left( 2 \pi P \omega^{-1} \right) \right]^{\frac{1}{2}}}{\left[ \operatorname{det} \! \left( 2 \pi P \right) \right]^{\frac{\omega}{2}}}
\ee
\item The product of Gaussian components is a Gaussian component, more precisely \cite{willmay2005,will2003}
\be
	\alpha_1 \, \mathcal{N}\!\left( x; \hat{x}_1, P_1 \right) \cdot \alpha_2 \, \mathcal{N}\!\left( x; \hat{x}_2, P_2 \right) = \alpha_{12} \, \mathcal{N}\!\left( x; \hat{x}_{12}, P_{12} \right)
\label{product}
\ee
where
\begin{eqnarray}
P_{12} & = & \left( P_1^{-1} + P_2^{-1} \right)^{-1} \label{naive1} \\
\hat{x}_{12} & = &  P_{12} \left( P_1^{-1} \hat{x}_1 + P_2^{-1} \hat{x}_2 \right) \vspace{1mm} \label{naive2} \\
\alpha_{12} & = & \alpha_1 \, \alpha_2 \, \mathcal{N}\!\left(  \hat{x}_1 - \hat{x}_2; 0, P_1 + P_2 \right)   
\label{alpha_12}
\end{eqnarray}
\item Due to (\ref{product}) and the distributive property, the product of GMs is a GM. In particular, if $\lpdf{}{a}{\cdot}$ and $\lpdf{}{b}{\cdot}$ have 
$N_G^{a}$ and, respectively, $N_G^{b}$ Gaussian components, then $\lpdf{}{a}{\cdot}\lpdf{}{b}{\cdot}$ will have $N_G^{a} N_G^{b}$ components.
\item Exponentiation of a GM does not provide, in general, a GM. 
\end{itemize}
As a consequence of the latter observation, the fusion (\ref{2-agent-GMs})-(\ref{2-agent-fusion}) does not provide a GM.
Hence, in order to preserve the GM form of the location PDF throughout the computations a suitable approximation of the GM exponentiation has to be devised.
In \cite{jul2006} it is suggested to use the following approximation
\be
	\left[ \displaystyle{\sum_{j=1}^{N_G}} \alpha_{j} \, \mathcal{N}\!\left( x; \hat{x}_{j}, P_{j} \right) \right]^{\omega} \cong \displaystyle{\sum_{j=1}^{N_G}} \left[ \alpha_{j} \mathcal{N}\!\left( x; \hat{x}_{j}, P_{j} \right) \right]^{\omega} = \displaystyle{\sum_{j=1}^{N_G}} \alpha_{j}^{\omega} \, \beta\!\left( \omega, P_{j} \right) \, \mathcal{N}\!\left( x; \hat{x}_{j}, \dfrac{P_{j}}{\omega} \right)
\label{approx}
\ee
As a matter of fact, the above approximation seems reasonable whenever the cross-products of the different terms in the GM are negligible for all $x$; this, in turn, holds provided that the centers $\hat{x}_{i}$ and $\hat{x}_{j}$, $i \neq j$, of the Gaussian components are well separated, as measured by the respective covariances $P_{i}$ and $P_{j}$.
In geometrical terms, the more separated are the confidence ellipsoids of the Gaussian components the smaller should be the error involved in the approximation (\ref{approx}).
In mathematical terms, the conditions for the validity of (\ref{approx}) can be expressed in terms of Mahalanobis distance \cite{mahal1925} inequalities
of the form
$$
\ba{rcl}
 \left( \hat{x}_{i} - \hat{x}_{j} \right)^{\top} P^{-1}_{i} \left( \hat{x}_{i} - \hat{x}_{j} \right) & \gg & 1 \vspace{1mm} \\
 \left( \hat{x}_{i} - \hat{x}_{j} \right)^{\top} P^{-1}_{j} \left( \hat{x}_{i} - \hat{x}_{j} \right) & \gg & 1  
\ea
$$
Provided that the use of (\ref{approx}) is preceded by a suitable merging step that fuses Gaussian components with Mahalanobis (or other type of)
distance below a given threshold, the approximation seems reasonable.
To this end, the merging algorithm proposed by Salmond \cite{salm1988,salm1990} represents a good tool. 

Exploiting (\ref{approx}), the fusion (\ref{2-agent-fusion}) can be approximated as follows:
\be
	\overline{s}(x) = \dfrac{\displaystyle{\sum_{i=1}^{N_G^{a}}} \displaystyle{\sum_{j=1}^{N_G^{b}}} \, \alpha_{ij}^{ab} \, \mathcal{N}\!\left( x; \hat{x}_{ij}^{ab}, P_{ij}^{ab} \right)}  {\displaystyle{\int}  \displaystyle{\sum_{i=1}^{N_G^{a}}} \displaystyle{\sum_{j=1}^{N_G^{b}}} \, \alpha_{ij}^{ab} \, \mathcal{N}\!\left( x;  \hat{x}_{ij}^{ab}, P_{ij}^{ab} \right) dx} = \dfrac{\displaystyle{\sum_{i=1}^{N_G^{a}}} \displaystyle{\sum_{j=1}^{N_G^{b}}} \, \alpha_{ij}^{ab} \, \mathcal{N}\!\left( x; \hat{x}_{ij}^{ab}, P_{ij}^{ab} \right)} { \displaystyle{\sum_{i=1}^{N_G^{a}}} \displaystyle{\sum_{j=1}^{N_G^{b}}} \, \alpha_{ij}^{ab} }
\label{f1}
\ee
where
\begin{IEEEeqnarray}{rCl}
P_{ij}^{ab} & = & \left[ \omega \left( P_{i}^{a} \right)^{-1} + (1-\omega) \left( P_{j}^{b} \right)^{-1} \right]^{-1}  \vspace{1mm} \label{f2} \\ \hat{x}_{ij}^{ab} & = & P_{ij}^{ab} \left[ \omega \left( P_{i}^{a} \right)^{-1} \hat{x}_{i}^{a} +
(1-\omega) \left( P_{j}^{b} \right)^{-1} \hat{x}_{j}^{b} \right]  \vspace{1mm} \label{f3} \\
\alpha_{ij}^{ab} & = & \left( \alpha_{i}^{a} \right)^{\omega} \, \left( \alpha_{j}^{b} \right)^{1-\omega} \beta\!\left( \omega, P_{i}^{a} \right) \beta\!\left( 1 - \omega, P_{j}^{b} \right) \mathcal{N}\!\left( \hat{x}_{i}^{a} - \hat{x}_{j}^{b}; \, 0, \, \frac{P_{i}^{a}}{\omega} + \frac{P_{j}^{b}}{1-\omega} \right) \label{f4}
\end{IEEEeqnarray}
Notice that (\ref{f1})-(\ref{f4}) amounts to performing a CI fusion on any possible pair formed by a Gaussian component of agent $a$ and a Gaussian component of agent $b$.
Further the coefficient $\alpha_{ij}^{ab}$ of the resulting (fused) component includes a factor $\mathcal{N}\!\left( \hat{x}_{i}^{a} - \hat{x}_{j}^{b}; \, 0, P_{i}^{a}/\omega + P_{j}^{b}/\left( 1 - \omega \right) \right)$ that measures the separation of the two fusing components $\left( \hat{x}_{i}^{a}, P_{i}^{a} \right)$ and $\left( \hat{x}_{j}^{b}, P_{j}^{b} \right)$. 
In (\ref{f1}) it is reasonable to remove Gaussian components with negligible coefficients $\alpha_{ij}^{ab}$.
This can be done either by fixing a threshold for such coefficients or by checking whether the Mahalanobis distance 
\be
	\sqrt{\left( x_{i}^{a} - x_{i}^{b} \right)^{\top} \left(\dfrac{P_{a}}{1 - \omega} + \dfrac{P_{b}}{\omega} \right)^{-1} \left( x_{i}^{a} - x_{i}^{b} \right)}
\ee
falls below a given threshold.
The fusion (\ref{2-agent-fusion}) can be easily extended to $N \ge 2$ agents by sequentially applying the pairwise fusion (\ref{f2})-(\ref{f3}) $N - 1$ times. Note that, by the associative and commutative properties of multiplication, the ordering of pairwise fusions is irrelevant.

\section{Consensus GM-CPHD filter}
\label{distributedgmcphd}

This section presents the proposed \textit{Consensus Gaussian Mixture-CPHD} (CGM-CPHD) filter algorithm.
The sequence of operations carried out at each sampling interval $k$ in each node $i \in \mathcal{N}$ of the network is reported in Table \ref{alg:dgmcphd}.
All nodes $i \in \mathcal{N}$ operate in parallel at each sampling interval $k$ in the same way, each starting from its own previous estimates of the cardinality PMF and location PDF in GM form
\be
	\left\{ \cdd[k-1]{i}{n} \right\}_{n=0}^{n_{max}} \, , \left\{ \left( \alpha_{j}^{i}, \hat{x}_{j}^{i}, P_{j}^{i} \right)_{k-1} \right\}_{j=1}^{\left(N_G^{i}\right)_{k-1}}
\ee
and producing, at the end of the various steps listed in Table \ref{alg:dgmcphd}, its new estimates of the CPHD as well as of the
object set, i.e.
\be
	\left\{ \cdd[k]{i}{n} \right\}_{n=0}^{n_{max}} \, , \left\{ \left( \alpha_{j}^{i}, \hat{x}_{j}^{i}, P_{j}^{i} \right)_{k} \right\}_{j=1}^{\left(N_G^{i}\right)_{k}} \, , \widehat{X}_{k}^{i}
\ee
A brief description of the sequence of steps of the CGM-CPHD algorithm is in order.
\begin{enumerate}
\item First, each node $i$ performs a local GM-CPHD filter update exploiting the multi-object dynamics and the local measurement set $Y_{k}^{i}$.
      The details of the GM-CPHD update (prediction and correction) can be found in \cite{vo-vo-cantoni}. 
      A merging step, to be described later, is introduced after the local update and before the consensus phase in order to reduce the number of Gaussian components and, hence, alleviate both the communication and the computation burden.
\item  Then, consensus takes place in each node $i$ involving the subnetwork $\mathcal{N}^{i}$.
       Each node exchanges information (i.e., cardinality PMF and GM representation of the location PDF) with the neighbors; more precisely node $i$ transmits its data to nodes $j$ such that $i \in \mathcal{N}^{j}$ and waits until it receives data from $j \in \mathcal{N}^{i} \backslash \{ i \}$.
       Next, node $i$ carries out the GM-GCI fusion in (\ref{card}) and (\ref{f1})-(\ref{f4}) over $\mathcal{N}^{i}$. Finally, a merging step is applied to reduce the joint communication-computation burden for the next consensus step.
       This procedure is repeatedly applied for an appropriately chosen number $L \geq 1$ of consensus steps. 
\item After the consensus, the resulting GM is further simplified by means of a pruning step to be described later.
      Finally, an estimate of the object set is obtained from the cardinality PMF and the pruned location GM via an 
      estimate extraction step to be described later.   
\end{enumerate}

\begin{table}[!h]
\renewcommand{\arraystretch}{1.3}
\caption{Distributed GM-CPHD (DGM-CPHD) filter pseudo-code}
\label{alg:dgmcphd}
\centering
\hrulefill
\hrule
\begin{algorithmic}[0]
	\Procedure{DGM-CPHD(node $i$, time $k$)}{}
		\State \bfbox{Local Filtering}
		\State \textsc{Local GM-CPHD Prediction} \Comment{See (\ref{eq:cphdpred}) and \cite{vo-vo-cantoni}}
		\State \textsc{Local GM-CPHD Correction} \Comment{See (\ref{eq:cphdcor}) and \cite{vo-vo-cantoni}}
		\State \textsc{Merging} \Comment{See Table \ref{alg:merging}}\vspace{0.5em}
		\State \bfbox{Information Fusion}
		\For{${l} = 1, \dots, L$}
			\State \textsc{Information Exchange}
			\State \textsc{GM-GCI Fusion} \Comment{See (\ref{card}) and (\ref{f1})}
			\State \textsc{Merging} \Comment{See Table \ref{alg:merging}}
		\EndFor\vspace{0.5em}
		\State \bfbox{Extraction}
		\State \textsc{Pruning} \Comment{See Table \ref{alg:pruning}}
		\State \textsc{Estimate Extraction} \Comment{See Table \ref{alg:extraction}}
	\EndProcedure
\end{algorithmic}
\hrule\hrule
\end{table}
While the local GM-CPHD update (prediction and correction) is thoroughly described in the literature \cite{vo-vo-cantoni} and the
GM-GCI fusion has been thoroughly dealt with in the previous section, the remaining steps (merging, pruning and estimate extraction) are
outlined in the sequel. 

\subsection*{Merging}
Recall \cite{vo-ma,vo-vo-cantoni} that the prediction step (\ref{eq:cphdpred}) and the correction step (\ref{eq:cphdcor}) of the PHD/CPHD filter make the number of Gaussian components increase. 
Hence, in order to avoid an unbounded growth with time of such components and make the GM-(C)PHD filter practically implementable, 
suitable component reduction strategies have to be adopted. 
In \cite[Section III.C, Table II]{vo-ma}, a reduction procedure based on truncation of components with low weights, merging of similar components and pruning, has been presented.
For use in the proposed CGM-CPHD algorithm, it is convenient to deal with merging and pruning in a separate way.
A pseudo-code of the merging algorithm is in Table \ref{alg:merging}.

\begin{table}[!h]
\renewcommand{\arraystretch}{1.3}
\caption{Merging pseudo-code}
\label{alg:merging}
\centering
\hrulefill\hrule
\begin{algorithmic}[0]
	\Procedure{Merging}{$\left\{ \hat{x}_{i}, P_{i}, \alpha_{i} \right\}_{i = 1}^{N_{G}}, \gamma_{m}$}
		\State $t = 0$
		\State $\mathcal{I} = \left\{ 1, \dots, N_{G} \right\}$
		\Repeat
			\State $t = t + 1$
			\State $j = \operatorname{arg \underset{i \in \mathcal{I}}{max}}\alpha_{i}$
			\\
			\State $\mathcal{MC} = \left\{ i \in \mathcal{I}  \Big | \left( \hat{x}_{j} - \hat{x}_{i} \right)^{\top} P_{i}^{-1} \left( \hat{x}_{j} - \hat{x}_{i} \right) \le \gamma_{m} \right\}$
			\\
			\State $\bar \alpha_{t} = \displaystyle \sum_{i \in \mathcal{MC}} \alpha_{i}$
			\\
			\State $\bar{x}_{t} = \dfrac{1}{\bar \alpha_{k}}\displaystyle\sum_{i \in \mathcal{MC}} \alpha_{i} \hat{x}_{i}$
			\\
			\State $\bar{P}_{t} = \dfrac{1}{\bar \alpha_{k}}\displaystyle\sum_{i \in \mathcal{MC}} \alpha_{i} \left[ P_{i} + \left( \bar{x}_{k} - \hat{x}_{i} \right)\left( \bar{x}_{k} - \hat{x}_{i} \right)^{\top} \right]$
			\\
			\State $\mathcal{I} = \mathcal{I} \backslash \mathcal{MC}$
		\Until{$\mathcal{I} \ne \emptyset$}
		\State \textbf{return} $\left\{ \bar{x}_{i}, \bar{P}_{i}, \bar \alpha_{i} \right\}_{i = 1}^{t}$
	\EndProcedure
\end{algorithmic}
\hrule\hrule
\end{table}

\subsection*{Pruning}
To prevent the number of Gaussian components from exceeding a maximum allowable value, say $N_{max}$, only the first 
$N_{max}$ Gaussian components are kept  while the other are removed from the GM.
The pseudo-code of the pruning algorithm is in Table \ref{alg:pruning}.

\begin{table}[!h]
\renewcommand{\arraystretch}{1.3}
\caption{Pruning pseudo-code}
\label{alg:pruning}
\centering
\hrulefill\hrule
\begin{algorithmic}[0]
	\Procedure{Pruning}{{$\left\{ \hat{x}_{i}, P_{i}, \alpha_{i} \right\}_{i = 1}^{N_{G}}, N_{max}$}}
		\If{$N_{G} > N_{max}$}
			\State $\mathcal{I} = \left\{ \mbox{indices of the } N_{max} \mbox{ Gaussian components with highest weights } \alpha_{i} \right\}$
			\State \textbf{return} $\left\{ \hat{x}_{i}, P_{i}, \alpha_{i} \right\}_{i \in \mathcal{I}}$
		\EndIf
		\State \textbf{return} $\left\{ \hat{x}_{i}, P_{i}, \alpha_{i} \right\}_{i = 1}^{N_{max}}$
	\EndProcedure
\end{algorithmic}
\hrule\hrule
\end{table}

Since the CPHD filter directly estimates the cardinality distribution, the estimated number of objects can be obtained via \textit{Maximum A Posteriori} (MAP) estimation, i.e.
\be
	\hat n_{k} = \underset{n}{\operatorname{max}} \, \cd[k]{n}
\label{MAP}
\ee
Given the MAP-estimated number of objects $\hat n_{k}$, estimate extraction is performed via the algorithm in Table \ref{alg:extraction} \cite[Section III.C, Table III]{vo-ma}.
The algorithm extracts the $\hat n_{k}$ local maxima (peaks) of the estimated location PDF keeping those for which the corresponding weights are above a preset threshold $\gamma_e$.

\begin{table}[!h]
\renewcommand{\arraystretch}{1.3}
\caption{Estimate extraction pseudo-code}
\label{alg:extraction}
\centering
\hrulefill\hrule
\begin{algorithmic}[0]
	\Procedure{Estimate Extraction}{{$\left\{ \hat{x}_{i}, P_{i}, \alpha_{i} \right\}_{i = 1}^{N_{G}}, \hat n_{k}, \gamma_{e}$}}
		\State $\mathcal{I} = \left\{ \mbox{indices of the } \hat n_{k} \mbox{ Gaussian components with highest weights } \alpha_{i} \right\}$
		\State $\hat X_{t} = \emptyset$
		\For{$i \in \mathcal{I}$}
			\If{$\alpha_{i} \hat n_{k} > \gamma_{e}$}
				\State $\hat X_{t} = \hat X_{t} \cup {\hat{x}_{i}}$
			\EndIf
		\EndFor
		\State \textbf{return}  $\hat X_{t}$
	\EndProcedure
\end{algorithmic}
\hrule\hrule
\end{table}

\section{Performance evaluation}
To assess performance of the proposed CGM-CPHD algorithm described in section \ref{distributedgmcphd}, a $2$-dimensional (planar)  multi-object tracking scenario is considered over a surveillance area of $50\times50 \, [km^2]$, wherein the sensor network of Fig. \ref{fig:4toa3doa} is deployed. 
The scenario consists of $6$ objects as depicted in Fig. \ref{fig:6trajectories}.

\begin{figure}[h!]
\centering
\includegraphics[width=0.6\columnwidth]{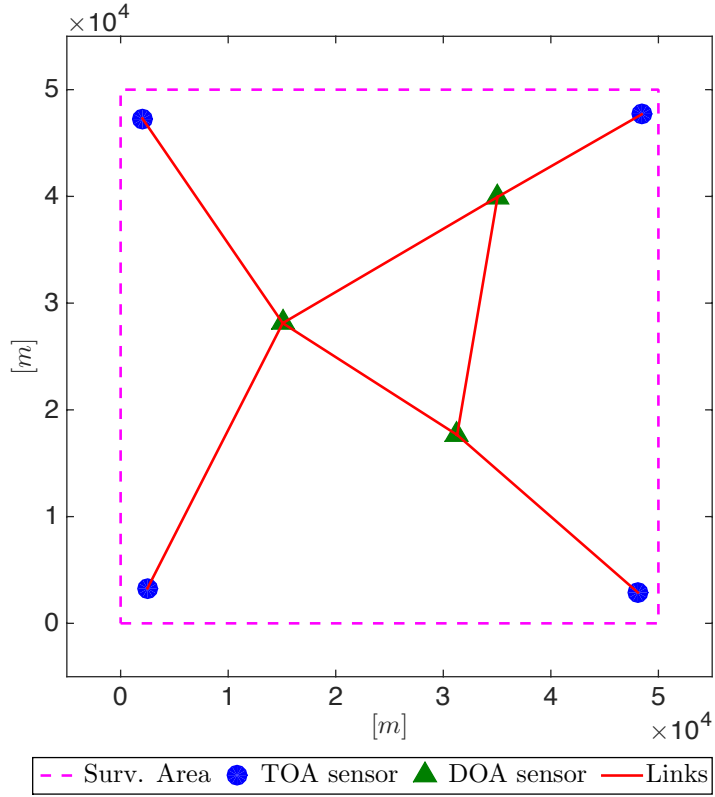}
\caption{Network with 7 sensors: 4 TOA and 3 DOA.}
\label{fig:4toa3doa}
\end{figure}

\begin{figure}[h!]
\centering
\includegraphics[width=0.6\columnwidth]{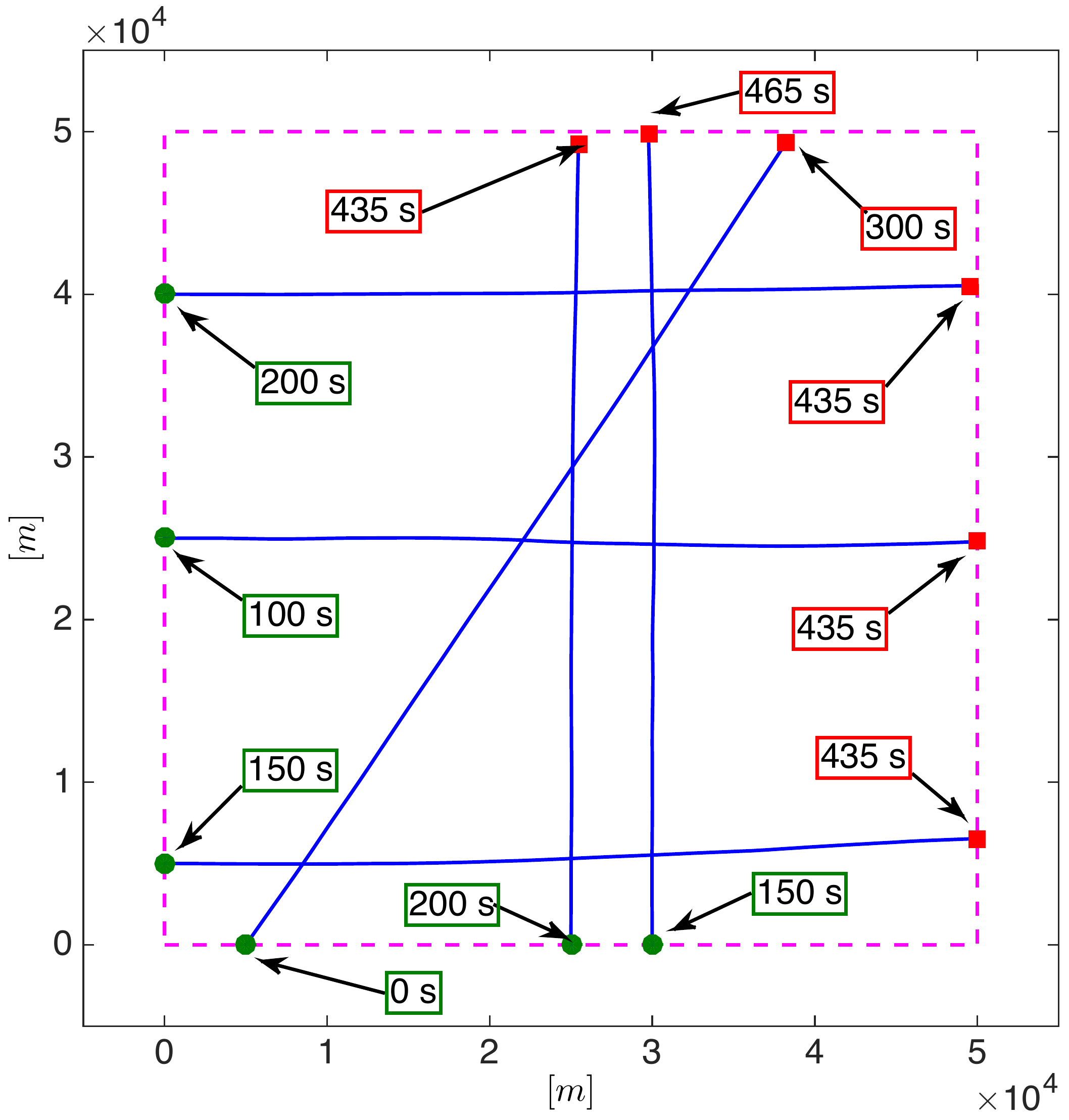}
\caption{object trajectories considered in the simulation experiment. The start/end point for each trajectory is denoted, respectively, by $\bullet\backslash\blacksquare$.}
\label{fig:6trajectories}
\end{figure}

The object state is denoted by $x = \left[ p_{x}, \, \dot{p}_{x}, \,p_{y}, \,\dot{p}_{y} \right]^{\top}$ where $(p_{x}, p_{y})$ and $( \dot{p}_{x}, \dot{p}_{y} )$ represent the object Cartesian position and, respectively, velocity components.
The motion of objects is modeled by the filters according to the nearly-constant velocity model:
\be
x_{k + 1} = \left[ \ba{cccc}
1 & T_{s} & 0 & 0	\\
0 & 1 	  & 0 & 0		\\
0 & 0 	  & 1 & T_{s} \\
0 & 0 	  & 0 & 1		\ea \right] x_{k} + w_{k} \, , \qquad
Q = \sigma_{w}^{2} \left[ \ba{cccc}
\frac{1}{4}T_{s}^{4} & \frac{1}{2}T_{s}^{3} & 0 & 0 \\
\frac{1}{2}T_{s}^{3} & T_{s}^{2} & 0 & 0 \\
0 & 0 & \frac{1}{4}T_{s}^{4} & \frac{1}{2}T_{s}^{3}\\
0 & 0 & \frac{1}{2}T_{s}^{3} & T_{s}^{2} \ea \right]
\ee
where $\sigma_{w} = 2 \, [m/s^{2}]$ and the sampling interval is $T_{s} = 5 \, [s]$.

As it can be seen from Fig. \ref{fig:4toa3doa}, the sensor network considered in the simulation consists of $4$range-only (\textit{Time Of Arrival}, TOA) and $3$ bearing-only (\textit{Direction Of Arrival}, DOA) sensors characterized by the following measurement functions:
\be
\begin{array}{rcl}
h^{i}(x) = \left\{ \ba{ll} \angle [ \left( p_{x} - x^{i} \right) + j \left( p_{y} - y^{i} \right)], & \mbox{DOA} \\[0.5em]
                                \sqrt{ \left( p_{x} - x^{i} \right)^2+ \left( p_{y} - y^{i} \right)^2 }, & \mbox{TOA}
\ea
                                \right.
\end{array}
\ee
where $( x^{i}, y^{i} )$ represents the known position of sensor $i$. The standard deviation of DOA and TOA measurement noises are taken respectively as $\sigma_{DOA} = 1\, [\mbox{}^{\circ}]$ and $\sigma_{TOA} = 100 \, [m]$. Because of the non linearity of the aforementioned sensors, the UKF \cite{juluhl2004} is exploited in each sensor in order to update means and covariances of the Gaussian components.

Clutter is modeled as a Poisson process with parameter $\lambda_{c} = 5$ and uniform spatial distribution over the surveillance area; the probability of object detection is $P_{D} = 0.99$.

In the considered scenario, objects pass through the surveillance area with no prior information for object birth locations. 
Accordingly, a $40$-component GM
\be
d_b(x) = \displaystyle \sum_{j = 1}^{40} \alpha_{j} \, \mathcal{N}\!\left( x; \hat{x}_{j}, P_{j} \right)
\ee
has been hypothesized for the birth intensity. Fig. \ref{fig:borderlineinit} gives a pictorial view of $d_b(x)$; notice that the center of the Gaussian components are regularly placed along the border of the surveillance region. For all components, the same covariance $P_{j} = \operatorname{diag}\!\left( 10^{6}, \,10^{4},10^{6}, \,10^{4} \right)$ and the same coefficient $\alpha_{j} = \alpha$, such that $n_{b} = 40 \alpha$ is the expected number of new-born objects, have been assumed.
\begin{figure}[h!]
\centering
\includegraphics[width=0.6\columnwidth]{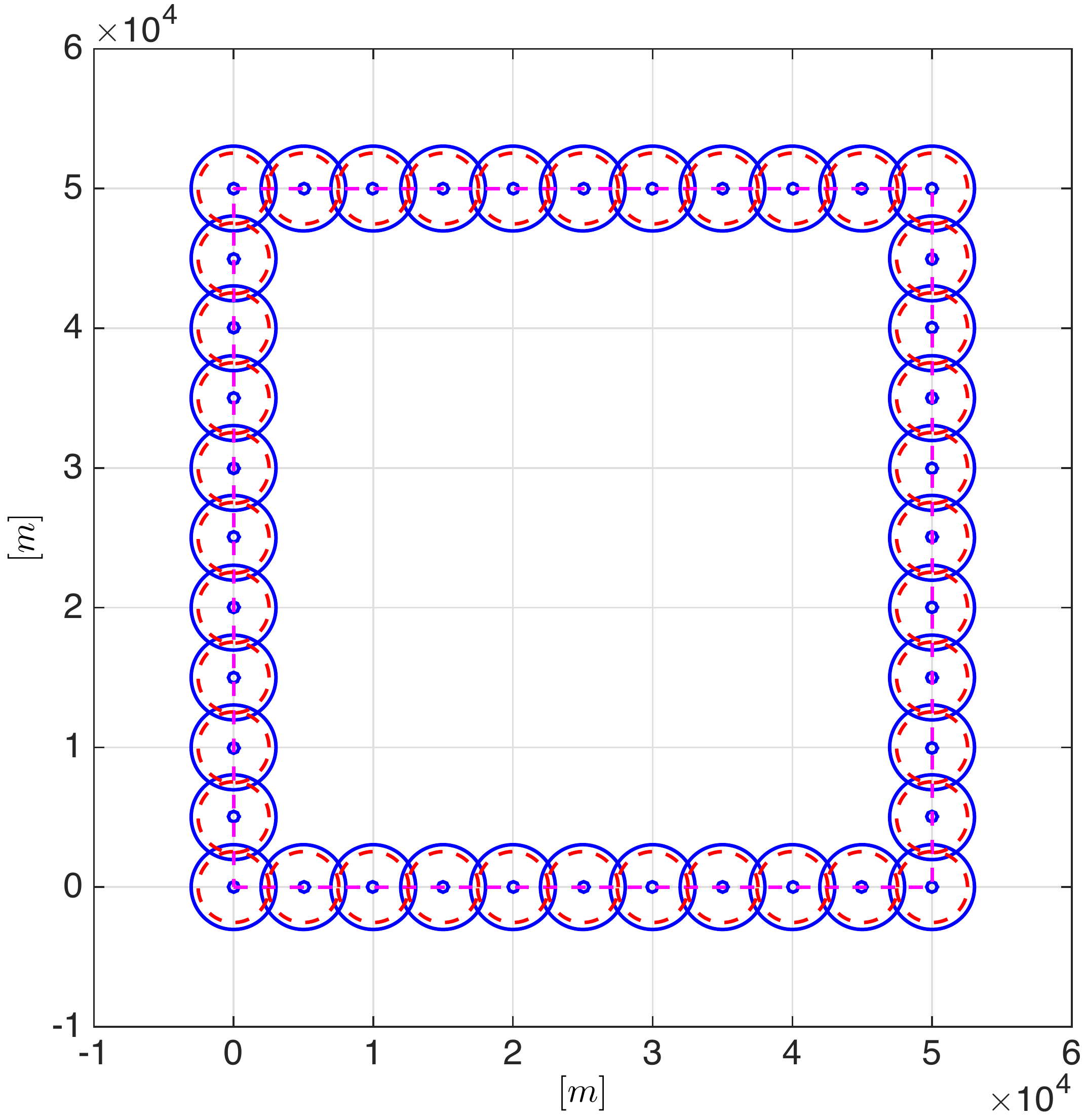}
\caption{Borderline position GM initialization. The symbol $\circ$ denotes the component mean while the blue solid line and the red dashed line are, respectively, their $3\sigma$ and $2\sigma$ confidence regions.}
\label{fig:borderlineinit}
\end{figure}

The proposed consensus CGM-CPHD filter is compared to an analogous filter, called \textit{Global GM-CPHD} (GGM-CPHD), that at each sampling interval performs a global fusion among all network nodes.
Multi-object tracking performance is evaluated in terms of the \textit{Optimal SubPattern Assignment} (OSPA) metric \cite{schvovo2008}.
The reported metric is averaged over $N_{mc} = 200$ Monte Carlo trials for the same object trajectories but different, independently generated, clutter and measurement noise realizations. The duration of each simulation trial is fixed to $500 \, [s]$ ($100$ samples).

The parameters of the CGM-CPHD and GGM-CPHD filters have been chosen as follows: the survival probability is $P_{S} = 0.99$; the maximum number of Gaussian components is $N_{max} = 25$; the merging threshold is $\gamma_{m} = 4$; the truncation threshold is $\gamma_{t} = 10^{-4}$; the extraction threshold is $\gamma_{e} = 0.5$; the weight of each Gaussian component of the birth PHD function is chosen as $\alpha_{j} = 1.5 \cdot 10^{-3}$.

Figs. \ref{fig:cardggmcphd}, \ref{fig:carddgmcphd1cs}, \ref{fig:carddgmcphd2cs} and \ref{fig:carddgmcphd3cs} display the statistics (mean and standard deviation) of the estimated number of objects obtained with CGM-CPHD (Fig. \ref{fig:cardggmcphd}) and CGM-CPHD with $L = 1$ (Fig. \ref{fig:carddgmcphd1cs}), $L = 2$ (Fig. \ref{fig:carddgmcphd2cs}) and $L = 3$ (Fig. \ref{fig:carddgmcphd3cs}) consensus steps. Fig. \ref{fig:6tospa} reports the OSPA metric (with Euclidean distance, $p = 2$, and cutoff parameter $c = 600$) for the same filters.
As it can be seen from Fig. \ref{fig:6tospa}, the performance obtained with three consensus steps is significantly better than with a single one, and comparable with the one given by the non scalable GGM-CPHD filter which performs a global fusion over all network nodes. Similar considerations hold for cardinality estimation.
These results show that by applying consensus, for a suitable number of steps, performance of distributed scalable algorithms is comparable to the one provided by the non scalable GGM-CPHD.

\begin{figure}[h!]
\centering
\includegraphics[width=\columnwidth]{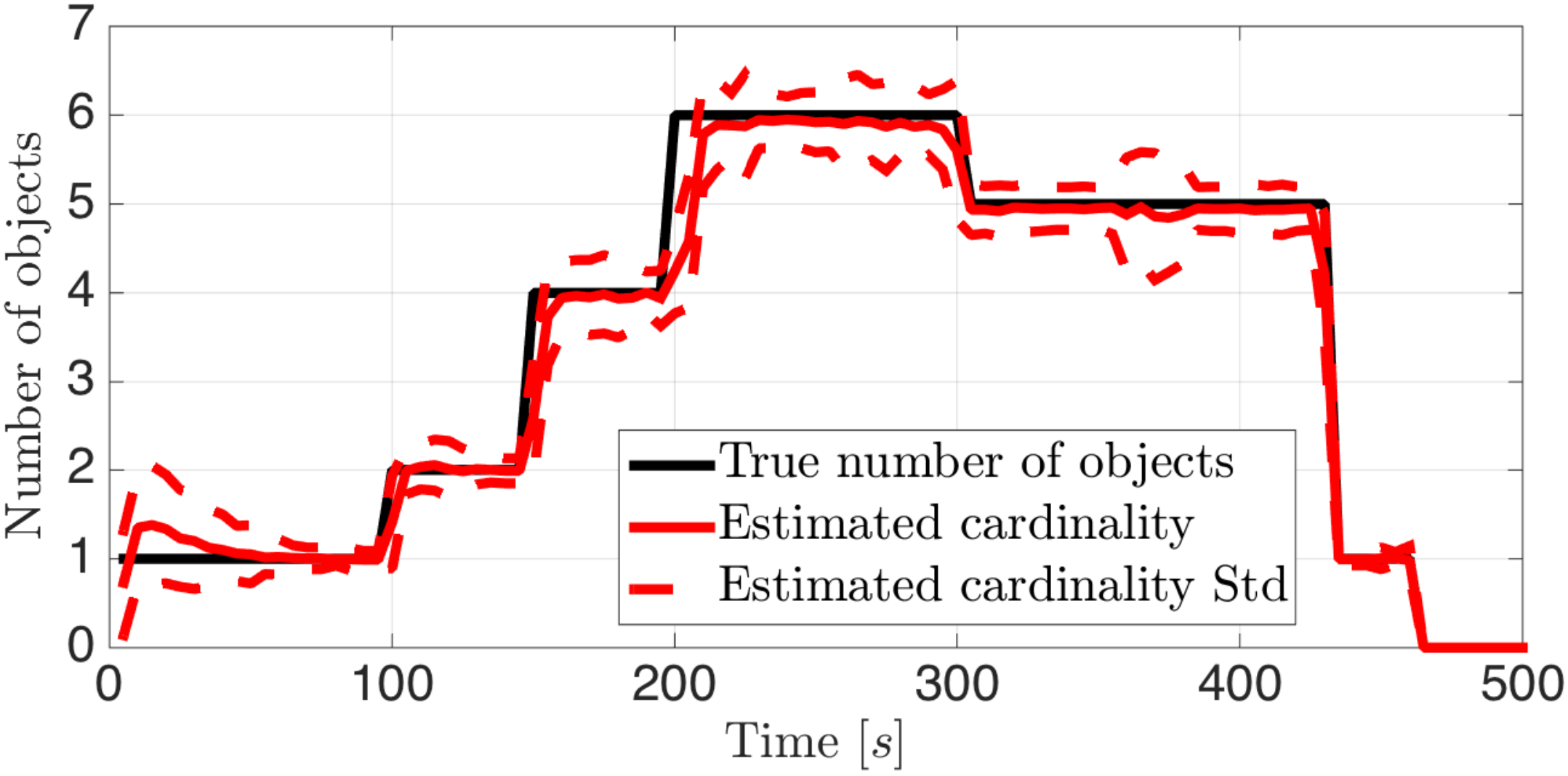}
\caption{Cardinality statistics for GGM-CPHD.}
\label{fig:cardggmcphd}
\end{figure}

\begin{figure}[h!]
\centering
\includegraphics[width=\columnwidth]{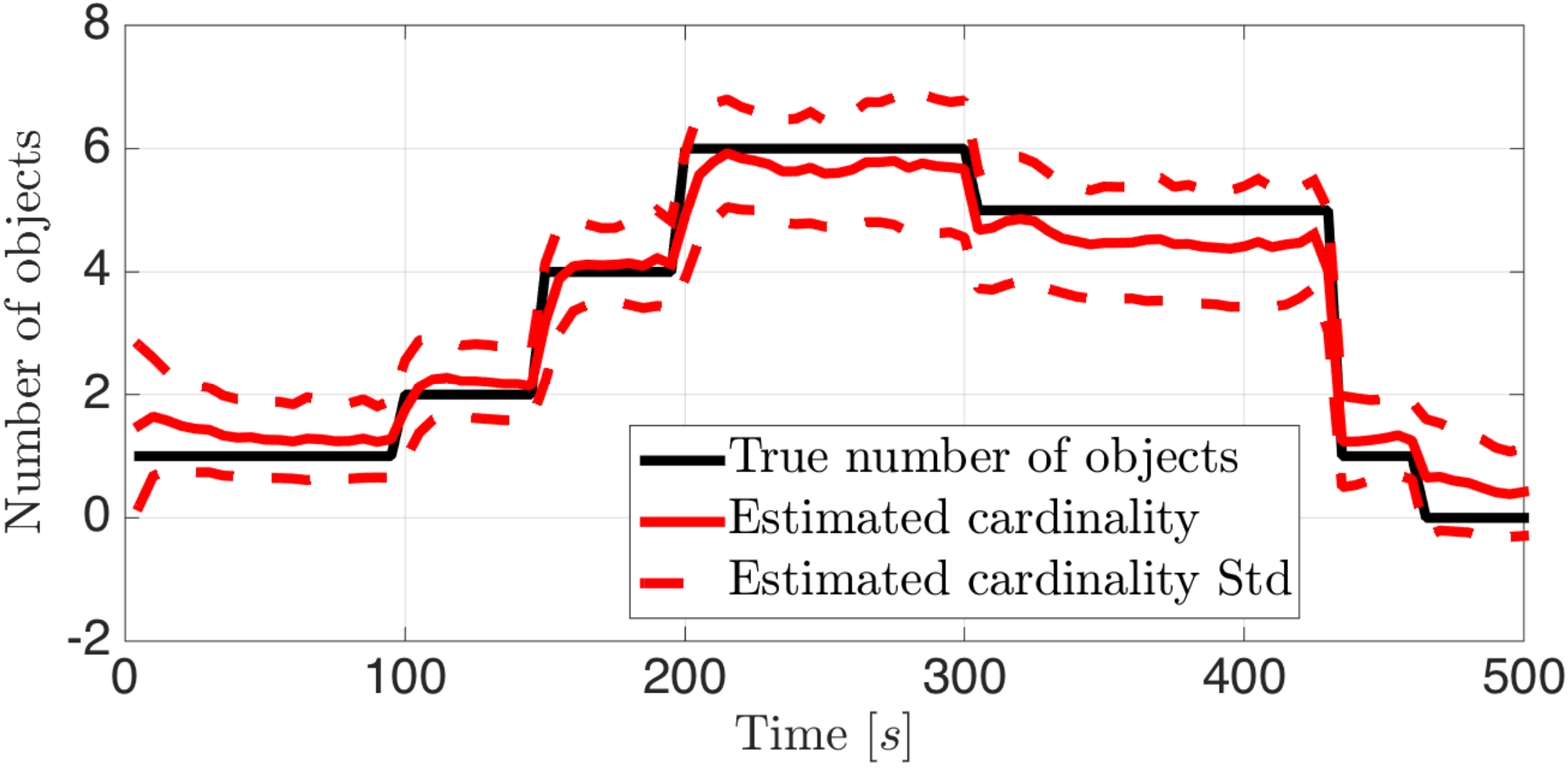}
\caption{Cardinality statistics for CGM-CPHD with $L = 1$ consensus step.}
\label{fig:carddgmcphd1cs}
\end{figure}

\begin{figure}[h!]
\centering
\includegraphics[width=\columnwidth]{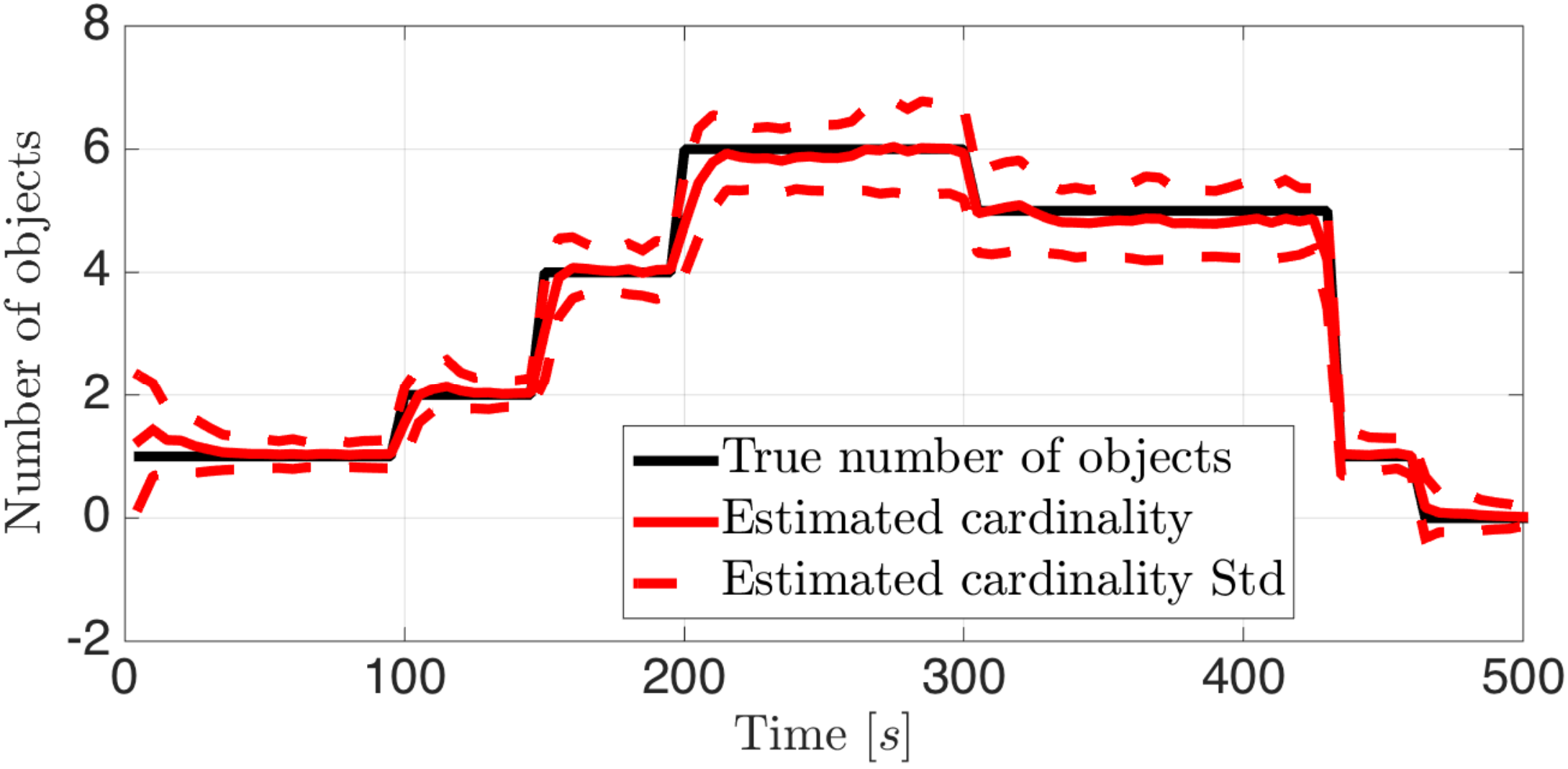}
\caption{Cardinality statistics for CGM-CPHD with $L = 2$ consensus steps.}
\label{fig:carddgmcphd2cs}
\end{figure}

\begin{figure}[h!]
\centering
\includegraphics[width=\columnwidth]{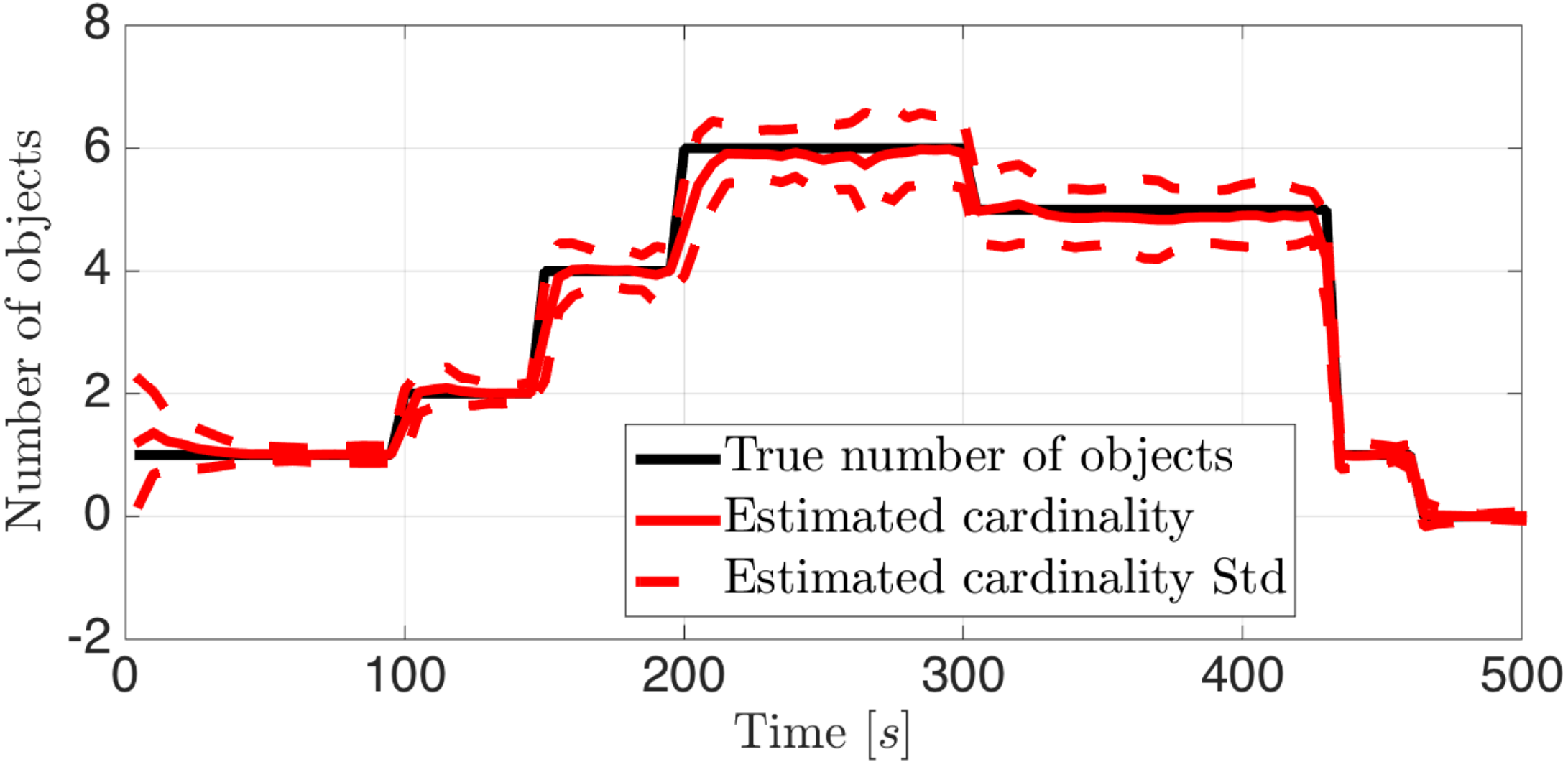}
\caption{Cardinality statistics for CGM-CPHD with $L = 3$ consensus steps.}
\label{fig:carddgmcphd3cs}
\end{figure}

\begin{figure}[h!]
\centering
\includegraphics[width=\columnwidth]{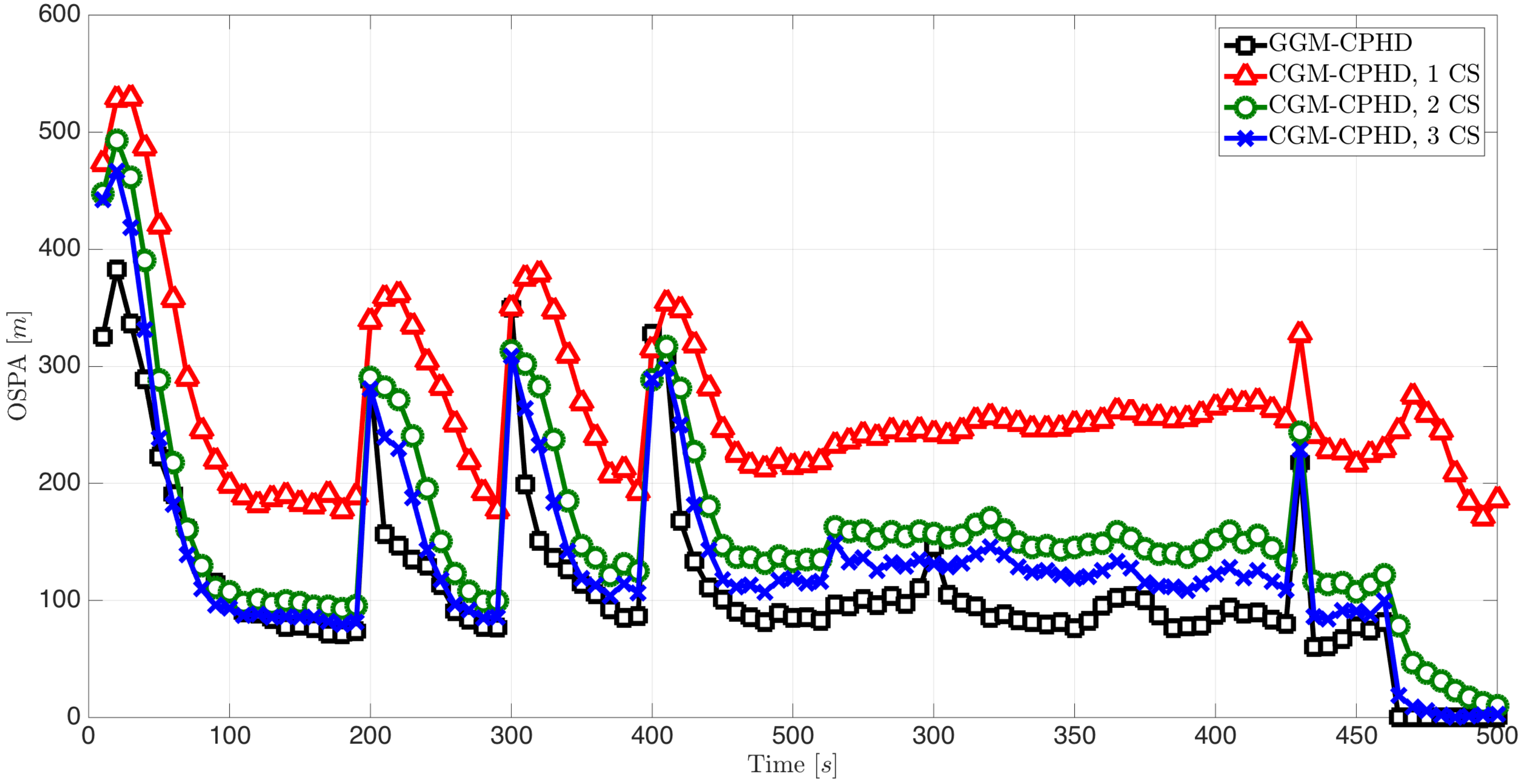}
\caption{Performance comparison, using OSPA, between GGM-CPHD and CGM-CPHD respectively with $L = 1$, $L = 2$ and $L = 3$ consensus steps.}
\label{fig:6tospa}
\end{figure}

To provide an insightful view of the consensus process, Fig. \ref{fig:benefits_consensus} shows the GM representation within a particular sensor node (TOA sensor $1$ in Fig. \ref{fig:4toa3doa}) before consensus and after $L = 1, 2, 3$ consensus steps. It can be noticed that consensus steps provide a progressive refinement of multi-object information. In fact, before consensus the CPHD reveals the presence of several false objects 
(see Fig. \ref{fig:benefits_consensus-a}); this is clearly due to the fact that the considered TOA sensor does not guarantee observability. Then, in the subsequent consensus steps (see Figs. \ref{fig:benefits_consensus-b}-\ref{fig:benefits_consensus-d}), the weights of the Gaussian components corresponding to false objects as well as the covariance of the component relative to the true object are progressively reduced. Summing up, just one consensus step leads to acceptable performance but, in this case, the distributed algorithm tends to be less responsive and accurate in estimating
object number and locations when compared to the centralized one (GGM-CPHD). However, by performing additional consensus steps, it is possible to retrieve performance comparable to the one of GGM-CPHD.

\begin{figure}[h!]
        \centering
        \subfloat[][GM before exploiting consensus.]{
                \includegraphics[width=0.45\textwidth]{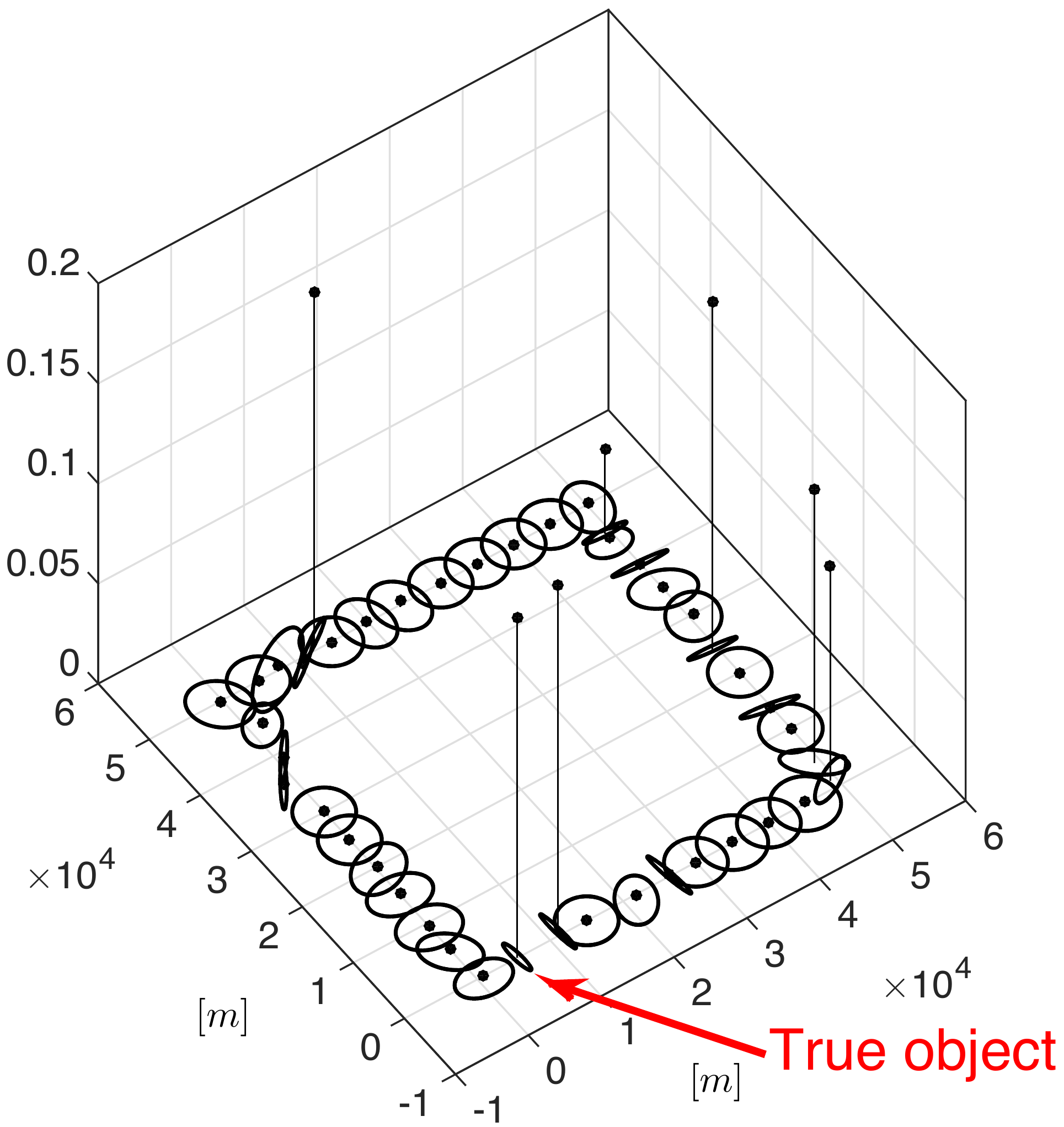}
                \label{fig:benefits_consensus-a}
        }
        \subfloat[][GM after the first consensus step ($L = 1$)]{
                \includegraphics[width=0.45\textwidth]{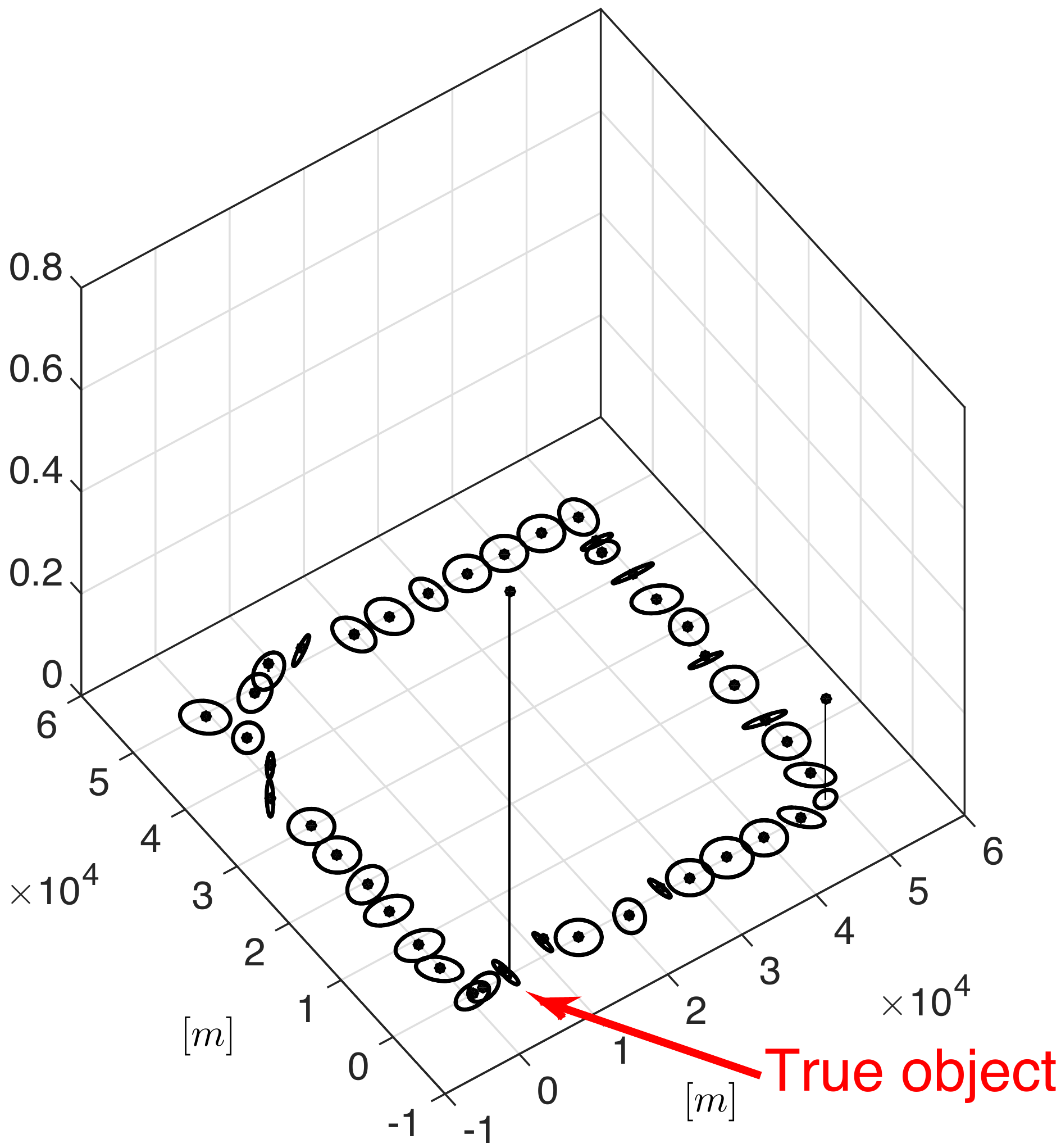}
		\label{fig:benefits_consensus-b}
        }

        \subfloat[][GM after the second consensus step ($L = 2$)]{
                \includegraphics[width=0.45\textwidth]{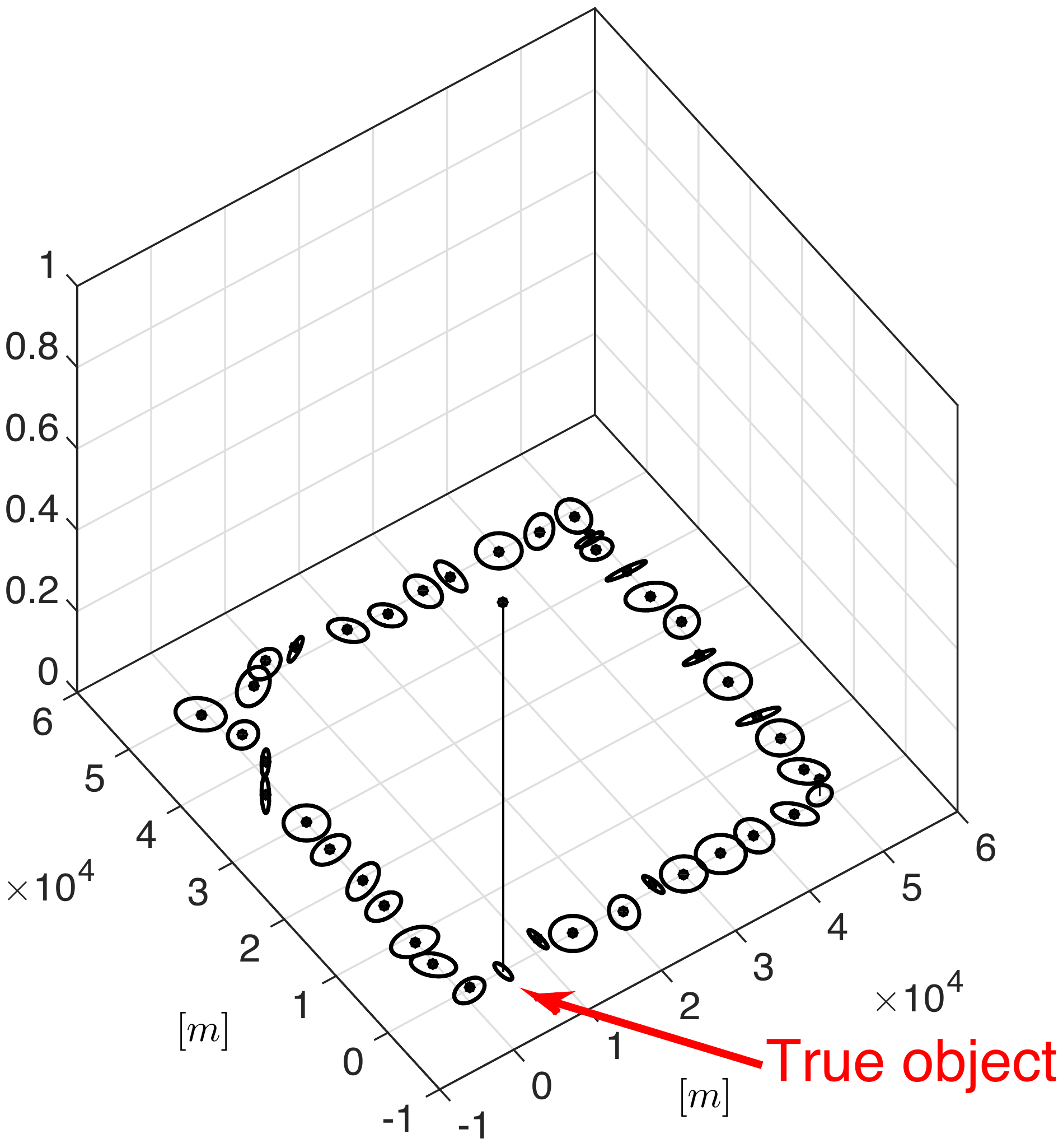}
		\label{fig:benefits_consensus-c}
        }
        \subfloat[][GM after the third consensus step ($L = 3$)]{
                \includegraphics[width=0.45\textwidth]{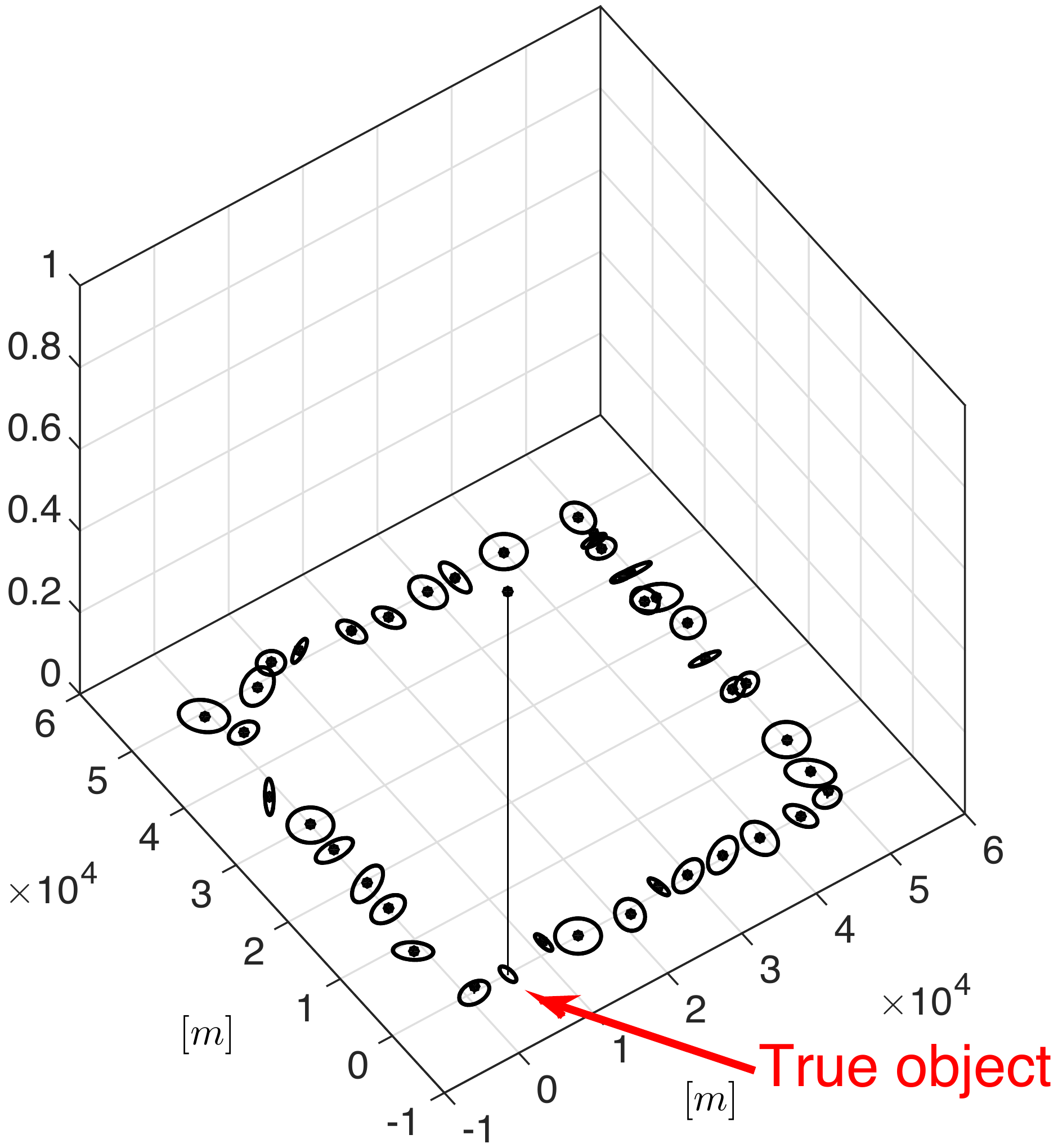}
                \label{fig:benefits_consensus-d}
        }
	\caption{GM representation within TOA sensor $1$ (see Fig. 2) at the initial time instant throughout the consensus process. Means and $99.7\%$ confidence ellipses of the Gaussian components of the location PDF are displayed in the horizontal plane while weights of such components are on the vertical axis.}    
       	\label{fig:benefits_consensus}
\end{figure}

In order to have an idea about the data communication requirements of the proposed algorithm, it is useful to examine the statistics of the number of Gaussian components of each node just before carrying out a consensus step; this determines, in fact, the number of data to be exchanged between nodes.
To this end, Table \ref{tab:gausscomponent} reports the average, standard deviation, minimum and maximum number of components in each node and at each consensus step during the simulation. Notice that, in the considered case-study, the average number of Gaussian components to be transmitted is about $20 \div 30$. Since each component involves $15$ real numbers ($4$ for the mean, $10$ for the covariance and $1$ for the weight), the overall average communication load per consensus step is about $1.2 \div 1.8$ kbytes if a $4$-byte single-precision floating-point representation is adopted for each real number.

To give a rough idea of the computation time, our MATLAB implementation on a Personal Computer with $4.2$ $GHz$ clock exhibited an average processing time per consensus step in the range of $45 \div 55$ $ms$. Clearly, this time can be significantly reduced with a different, e.g. C/C++, 
implementation.

The CCPHD filter has also been successfully applied in scenarios with range and/or Doppler sensors \cite{batchifan2013eusipco} and with passive multi-receiver radar systems \cite{batchifan2013irs}.

\begin{table}[h!]
	\renewcommand{\arraystretch}{1.3}
	\setlength\arrayrulewidth{0.5pt}\arrayrulecolor{black} 
	\setlength\doublerulesep{0.5pt}\doublerulesepcolor{black} 
	\caption{Number of Gaussian components before consensus for CGM-CPHD}
	\label{tab:gausscomponent}
	\centering
	~\\
	\scalebox{1}{
	\begin{tabular}{>{\columncolor[gray]{.95}}c||c|c||c|}
		{\textbf{TOA}$_\mb{1}$} & \multicolumn{1}{>{\columncolor[gray]{.95}}c|}{AVG} & \multicolumn{1}{>{\columncolor[gray]{.95}}c||}{STD} & \multicolumn{1}{>{\columncolor[gray]{.95}}c|}{MAX} \\
		\hline
		\hline
		\textit{$L = 1$} & $29.39$ & $12.98$ & $63$ \\
		\hline
		\textit{$L = 2$} & $10.63$ & $6.97$ & $48$ \\
		\hline
		\textit{$L = 3$} & $5.18$ & $4.58$ & $48$ \\
		\hline
	\end{tabular}
	}
	\scalebox{1}{
	\begin{tabular}{>{\columncolor[gray]{.95}}c||c|c||c|}
		{\textbf{TOA}$_\mb{2}$} & \multicolumn{1}{>{\columncolor[gray]{.95}}c|}{AVG} & \multicolumn{1}{>{\columncolor[gray]{.95}}c||}{STD} & \multicolumn{1}{>{\columncolor[gray]{.95}}c|}{MAX} \\
		\hline
		\hline
		\textit{$L = 1$} & $29.34$ & $13.17$ & $66$ \\
		\hline
		\textit{$L = 2$} & $10.33$ & $7.11$ & $58$ \\
		\hline
		\textit{$L = 3$} & $4.61$ & $4.49$ & $58$ \\
		\hline
	\end{tabular}
	}
	~\\~\\~\\
	\scalebox{1}{
	\begin{tabular}{>{\columncolor[gray]{.95}}c||c|c||c|}
		{\textbf{TOA}$_\mb{3}$} & \multicolumn{1}{>{\columncolor[gray]{.95}}c|}{AVG} & \multicolumn{1}{>{\columncolor[gray]{.95}}c||}{STD} & \multicolumn{1}{>{\columncolor[gray]{.95}}c|}{MAX} \\
		\hline
		\hline
		\textit{$L = 1$} & $29.08$ & $13.13$ & $63$ \\
		\hline
		\textit{$L = 2$} & $10.42$ & $7.06$ & $58$ \\
		\hline
		\textit{$L = 3$} & $4.7$ & $4.49$ & $57$ \\
		\hline
	\end{tabular}
	}
	\scalebox{1}{
	\begin{tabular}{>{\columncolor[gray]{.95}}c||c|c||c|}
		{\textbf{TOA}$_\mb{4}$} & \multicolumn{1}{>{\columncolor[gray]{.95}}c|}{AVG} & \multicolumn{1}{>{\columncolor[gray]{.95}}c||}{STD} & \multicolumn{1}{>{\columncolor[gray]{.95}}c|}{MAX} \\
		\hline
		\hline
		\textit{$L = 1$} & $29.72$ & $13.06$ & $67$ \\
		\hline
		\textit{$L = 2$} & $10.79$ & $6.97$ & $52$ \\
		\hline
		\textit{$L = 3$} & $5.25$ & $4.6$ & $52$ \\
		\hline
	\end{tabular}
	}
	~\\~\\~\\
	\scalebox{1}{
	\begin{tabular}{>{\columncolor[gray]{.95}}c||c|c||c|}
		{\textbf{DOA}$_\mb{1}$} & \multicolumn{1}{>{\columncolor[gray]{.95}}c|}{AVG} & \multicolumn{1}{>{\columncolor[gray]{.95}}c||}{STD} & \multicolumn{1}{>{\columncolor[gray]{.95}}c|}{MAX} \\
		\hline
		\hline
		\textit{$L = 1$} & $24.12$ & $15$ & $52$ \\
		\hline
		\textit{$L = 2$} & $5.65$ & $4.96$ & $53$ \\
		\hline
		\textit{$L = 3$} & $4.22$ & $4.25$ & $51$ \\
		\hline
	\end{tabular}
	}
	\scalebox{1}{
	\begin{tabular}{>{\columncolor[gray]{.95}}c||c|c||c|}
		{\textbf{DOA}$_\mb{2}$} & \multicolumn{1}{>{\columncolor[gray]{.95}}c|}{AVG} & \multicolumn{1}{>{\columncolor[gray]{.95}}c||}{STD} & \multicolumn{1}{>{\columncolor[gray]{.95}}c|}{MAX} \\
		\hline
		\hline
		\textit{$L = 1$} & $23.61$ & $15.32$ & $53$ \\
		\hline
		\textit{$L = 2$} & $5.6$ & $4.94$ & $50$ \\
		\hline
		\textit{$L = 3$} & $4.19$ & $4.24$ & $48$ \\
		\hline
	\end{tabular}
	}
	~\\~\\~\\
	\scalebox{1}{
	\begin{tabular}{>{\columncolor[gray]{.95}}c||c|c||c|}
		{\textbf{DOA}$_\mb{3}$} & \multicolumn{1}{>{\columncolor[gray]{.95}}c|}{AVG} & \multicolumn{1}{>{\columncolor[gray]{.95}}c||}{STD} & \multicolumn{1}{>{\columncolor[gray]{.95}}c|}{MAX} \\
		\hline
		\hline
		\textit{$L = 1$} & $23.45$ & $15.23$ & $52$ \\
		\hline
		\textit{$L = 2$} & $4.92$ & $4.78$ & $58$ \\
		\hline
		\textit{$L = 3$} & $4.22$ & $4.35$ & $60$ \\
		\hline
	\end{tabular}
	}
\end{table}

% CENTRALIZED MULTI-OBJECT TRACKING
\chapter{Centralized multi-object tracking}
\label{chap:mot}
\spminitoc
In this chapter a new approximation of the $\delta$-GLMB filter is presented \cite{mdglmbf}.
It will be shown that the GLMB distribution can be used to construct a principled approximation to an arbitrary labeled RFS density that matches the PHD and the cardinality distribution \cite{papi2014}.
The resulting filter is referred to as the Marginalized $\delta$-GLMB (M$\delta$-GLMB) tracking filter since it can be interpreted as a \textit{marginalization over the data associations}.
The proposed filter is consequently computationally cheaper than the $\delta$-GLMB filter while still preserving key summary statistics of the multi-object posterior.
Importantly, the M$\delta$-GLMB filter facilitates tractable multi-sensor multi-object tracking.
Unlike PHD \cite{vo-ma}, CPHD \cite{vo-vo-cantoni} and Multi-Bernoulli based filters \cite{Vo2009}, the proposed approximation accommodates statistical dependence between objects.
An alternative derivation of the LMB filter proposed in \cite{lmbf} based on the newly proposed M$\delta$-GLMB filter is presented \cite{mdglmbf}.

\section{The $\delta$-GLMB filter}
\label{sec:dglmb}
An efficient approach to multi-object tracking was presented in \cite{vovo1} using the $\delta$-GLMB distribution (\ref{eq:dglmbpdf}) (or equivalently (\ref{eq:dglmbpdf2})), i.e.
\bie
	\boldsymbol{\pi}(\lb{X}) & = & \Delta(\lb{X})\sum_{\left(I,\xi\right)\in\mathcal{F}\left(\mathbb{L}\right)\times\Xi}w^{\left(I,\xi\right)}\delta_{I}\left(\mathcal{L}\left(\lb{X}\right)\right)\left[p^{\left(\xi\right)}\right]^{\lb{X}} \, ,\\
	& = & \Delta(\lb{X})\sum_{I\in\mathcal{F}\left(\mathbb{L}\right)}\delta_{I}\left(\mathcal{L}\left(\lb{X}\right)\right)\sum_{\xi\in\Xi}w^{\left(I,\xi\right)}\left[p^{\left(\xi\right)}\right]^{\lb{X}} \, .
\eie

The $\delta$-GLMB density naturally arises in multi-object tracking problems when using the standard detection based measurement model .
In the following, the prediction and update steps of the $\delta$-GLMB filter is briefly recalled; details can be found in \cite{vovo1,vovo2}.

\subsection{$\delta$-GLMB prediction}
The standard multi-object dynamic model is described as follows.
Given the multi-object state $\lb{Z}$, each state $(\zeta, \ell_{-}) \in \lb{Z}$ either continues to exist at the next time step with probability $P_{S}(\zeta, \ell_{-})$ and evolves to a new state $(x, \ell)$ with probability density $\varphi_{k|k-1}(x|\zeta,\ell_{-})$, or disappears with probability $1 - P_{S}(\zeta,\ell_{-})$.
Note that the label of the objects is preserved in the transition, only the kinematic part of state changes.
Assuming that $\lb{Z}$ has distinct labels and that conditional on $\lb{Z}$, the transition of the kinematic states are mutually independent, then the set $\lb{W}$ of surviving objects at the next time is a labeled multi-Bernoulli RFS \cite{vovo1}
\begin{equation}
	\boldsymbol{f}_{S}\!\left( \mathbf{W} | \mathbf{Z} \right) = \Delta\!\left( \mathbf{W} \right) \Delta\!\left( \mathbf{Z} \right) \, 1_{\mathcal{L}\left( \mathbf{Z} \right)}\left( \mathcal{L}\!\left( \mathbf{W} \right) \right) \left[ \Phi\left( \mathbf{W}; \cdot \right) \right]^{\mathbf{Z}} \, ,\label{eq:survivorpdf}
\end{equation}
where
\begin{equation}
	\Phi\!\left( \mathbf{W}; \zeta, \ell_{-} \right) = \sum_{\left( x, \ell \right) \in \mathbf{W}} \delta_{\ell_{-}}\left( \ell \right) P_{S}\left( \zeta, \ell_{-} \right) f\left( x | \zeta, \ell_{-} \right) + \left[ 1 - 1_{\mathcal{L}\left( \mathbf{W} \right)}\left( \ell_{-} \right) \right] \left( 1 - P_{S}\left( \zeta, \ell_{-} \right) \right) \, .
\end{equation}
The $\Delta\!\left( \mathbf{Z} \right)$ in (\ref{eq:survivorpdf}) ensures that only $\mathbf{Z}$ with distinct labels are considered.

The set of new objects born at the next time step is distributed according to 
\be
	\boldsymbol{f}_{B}(\lb{Y}) = \Delta(\lb{Y}) \, w_{B}(\lbs{\lb{Y}})\left[p_{B}\right]^{\lb{Y}}\label{eq:Birth_transition}
\ee
The birth density $\boldsymbol{f}_{B}(\cdot)$ is defined on $\mathbb{X} \times \mathbb{L}$ and $\boldsymbol{f}_{B}(\lb{Y}) = 0$ if $\lb{Y}$ contains any element $\lb{x}$ with $\lbs{\lb{x}} \notin \mathbb{L}$.
The birth model (\ref{eq:Birth_transition}) includes both labeled Poisson and labeled multi-Bernoulli densities.
The multi-object state at the next time $\mathbf{X}$ is the superposition of surviving objects and new born objects, i.e. $\mathbf{X} = \mathbf{W} \cup \mathbf{Y}$.
Since the label spaces $\mathbb{L}$ and $\mathbb{B}$ are disjoint, the labeled birth objects and surviving objects are independent.
Thus the multi-object transition density turns out to be the product of the transition density (\ref{eq:survivorpdf}) and the density of new objects (\ref{eq:Birth_transition})
\begin{equation}
	\boldsymbol{f}\!\left( \mathbf{X} | \mathbf{Z} \right) = \boldsymbol{f}_{S}\!\left( \mathbf{X} \cap \left( \mathbb{X} \times \mathbb{L} \right) | \mathbf{Z} \right) \, \boldsymbol{f}_{B}\left( \mathbf{X} - \left( \mathbb{X} \times \mathbb{L} \right) \right) \, .
\end{equation}
Additional details can be found in \cite[Subsection IV.D]{vovo1}.

If the current multi-object prior density is a $\delta$-GLMB of the form (\ref{eq:dglmbpdf}), then the multi-object prediction density is also a $\delta$-GLMB given by 
\be
	\boldsymbol{\pi}_{k|k-1}(\lb{X}) = \Delta(\lb{X}) \sum_{\left(I,\xi\right) \in \mathcal{F}(\lbsp) \times \Xi} w_{k|k-1}^{\left( I, \xi \right)} \delta_{I}\left(\lbs{\lb{X}}\right)\left[p_{k|k-1}^{\left(\xi\right)}\right]^{\lb{X}}\label{eq:dglmbpredictedpdf}
\ee
where
\bie
	w_{k|k-1}^{\left(I,\xi\right)} & = & w_{B}(I\backslash\lbsp[-])w_{S}^{(\xi)}(I\cap\lbsp[-]) \, ,\\
	p_{k|k-1}^{(\xi)}(x,\ell) & = & 1_{\lbsp[-]}(\ell)p_{S}^{(\xi)}(x,\ell)+1_{\mathbb{B}}(\ell)p_{B}(x,\ell) \, ,\\
	p_{S}^{(\xi)}(x,\ell) & = & \frac{\left\langle P_{S}(\cdot,\ell) \, \varphi_{k|k-1}(x|\cdot,\ell),p_{k-1}^{(\xi)}(\cdot,\ell)\right\rangle }{\eta_{S}^{(\xi)}(\ell)} \, ,\\
	\eta_{S}^{(\xi)}(\ell) & = & \left\langle P_{S}(\cdot,\ell),p_{k-1}^{(\xi)}(\cdot,\ell)\right\rangle \, ,\\
	w_{S}^{(\xi)}(L) & = & [\eta_{S}^{(\xi)}]^{L}\sum_{J\subseteq\lbsp[-]}1_{J}(L)[1-\eta_{S}^{(\xi)}]^{J-L}w_{k}^{(J,\xi)} \, .
\eie

\subsection{$\delta$-GLMB update}
The standard multi-object observation model is described as follows.
For a given multi-object state $\lb{X}$, each state $\mathbf{x}\in\lb{X}$ is either detected with probability $P_{D}\left(\mathbf{x}\right)$ and generates a point $y$ with likelihood $g(y|\mathbf{x})$, or missed with probability $1 - P_{D}\left(\mathbf{x}\right)$, i.e. $\mathbf{x}$ generates a Bernoulli RFS with parameter $(P_{D}(\mathbf{x}), g(\cdot|\mathbf{x}))$.
Assuming that conditional on $\lb{X}$ these Bernoulli RFSs are independent, then the set $W \subset \mathbb{Y}$ of detected points (non-clutter measurements) is a multi-Bernoulli RFS with parameter set $\left\{ \left( P_{D}(\mathbf{x}), g(\cdot|\mathbf{x}) \right) : \mathbf{x} \in \lb{X} \right\}$.
The set $\mathcal{C} \subset \mathbb{Y}$ of false observations (or clutter), assumed independent of the detected points, is modeled by a Poisson RFS with intensity function $\kappa(\cdot)$.
The multi-object observation $Y$ is the superposition of the detected points and false observations, i.e. $Y = W \cup \mathcal{C}$.
Assuming that, conditional on $\lb{X}$, detections are independent, and that clutter is independent of the detections, the multi-object likelihood is given by 
\be
	g_{k}(Y|\lb{X})=e^{-\left\langle \kappa ,1\right\rangle}\kappa^{Y}\sum_{\theta \in \Theta (\mathcal{L}(\lb{X}))}\left[\psi_{Y}(\cdot ;\theta )\right] ^{\lb{X}} \, ,\label{eq:RFSmeaslikelihood0}
\ee
where $\Theta (I)$ is the set of mappings $\theta:I \rightarrow \{0,1,...,M\},$ such that $\theta(i) = \theta( i^{\prime} ) > 0$ implies $i = i^{\prime}$, and
\be
	\psi_{Y}(x, \ell; \theta ) =	\left\{ 
							\begin{array}{ll}
								\dfrac{P_{D}(x,\ell) \, g_{k}(y_{\theta(\ell)}|x,\ell)}{\kappa (y_{\theta(\ell)})} \, , & \text{if }\theta (\ell) > 0 \\ 
								1-P_{D}(x,\ell) \, , & \text{if }\theta(\ell) = 0
							\end{array} \right. \, .\label{eq:molikelihood}
\ee
Note that an association map $\theta$ specifies which tracks generated which measurements, i.e. track $\ell$ generates measurement $y_{\theta(\ell)}\in Y$, with undetected tracks assigned to $0$.
The condition ``$\theta(i)=\theta (i^{\prime}) > 0$ implies $i = i^{\prime}$'', means that a track can generate at most one measurement, and a measurement can be assigned to at most one track, at one time instant.
Additional details can be found in \cite[Subsection IV.C]{vovo1}.

If the current multi-object prediction density is also a $\delta$-GLMB of the form (\ref{eq:dglmbpdf}), then the multi-object posterior density is a $\delta$-GLMB given by
\be
	\boldsymbol{\pi}_{k}\left(\lb{X}\right) = \Delta(\lb{X}) \,
	\sum_{\left(I,\xi\right)\in\mathcal{F}\left(\mathbb{L}\right)\times\Xi} \,
	\sum_{\theta\in\Theta(I)} w_{k}^{\left(I,\xi,\theta\right)}(Y_{k})\delta_{I}\left(\mathcal{L}\left(\lb{X}\right)\right)\left[p_{k}^{\left(\xi,\theta\right)}\right]^{\lb{X}}\label{eq:dglmbupdatedpdf}
\ee
where $\Theta(I)$ denotes the subset of the current maps with domain $I$, and
\bie
	w_{k}^{(I,\xi,\theta)}(Y_{k}) & \propto & w_{k|k-1}^{\left(I,\xi\right)}\left[\eta_{Y_{k}}^{(\xi,\theta)}(\ell)\right]^{I} \, ,\\
	\eta_{Y_{k}}^{(\xi,\theta)}(\ell) & = & \left\langle p_{k|k-1}^{(\xi)}(\cdot,\ell),\psi_{Y_{k}}(\cdot,\ell;\theta)\right\rangle \, ,\\
	p_{k}^{\left(\xi,\theta\right)}\left(x, \ell\right) & = & \frac{p_{k|k-1}^{(\xi)}(x,\ell)\psi_{Y_{k}}(x,\ell;\theta)}{\eta_{Y_{k}}^{(\xi,\theta)}(\ell)} \, ,\\
	\psi_{Y_{k}}(x,\ell;\theta) & = &	\begin{cases}
							\dfrac{P_{D}(x,\ell) \, g(y_{\theta(\ell)}|x,\ell)}{\kappa(y_{\theta(\ell)})} \, , & \mbox{if } \theta(\ell)>0\\
							1-P_{D}(x,\ell) \, , & \mbox{if } \theta(\ell)=0
						\end{cases} \, . 
\eie
Notice that the new association maps $\theta$ can be added (stacked) to their respective association histories $\xi$ in order to have again the more compact form (\ref{eq:dglmbpdf}) for the updated $\delta$-GLMB (\ref{eq:dglmbupdatedpdf}).

\section{GLMB approximation of multi-object densities}
\label{sce:glmbapprox}
In this section a GLMB approximation of the multi-object density with statistically dependent objects is proposed.
In particular a GLMB density that matches the multi-object density of interest in both the PHD and cardinality distribution is derived.
The strategy is inspired by Mahler's iid cluster approximation in the CPHD filter \cite{mahler2}, which has proven to be very effective in practical multi-object filtering problems \cite{mah2014book,geoschwil2009,svewinsve2009}.
Moreover, the GLMB approximation captures object dependencies whereas the iid cluster approximation assumes that individual objects are iid from a common single-object density.
Proof of the result is given in the appendix \ref{chap:appendix}.

Our result follows from the observation that any labeled RFS density ${\boldsymbol{\pi}}(\cdot)$ on ${\mathcal{F}}({\mathbb{X}}{\mathcal{\times}}{\mathbb{L}})$ can be written as 
\begin{equation}
	{\boldsymbol{\pi}}({\lb{X}})=w({\mathcal{L}}({\lb{X}}))p({\lb{X}})\label{eq:GLMBjoint}
\end{equation}
where 
\bie
	w(\{\ell_{1},\ldots,\ell_{n}\}) & \triangleq & \int{\boldsymbol{\pi}}(\{(x_{1},\ell_{1}),\ldots,(x_{n},\ell_{n})\}) dx_{1} \cdots dx_{n} \label{eq:GLMB1}\\
	p(\{(x_{1},\ell_{1}),\ldots,(x_{n},\ell_{n})\}) & \triangleq & \frac{{\boldsymbol{\pi}}(\{(x_{1},\ell_{1}),\ldots,(x_{n},\ell_{n})\})}{w(\{\ell_{1},\ldots,\ell_{n}\})}\label{eq:GLMB2}
\eie
It is implicitly assumed that $p({\lb{X}})$ is defined to be zero whenever $w({\mathcal{L}}({\lb{X}}))$ is zero.

Note that since ${\boldsymbol{\pi}}(\cdot)$ is symmetric in its arguments, the integral in (\ref{eq:GLMB1}) and $w$ is symmetric in $\ell_{1},\ldots,\ell_{n}$.
Moreover, since $\sum_{L\in{\mathcal{F}}({\mathbb{L}})}w(L)=1$, $w$ is indeed a probability distribution on ${\mathcal{F}}({\mathbb{L}})$ and can be interpreted as the probability that the labeled RFS has label set $\{\ell_{1},\ldots,\ell_{n}\}$. For $w(\{\ell_{1},\ldots,\ell_{n}\})>0$, $p(\{(x_{1},\ell_{1}),\ldots,(x_{n},\ell_{n})\})$\ is the joint probability density (on ${\mathbb{X}}^{n}$) of the kinematic states $x_{1},\ldots,x_{n}$ given that their corresponding labels are $\ell_{1},\ldots,\ell_{n}$.
\begin{pro}[GLMB approximation]\label{pro:glmbapprox}~\\
	The GLMB that matches the cardinality distribution and the PHD of a labeled RFS with density ${\boldsymbol{\pi}}(\cdot)$ is given by:
	\be
		\clmb[\hat]{}{}{\lb{X}} = \sum_{I\in{\mathcal{F}}({\mathbb{L}})}w^{(I)}({\mathcal{L}}({\lb{X}}))\left[p^{(I)}\right]^{{\lb{X}}}\label{eq:ApproximateGLMB}
	\ee
	where,
	\bie
		w^{(I)}(L) & = & \delta_{I}(L)w(I) \label{eq:MarginalizeGeneral-1}\\
		p^{(I)}(x,\ell) & = & 1_{I}(\ell)p_{I\backslash\{\ell\}}(x,\ell)  \label{eq:MarginalizeGeneral-2}\\
		p_{\{\ell_{1},\ldots,\ell_{n}\}}(x,\ell) & = & \int p(\left\{ (x,\ell),(x_{1},\ell_{1}),\ldots,(x_{n},\ell_{n})\right\} )d\left(x_{1},\ldots,x_{n}\right)\label{eq:MarginalizeGeneral-3}
	\eie
\end{pro}

\begin{rem}
	Note that in \cite[Section V]{beavovo2014} a GLMB was proposed to approximate a particular family of labeled RFS densities that arises from multi-object filtering with merged measurements.
	By applying Proposition \ref{pro:glmbapprox}, it can be shown that the approximation used in \cite[Section V]{beavovo2014} preserves the cardinality distribution and first moment.
\end{rem}

\begin{rem}
	In multi-object tracking, the matching of the labeled PHDs $\hat{d}(\cdot,\ell)$ and $d(\cdot,\ell)$ in Proposition \ref{pro:glmbapprox} is a stronger result than simply matching the unlabeled PHDs alone.
\end{rem}

\section{Marginalizations of the $\delta$-GLMB density}
\label{sec:marginal}
One of the main factors increasing the computational burden of the $\delta$-GLMB filter \cite{vovo2} is the exponential growth of the number of hypotheses in the update of the prior (\ref{eq:dglmbupdatedpdf}), which gives rise to the explicit sum over an association history variable $\xi$.
Further, note that the number of association histories is considerably increased in a multi-sensor context.
The idea behind the M$\delta$-GLMB filter is to construct a principled GLMB approximation $\clmb[\hat]{}{}{\cdot}$ to the posterior density $\clmb{}{}{\cdot}$ by means of a marginalization with respect to the association histories $\xi$.
In this section, the time index is dropped for the sake of simplicity and the attention is devoted to the M$\delta$-GLMB approximation corresponding to any $\delta$-GLMB density.
\begin{pro}[The M$\delta$-GLMB density]\label{pro:mdglmb}~\\
	A Marginalized $\delta$-GLMB (M$\delta$-GLMB) density $\clmb[\hat]{}{}{\cdot}$ corresponding to the $\delta$-GLMB density $\clmb{}{}{\cdot}$ in (\ref{eq:dglmbpdf}) is a probability density of the form
	\bie
		\clmb[\hat]{}{}{\lb{X}} & = & \dli{\lb{X}} \sum_{I \in \mathcal{F}(\lbsp)} \delta_{I}(\lbs{\lb{X}}) w^{(I)} \left[p^{(I)}\right]^{\lb{X}} \, ,\label{eq:mdglmbpdf}\\
		w^{(I)} & = & \sum_{\xi\in\Xi}w^{(I,\xi)} \, ,\label{eq:mdglmb:w}\\
		p^{(I)}(x,\ell) & = & 1_{I}(\ell)\frac{1}{w^{(I)}}\sum_{\xi\in\Xi}w^{(I,\xi)}p^{(\xi)}(x,\ell) \, .\label{eq:mdglmb:p}
	\eie
	The M$\delta$-GLMB (\ref{eq:mdglmbpdf}) preserves both (labeled) PHD and cardinality distribution of the original $\delta$-GLMB density $\clmb{}{}{\cdot}$ in (\ref{eq:dglmbpdf}).
\end{pro}

Note now that the M$\delta$-GLMB density provided by Proposition \ref{pro:mdglmb} is a $\delta$-GLMB which aggregates the association histories $\xi$ relative to the label set $I$ via (\ref{eq:mdglmb:w})-(\ref{eq:mdglmb:p}).
In fact, the marginalized density (\ref{eq:mdglmbpdf}) has no association history set $\Xi$.

The M$\delta$-GLMB density can be exploited to construct an efficient recursive multi-object tracking filter by calculating the M$\delta$-GLMB approximation step after the $\delta$-GLMB update, and predicting forward in time using the $\delta$-GLMB prediction.

\section{The M$\delta$-GLMB filter}
\label{sec:mdglmbf}
The M$\delta$-GLMB filter propagates a M$\delta$-GLMB multi-object posterior density forward in time via the multi-object Bayesian recursion (\ref{eq:MTBayesPred})-(\ref{eq:MTBayesUpdate}).
Since the M$\delta$-GLMB is still a $\delta$-GLMB, the equations derived for the $\delta$-GLMB filter \cite{vovo2} as a closed-form solution for the eqs. (\ref{eq:MTBayesPred})-(\ref{eq:MTBayesUpdate}) still hold.
In the following, the prediction and update steps for the $\delta$-GLMB filter are briefly recalled; additional details can be found in
\cite{vovo2,mdglmbf}.

\subsection{M$\delta$-GLMB prediction}
Given the multi-object state $\lb{Z}$, each state $(\zeta, \ell_{-}) \in \lb{Z}$ either continues to exist at the next time step with probability $P_{S}(\zeta, \ell_{-})$ and evolves to a new state $(x,\ell)$ with probability density $\varphi_{k|k-1}(x|\zeta,\ell_{-})$, or dies with probability $1-P_{S}(\zeta,\ell_{-})$.
The set of newborn objects at the next time step is distributed according to the LMB
\begin{equation}
	\boldsymbol{f}_{B}=\left\{ \left( r_{B}^{(\ell )},p_{B}^{(\ell )}\right)\right\}_{\ell \in \mathbb{B}}  \label{eq:MDGLMB:Birth_transition}
\end{equation}
with $\mathbb{L}_{-}\mathbb{\cap B=\varnothing}$ and $w_{B}(\lb{X})$ defined in (\ref{eq:lmb:w}).
The multi-object state at the current time $\lb{X}_{k}$ is the superposition of surviving objects and new born objects.

If the current multi-object prior density is a M$\delta$-GLMB of the form (\ref{eq:mdglmbpdf}), then the multi-object prediction density is also a M$\delta$-GLMB given by
\begin{equation}
	{\mathbf{\boldsymbol{\pi}}}_{k|k-1}({\lb{X}}) = \Delta({\lb{X}})\sum_{I \in \mathcal{\mathcal{F}}(\lbsp)} \delta_{I}\left(\mathcal{L}\left(\lb{X}\right)\right) \, w_{k|k-1}^{\left( I \right)} \left[p_{k|k-1}^{\left( I \right)}\right]^{\lb{X}} \, ,\label{eq:mdglmbpredictedpdf}
\end{equation}
where
\bie
	w_{k|k-1}^{\left( I \right)} & = & w_{B}(I \backslash \lbsp[-]) \, w_{S}^{(I)}(I \cap \lbsp[-]) \, ,\\
	p_{k|k-1}^{\left( I \right)}(x,\ell) & = & 1_{\lbsp[-]}(\ell)p_{S}^{(I)}(x,\ell)+1_{\mathbb{B}}(\ell)p_{B}(x,\ell) \, ,\\
	p_{S}^{(I)}(x, \ell) & = & \frac{\left\langle P_{S}(\cdot,\ell)\varphi_{k|k-1}(x | \cdot, \ell), p_{k-1}^{(I)}(\cdot, \ell)\right\rangle }{\eta_{S}^{(I)}(\ell)} \, ,\\
	\eta_{S}^{( I )}(\ell) & = & \left\langle P_{S}(\cdot, \ell), p_{k-1}^{(I)}(\cdot, \ell)\right\rangle \, ,\\
	w^{(I)}_{S}(L) & = & [\eta_{S}^{(I)}]^{L}\!\!\!\sum_{J\subseteq\lbsp[-]}1_{J}(L)[1-\eta_{S}^{(I)}]^{J-L}w_{k}^{(J)} \, ,
\eie
which is exactly the $\delta$-GLMB prediction step (\ref{eq:dglmbpredictedpdf}) with no association histories from previous time step, i.e. $\Xi = \varnothing$, and with the convention of having the superscript $(I)$ instead of $(\xi)$ due to the marginalization (\ref{eq:mdglmb:w})-(\ref{eq:mdglmb:p}).

%\begin{rem}\label{rem:maxhppred}
%	The number of $w_{k+1|k}^{\left(I\right)}$ and $p_{k+1|k}^{\left(I\right)}$ stored/computed after the M$\delta$-GLMB prediction step (\ref{eq:mdglmbpredictedpdf}) is $\left|\mathcal{F}(\mathbb{L}_{0:k+1})\right|$, while for the quantities $w_{k+1|k}^{\left(I, \xi\right)}$ and $p_{k+1|k}^{\left(\xi\right)}$ of the $\delta$-GLMB prediction (\ref{eq:dglmbpredictedpdf}) the number is, respectively, $\left|\mathcal{F}(\mathbb{L}_{0:k+1}) \times \Xi \right|$ and $\left| \Xi \right|$.
%	Notice that the number of weights stored/computed by the M$\delta$-GLMB is significantly lower than the one of $\delta$-GLMB.
%	As for the number of stored/computed location PDFs, it is worth noticing that the growth rate of the association histories $\xi \in \Xi$ is super-exponential with time \cite{vovo1,vovo2}, while, as a result of the superposition of surviving objects and new born objects, the growth rate of the cardinality of $\mathcal{F}(\mathbb{L}_{0:k+1})$ is exponential.
%\end{rem}
\begin{rem}\label{rem:maxhppred}
	The number of components $\left( w_{k+1|k}^{\left(I\right)}, p_{k+1|k}^{\left(I\right)} \right)$ computed after the M$\delta$-GLMB prediction step (\ref{eq:mdglmbpredictedpdf}) is $\left|\mathcal{F}(\mathbb{L}_{0:k+1})\right|$.
	On the other hand, the number of components $\left( w_{k+1|k}^{\left(I, \xi\right)},\right.$ $\left.p_{k+1|k}^{\left(\xi\right)} \right)$ after the $\delta$-GLMB prediction (\ref{eq:dglmbpredictedpdf}) is $\left|\mathcal{F}(\mathbb{L}_{0:k+1}) \times \Xi \right|$ for $w_{k+1|k}^{\left(I, \xi\right)}$ and $\left| \Xi \right|$ for $p_{k+1|k}^{\left(\xi\right)}$.
	Notice that the number of weights $w_{k+1|k}^{\left(I\right)}$ of the M$\delta$-GLMB is significantly lower than the $w_{k+1|k}^{\left(I, \xi\right)}$ of $\delta$-GLMB.
	As for the number of location PDFs $p_{k+1|k}^{\left(\cdot\right)}$, it is worth noticing that the growth rate of the association histories $\xi \in \Xi$ is super-exponential with time \cite{vovo1,vovo2}, while the growth rate of the cardinality of $\mathcal{F}(\mathbb{L}_{0:k+1})$ is exponential.
\end{rem}

The use of the M$\delta$-GLMB approximation further reduces the number of hypotheses in the posterior density while preserving the PHD and cardinality distribution.
Moreover, the M$\delta$-GLMB is in a form that is suitable for efficient and tractable information fusion (i.e. multi-sensor processing) as will be shown in the next subsection.

\subsection{M$\delta$-GLMB update}
Given a multi-object state $\lb{X}$, each state $(x, \ell) \in \lb{X}$ is either detected with probability $P_{D}\left( x, \ell \right)$ and generates a measurement $y$ with likelihood $g_{k}(y|x, \ell)$, or missed with probability $1 - P_{D}(x, \ell)$.
The multi-object observation $Y = \{y_{1}, \dots, y_{M}\}$ is the superposition of the detected points and Poisson clutter with intensity function $\kappa(\cdot)$.

If the current multi-object prediction density is a M$\delta$-GLMB of the form (\ref{eq:mdglmbpdf}), then the multi-object posterior density is a $\delta$-GLMB given by
\be
	\clmb{k}{}{\lb{X}} = \dli{\lb{X}} \sum_{I \in \mathcal{F}\left( \lbsp \right)} \, \sum_{\theta\in\Theta(I)} \delta_{I}\left(\mathcal{L}\left(\lb{X}\right)\right) w_{k}^{\left(I,\theta\right)}(Y_{k}) \left[p_{k}^{\left(I,\theta\right)} \right]^{\lb{X}} \, ,\label{eq:mdglmbupdatedpdf}
\ee
where $\Theta(I)$ denotes the subset of the current maps with domain $I$, and
\bie
	w_{k}^{(I,\theta)}(Y_{k}) & \propto & w_{k|k-1}^{\left(I\right)}\left[\eta_{Y_{k}}^{(\theta)}(\ell)\right]^{I} \, ,\\
	\eta_{Y_{k}}^{(I,\theta)}(\ell) & = & \left\langle p_{k|k-1}^{(I)}(\cdot,\ell),\psi_{Y_{k}}(\cdot,\ell;\theta)\right\rangle \, ,\\
	p_{k}^{\left(I,\theta\right)}\left( x \right) & = & \frac{p_{k|k-1}^{(I)}(x,\ell)\psi_{Y_{k}}(x,\ell;\theta)}{\eta_{Y_{k}}^{(I,\theta)}(\ell)} \, ,\\
	\psi_{Y_{k}}(x,\ell;\theta) & = & \begin{cases}
							\dfrac{P_{D}(x,\ell)g(y_{\theta(\ell)}|x,\ell)}{\kappa(y_{\theta(\ell)})}, & \mbox{if}\ \theta(\ell)>0\\
							1-P_{D}(x,\ell), & \mbox{if}\ \theta(\ell)=0
						\end{cases} \, .
\eie
Using (\ref{eq:mdglmb:w})-(\ref{eq:mdglmb:p}), the M$\delta$-GLMB density corresponding to the $\delta$-GLMB density in (\ref{eq:mdglmbupdatedpdf}) is a probability density of the form (\ref{eq:mdglmbpdf}) with
\bie
	w_{k}^{(I)} & = & \sum_{\theta\in\Theta(I)}w_{k}^{(I,\theta)} \, , \label{eq:mdglmb:wnxt}\\
	p_{k}^{(I)}(x,\ell) & = & 1_{I}(\ell)\frac{1}{w_{k}^{(I)}}\sum_{\theta\in\Theta(I)}w_{k}^{(I,\theta)}p_{k}^{(I,\theta)}(x,\ell) \, . \label{eq:mdglmb:pnxt}
\eie
The M$\delta$-GLMB density provided by (\ref{eq:mdglmb:wnxt})-(\ref{eq:mdglmb:pnxt}) preserves both PHD and cardinality distribution of the original $\delta$-GLMB density.

\begin{rem}\label{rem:maxhpup}
	Each hypothesis $I \in \mathcal{F}\left( \mathbb{L} \right)$ generates a set of $\left| \Theta(I) \right|$ new measurement-to-track association maps for the $\delta$-GLMB posterior.
	The number of components $\left( w_{k+1|k}^{\left(I, \theta\right)}, p_{k+1|k}^{\left(I, \theta\right)} \right)$ stored/computed after the M$\delta$-GLMB update step (\ref{eq:mdglmbupdatedpdf}) is $\left| \mathcal{F}\!\left( \mathbb{L} \right) \right| \cdot \sum_{I \in \mathcal{F}\left(\mathbb{L}\right)}\left| \Theta(I) \right|$.
	On the other hand, the number of hypotheses $\left( w_{k+1|k}^{\left(I, \xi, \theta \right)}, p_{k+1|k}^{\left(\xi, \theta\right)} \right)$ after the $\delta$-GLMB update (\ref{eq:dglmbupdatedpdf}) is $\left| \mathcal{F}\!\left( \mathbb{L} \right) \times \Xi \right| \cdot\sum_{I \in \mathcal{F}(\mathbb{L})}\left| \Theta(I) \right|$ for $w_{k+1|k}^{\left(I, \xi, \theta \right)}$ and $\left| \Xi \right| \cdot \sum_{I \in \mathcal{F}(\mathbb{L})}\left| \Theta(I) \right|$ for $p_{k+1|k}^{\left(\xi, \theta\right)}$.
	The same conclusions along the lines of Remark \ref{rem:maxhppred} hold.
\end{rem}

\begin{rem}\label{rem:maxhpmarginal}
	After the marginalization procedure (\ref{eq:mdglmb:wnxt})-(\ref{eq:mdglmb:pnxt}) only $\left|\mathcal{F}(\mathbb{L})\right|$ hypotheses are retained, as all the new contributions provided by the association maps $\left| \Theta(I) \right|$ are aggregated in a single component.
	Notice that $\left|\mathcal{F}(\mathbb{L})\right|$ is the same number of hypotheses produced during the prediction step (\ref{eq:mdglmbpredictedpdf}) (see Remark \ref{rem:maxhppred}).
	Thus, the prediction step (\ref{eq:mdglmbupdatedpdf}) sets the upper bound of the hypotheses that will be retained after each full M$\delta$-GLMB step.
\end{rem}

In the multi-sensor setting, for each time instant $k$, the update step (\ref{eq:mdglmbupdatedpdf}) (or (\ref{eq:dglmbupdatedpdf}) for the $\delta$-GLMB) is repeatedly evaluated using different measurement sets $Y^{i}$ provided by the sensors $i \in \ncal$ (see eq. (\ref{eq:msMTBayesUpdate})).
From Remarks \ref{rem:maxhpup} and \ref{rem:maxhpmarginal}, M$\delta$-GLMB is preferable to $\delta$-GLMB in terms of memory and computational requirements, since the number of remaining hypotheses after each sensor update step in (\ref{eq:msMTBayesUpdate}) is always set to $\left|\mathcal{F}(\mathbb{L})\right|$.
Note that this does not apply to $\delta$-GLMB due to the super-exponential memory/computational growth as reported in Remark \ref{rem:maxhppred}.
This is an important property of M$\delta$-GLMB since it yields a principled approximation which greatly decreases the need of pruning hypotheses with respect to $\delta$-GLMB \cite{vovo2}.
In fact, pruning in the $\delta$-GLMB might lead to poor performance in multi-sensor scenarios with low SNR (e.g. high clutter intensity, low detection probability, etc.) and limited storage/computational capabilities.
For instance, this may happen if a subset of the sensors do not detect one or multiple objects and hypotheses associated to the true tracks are removed due to pruning.
Furthermore, from a mathematical viewpoint, pruning between corrections generally produces a less informative and order-independent approximation of the posterior distribution in (\ref{eq:msMTBayesUpdate}).

\section{The LMB filter}
\label{sec:lmbf}
The LMB filter introduced in \cite{lmbf} is a single component approximation of a GLMB density (\ref{eq:glmbpdf}) that matches the unlabeled PHD.
An alternative derivation of the LMB approximation (\ref{eq:lmbpdf}) first proposed in \cite{lmbf} through a connection with the M$\delta$-GLMB
approximation is provided hereafter.

\subsection{Alternative derivation of the LMB Filter}
Recall that a LMB density is uniquely parameterized by a set of existence probabilities $r^{\left(\ell\right)}$ and corresponding track densities $p^{(\ell)}(\cdot)$:
\be
	\lmb{\lb{X}} = \dli{\lb{X}} \, w(\mathcal{L}(\mathbf{X}))p^{\mathbf{X}}\label{eq:lmbpdfnxt}
\ee
where
\bie
	w(L) & = & \prod_{\imath\in\mathbb{L}}\left(1-r^{\left(\imath\right)}\right)\prod_{\ell\in L}\frac{1_{\mathbb{L}}(\ell)r^{\left(\ell\right)}}{1-r^{(\ell)}}\\
	p(x,\ell) & = & p^{(\ell)}(x)
\eie

By extracting individual tracks from the M$\delta$-GLMB approximation $\clmb[\hat]{}{}{\cdot}$ (\ref{eq:mdglmbpdf}) it is possible to prove that the same expressions
for the existence probabilities and state densities originally proposed for the LMB filter in \cite{lmbf} can be obtained.
In fact
\bie
	r^{(\ell)} & \triangleq & \mbox{Pr}_{\hat{\boldsymbol{\pi}}}(\ell \in \mathcal{L}(\mathbf{X}))\\
	& = & \sum_{\ell \in L}w(L)\\
	& = & \sum_{I\in\mathcal{F}(\mathbb{L})}1_{I}(\ell)w(L)\\
	& = & \sum_{I \in \mathcal{F}(\mathbb{L})}1_{I}(\ell) w^{\left(I\right)} \, ,\\
	& = & \sum_{(I, \xi) \in \mathcal{F}(\mathbb{L}) \times \Xi} 1_{I}(\ell) w^{\left( I,\xi \right)} \, ,
\eie
\bie
	p^{(\ell)}(x) & \triangleq & \dfrac{\mbox{Pr}_{\hat{\boldsymbol{\pi}}}((x,\ell) \in \mathbf{X})}{\mbox{Pr}_{\hat{\boldsymbol{\pi}}}(\ell \in \mathcal{L}(\mathbf{X}))}\label{eq:lmb:defphd} \\
	& = & \dfrac{\displaystyle\int\hat{\boldsymbol{\pi}}((x,\ell)\cup\mathbf{X})\delta\mathbf{X}}{r^{(\ell)}}\\
	& = & \dfrac{\hat{d}(x,\ell)}{r^{(\ell)}} \, ,\label{eq:lmb:phd}
\eie
where the notation for the numerator in (\ref{eq:lmb:defphd}) is defined as per \cite[eq. 11.111]{mahler}, while the numerator of (\ref{eq:lmb:phd}) follows from \cite[eq. 11.112]{mahler}.
Notice that the numerator is precisely the PHD $\hat{d}(\cdot)$ corresponding to $\clmb[\hat]{}{}{\cdot}$, which by Proposition \ref{pro:glmbapprox} exactly matches the PHD $d(\cdot)$ corresponding to $\clmb{}{}{\cdot}$.
It can be easily verified that
\be
	\hat{d}(x,\ell) = \sum_{I \in \mathcal{F}(\mathbb{L})}1_{I}(\ell)w^{(I)}p^{(I)}(x,\ell)
\ee
 and consequently
\bie
	p^{(\ell)}(x) & = & \frac{1}{r^{(\ell)}}\sum_{I \in \mathcal{F}(\mathbb{L})}1_{I}(\ell)w^{(I)}p^{(I)}(x,\ell)\\
	& = & \frac{1}{r^{(\ell)}}\sum_{(I, \xi) \in \mathcal{F}(\mathbb{L}) \times \Xi} 1_{I}(\ell)w^{(I,\xi)} p^{(\xi)}(x,\ell)
\eie

\begin{rem}
	The property of matching the labeled PHD of the $\delta$-GLMB does not hold for the LMB tracking filter, as shown in \cite[Section III]{lmbf}, due to the imposed multi-Bernoulli structure for the cardinality distribution.
This can be also proved by combining (\ref{eq:lmbpdfnxt}) with Proposition \ref{pro:glmbapprox}, thus showing again that only the unlabeled PHD of the LMB matches exactly the unlabeled PHD of the original $\delta$-GLMB density.
\end{rem}

As suggested by its name, the LMB filter propagates an LMB multi-object posterior density forward in time \cite{lmbf}.
Its recursion will be briefly reviewed in the next subsections.
Additional details on its implementation can be found in \cite{lmbf}.

\subsection{LMB prediction}
Given the previous multi-object state $\lb{Z}$, each state $(\zeta,\ell_{-})\in \lb{Z}$ either continues to exist at the next time step with probability $P_{S}(\zeta,\ell_{-})$ and evolves to a new state $(x,\ell)$ with probability density $\varphi_{k|k-1}(x|\zeta,\ell_{-})$, or dies with probability $1 - P_{S}(\zeta,\ell_{-})$.
The set of new objects born at the next time step is distributed according to
\begin{equation}
	\boldsymbol{f}_{B}=\left\{ \left( r_{B}^{(\ell )},p_{B}^{(\ell )}\right)\right\}_{\ell \in \mathbb{B}}  \label{eq:LMB:Birth_transition}
\end{equation}
with $\mathbb{L}_{-}\mathbb{\cap B=\varnothing }$. The multi-object state at the current time $\lb{X}_{k}$ is the superposition of surviving objects and new born objects.

Suppose that the multi-object posterior density at time $k-1$ is LMB with parameter set $\boldsymbol{\pi }_{k-1}=\left\{ \left( r^{(\ell)},p^{(\ell )}\right) \right\}_{\ell \in \mathbb{L}_{-}}$. Then, it has been shown in \cite{lmbf} that the multi-object predicted density is also LMB with parameter set
\begin{equation}
	\boldsymbol{\pi }_{k|k-1}=\left\{ \left( r_{S}^{(\ell )},p_{S}^{(\ell)}\right) \right\}_{\ell \in \mathbb{L}_{-}}\cup \left\{ \left(r_{B}^{(\ell )},p_{B}^{(\ell )}\right) \right\}_{\ell \in \mathbb{B}} \, ,
\label{eq:lmbpredictedpdf}
\end{equation}
where 
\bie
	r_{S}^{(\ell )} & = &\eta_{S}(\ell ) \, r^{(\ell )} \, ,\label{eq:LMBupdate1}\\
	p_{S}^{(\ell )}(x) & = & \dfrac{\left\langle P_{S}(\cdot ,\ell )\varphi_{k|k-1}(x|\cdot ,\ell),p(\cdot ,\ell )\right\rangle}{\eta_{S}(\ell)} \, ,\label{eq:LMBupdate2}\\
	\eta_{S}(\ell ) & = &\left\langle P_{S}(\cdot ,\ell ),p(\cdot ,\ell)\right\rangle \, .
\eie

\subsection{LMB update}
Given a multi-object state $\lb{X}$, each state $(x, \ell)\in \lb{X}$ is either detected with probability $P_{D}\left( x,\ell \right)$ and generates a measurement $y$ with likelihood $g_{k}(y|x, \ell)$, or missed with probability $1-P_{D}(x, \ell )$. The multi-object observation $Y=\{y_{1},...,y_{M}\}$ is the superposition of the detected points and Poisson clutter with intensity function $\kappa(\cdot)$.
Assuming that, conditional on $\lb{X}$, detections are independent, and that clutter is independent of the detections, the multi-object likelihood is the same of the M$\delta$-GLMB Update in (\ref{eq:molikelihood}).

Suppose that the multi-object predicted density is LMB with parameter set
\be
	\boldsymbol{\pi}_{k|k-1} = \left\{ \left(r_{k|k-1}^{(\ell )},p_{k|k-1}^{(\ell )}\right) \right\}_{\ell \in \mathbb{L}} \, .
\ee
The multi-object posterior density cannot be evaluated in a closed form \cite{lmbf}.
Specifically, the LMB $\boldsymbol{\pi}_{k|k-1}$ has to be converted to a $\delta$-GLMB density by (possibly) generating all possible combinations of label sets $I \in \mathcal{F}\!\left( \lbs{\lb{X}} \right)$ from the constituent Bernoulli components \cite{lmbf}.
Further, the $\delta$-GLMB $\boldsymbol{\pi}_{k|k-1}(\cdot)$ resulting from the conversion process undergoes the $\delta$-GLMB update step (\ref{eq:dglmbupdatedpdf}).
Finally, the GLMB posterior $\boldsymbol{\pi}_{k}(\cdot)$ is approximated by an LMB that matches exactly the first moment of the unlabeled posterior intensity.
Thus, the updated LMB is given by
\be
	\clmb{k}{}{\lb{X}} = \left\{\left( r_{k}^{(\ell )},p_{k}^{(\ell )}\right) \right\}_{\ell \in \mathbb{L}} \, ,\label{eq:lmbupdatedpdf}
\ee
where 
\bie
	r_{k}^{(\ell )} & = & \sum_{(I,\theta )\in \mathcal{F}(\lbsp)\times\Theta (I)} \! 1_{I}(\ell) \, w_{k}^{(I,\theta )}(Y_{k}) \, ,\label{eq:existenceBernoulli} \\
	p_{k}^{(\ell )}(x) & = & \frac{1}{r_{k}^{(\ell)}}\sum_{(I,\theta) \in \mathcal{F}(\lbsp)\times \Theta (I)} 1_{I}(\ell) \, w_{k}^{(I,\theta )}(Y_{k}) \, p_{k}^{(\theta )}(x,\ell) \, , \label{eq:spatialBernoulli} \\
	w^{(I,\theta )}(Y_{k}) & \propto & \left[ \eta_{Y_{k}}^{(\theta )}\right]^{I}\prod\limits_{\ell \in \lbsp \backslash I}\left( 1-r_{k|k-1}^{(\ell)}\right) \prod \limits_{i\in I}1_{\lbsp}(i)r_{k|k-1}^{(i)} \, ,\\
	p_{k}^{(\theta )}(x,\ell) & = & \frac{p_{k|k-1}(x,\ell) \, \psi_{Y_{k}}(x,\ell ;\theta )}{\eta_{Y_{k}}^{(\theta )}(\ell )} \, ,\\
	\eta_{Y_{k}}^{( \theta )}(\ell) & = & \left\langle p_{k|k-1}(\cdot, \ell ),\psi_{Y_{k}}(\cdot ,\ell ;\theta )\right\rangle \, ,\\
	\psi_{Y_{k}}(x,\ell;\theta) & = & \begin{cases}
							\dfrac{P_{D}(x,\ell)g(y_{\theta(\ell)}|x,\ell)}{\kappa(y_{\theta(\ell)})}, & \mbox{if}\ \theta(\ell)>0\\
							1-P_{D}(x,\ell), & \mbox{if}\ \theta(\ell)=0
						\end{cases} \, .
\eie
Additional details for an efficient implementation of the LMB filter can be found in \cite{lmbf}.

\section{Performance evaluation}
To assess performance of the proposed Marginalized $\delta$-GLMB (M$\delta$-GLMB) tracker, a $2$-dimensional multi-object tracking scenario is considered over a surveillance area of $50\times50 \, [km^{2}]$.
Two sensor sets are used to represent scenarios with different observability capabilities.
In particular:
\begin{enumerate}[label=\textsc{\roman*}.]
	\item a single radar in the middle of the surveillance region is used as it guarantee observability;
	\item a set of $3$ \textit{range-only} (Time Of Arrival, TOA) sensors, deployed as shown in Fig. \ref{fig:3toa}, are used as they do not guarantee observability individually, but information from different sensors need to be combined to achieve it.
\end{enumerate}

The scenario consists of $5$ objects as depicted in Fig. \ref{fig:cmot:5trajectories}.
\begin{figure}[h!]
	\centering
	\includegraphics[width=0.6\columnwidth]{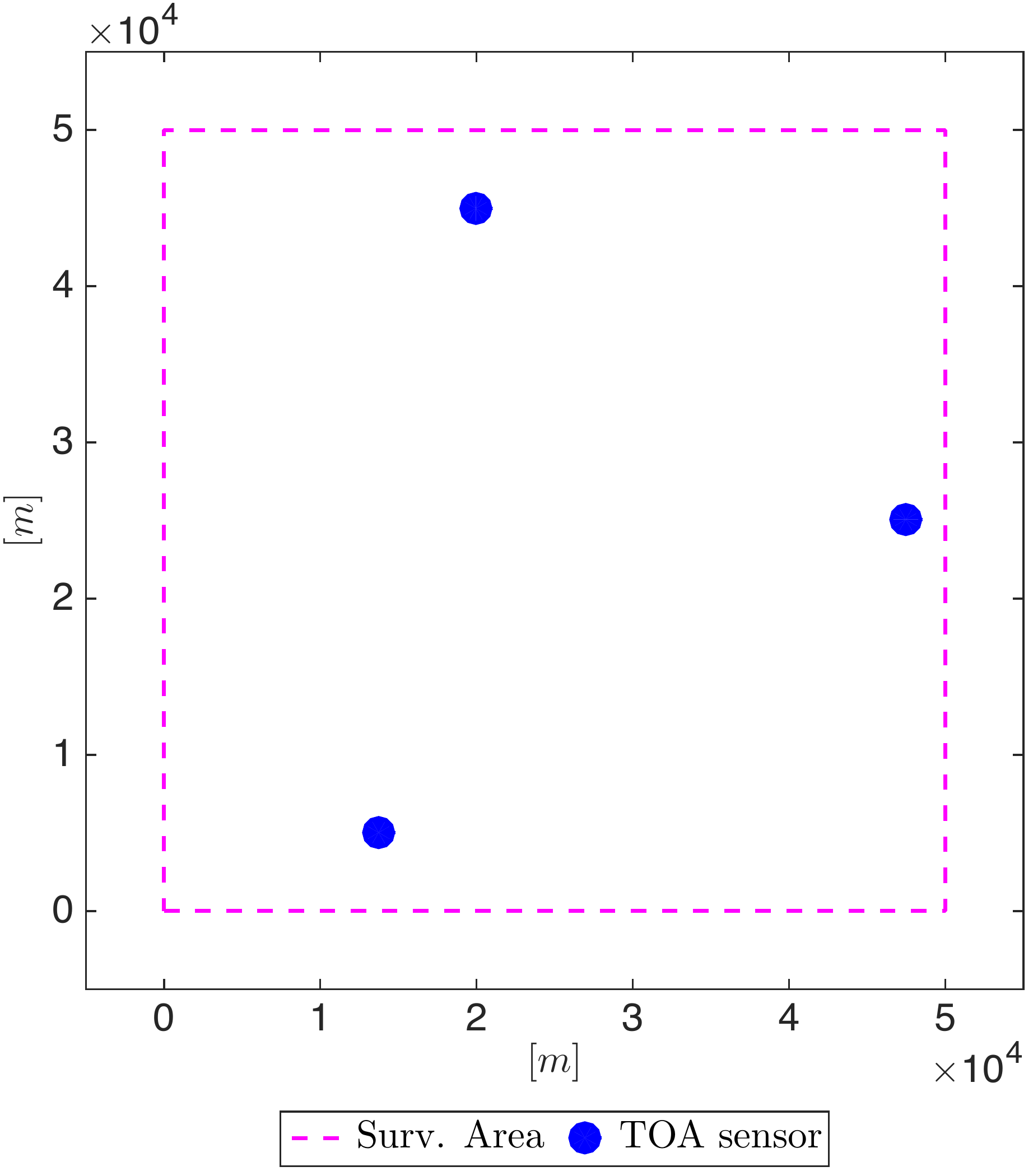}
	\caption{Network with 3 TOA sensors.}
	\label{fig:3toa}
\end{figure}
\begin{figure}[h!]
	\centering
	\includegraphics[width=0.6\columnwidth]{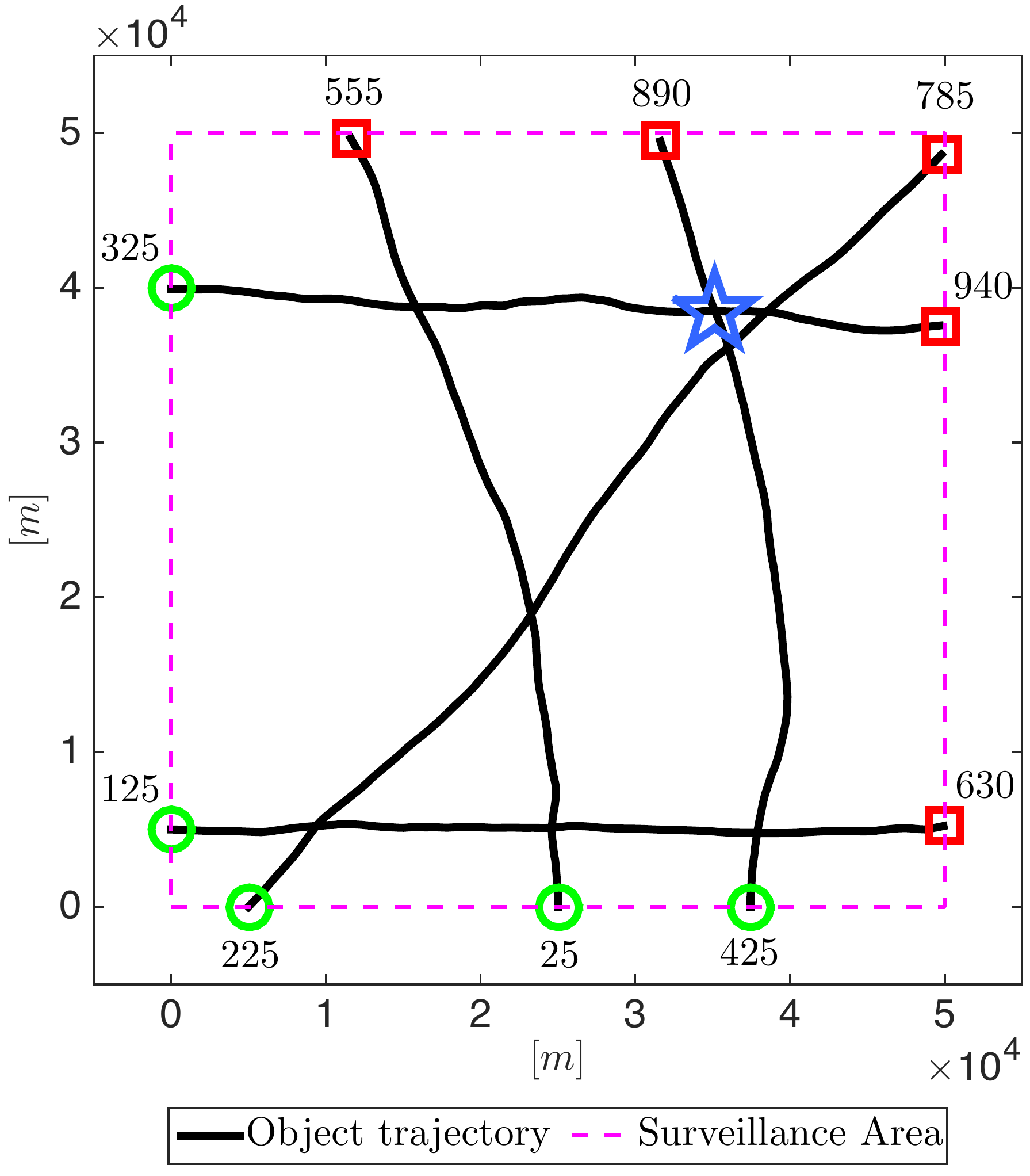}
	\caption{Object trajectories considered in the simulation experiment. The start/end point for each trajectory is denoted, respectively, by $\bullet\backslash\blacksquare$. The {\Large$\star$} indicates a rendezvous point.}
	\label{fig:cmot:5trajectories}
\end{figure}
For the sake of comparison, the M$\delta$-GLMB is also compared with the $\delta$-GLMB ($\delta$-GLMB) \cite{vovo1,vovo2} and LMB (LMB) \cite{lmbf} filters.
The three tracking filters are implemented using Gaussian Mixtures to represent their predicted and updated densities \cite{vovo2,lmbf}.
Due to the non linearity of the sensors, the \textit{Unscented Kalman Filter} (UKF) \cite{juluhl2004} is exploited to update means and covariances of the Gaussian components.

The kinematic object state is denoted by $x = \left[ p_{x}, \, \dot{p}_{x}, \, p_{y}, \, \dot{p}_{y} \right]^{\top}$, i.e. the planar position and velocity.
The object motion is modeled according to the Nearly-Constant Velocity (NCV) model \cite{far1985v1,far1985v2,book0,book}:
\be
x_{k + 1} = \left[ \begin{array}{cccc}
1 & T_{s} & 0 & 0	\\
0 & 1 	  & 0 & 0		\\
0 & 0 	  & 1 & T_{s} \\
0 & 0 	  & 0 & 1		\end{array} \right] x_{k} + w_{k} \, , \qquad
Q = \sigma_{w}^{2} \left[ \begin{array}{cccc}
\frac{1}{4}T_{s}^{4} & \frac{1}{2}T_{s}^{3} & 0 & 0 \\
\frac{1}{2}T_{s}^{3} & T_{s}^{2} & 0 & 0 \\
0 & 0 & \frac{1}{4}T_{s}^{4} & \frac{1}{2}T_{s}^{3}\\
0 & 0 & \frac{1}{2}T_{s}^{3} & T_{s}^{2} \end{array} \right]
\ee
where: $w_{k}$ is a white noise with zero mean and covariance matrix $Q$, $\sigma_{w} = 5 \, [m/s^{2}]$ and the sampling interval is $T_{s} = 5\,[s]$.

The radar has the following measurement function:
\begin{equation}
\begin{array}{c}
h(x) = \left[ \begin{array}{ll}
				\angle [ \left( p_{x} - x^{r} \right) + j \left( p_{y} - y^{r} \right)] \\[0.5em]
                                	\sqrt{ \left( p_{x} - x^{r} \right)^2+ \left( p_{y} - y^{r} \right)^2 }
			\end{array} \right]
\end{array}
\end{equation}
where $( x^{r}, y^{r} )$ represents the known position of the radar and the standard deviation of the measurement noise is $\sigma_{\theta} = 1 \, [\mbox{}^{\circ}]$ for the azimuth and $\sigma_{r} =  100 \, [m]$ for the range.
The measurement functions of the $3$ TOA of Fig. \ref{fig:3toa} are:
\begin{equation}
h(x) = \sqrt{ \left( p_{x} - x^{s} \right)^2+ \left( p_{y} - y^{s} \right)^2 } \, ,
\end{equation}
where $( x^{s}, y^{s} )$ represents the known position of sensor (indexed with) $s$. The standard deviation of the TOA measurement noise is taken as $\sigma_{TOA} = 100 \, [m]$.

The clutter is characterized by a Poisson process with parameter $\lambda_{c} = 15$.
The probability of object detection is $P_{D} = 0.85$.

In the considered scenario, objects pass through the surveillance area with partial prior information for object birth locations. 
Accordingly, a $10$-component LMB RFS $\boldsymbol{\pi}_{B} = \left\{ \left( r^{\left( \ell \right)}_{B},\right.\right.$ $\left.\left.p^{\left( \ell \right)}_{B} \right) \right\}_{\ell \in \mathbb{B}}$ has been hypothesized for the birth process.
Table \ref{tab:cmot:borderlineinit} gives a detailed summary of such components.
\begin{table}[h!]
	\setlength\arrayrulewidth{0.5pt}\arrayrulecolor{black} 
	\setlength\doublerulesep{0.5pt}\doublerulesepcolor{black} 
	\caption{Components of the LMB RFS birth process at a given time $k$.}
	\label{tab:cmot:borderlineinit}
	\centering
	$r^{\left( \ell \right)} = 0.09$\\
	$p^{\left( \ell \right)}_{B}(x) = \mathcal{N}\!\left( x;\, m^{\left( \ell \right)}_{B}, P_{B} \right)$\\
	$P_{B} = \operatorname{diag}\!\left( 10^{6}, 10^{4}, 10^{6}, 10^{4} \right)$\\\vspace{0.5em}
	\scalebox{1}{
	\begin{tabular}{>{\columncolor[gray]{.95}}c||c|c|c||}
		{\textbf{Label}} & $\left( k, \, 1 \right)$ & $\left( k, \, 2 \right)$ & $\left( k, \, 3 \right)$ \\
		\hline
		$m^{\left( \ell \right)}_{B}$ & $\left[ 0, \, 0, \, 40000,\, 0 \right]^{\top}$ & $\left[ 0, \, 0, \, 25000,\, 0 \right]^{\top}$ & $\left[ 0, \, 0, \, 5000,\, 0 \right]^{\top}$\\
		\hline
		\hline
	\end{tabular}
	}\\\vspace{0.5em}
	\scalebox{1}{
	\begin{tabular}{>{\columncolor[gray]{.95}}c||c|c|c||}
		{\textbf{Label}} & $\left( k, \, 4 \right)$ & $\left( k, \, 5 \right)$ & $\left( k, \, 6 \right)$\\
		\hline
		$m^{\left( \ell \right)}_{B}$ & $\left[ 5000, \, 0, \, 0,\, 0 \right]^{\top}$ & $\left[ 25000, \, 0, \, 0,\, 0 \right]^{\top}$ & $\left[ 36000, \, 0, \, 0,\, 0 \right]^{\top}$\\
		\hline
		\hline
	\end{tabular}
	}\\\vspace{0.5em}
	\scalebox{1}{
	\begin{tabular}{>{\columncolor[gray]{.95}}c||c|c||}
		{\textbf{Label}} & $\left( k, \, 7 \right)$ & $\left( k, \, 8 \right)$\\
		\hline
		$m^{\left( \ell \right)}_{B}$ & $\left[ 50000, \, 0, \, 15000,\, 0 \right]^{\top}$ & $\left[ 50000, \, 0, \, 40000,\, 0 \right]^{\top}$\\
		\hline
		\hline
	\end{tabular}
	}\\\vspace{0.5em}
	\scalebox{1}{
	\begin{tabular}{>{\columncolor[gray]{.95}}c||c|c||}
		{\textbf{Label}} & $\left( k, \, 9 \right)$ & $\left( k, \, 10 \right)$\\
		\hline
		$m^{\left( \ell \right)}_{B}$ & $\left[ 40000, \, 0, \, 50000,\, 0 \right]^{\top}$ & $\left[ 10000, \, 0, \, 50000,\, 0 \right]^{\top}$\\
		\hline
		\hline
	\end{tabular}
	}
\end{table}
Due to the partial prior information on the object birth locations, some of the LMB components cover a state space region where there is no birth.
Clutter measurements are, therefore, more prone to generate false objects.

Multi-object tracking performance is evaluated in terms of the \textit{Optimal SubPattern Analysis} (OSPA) metric \cite{schvovo2008} with Euclidean distance $p = 2$ and cutoff $c = 600 \, [m]$.
The reported metric is averaged over $100$ Monte Carlo trials for the same object trajectories but different, independently generated, clutter and measurement noise realizations. The duration of each simulation trial is fixed to $1000 \, [s]$ ($200$ samples).

The three tracking filters are coupled with the \textit{parallel CPHD look ahead strategy} described in \cite{vovo1,vovo2}. The CPHD \cite{vo-vo-cantoni} filter.

\subsection{Scenario 1: radar}
Figs. \ref{fig:1:cardMDGLMB}, \ref{fig:1:cardDGLMB} and \ref{fig:1:cardLMB} display the statistics (mean and standard deviation) of the estimated number of objects obtained, respectively, with M$\delta$-GLMB, $\delta$-GLMB and LMB tracking filters.
As it can be seen, all the algorithms estimate the object cardinality accurately, with no substantial differences.
This result implies that, in the presence of a single sensor guaranteeing observability, the approximations made by both M$\delta$-GLMB and LMB trackers are not critical in that they provide performance comparable to $\delta$-GLMB with the advantage of a cheaper computational burden and reduced storage requirements.
Note that the problems introduced by the rendezvous point (e.g. merged or lost tracks) are correctly tackled by all the algorithms.

Fig. \ref{fig:cmot1:ospa} shows the OSPA distance of the algorithms.
Note again that, in agreement with the estimated cardinality distributions, the OSPA distances are nearly identical.
\begin{figure}[h!]
        \begin{minipage}[t][][t]{\columnwidth}
	        	\centering
		\includegraphics[width=\columnwidth]{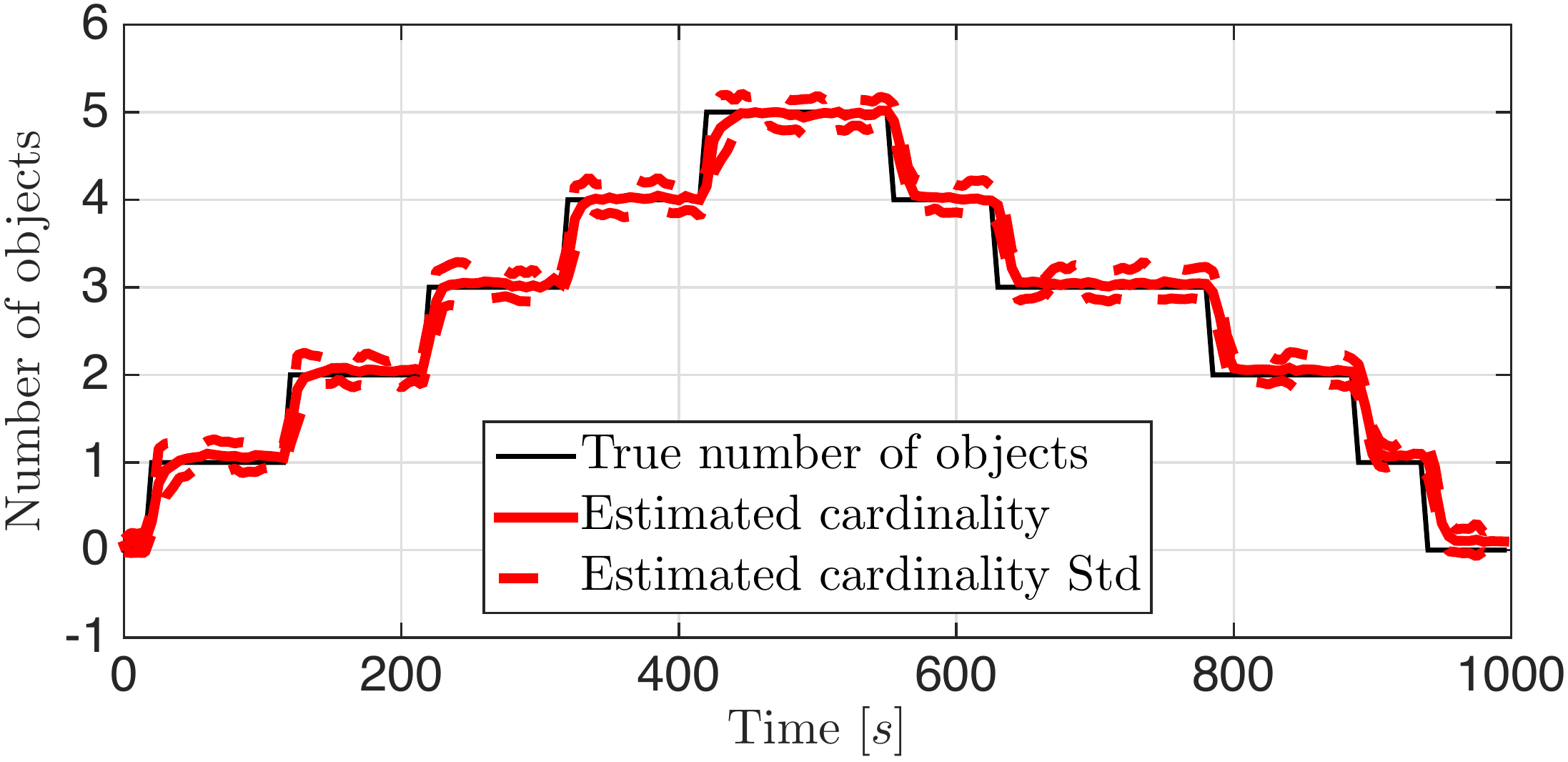}
		\caption{Cardinality statistics for M$\delta$-GLMB tracking filter using 1 radar.}
		\label{fig:1:cardMDGLMB}
        \end{minipage}\vspace{0.5em}
        \begin{minipage}[t][][t]{\columnwidth}
	        \centering
		\includegraphics[width=\columnwidth]{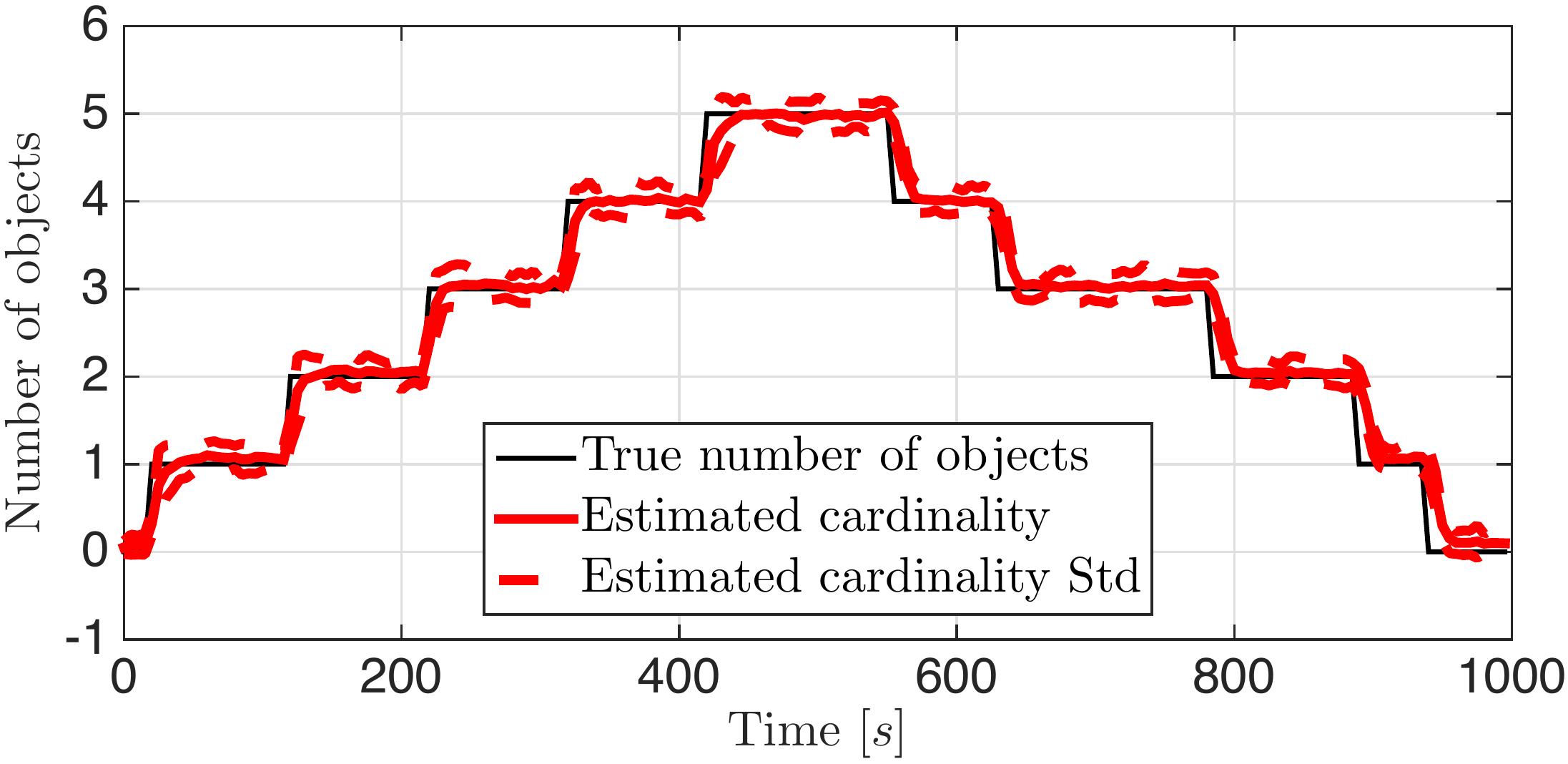}
		\caption{Cardinality statistics for $\delta$-GLMB tracking filter using 1 radar.}
		\label{fig:1:cardDGLMB}
        \end{minipage}
\end{figure}
\begin{figure}[h!]
        \begin{minipage}[t][][t]{\columnwidth}
	        	\centering
		\includegraphics[width=\columnwidth]{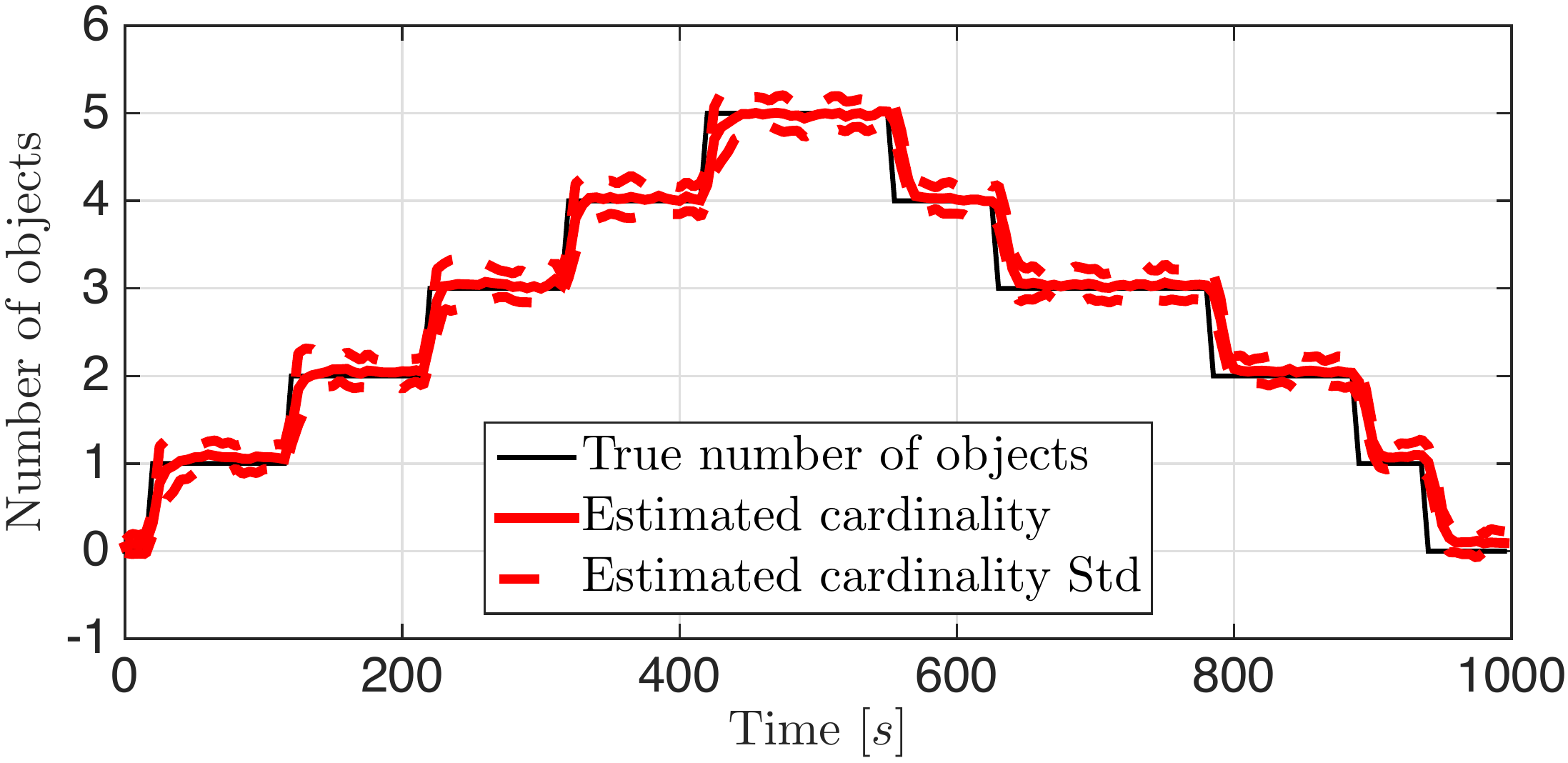}
		\caption{Cardinality statistics for LMB tracking filter using 1 radar.}
		\label{fig:1:cardLMB}
        \end{minipage}\vspace{0.5em}
        \begin{minipage}[t][][t]{\columnwidth}
        		\centering
		\includegraphics[width=\columnwidth]{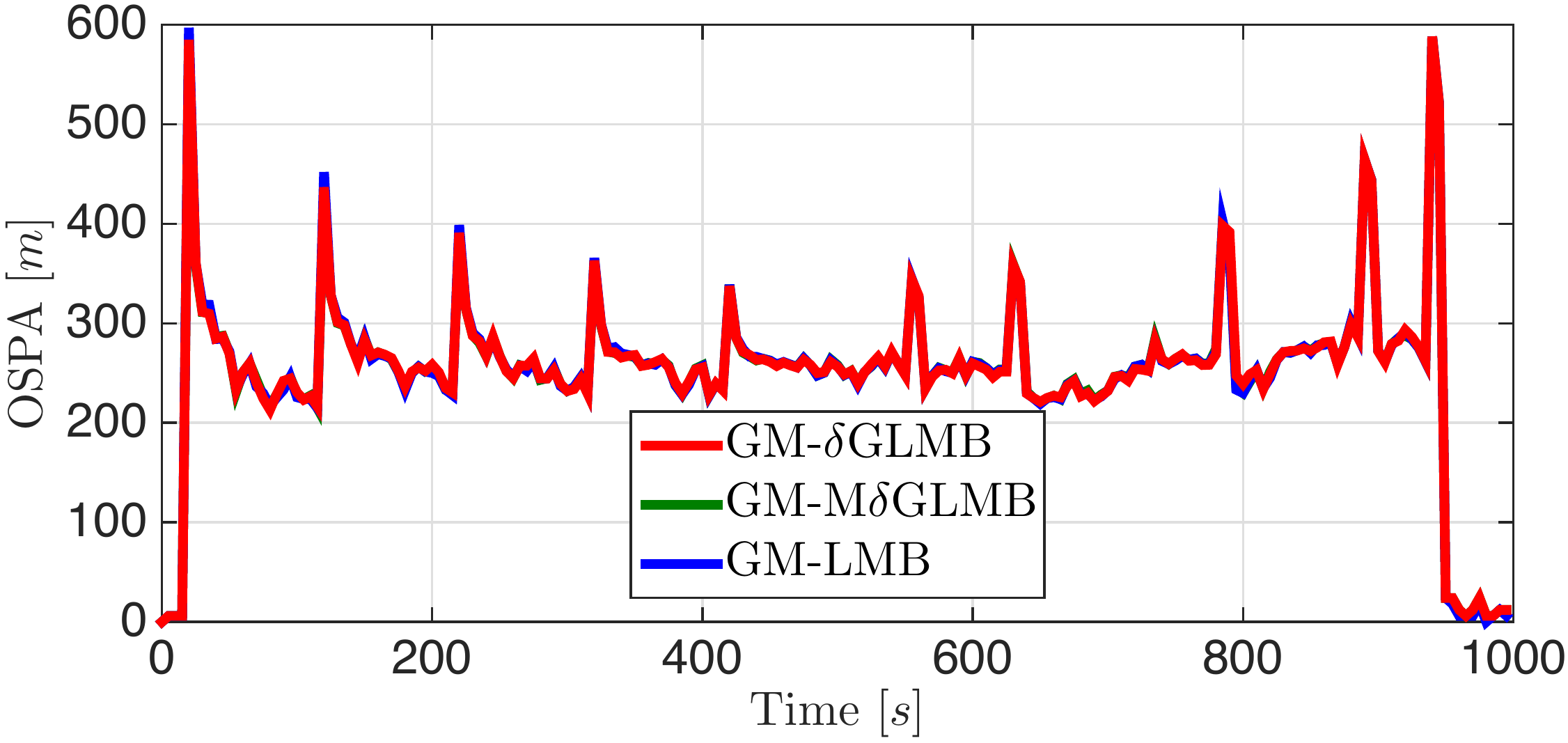}
		\caption{OSPA distance ($c = 600 \, [m]$, $p = 2$) using 1 radar.}
		\label{fig:cmot1:ospa}
        \end{minipage}
\end{figure}

\subsection{Scenario 2: 3 TOA sensors}
Figs. \ref{fig:2:cardMDGLMB}, \ref{fig:2:cardDGLMB} and \ref{fig:2:cardLMB} display the statistics (mean and standard deviation) of the estimated number of objects obtained, respectively, with M$\delta$-GLMB, $\delta$-GLMB and LMB trackers.
The M$\delta$-GLMB and the $\delta$-GLMB tracking filters estimate the object cardinality accurately, while the LMB exhibits poor performance and higher standard deviation due to losing some tracks when 4 or 5 objects are jointly present in the surveillance area.
It is worth noticing that the M$\delta$-GLMB tracker exhibits performance very close to that of $\delta$-GLMB and that the problems introduced by the rendezvous point are again succesfully tackled.

Fig. \ref{fig:cmot2:ospa} shows the OSPA distance.
Note that the OSPA of M$\delta$-GLMB is close to the one of $\delta$-GLMB, while LMB shows an overall higher error in agreement with the cardinality error due to losing tracks.
\begin{figure}[h!]
	\begin{minipage}[t][][t]{\columnwidth}
		\centering
		\includegraphics[width=\columnwidth]{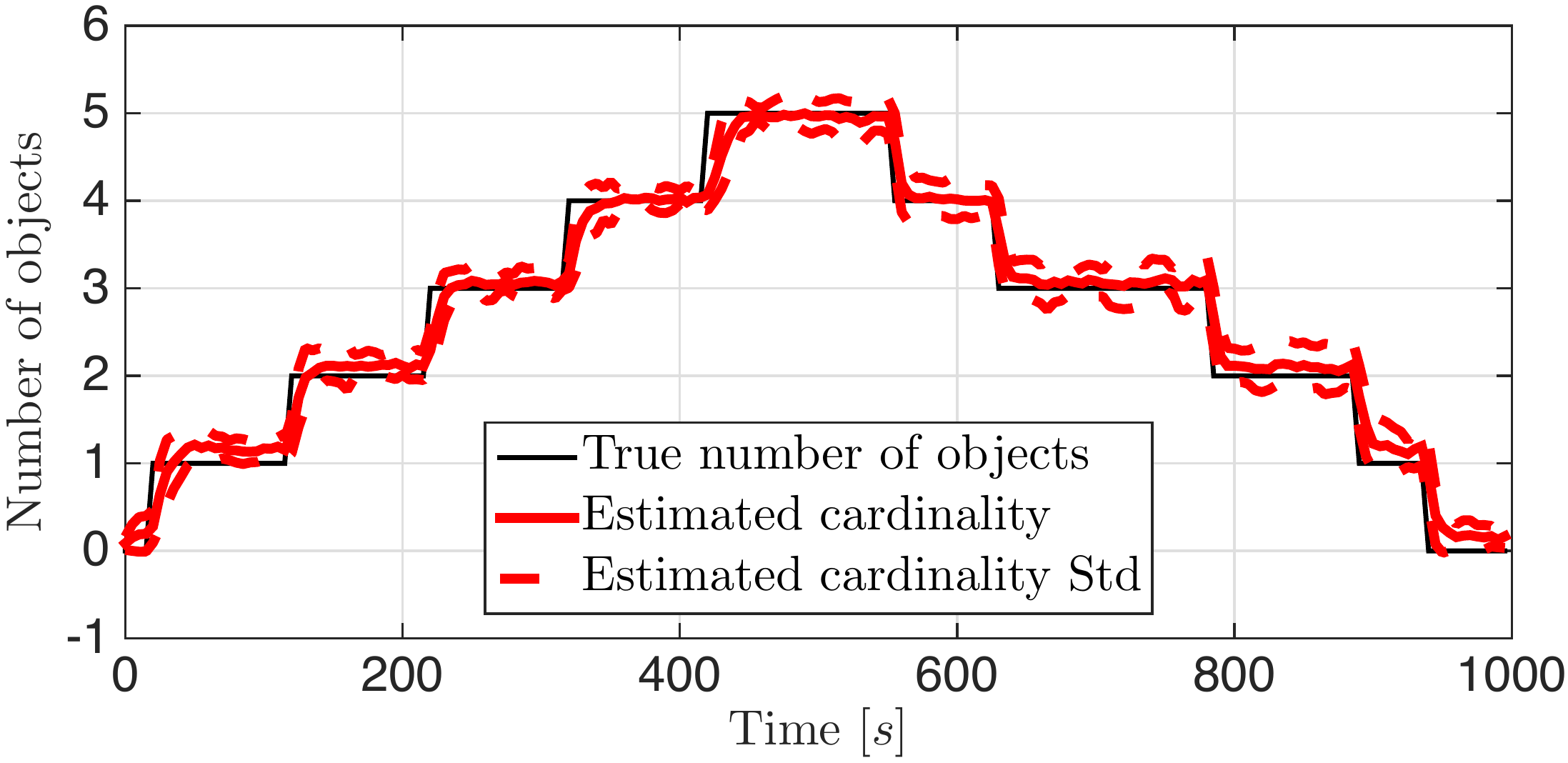}
		\caption{Cardinality statistics for M$\delta$-GLMB tracking filter using 3 TOA.}
		\label{fig:2:cardMDGLMB}
        \end{minipage}\vspace{0.5em}
        \begin{minipage}[t][][t]{\columnwidth}
        		\centering
		\includegraphics[width=\columnwidth]{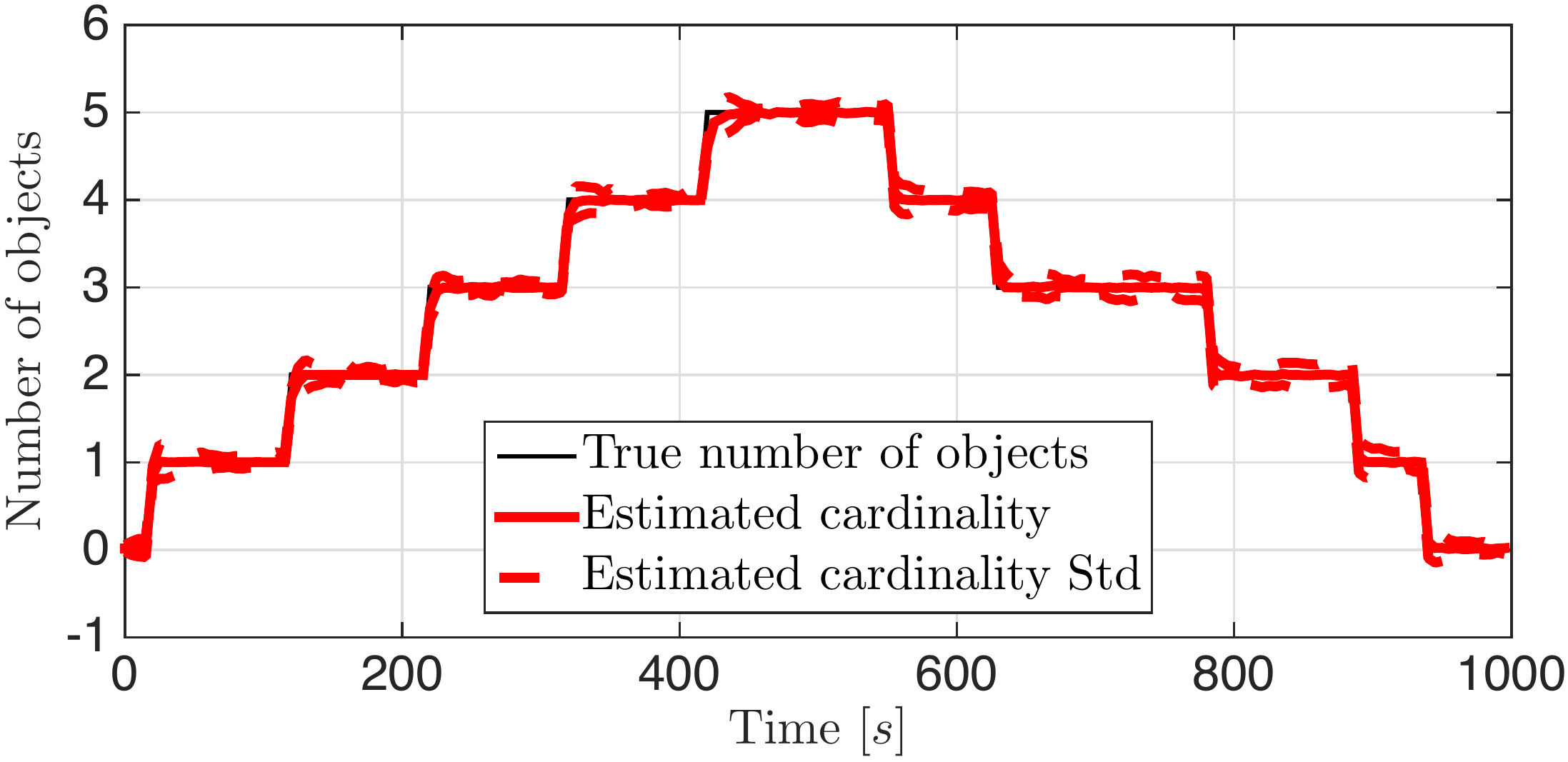}
		\caption{Cardinality statistics for $\delta$-GLMB tracking filter using 3 TOA.}
		\label{fig:2:cardDGLMB}
        \end{minipage}
\end{figure}
\begin{figure}[h!]
        \begin{minipage}[t][][t]{\columnwidth}
        		\centering
		\includegraphics[width=\columnwidth]{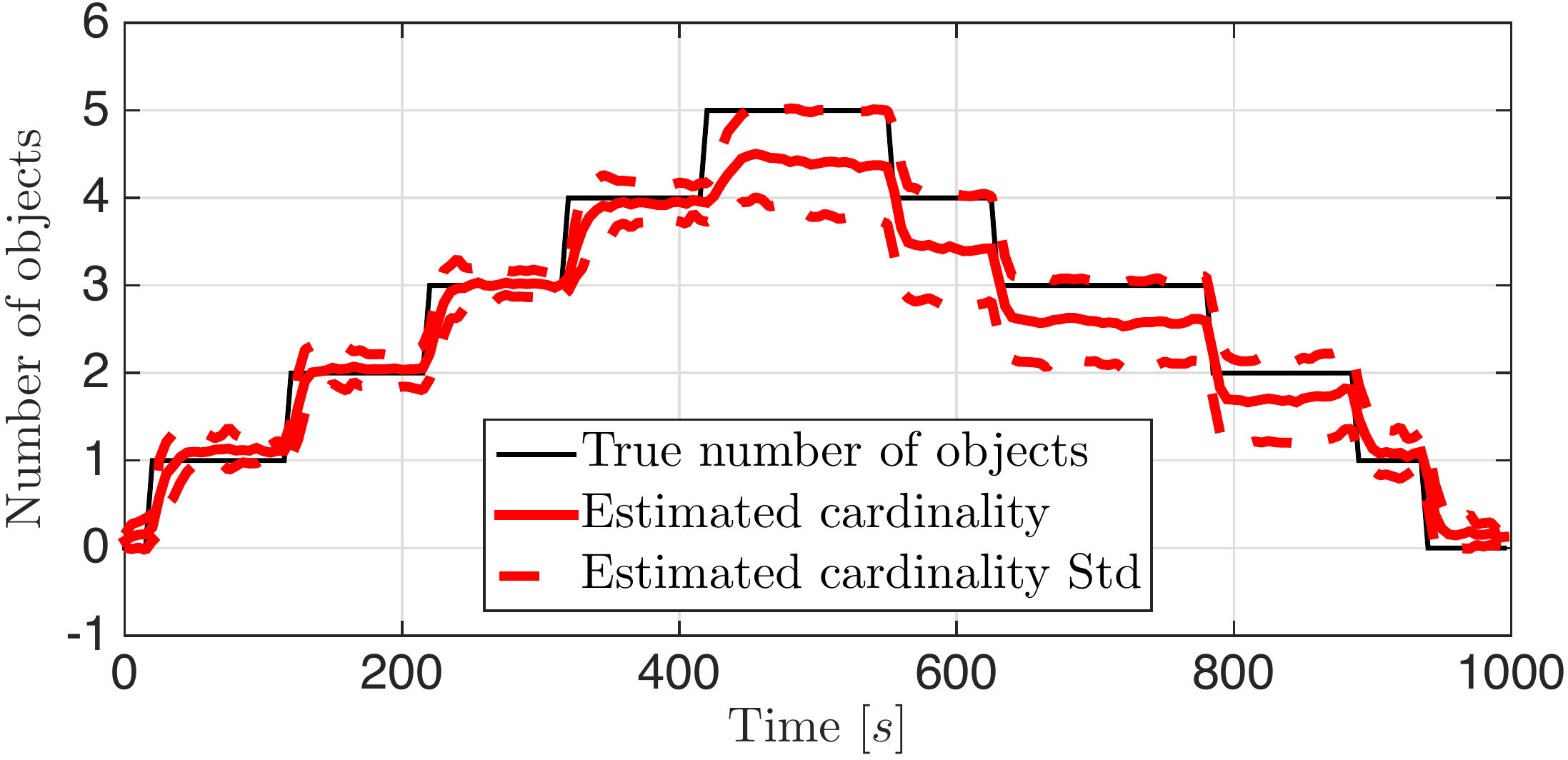}
		\caption{Cardinality statistics for LMB tracking filter using 3 TOA.}
		\label{fig:2:cardLMB}
        \end{minipage}\vspace{0.5em}
        \begin{minipage}[t][][t]{\columnwidth}
        		\centering
		\includegraphics[width=\columnwidth]{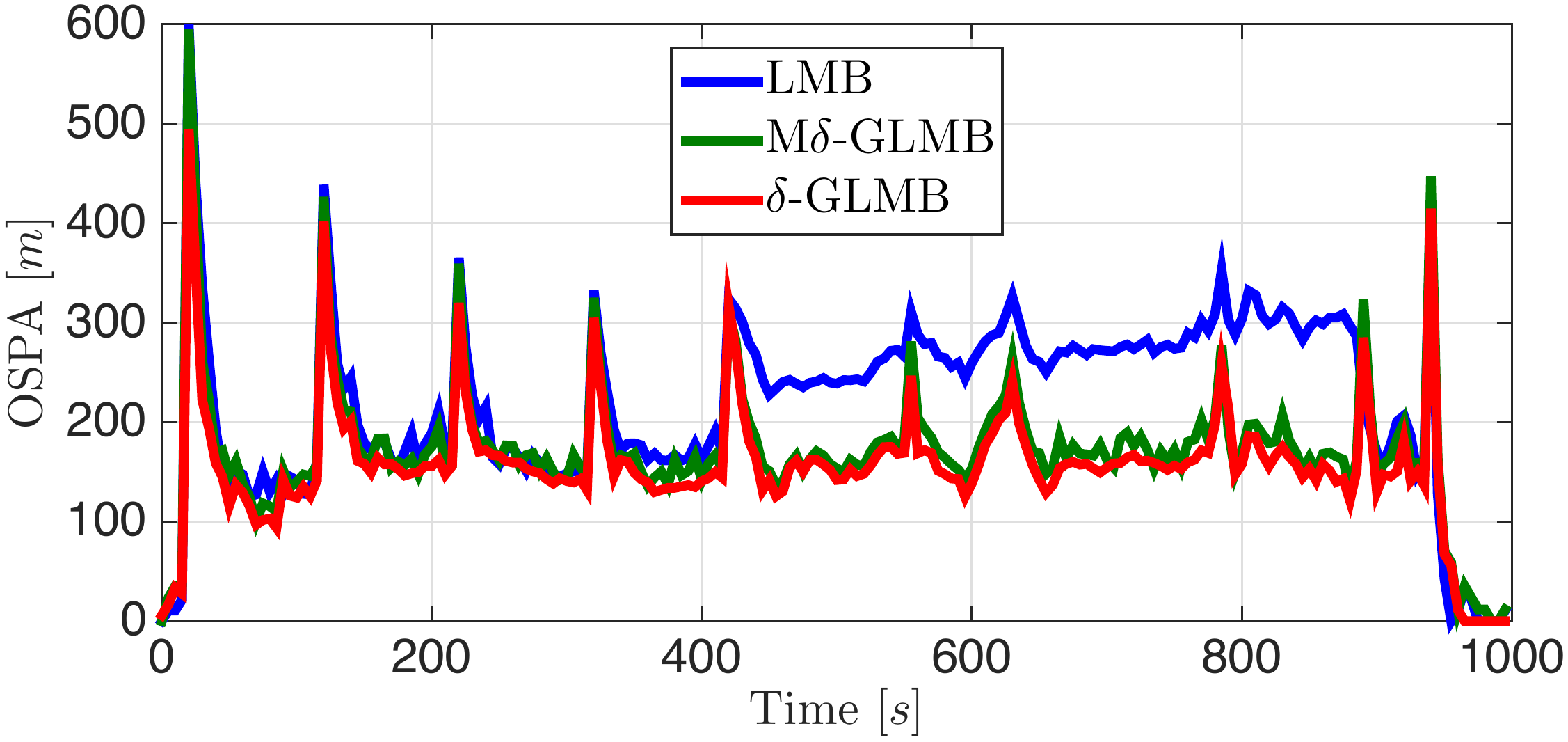}
		\caption{OSPA distance ($c = 600 \, [m]$, $p = 2$) using 3 TOA.}
		\label{fig:cmot2:ospa}
        \end{minipage}
\end{figure}

% DISTRIBUTED MULTI-OBJECT TRACKING
\chapter{Distributed multi-object tracking}
\label{chap:dmot}
\spminitoc
The focus of this chapter is on \textit{Distributed Multi-Object Tracking} (DMOT).
We adopt the RFS formulation since it provides the concept of \textit{probability density} for \textit{multi-object state} that allows us to directly extend existing tools in distributed estimation to the multi-object case.
Such a concept is not available in the \textit{Multiple Hypotheses Tracking} (MHT) and \textit{Joint Probabilistic Data Association} (JPDA) approaches \cite{reid,far1985v1,far1985v2,book0,book,BlPo}.

DMOF based on the RFS paradigm has previously been described in chapter \ref{chap:dmof} \cite{mahler-chap-2012,ccphd,emd,ccjpda}.
Specifically, an information-theoretic approach to robust distributed multi-object estimation based on the CCPHD filter has been proposed.
However, this formulation is not able to provide estimates of the object's trajectories or tracks.

Taking into account the above-mentioned considerations, the main contributions in the present chapter are \cite{fanvovo2015}:
\begin{itemize}
	\item the generalization of the robust distributed multi-object filtering approach of subsection \ref{ssec:KLA} to multi-object tracking;
	\item the development of the \textit{Consensus Marginalized $\delta$-Generalized Labeled Multi-Bernoulli} (CM$\delta$GLMB) and \textit{Consensus Labeled Multi-Bernoulli} (CLMB) filters as the first distributed multi-object trackers in the RFS framework.
\end{itemize}

The proposed solutions are based on the concept of labeled RFS introduced in subsection \ref{ssec:lrfss}, that enables the estimation of multi-object trajectories in a principled manner \cite{vovo1}.
Moreover, labeled RFS conjugate priors \cite{vovo1} have lead to the development of a tractable analytic multi-object tracking solution called the \textit{$\delta $-Generalized Labeled Multi-Bernoulli} ($\delta $-GLMB) filter \cite{vovo2}.
However, it is not known if this filter is amenable to DMOT.
Nonetheless, the M$\delta$-GLMB and the LMB filters \cite{mdglmbf,lmbf} are two efficient approximations of the $\delta$-GLMB filter that
\begin{enumerate}[label=\textsc{\roman*}.]
	\item have an appealing mathematical formulation that facilitates an efficient and tractable closed-form fusion rule for DMOT;
	\item preserve key summary statistics of the full multi-object posterior.
\end{enumerate}
In addition, labeled RFS-based trackers do not suffer from the so-called ``\textit{spooky effect}'' \cite{spooky} that degrades the performance in the presence of low detection probability like the multi-object filters \cite{vo-vo-cantoni,ccphd,emd}.

\section{Information fusion with labeled RFS}
\label{sec:lrfsnwgm}
In this section, the notion of KLA (\ref{KLA}) introduced in subsection \ref{ssec:KLA} will be extended to densities of labeled RFSs and will be shown to be equivalent to the \textit{normalized weighted geometric mean}.
Proofs of the results are provided in the appendix \ref{chap:appendix}.

\subsection{Labeled multi-object Kullback-Leibler average}
\label{ssec:lkla}

The benefit of using the RFS framework which provides the concept of probability density of the multi-object state, allows us to directly extend the notion of multi-object KLA originally devised in \cite{ccphd} to labeled multi-object densities \cite{fanvovo2015}.
Herewith, we adopt the measure theoretic notion of multi-object density given in \cite{VSD05}, which does not suffer of the unit compatibility problem in the integrals involving the product of powers of multi-object densities. Note that this density is equivalent to the FISST density as shown in \cite{VSD05}. Thus, the standard inner product notation is extended to multi-object densities
\be
	\left\langle \boldsymbol{f}, \,\boldsymbol{g}\right\rangle \triangleq \int \boldsymbol{f}(\mathbf{X})\,\boldsymbol{g}(\mathbf{X}) \delta \mathbf{X} \, .
\ee

The weighted KLA $\boldsymbol{f}_{KLA}$ of the labeled multi-object densities $\boldsymbol{f}^{i}$ is defined by 
\begin{IEEEeqnarray}{rCl}
	\boldsymbol{f}_{KLA} & \triangleq & \arg \inf_{\boldsymbol{f}} \sum_{i \in \mathcal{N}} \omega^{i} D_{KL}\left( \boldsymbol{f} \parallel \boldsymbol{f}^{i} \right), \label{eq:lkla}\\
	\omega^{i} & \ge & 0,~~~~~\sum_{i \in \mathcal{N}} \omega^{i} = 1. \label{eq:klaweights}
\end{IEEEeqnarray}

where  
\begin{equation}
	D_{KL}\left( \boldsymbol{f}\parallel \boldsymbol{g}\right) \triangleq \int \boldsymbol{f}(\mathbf{X})\log \!\dfrac{\boldsymbol{f}(\mathbf{X})}{\boldsymbol{g}(\mathbf{X})}\delta \mathbf{X}  \label{eq:dkl}
\end{equation}
is the KLD \cite{mahler1,mahler} between two multi-object densities $\boldsymbol{f}(\cdot)$ and $\boldsymbol{g}(\cdot)$.
\begin{thm}[KLA of labeled RFSs]\label{thm:kla}~\\
	The weighted KLA defined in (\ref{eq:lkla}) is the normalized weighted geometric mean of the multi-object densities $\boldsymbol{f}^{i}(\cdot)$, $i\in \mathcal{N}$, i.e.
	\begin{equation}
		\boldsymbol{f}_{KLA}=\bigoplus_{i\in \mathcal{N}}\left( \omega^{i} \odot \boldsymbol{f}^{i}\right) .  \label{eq:kla:gci}
	\end{equation}
\end{thm}
\noindent Theorem \ref{thm:kla} follows from Theorem 1 of \cite[Section III.B]{ccphd}.

\begin{rem}
The fusion rule (\ref{eq:kla:gci}) is the labeled multi-object version of the multi-object Chernoff fusion first proposed by Mahler \cite{mah2000}.
\end{rem}

%Note that the general form of the weighted geometric mean involves exponentiation of its terms.
%Hence, such an operation for GLMB and $\delta$-GLMB RFSs cannot be evaluated in a closed form solution due to the generation of the cross-terms (cfr. (\ref{eq:glmbpdf}) and (\ref{eq:dglmbpdf2})).
%On the other hand, the simpler forms of M$\delta$-GLMB and LMB RFSs allow exponentiation in a closed and neat form, thus making their densities two perfect candidates for extending the KLA (\ref{KLA}) to labeled RFSs.
%Therefore, it is shown that the KLAs of M$\delta$-GLMB and LMB densities are also, respectively, a M$\delta$-GLMB and an LMB density.
%In particular, it is established that the \textit{normalized weighted geometric mean}s of such densities are also M$\delta$-GLMB and LMB and a close form expression for them is derived
%
%In order to implement the fusion rule (\ref{eq:kla:gci}) with information provided by the agents in form M$\delta$-GLMB or LMB, it is necessary to study their normalized weighted geometric mean.
Subsequently, it will be shown that the KLAs of M$\delta$-GLMB and LMB densities are also, respectively, a M$\delta$-GLMB and an LMB densities.
In particular closed-form solutions for the normalized weighted geometric means are derived.
These results are necessary to implement the fusion rule (\ref{eq:kla:gci}) for M$\delta$-GLMB and LMB tracking filters.

\subsection{Normalized weighted geometric mean of M$\delta$-GLMB densities}
The following result holds.

\begin{thm}[NWGM of M$\delta$-GLMB RFSs]\label{thm:mdglmb:nwgm}~\\
	Let $\boldsymbol{\pi }^{\imath}(\cdot)$, $\imath =1,\dots ,\mathcal{I}$, be M$\delta$-GLMB densities on $\mathcal{F}(\mathbb{X}\mathcal{\times }\mathbb{L})$ and $\omega^{\imath}\in \left( 0,1 \right)$, $\imath =1,\dots ,\mathcal{I}$, such that $\sum_{\imath =1}^{\mathcal{I}}\omega^{\imath}=1$. Then the normalized weighted geometric mean is given by
	\begin{equation}
		{\bigoplus_{\imath =1}^{\mathcal{I}}}\,\left( \omega^{\imath}\odot \boldsymbol{\pi}^{\imath}\right)\left( \lb{X} \right) = \Delta\!(\mathbf{X})\sum_{L \in \mathcal{F}\!\left( \mathbb{L} \right)} \delta_{L}\!\left( \mathcal{L}\!\left( \mathbf{X} \right) \right) \, \overline{w}^{\left( L \right)}\left[ \overline{p}^{\left( L \right)}\right]^{\mathbf{X}}\label{eq:mdglmb:nwgm}
	\end{equation}
	where
	\begin{IEEEeqnarray}{rCl}
		\overline{p}^{\left( L \right)} & = & \dfrac{\displaystyle \prod_{\imath = 1}^{\mathcal{\mathcal{I}}} \left( p_{\imath}^{\left( L \right)} \right)^{\omega^{\imath}}}{\displaystyle\int \prod_{\imath = 1}^{\mathcal{I}} \left( p_{\imath}^{\left( L \right)} \right)^{\omega^{\imath}} dx}\\
		\overline{w}^{\left( L \right)} & = & \dfrac{\displaystyle\prod_{\imath = 1}^{\mathcal{\mathcal{I}}} \left( w_{\imath}^{\left( L \right)} \right)^{\omega^{\imath}} \left[ \int \prod_{\imath = 1}^{\mathcal{I}} \left( p_{\imath}^{\left( L \right)}\!\left( x, \cdot \right) \right)^{\omega^{\imath}} dx \right]^{L}}
			{\displaystyle \sum_{F \subseteq \mathbb{L}} \prod_{\imath = 1}^{\mathcal{I}} \left( w_{\imath}^{\left( F \right)} \right)^{\omega^{\imath}} \left[ \int \prod_{\imath = 1}^{\mathcal{I}} \left( p_{\imath}^{\left( F \right)}\!\left( x, \cdot \right) \right)^{\omega^{\imath}} dx \right]^{F}}
	\end{IEEEeqnarray}
\end{thm}

The fusion rule for M$\delta$-GLMBs follows by applying Theorem \ref{thm:mdglmb:nwgm} to find the KLA (\ref{eq:kla:gci}) of the M$\delta$-GLMB's $\boldsymbol{\pi}^{i}(\cdot)$, $i\in \mathcal{N}$. This is summarized in the following Proposition.

\begin{pro}[KLA of M$\delta$-GLMB RFSs]\label{pro:mdglmb:fusion}~\\
	Suppose that each agent $i \in \mathcal{N}$ is provided with an M$\delta$-GLMB $\boldsymbol{\pi}^{i}(\cdot)$ and that all agents share the same label space for the birth process, then the M$\delta$-GLMB components $\overline{p}^{\left( L \right)}(\cdot)$ and $\overline{w}^{\left( L \right)}$ of the KLA are given by
	\begin{IEEEeqnarray}{rCl}
	\overline{p}^{\left( L \right)}\!\left( x, \ell \right) & = & \bigoplus_{i \in \mathcal{N}} \left( \omega^{i} \odot p^{\left( L \right)} \right)\!\left( x, \ell \right) \, , \label{eq:mdglmbfusion:pdf}\\
		\overline{w}^{\left( L \right)} & = & \dfrac{\widetilde{w}^{\left( L \right)}}{\displaystyle \sum_{F \subseteq \mathbb{L}} \widetilde{w}^{\left( F \right)}} \, ,\label{eq:mdglmbfusion:w}
	\end{IEEEeqnarray}
	where 
	\begin{IEEEeqnarray}{rCl}
		\widetilde{w}^{\left( L \right)} & = & \displaystyle\prod_{i \in \mathcal{N}} \left( w_{i}^{\left( L \right)} \right)^{\omega^{i}} \, \left[ \int \widetilde{p}^{\left( L \right)}\!\left( x, \cdot \right) dx \right]^{L} \, ,\\
		\widetilde{p}^{\left( L \right)} & = & \prod_{i  \in \mathcal{N}} \left( p_{i}^{\left( L \right)} \right)^{\omega^{i}} \, .
	\end{IEEEeqnarray}
\end{pro}

\begin{rem}
Note that the quantities $\widetilde{w}^{\left( L \right)}$ and $\overline{p}^{\left( L \right)}$ can be independently determined using (\ref{eq:mdglmbfusion:pdf}) and (\ref{eq:mdglmbfusion:w}). Thus, the overall fusion procedure is fully parallelizable.
\end{rem}

\begin{rem}
Notice that (\ref{eq:mdglmbfusion:pdf}) is indeed the CI fusion rule \cite{juluhl1997} for the single-object PDFs.
\end{rem}

\subsection{Normalized weighted geometric mean of LMB densities}
The following result holds.

\begin{thm}[NWGM of LMB RFSs]\label{thm:lmb:nwgm}~\\
	Let $\boldsymbol{\pi }^{\imath}=\left\{ \left(r_{\imath}^{(\ell )},p_{\imath}^{(\ell )}\right) \right\} _{\ell \in \mathbb{L}}$, $\imath =1,\dots ,\mathcal{I}$, be LMB densities on $\mathcal{F}(\mathbb{X}\mathcal{\times }\mathbb{L})$ and $\omega^{\imath}\in \left( 0,1 \right)$, $\imath =1,\dots ,\mathcal{I}$, such that $\sum_{\imath =1}^{\mathcal{I}}\omega^{\imath}=1$. Then the normalized weighted geometric mean is given by
	\begin{equation}
		{\bigoplus_{\imath =1}^{\mathcal{I}}}\,\left( \omega^{\imath}\odot \boldsymbol{\pi}^{\imath}\right) =\left\{ \left( \overline{r}^{(\ell )},\overline{p}^{(\ell )}\right)\right\} _{\ell \in \mathbb{L}} \label{eq:lmb:nwgm}
	\end{equation}
	where
	\begin{IEEEeqnarray}{rCl}
		\overline{r}^{(\ell )} & = & \dfrac{\displaystyle \int \prod\limits_{\imath =1}^{\mathcal{I}}\left( r_{\imath}^{(\ell )}p_{\imath}^{(\ell )}(x)\right) ^{\omega^{\imath}}dx}{\displaystyle \prod\limits_{\imath =1}^{\mathcal{I}}\left( 1-r_{\imath}^{(\ell )}\right) ^{\omega^{\imath}}+\int \prod\limits_{\imath =1}^{\mathcal{I}}\left( r_{\imath}^{(\ell)}p_{\imath}^{(\ell )}(x)\right) ^{\omega^{\imath}}dx} \\
		\overline{p}^{(\ell )} &=&{\bigoplus_{\imath =1}^{\mathcal{I}}}\,\left( \omega^{\imath}\odot p_{\imath}^{(\ell )}\right)
	\end{IEEEeqnarray}
\end{thm}

The fusion rule for LMBs follows by applying Theorem \ref{thm:lmb:nwgm} to find
the KLA (\ref{eq:kla:gci}) of the LMB's $\left\{ \left( r_{i}^{(\ell
)},p_{i}^{(\ell )}\right) \right\} _{\ell \in \mathbb{L}}$, $i\in \mathcal{N}
$. This is summarized in the following Proposition.

\begin{pro}[KLA of LMB RFSs]\label{pro:lmb:fusion}~\\
Suppose that each agent $i \in \mathcal{N}$ is provided with an LMB $\boldsymbol{\pi}^{i}(\cdot)$ and that all agents share the same label space for the birth process, then the LMB components $\left\{ \left( \overline{r}^{(\ell )},\overline{p}^{(\ell)}\right) \right\} _{\ell \in \mathbb{L}}$ of the KLA are given by
\begin{IEEEeqnarray}{rCl}
		\overline{r}^{(\ell)} &=& \dfrac{\widetilde{r}^{(\ell)}}{\widetilde{q}^{(\ell)} + \widetilde{r}^{(\ell)}} \, ,\label{eq:lmbfusion:ex}\\
		\overline{p}^{(\ell)}\!(x) & = & \bigoplus_{i \in \mathcal{N}} \left( \omega^{i} \odot p_{i}^{(\ell)} \right)\!\left(x\right) \, , \label{eq:lmbfusion:pdf}
	\end{IEEEeqnarray}where 
\begin{IEEEeqnarray}{rCl}
		\widetilde{r}^{(\ell)} & = & \displaystyle\int\prod_{i \in \mathcal{N}} \left( r_{i}^{(\ell)} p_{i}^{(\ell)}(x) \right)^{\omega^{i}}dx \label{eq:fusion:barex} \, ,\\
		\widetilde{q}^{(\ell)} & = & \prod_{i \in \mathcal{N}}\left( 1 - r_{i}^{(\ell)} \right)^{\omega^{i}} \label{eq:fusion:barnex} \, .
	\end{IEEEeqnarray}
\end{pro}

\begin{rem}
Note that each Bernoulli component $\left( \overline{r}^{(\ell )},\overline{p}^{(\ell
)}\right) $ can be independently determined using eqs. (\ref{eq:lmbfusion:ex}) 
%, (\ref{eq:lmbfusion:pdf}) and (\ref{eq:fusion:nex}).
and (\ref{eq:lmbfusion:pdf}). Thus, the overall fusion procedure is fully
parallelizable.
\end{rem}

\begin{rem}
Notice that also in this case eq. (\ref{eq:lmbfusion:pdf}) is the CI fusion rule \cite{juluhl1997} for the single-object PDFs.
\end{rem}

\subsection{Distributed Bayesian multi-object tracking via consensus}
\label{sec:fmttf}
At time $k$, the global KLA (\ref{eq:kla:gci}) which would require all the local multi-object densities $\boldsymbol{\pi}_{k}^{i}(\cdot)$, $i\in \mathcal{N}$, to be available, can be computed in a distributed and scalable way by iterating regional averages via the consensus algorithm \cite[Section III.A]{ccphd}\cite{cp} described in subsection \ref{ssec:consensus}.
Thus, each agent $i \in \mathcal{N}$ iterates the consensus steps
\begin{equation}
	\boldsymbol{\pi }_{k,l}^{i}=\bigoplus_{j\in \mathcal{N}^{i}}\left( \omega^{i,j}\odot \boldsymbol{\pi}_{k,l-1}^{j}\right) \, , \label{eq:consenus:alg}
\end{equation}
with $\boldsymbol{\pi}_{k,0}^{i}(\cdot) = \boldsymbol{\pi}_{k}^{i}(\cdot)$; $\omega^{i,j}\geq 0$, satisfying $\sum_{j\in \mathcal{N}^{i}}~\omega^{i,j}=1$, are the consensus weights relating agent $i$ to nodes $j\in \mathcal{N}^{i}$.
Using the same arguments as the single-object case (subsection \ref{ssec:consensus}), it follows that, at time $k$, if the consensus matrix is primitive and doubly stochastic, the consensus iterate of each node in the network converges to the global unweighted KLA (\ref{eq:uwgeomean}) of the multi-object posterior densities as the number of consensus iterations $l$ tends to infinity.
Convergence follows along the same line as in \cite{Calafiore,cp} since $\mathcal{F}\!\left( \mathbb{X} \times \mathbb{L} \right)$ is a metric space \cite{mahler}.
In practice, the iteration is stopped at some finite $l$.
\begin{rem}
	The consensus iteration (\ref{eq:consenus:alg}) is the multi-object counterpart of equation (\ref{eq:consensuspdf}) reviewed in Subsection \ref{ssec:consensus}. 
\end{rem}

For M$\delta$-GLMB multi-object densities, (\ref{eq:consenus:alg}) can be computed via (\ref{eq:mdglmbfusion:pdf}) and (\ref{eq:mdglmbfusion:w}), while for LMB densities by means of (\ref{eq:lmbfusion:ex}) and (\ref{eq:lmbfusion:pdf}).
In the present work, each single-object density $p_{i}$ is represented by a \textit{Gaussian Mixture} (GM) of the form
\begin{equation}
	p(x)=\sum_{j=1}^{N_{G}}\alpha _{j}\,\mathcal{N}\!\left(x;x_{j},P_{j}\right) .  \label{eq:gmpdf}
\end{equation}
Note that the fusion rules (\ref{eq:mdglmbfusion:pdf}) and (\ref{eq:lmbfusion:pdf}) involve exponentiation and multiplication of GM in the form of (\ref{eq:gmpdf}) which, in general, does not provide a GM.
Hence, the same solution conceived in subsection \ref{ssec:ccphd:gmfusion} for fusing GMs is also adopted here.

The other common approach for representing a single object location PDF $p$ is using \textit{particles}. Information fusion involving convex combinations of Dirac delta functions requires, at the best of our knowledge, the exploitation of additional techniques like kernel density estimation \cite{emd}, least square estimation \cite{clike1,clike2} or parametric model approaches \cite{coates} which increase the in-node computational burden. Moreover, the local filtering steps are also more resource demanding with respect to a GM implementation. At this stage, it is preferred to follow the approach devised in \cite{ccphd} by adopting a GM representation.

\section{Consensus labeled RFS information fusion}
\label{sec:clmbif}
In this section, two novel fully distributed and scalable multi-object tracking algorithms are described by exploiting Propositions \ref{pro:mdglmb:fusion} and \ref{pro:lmb:fusion} along with consensus \cite{Olfati,Xiao,Calafiore,cp} to propagate information throughout the network \cite{fanvovo2015}.
Pseudo-codes of the algorithms are also provided.

\subsection{Consensus M$\delta$-GLMB filter}
\label{ssec:cmdglmbf}
This subsection describes the novel \textit{Gaussian Mixture - Consensus Marginalized $\delta$-Generalized Labelled Multi-Bernoulli} (GM-CM$\delta$GLMB) filter algorithm.
The operations reported in Table \ref{alg:cmdglmb} are sequentially carried out locally by each agent $i\in \mathcal{N}$ of the network. Each node operates autonomously at each sampling interval $k$, starting from its own previous estimates of the multi-object distribution $\boldsymbol{\pi}^{i}(\cdot)$, having location PDFs $p^{(I)}\!\left( x, \ell \right)$, $\forall \,\ell \in I$, $I \in \mathcal{F}\!\left( \mathbb{L} \right)$, represented
with a GM, and producing, at the end of the sequence of operations, its new multi-object distribution $\boldsymbol{\pi }^{i}(\cdot)=\hat{\boldsymbol{\pi }}_{N}^{i}(\cdot)$ as an outcome of the Consensus procedure.
The steps of the GM-C$\delta$GLMB algorithm follow hereafter.
\begin{enumerate}
\item Each agent $i \in \mathcal{N}$ locally performs a GM-$\delta$GLMB prediction (\ref{eq:mdglmbpredictedpdf}) and update (\ref{eq:mdglmbupdatedpdf}). Additional details of the two procedures can be found in \cite[Section IV.B]{vovo2}.
\item Consensus takes place in each node $i$ involving the subnetwork $\mathcal{N}^{i}$.
More precisely, at each consensus step, node $i$ transmits its data to nodes $j$ such that $i\in \mathcal{N}^{j}$ and waits until it receives data from $j\in \mathcal{N}^{i}\backslash \{i\}$.
Next, node $i$ carries out the fusion rule of Proposition \ref{pro:mdglmb:fusion} over $\mathcal{N}^{i}$, i.e. performs (\ref{eq:consenus:alg}) using local information and information received by $\mathcal{N}^{i}$.
Finally, a merging step for each location PDF is applied to reduce the joint communication-computation burden for the next consensus step.
This procedure is repeatedly applied for a chosen number $N\geq 1$ of consensus steps.
\item After consensus, an estimate of the object set is obtained from the cardinality PMF and the location PDFs via an estimate extraction described in Table \ref{alg:mdglmb:estextr}.
\end{enumerate}

\begin{table}[h!]
\caption{Gaussian Mixture - Consensus Marginalized $\delta$-Generalized Labeled Multi-Bernoulli (GM-CM$\delta$GLMB) filter}
\label{alg:cmdglmb}\renewcommand{\arraystretch}{1.3}
\hrulefill\hrule
\begin{algorithmic}[0]
	\Procedure{GM-CM$\delta$GLMB}{\textsc{Node} $i$, \textsc{Time} $k$}
		\State \textsc{Local Prediction} \Comment{See (\ref{eq:mdglmbpredictedpdf}) and \cite[Table 2, Section V]{vovo2}}\vspace{0.5em}
		\State \textsc{Local Update} \Comment{See (\ref{eq:mdglmbupdatedpdf}) and \cite[Table 1, Section IV]{vovo2}}\vspace{0.5em}
		\State \textsc{Marginalization} \Comment{See (\ref{eq:mdglmb:wnxt}) and (\ref{eq:mdglmb:pnxt})}\vspace{0.5em}
		\For{$n = 1, \dots, N$}
			\State \textsc{Information Exchange}
			\State GM-M$\delta$DGLMB \textsc{Fusion} \Comment{See eqs. (\ref{eq:mdglmbfusion:pdf}) and (\ref{eq:mdglmbfusion:w})}
			\State GM \textsc{Merging} \Comment{See \cite[Table II, Section III.C]{vo-ma}}
		\EndFor\vspace{0.5em}
		\State \textsc{Estimate Extraction} \Comment{See algorithm in Table \ref{alg:mdglmb:estextr}}
	\EndProcedure
\end{algorithmic}
\hrule\hrulefill
\end{table}

\begin{table}[h!]
\caption{GM-CM$\delta$GLMB estimate extraction}
\label{alg:mdglmb:estextr}\renewcommand{\arraystretch}{1.3}
\hrulefill\hrule
	\begin{algorithmic}[0]
		\State \textbf{\textsc{Input:}} $\boldsymbol{\pi}$, $N_{max}$
		\State \textbf{\textsc{Output:}} $\hat{\mathbf{X}}$
	\end{algorithmic}
	\hrule\vspace{1mm} 
	\begin{algorithmic}[0]
		\For{$c = 1, \dots, C_{max}$}
			\State $\displaystyle{\rho(c) = \sum_{I \in \mathcal{F}_{c}(\mathbb{L})} w^{(I)}}$
		\EndFor
		\State $\displaystyle{\hat{C} = \arg \max_{c} \rho(c)}$
		\State $\displaystyle \hat{I} = \arg \max_{I \in \mathcal{F}_{\hat{C}}\!\left( \mathbb{L} \right)} w^{(I)}$\vspace{0.5em}
		\State $\displaystyle{\hat{\mathbf{X}} = \left\{ \left( \hat{x}, \hat{\ell}\right)\!: \hat{\ell} \in \hat{I},\, \hat{x} = \arg \max_{x}p^{(\hat{I})}(x, \hat{\ell}) \right\}}$
\end{algorithmic}
\hrule\hrulefill
\end{table}

\subsection{Consensus LMB filter}
\label{ssec:clmbf}
This subsection describes the novel \textit{Gaussian Mixture - Consensus Labelled Multi-Bernoulli} (GM-CLMB) filter algorithm. 
The operations reported in Table \ref{alg:clmb} will be sequentially carried out locally by each agent $i\in \mathcal{N}$ of the network.
Each node operates autonomously at each sampling interval $k$, starting from its own previous estimates of the multi-object distribution $\boldsymbol{\pi }^{i}(\cdot)$, having location PDFs $p^{(\ell )}$, $\forall \,\ell \in \mathbb{L}$, represented with a GM, and producing, at the end of the sequence of operations, its new multi-object distribution $\boldsymbol{\pi }^{i}(\cdot) = \hat{\boldsymbol{\pi }}_{N}^{i}(\cdot)$ as an outcome of the consensus procedure.
A summary description of the steps of the GM-CLMB algorithm follows.

\begin{enumerate}
\item Each agent $i \in \mathcal{N}$ locally performs a GM-LMB prediction (\ref{eq:lmbpredictedpdf}) and update (\ref{eq:lmbupdatedpdf}). The update procedure involves two more steps: $\left. \textsc{i} \right)$ a GM $\delta$-GLMB distribution is created from the predicted GM-LMB and updated using the local measurement set $Y_{k}$; $\left. \textsc{ii} \right)$ the updated GM $\delta$-GLMB is converted back to a GM-LMB
distribution. The details of the GM-LMB prediction and update can be found in \cite{lmbf,vovo2}.

\item Consensus takes place in each node $i$ involving the subnetwork $\mathcal{N}^{i}$. More precisely, at each consensus step, node $i$ transmits its data to nodes $j$ such that $i\in \mathcal{N}^{j}$ and waits until it receives data from $j\in \mathcal{N}^{i}\backslash \{i\}$. Next, node $i$ carries out the fusion rule of Proposition \ref{pro:lmb:fusion} over $\mathcal{N}^{i}$, i.e. performs (\ref{eq:consenus:alg}) using local information and information received by $\mathcal{N}^{i}$. Finally, a merging step for each location PDF is applied to reduce the joint communication-computation burden for the next consensus step. This procedure is repeatedly applied for a chosen number $N\geq 1$ of
consensus steps.

\item After consensus, an estimate of the object set is obtained from the cardinality PMF and the location PDFs via an estimate extraction described in Table \ref{alg:clmb:estextr}.
\end{enumerate}

\begin{table}[h!]
\caption{Gaussian Mixture - Consensus Labeled Multi-Bernoulli (GM-CLMB) filter}
\label{alg:clmb}\renewcommand{\arraystretch}{1.3}
\hrulefill\hrule
\begin{algorithmic}[0]
	\Procedure{GM-CLMB}{\textsc{Node} $i$, \textsc{Time} $k$}
		\State \textsc{Local Prediction} \Comment{See eq. (\ref{eq:lmbpredictedpdf}) and \cite[Proposition 2, Section III.A]{lmbf}}\vspace{0.5em}
		\State GM-LMB $\rightarrow$ GM $\delta$-GLMB \Comment{See \cite[Section IV.C.1]{lmbf}}
		\State \textsc{Local Update} \Comment{See eq. (\ref{eq:lmbupdatedpdf}) and \cite[Table 1, Section IV]{vovo2}}
		\State GM $\delta$-GLMB $\rightarrow$ GM-LMB \Comment{See \cite[Proposition 4, Section III.B]{lmbf}}\vspace{0.5em}
		\For{$n = 1, \dots, N$}
			\State \textsc{Information Exchange}
			\State GM-LMB \textsc{Fusion} \Comment{See eqs. (\ref{eq:lmbfusion:ex}) and (\ref{eq:lmbfusion:pdf})}
			\State GM \textsc{Merging} \Comment{See \cite[Table II, Section III.C]{vo-ma}}
		\EndFor\vspace{0.5em}
		\State \textsc{Estimate Extraction} \Comment{See algorithm in Table \ref{alg:clmb:estextr}}
	\EndProcedure
\end{algorithmic}
\hrule\hrulefill
\end{table}

\begin{table}[h!]
\caption{GM-CLMB estimate extraction}
\label{alg:clmb:estextr}\renewcommand{\arraystretch}{1.3}
\hrulefill\hrule
\begin{algorithmic}[0]
	\State \textbf{\textsc{Input:}} $\boldsymbol{\pi} = \left\{ r^{(\ell)},p^{(\ell)} \right\}_{\ell \in \mathbb{L}}$, $N_{max}$
	\State \textbf{\textsc{Output:}} $\hat{\mathbf{X}}$
\end{algorithmic}
\hrule
\begin{algorithmic}[0]
\For{$c = 1, \dots, C_{max}$}
	\State $\displaystyle{\rho(c) = \sum_{L \in \mathcal{F}_{c}(\mathbb{L})} w(L)}$
\EndFor
\State $\displaystyle{\hat{C} = \arg \max_{c} \rho(c)}$
\State $\hat{\mathbb{L}} = \varnothing$
\For{$\hat{c} = 1, \dots, \hat{C}$}
	\State $\displaystyle{\hat{\mathbb{L}} = \hat{\mathbb{L}} \cup \arg \max_{\ell \in \mathbb{L}\backslash\hat{\mathbb{L}}} r^{(\ell)}}$
\EndFor\vspace{0.5em}
\State $\displaystyle{\hat{\mathbf{X}} = \left\{ \left( \hat{x}, \hat{\ell}\right)\!: \hat{\ell} \in \hat{\mathbb{L}},\, \hat{x} = \arg \max_{x}p^{(\hat{\ell})}(x) \right\}}$
\end{algorithmic}
\hrule\hrulefill
\end{table}

\section{Performance evaluation}
\label{sec:performance}
To assess performance of the proposed Gaussian Mixture Consensus Marginalized $\delta$-GLMB (GM-CM$\delta$GLMB) and LMB (GM-CLMB) described in section \ref{sec:clmbif}, a $2$-dimensional multi-object tracking scenario is considered over a surveillance area of $50\times50 \, [km^{2}]$, wherein the sensor network of Fig. \ref{fig:mot:4toa3doa} is deployed. The scenario consists of $5$ objects as depicted in Fig. \ref{fig:5trajectories}.
For the sake of comparison, the trackers are also compared with the Gaussian Mixture Consensus CPHD (GM-CCPHD) filter of \cite{ccphd} which, however, does not provide labeled tracks.
\begin{figure}[h!]
        \begin{minipage}[t][][t]{\columnwidth}
	        \centering
		\includegraphics[width=0.6\columnwidth]{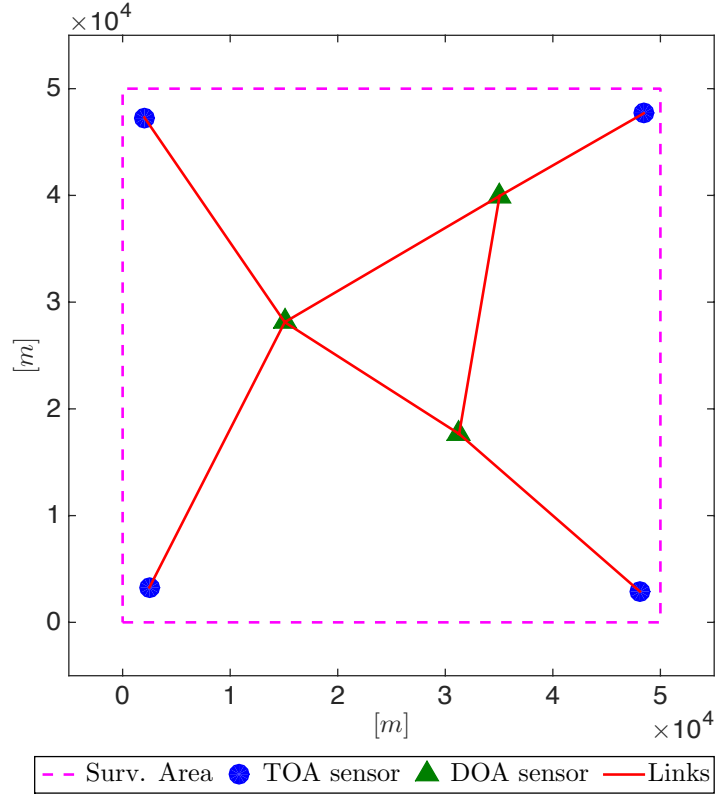}
		\caption{Network with 7 sensors: 4 TOA and 3 DOA.}
		\label{fig:mot:4toa3doa}
        \end{minipage}
\end{figure}
\begin{figure}[h!]
        \begin{minipage}[t][][t]{\columnwidth}
 	       	\centering
		\includegraphics[width=0.6\columnwidth]{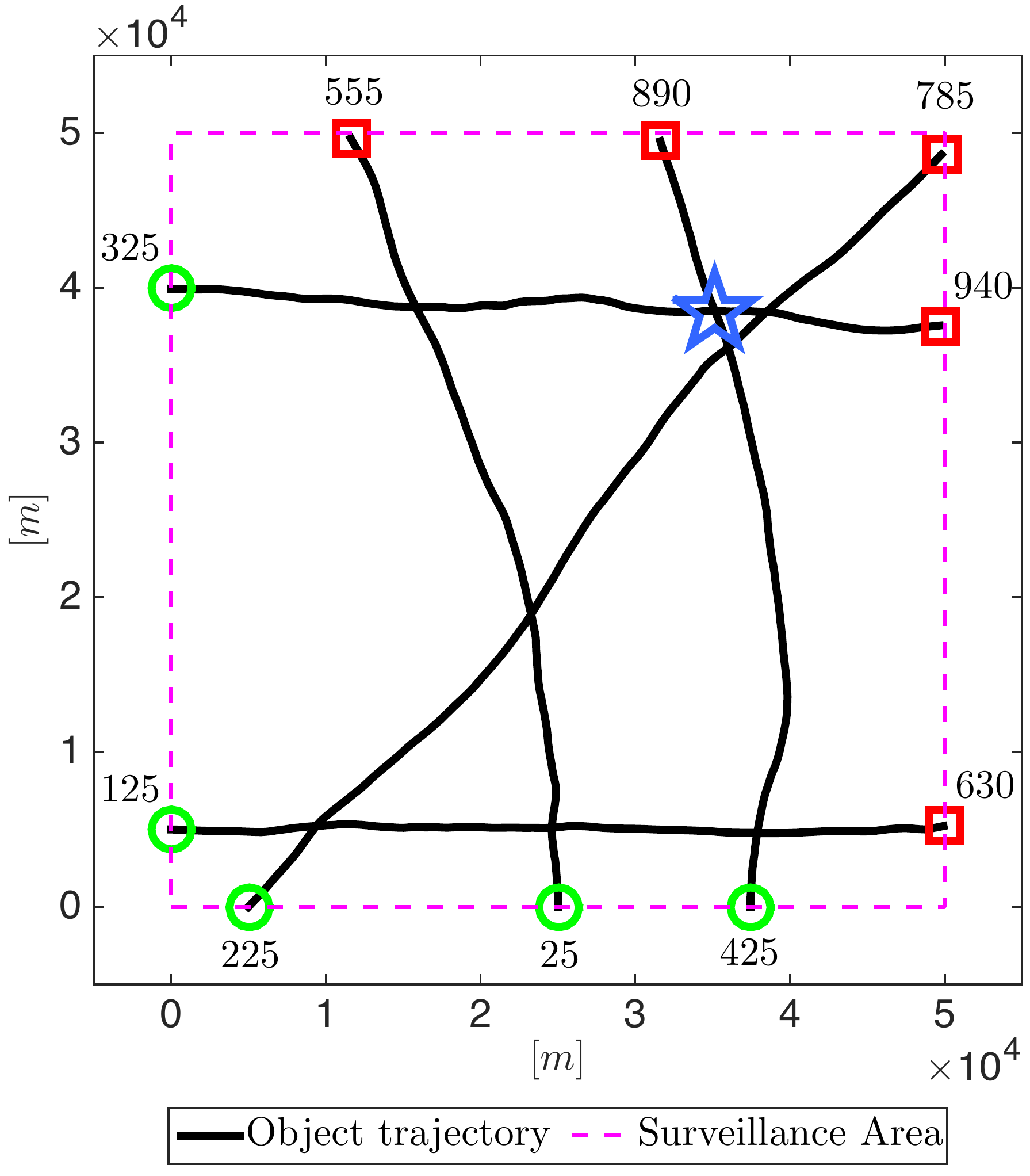}
		\caption{Tbject trajectories considered in the simulation experiment. The start/end point for each trajectory is denoted, respectively, by $\bullet\backslash\blacksquare$. The {\Large$\star$} indicates a rendezvous point.}
		\label{fig:5trajectories}
        \end{minipage}
\end{figure}

The kinematic object state is denoted by $x = \left[ p_{x}, \, \dot{p}_{x}, \, p_{y}, \, \dot{p}_{y} \right]^{\top}$, i.e. the planar position and velocity.
The motion of objects is modeled by the filters according to the Nearly-Constant Velocity (NCV) model \cite{far1985v1,far1985v2,book0,book}:
\be
x_{k + 1} = \left[ \begin{array}{cccc}
1 & T_{s} & 0 & 0	\\
0 & 1 	  & 0 & 0		\\
0 & 0 	  & 1 & T_{s} \\
0 & 0 	  & 0 & 1		\end{array} \right] x_{k} + w_{k} \, , \qquad
Q = \sigma_{w}^{2} \left[ \begin{array}{cccc}
\frac{1}{4}T_{s}^{4} & \frac{1}{2}T_{s}^{3} & 0 & 0 \\
\frac{1}{2}T_{s}^{3} & T_{s}^{2} & 0 & 0 \\
0 & 0 & \frac{1}{4}T_{s}^{4} & \frac{1}{2}T_{s}^{3}\\
0 & 0 & \frac{1}{2}T_{s}^{3} & T_{s}^{2} \end{array} \right] \, ,
\ee
where: $w_{k}$ is a white noise with zero mean and covariance matrix $Q$, $\sigma_{w} = 5 \, [m/s^{2}]$ and the sampling interval is $T_{s} = 5\,[s]$.

The sensor network considered in the simulation (see Fig. \ref{fig:mot:4toa3doa}) consists of $4$ \textit{range-only} (Time Of Arrival, TOA) and $3$ \textit{bearing-only} (Direction Of Arrival, DOA) sensors characterized by the following measurement functions:
\be
\begin{array}{c}
h^{i}(x) = \left\{ \begin{array}{ll}
				\angle [ \left( p_{x} - x^i \right) + j \left( p_{y} - y^i \right)], & \mbox{if $i$ is a DOA sensor} \\[0.5em]
                                	\sqrt{ \left( p_{x} - x^i \right)^2+ \left( p_{y} - y^i \right)^2 }, & \mbox{if $i$ is a TOA sensor}
			\end{array} \right.
\end{array}
\ee
where $( x^i, y^i )$ represents the known position of sensor $i$. The standard deviation of DOA and TOA measurement noises are taken respectively as $\sigma_{DOA} = 1 \, [\mbox{}^{\circ}]$ and $\sigma_{TOA} = 100 \, [m]$. Because of the non linearity of the aforementioned sensors, the \textit{Unscented Kalman Filter} (UKF) \cite{juluhl2004} is exploited in each sensor in order to update means and covariances of the Gaussian components.

Three different scenarios will be considered, each of which has a different clutter Poisson process with parameter $\lambda_{c}$ and a probability of object detection $P_{D}$.
\begin{itemize}
	\item \textbf{High SNR}: $\lambda_{c} = 5$, $P_{D} = 0.99$. These parameters were used in the work \cite{ccphd} and, therefore, will be used as a first comparison test.
	\item \textbf{Low SNR}: $\lambda_{c} = 15$, $P_{D} = 0.99$. These parameters try to describe, in a realistic way, a scenario characterized by high clutter rate $\lambda_{c}$.
	\item \textbf{Low $\mathbf{P_{D}}$}: $\lambda_{c} = 5$, $P_{D} = 0.7$. These parameters test the distributed algorithms in the presence of severe misdetection.
\end{itemize}
All the above-mentioned case studies have, for each sensor, a uniform clutter spatial distribution over the surveillance area.

In the considered scenario, objects pass through the surveillance area with partial prior information for object birth locations. 
Accordingly, a $10$-component LMB RFS $\boldsymbol{\pi}_{B} = \left\{ \left( r^{\left( \ell \right)}_{B}, p^{\left( \ell \right)}_{B} \right) \right\}_{\ell \in \mathbb{B}}$ has been hypothesized for the birth process.
Table \ref{tab:borderlineinit} gives a detailed summary of such components.
\begin{table}[h!]
	\renewcommand{\arraystretch}{1.3}
	\setlength\arrayrulewidth{0.5pt}\arrayrulecolor{black} 
	\setlength\doublerulesep{0.5pt}\doublerulesepcolor{black} 
	\caption{Components of the LMB RFS birth process at a given time $k$.}
	\label{tab:borderlineinit}
	\centering
	$r^{\left( \ell \right)} = 0.09$\\
	$p^{\left( \ell \right)}_{B}(x) = \mathcal{N}\!\left( x;\, m^{\left( \ell \right)}_{B}, P_{B} \right)$\\
	$P_{B} = \operatorname{diag}\!\left( 10^{6}, 10^{4}, 10^{6}, 10^{4} \right)$\\\vspace{0.5em}
	\scalebox{1}{
	\begin{tabular}{>{\columncolor[gray]{.95}}c||c|c|c||}
		{\textbf{Label}} & $\left( k, \, 1 \right)$ & $\left( k, \, 2 \right)$ & $\left( k, \, 3 \right)$ \\
		\hline
		$m^{\left( \ell \right)}_{B}$ & $\left[ 0, \, 0, \, 40000,\, 0 \right]^{\top}$ & $\left[ 0, \, 0, \, 25000,\, 0 \right]^{\top}$ & $\left[ 0, \, 0, \, 5000,\, 0 \right]^{\top}$\\
		\hline
		\hline
	\end{tabular}
	}\vspace{0.5em}
	\scalebox{1}{
	\begin{tabular}{>{\columncolor[gray]{.95}}c||c|c|c||}
		{\textbf{Label}} & $\left( k, \, 4 \right)$ & $\left( k, \, 5 \right)$ & $\left( k, \, 6 \right)$\\
		\hline
		$m^{\left( \ell \right)}_{B}$ & $\left[ 5000, \, 0, \, 0,\, 0 \right]^{\top}$ & $\left[ 25000, \, 0, \, 0,\, 0 \right]^{\top}$ & $\left[ 36000, \, 0, \, 0,\, 0 \right]^{\top}$\\
		\hline
		\hline
	\end{tabular}
	}\vspace{0.5em}
	\scalebox{1}{
	\begin{tabular}{>{\columncolor[gray]{.95}}c||c|c||}
		{\textbf{Label}} & $\left( k, \, 7 \right)$ & $\left( k, \, 8 \right)$\\
		\hline
		$m^{\left( \ell \right)}_{B}$ & $\left[ 50000, \, 0, \, 15000,\, 0 \right]^{\top}$ & $\left[ 50000, \, 0, \, 40000,\, 0 \right]^{\top}$\\
		\hline
		\hline
	\end{tabular}
	}\vspace{0.5em}
	\scalebox{1}{
	\begin{tabular}{>{\columncolor[gray]{.95}}c||c|c||}
		{\textbf{Label}} & $\left( k, \, 9 \right)$ & $\left( k, \, 10 \right)$\\
		\hline
		$m^{\left( \ell \right)}_{B}$ & $\left[ 40000, \, 0, \, 50000,\, 0 \right]^{\top}$ & $\left[ 10000, \, 0, \, 50000,\, 0 \right]^{\top}$\\
		\hline
		\hline
	\end{tabular}
	}
\end{table}
Due to the partial prior information on the object birth locations, some of the LMB components cover a state space region where there is no birth. Therefore, clutter measurements are more prone to generate false objects.

Multi-object tracking performance is evaluated in terms of the \textit{Optimal SubPattern Assignment} (OSPA) metric \cite{schvovo2008} with Euclidean distance $p = 2$ and cutoff $c = 600 \, [m]$.
The reported metric is averaged over $100$ Monte Carlo trials for the same object trajectories but different, independently generated, clutter and measurement noise realizations. The duration of each simulation trial is fixed to $1000 \, [s]$ ($200$ samples).

The GM-CM$\delta$GLMB and the GM-CLMB are capped to $20000$, $8000$ and $3000$ hypotheses \cite{vovo1,vovo2}, respectively, for the High SNR, Low SNR and Low $P_{D}$ scenario, and are coupled with the \textit{parallel CPHD look ahead strategy} described in \cite{vovo1,vovo2}. The CPHD filter is similarly capped, for each case study, to the same number of components through pruning and merging of mixture components.

The parameters of the GM-CCPHD filter have been chosen as follows: the survival probability is $P_{s} = 0.99$; the maximum number of Gaussian components is $N_{max} = 25$; the merging threshold is $\gamma_{m} = 4$; the truncation threshold is $\gamma_{t} = 10^{-4}$; the extraction threshold is $\gamma_{e} = 0.5$; the birth intensity function is the PHD of the LMB RFS of Table \ref{tab:borderlineinit}.
A single consensus step $L = 1$ is employed for all the simulations.

\subsection{High SNR}
Figs. \ref{fig:1:cardCCPHD}, \ref{fig:1:cardCLMB} and \ref{fig:1:cardCMDGLMB} display the statistics (mean and standard deviation) of the estimated number of objects obtained, respectively, with GM-CCPHD, GM-CLMB and GM-CM$\delta$GLMB.
As it can be seen, all the distributed algorithms estimate the object cardinality accurately, with the GM-CM$\delta$GLMB exhibiting better estimated cardinality variance.
Note that the difficulties introduced by the rendezvous point (e.g. merged or lost tracks) are correctly tackled by all the distributed algorithms.

Fig. \ref{fig:1:ospa} shows the OSPA distance for the three algorithms.
The improved localization performance of GM-CLMB and GM-CM$\delta$GLMB is attributed to two factors: (a) the ``spooky effect'' \cite{spooky} causes GM-CCPHD filter to temporarily drop objects which are subjected to missed detections and to declare multiple estimates for existing tracks in place of the dropped objects, and (b) the two trackers are generally able to better localize objects due to a more accurate propagation of the posterior density.
Note that GM-CLMB and GM-CM$\delta$GLMB have similar performance since the approximations introduced by the LMB tracker (see (\ref{eq:existenceBernoulli})-(\ref{eq:spatialBernoulli})) are negligible in the case of high SNR.
\begin{figure}[h!]
        \begin{minipage}[t][][t]{\columnwidth}
        		\centering
		\includegraphics[width=\columnwidth]{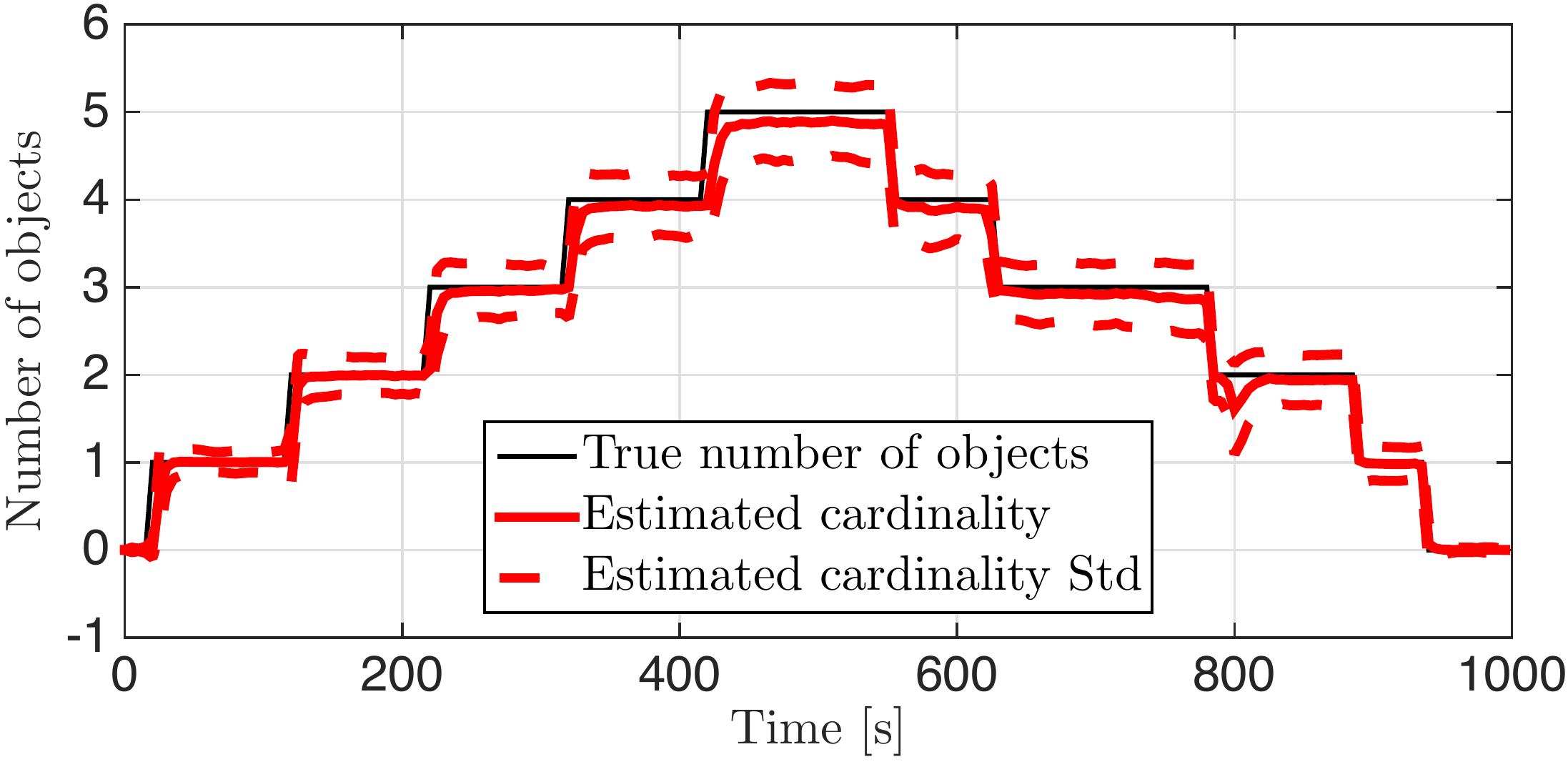}
		\caption{Cardinality statistics for the GM-CCPHD filter under high SNR.}
		\label{fig:1:cardCCPHD}
        \end{minipage}\vspace{0.5em}
        \begin{minipage}[t][][t]{\columnwidth}
	        \centering
		\includegraphics[width=\columnwidth]{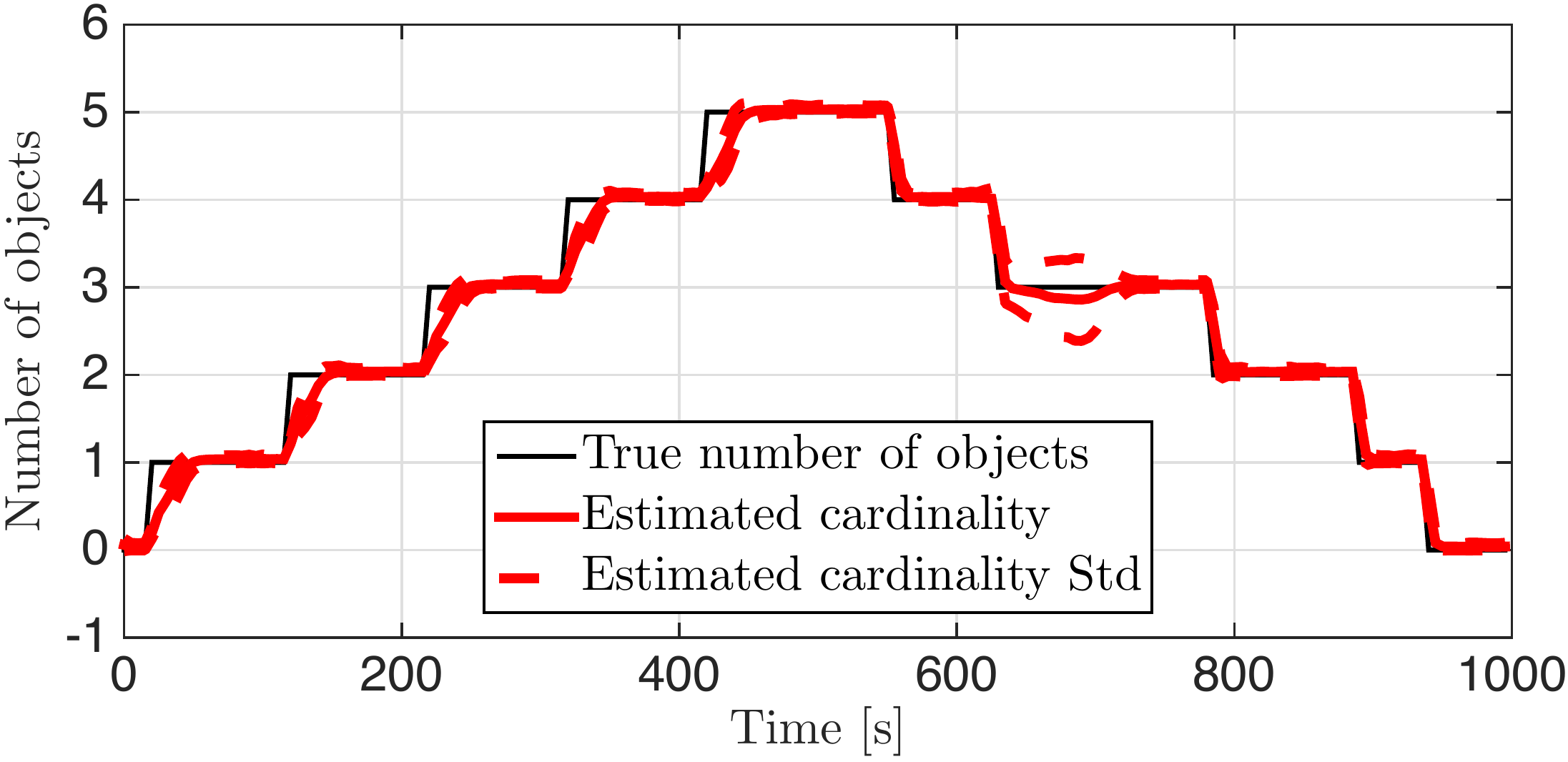}
		\caption{Cardinality statistics for the GM-CLMB tracker under high SNR.}
		\label{fig:1:cardCLMB}
        \end{minipage}
\end{figure}
\begin{figure}[h!]
        \begin{minipage}[t][][t]{\columnwidth}
	        \centering
		\includegraphics[width=\columnwidth]{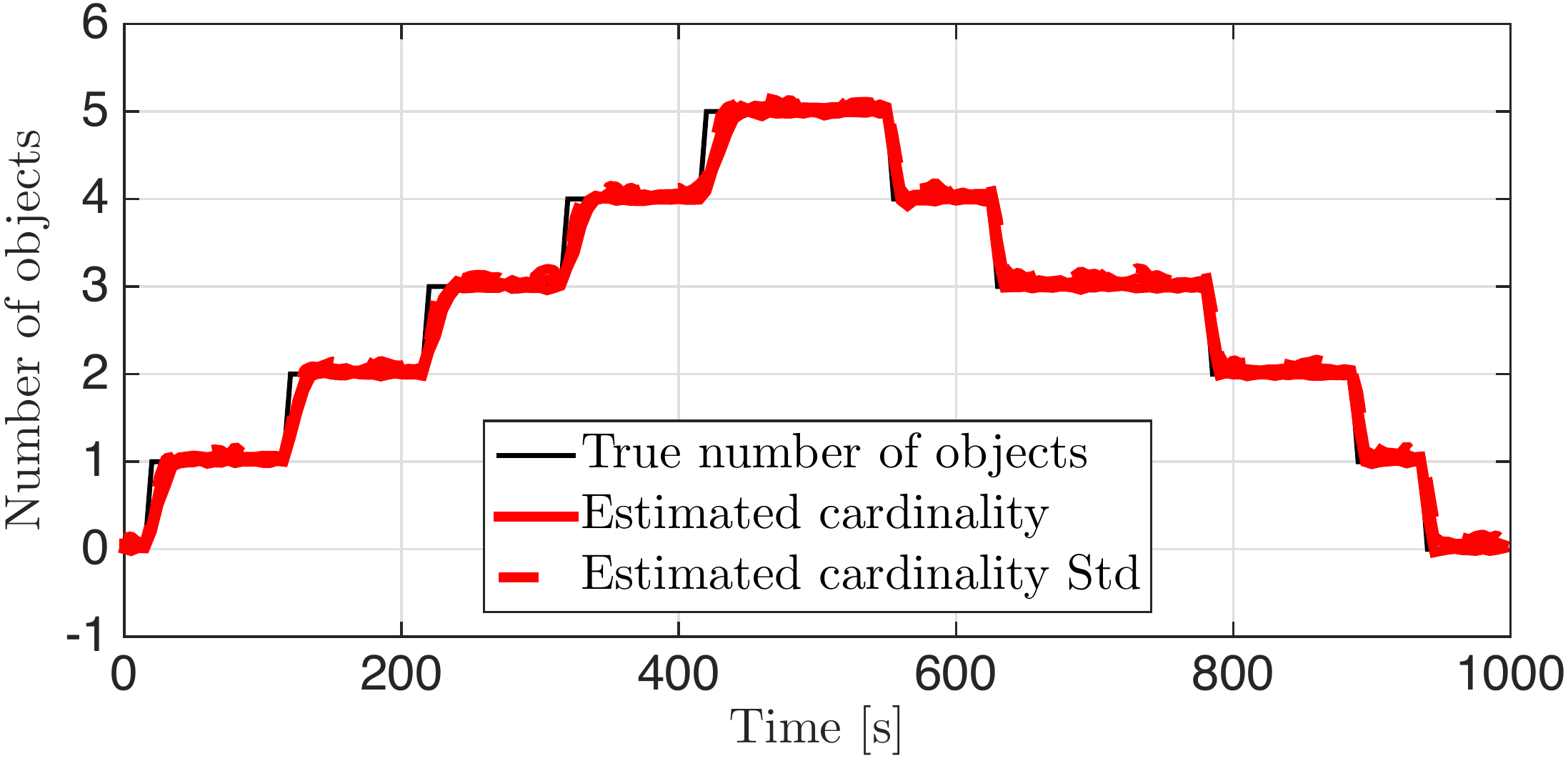}
		\caption{Cardinality statistics for the GM-CM$\delta$GLMB tracker under high SNR.}
		\label{fig:1:cardCMDGLMB}
        \end{minipage}\vspace{0.5em}
        \begin{minipage}[t][][t]{\columnwidth}
	        \centering
		\includegraphics[width=\columnwidth]{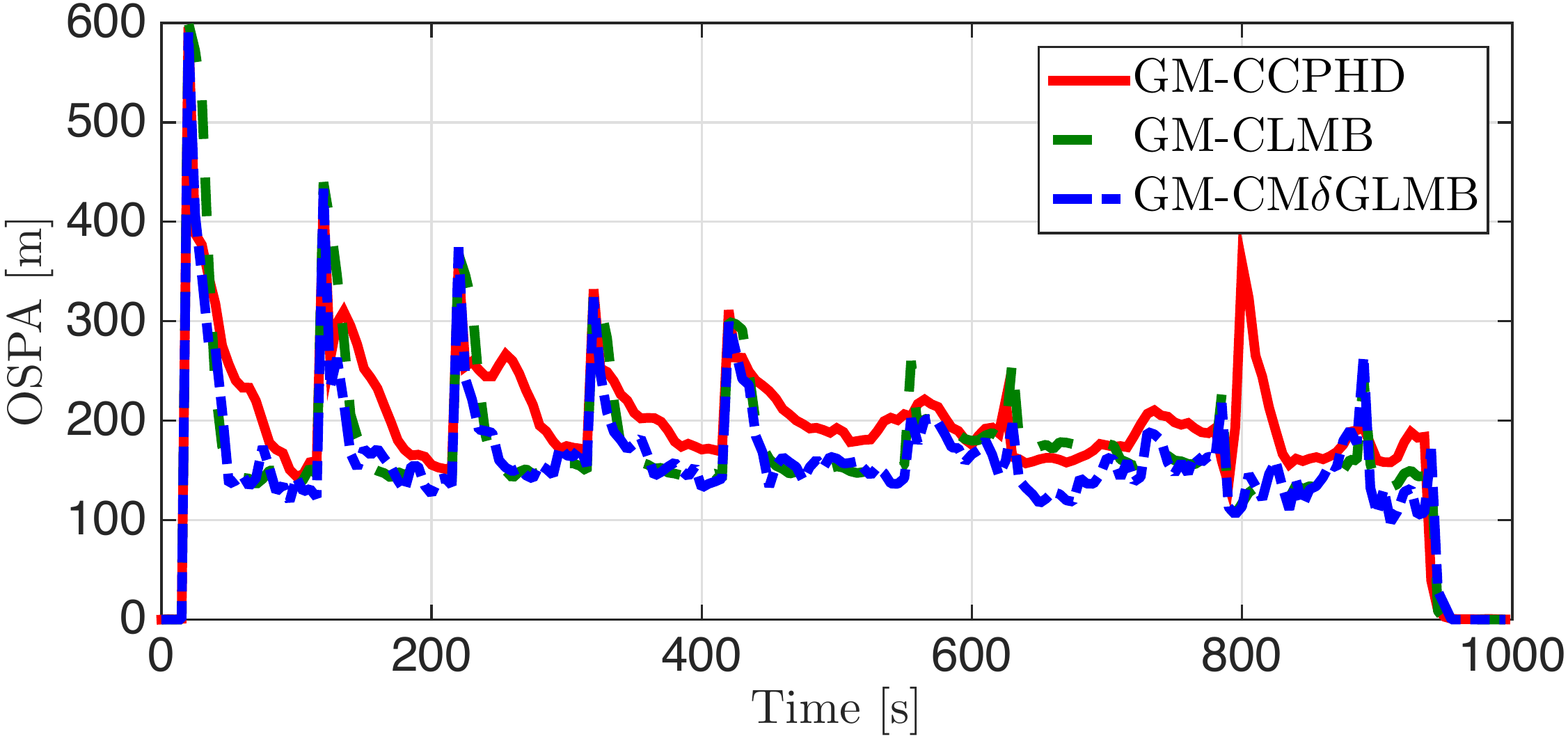}
		\caption{OSPA distance ($c = 600 \, [m]$, $p = 2$) under high SNR.}
		\label{fig:1:ospa}
        \end{minipage}
\end{figure}

\subsection{Low SNR}
Figs. \ref{fig:2:cardCCPHD} and \ref{fig:2:cardCMDGLMB} display the statistics (mean and standard deviation) of the estimated number of objects obtained, respectively, with the GM-CCPHD and the GM-CM$\delta$GLMB.
As it can be seen, the two distributed algorithms estimate the object cardinality accurately, with the GM-CM$\delta$GLMB exhibiting again better estimated cardinality variance.

Note that the GM-CLMB fails to track the objects. The problem is due to the approximation (\ref{eq:existenceBernoulli})-(\ref{eq:spatialBernoulli}) made to convert a $\delta$-GLMB to an LMB, becoming significant with low SNR.
In particular, each local tracker fails to properly set the existence probability of the tracks for three main factors: (a) no local observability, (b) high clutter rate and (c) loss of the full posterior cardinality distribution after the probability density conversion.
By having low existence probabilities, the extraction of the tracks fails even if the single object densities are correctly propagated in time.

Fig. \ref{fig:2:ospa} shows the OSPA distance for the current scenario.
As for the previous case study, GM-CM$\delta$GLMB outperforms GM-CCPHD.
\begin{figure}[h!]
        \begin{minipage}[t][][t]{\columnwidth}
        		\centering
		\includegraphics[width=0.9\columnwidth]{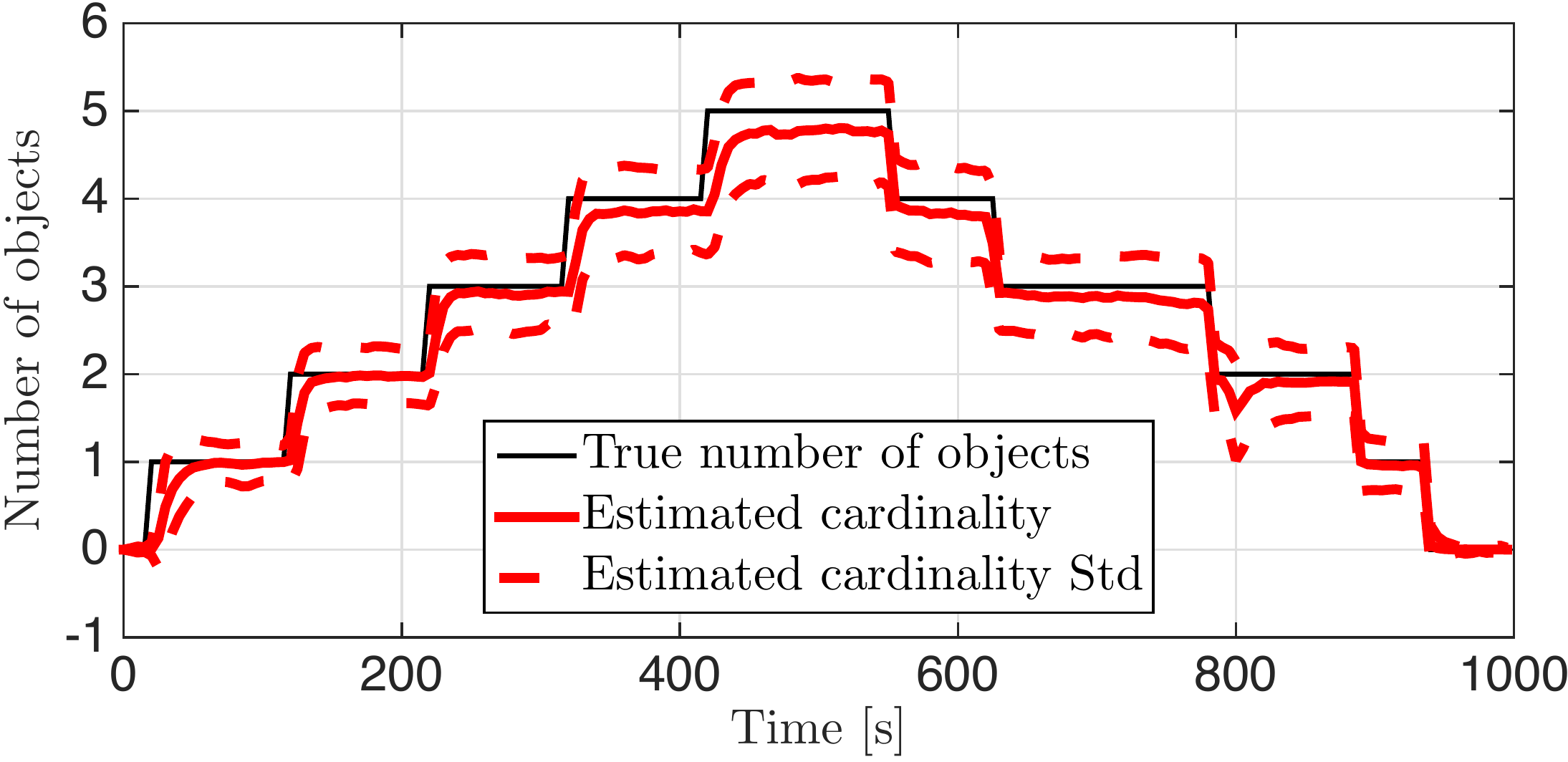}
		\caption{Cardinality statistics for the GM-CCPHD filter under low SNR.}
		\label{fig:2:cardCCPHD}
        \end{minipage}\vspace{0.5em}
        \begin{minipage}[t][][t]{\columnwidth}
        		\centering
		\includegraphics[width=0.9\columnwidth]{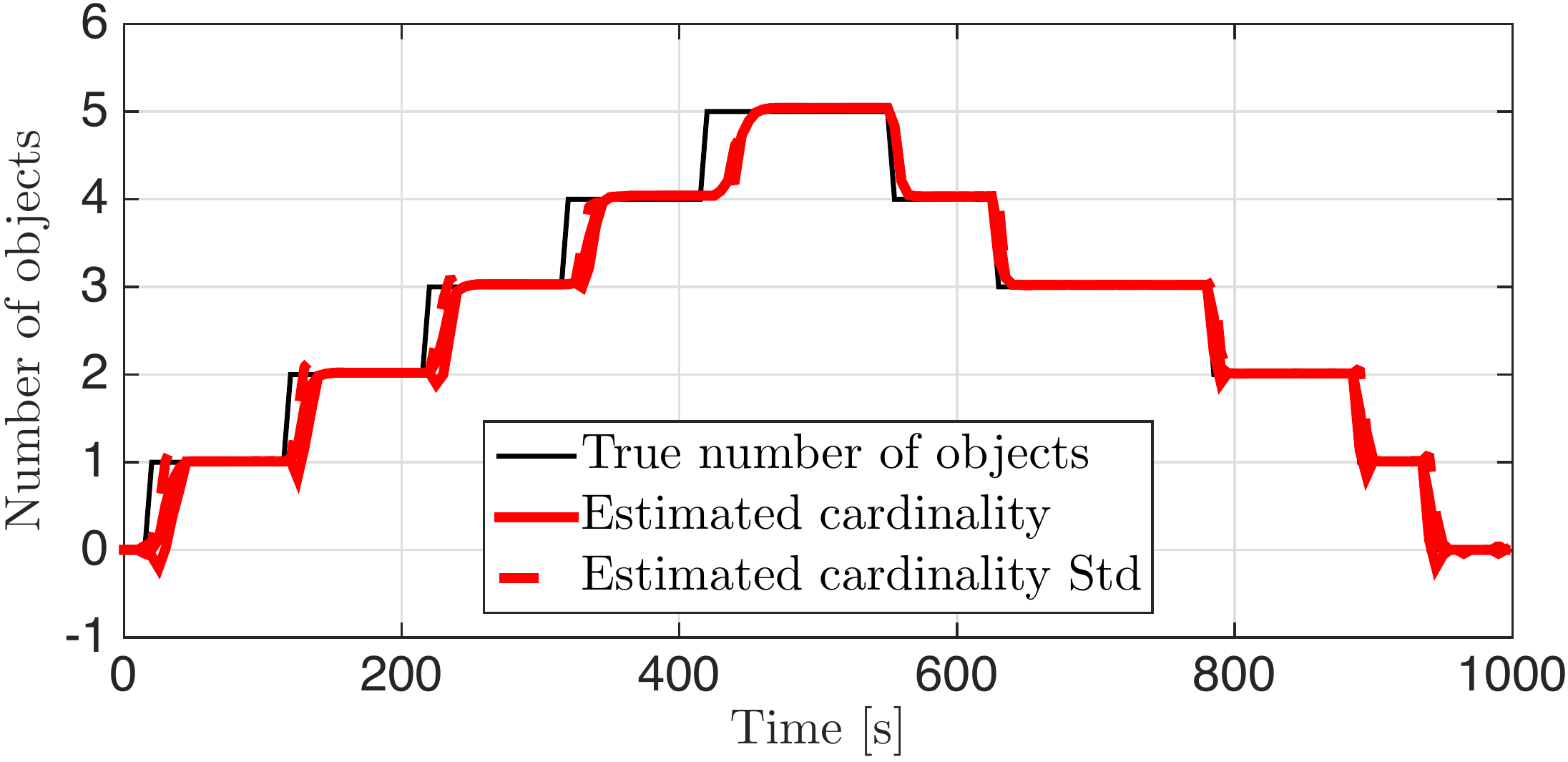}
		\caption{Cardinality statistics for the GM-CM$\delta$GLMB tracker under low SNR.}
		\label{fig:2:cardCMDGLMB}
        \end{minipage}\vspace{0.5em}
        \begin{minipage}[t][][t]{\columnwidth}
        		\centering
		\includegraphics[width=0.9\columnwidth]{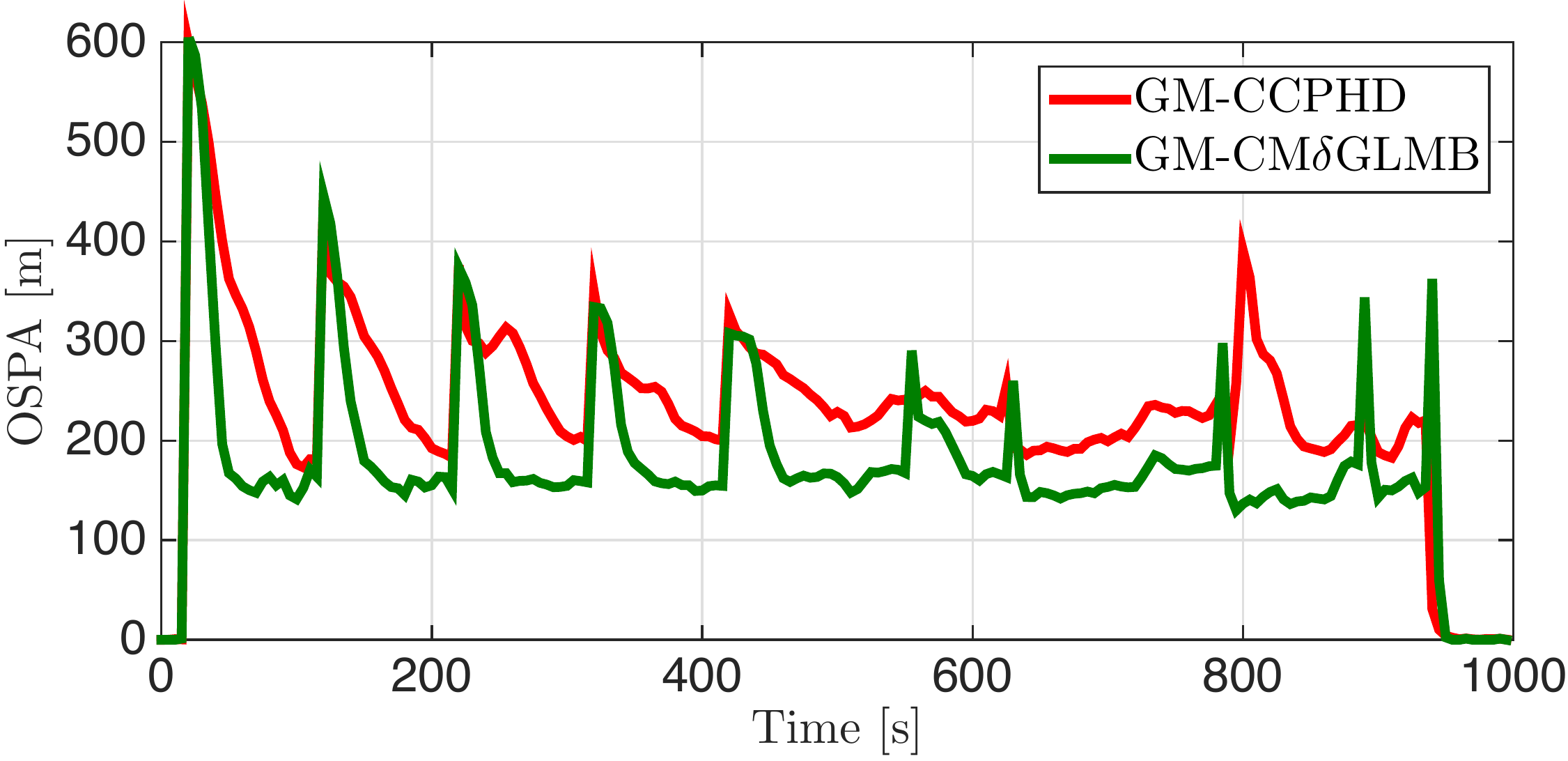}
		\caption{OSPA distance ($c = 600 \, [m]$, $p = 2$) under low SNR.}
		\label{fig:2:ospa}
        \end{minipage}
\end{figure}

\subsection{Low $P_{D}$}
Figs. \ref{fig:3:cardCMDGLMB} displays the statistics (mean and standard deviation) of the estimated number of objects obtained with the GM-CM$\delta$GLMB.
It is worth pointing out that the only working distributed algorithm is, indeed, the GM-CM$\delta$GLMB and that it exhibits good performance in terms of average number of estimated objects in a very tough scenario with $P_{D} = 0.7$.

Fig. \ref{fig:3:ospa} shows the OSPA distance for the current scenario.
\begin{figure}[h!]
        \begin{minipage}[t][][t]{\columnwidth}
        		\centering
		\includegraphics[width=\columnwidth]{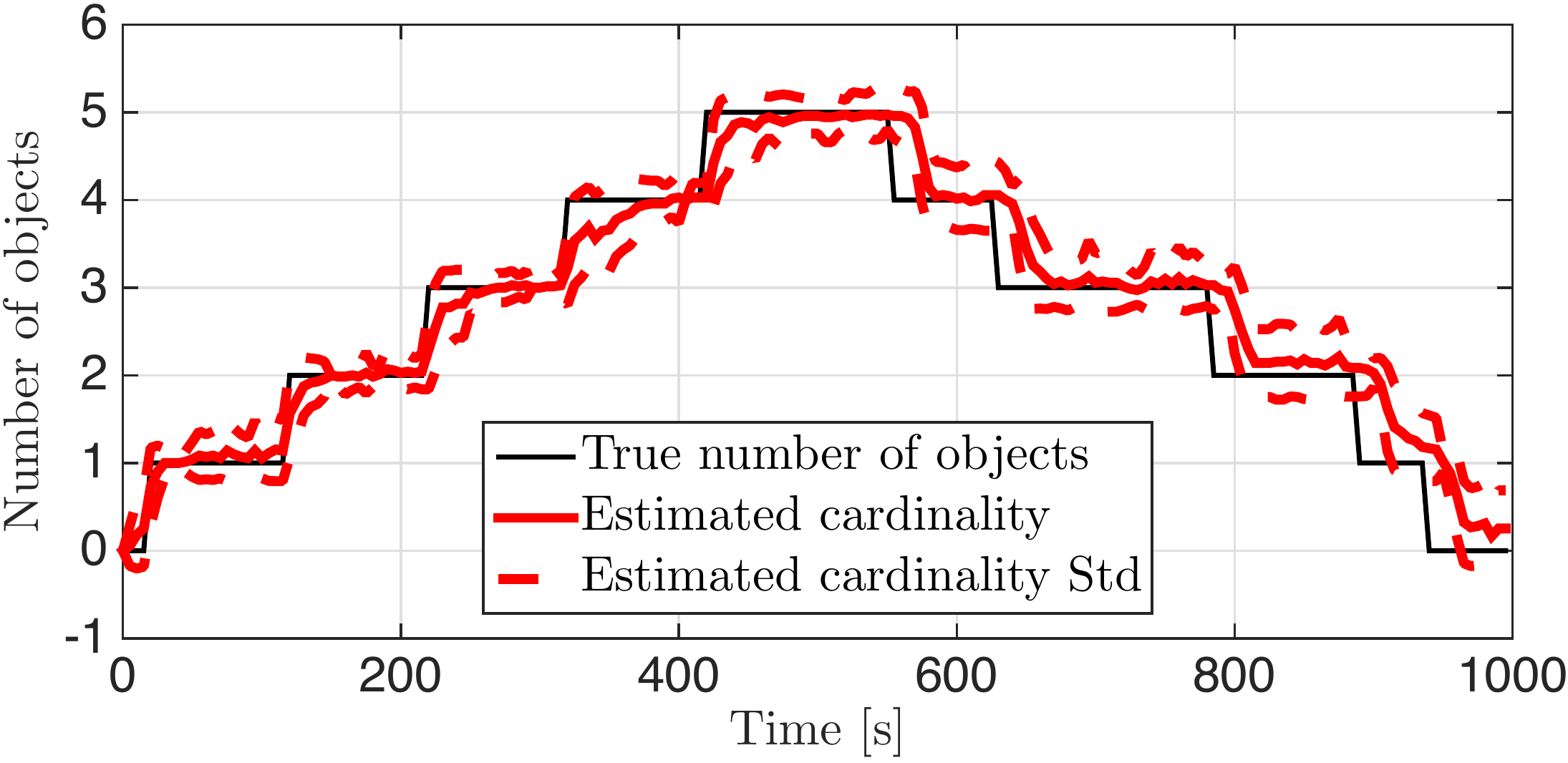}
		\caption{Cardinality statistics for the GM-CM$\delta$GLMB tracker under low $P_{D}$.}
		\label{fig:3:cardCMDGLMB}
        \end{minipage}\vspace{0.5em}
        \begin{minipage}[t][][t]{\columnwidth}
        		\centering
		\includegraphics[width=\columnwidth]{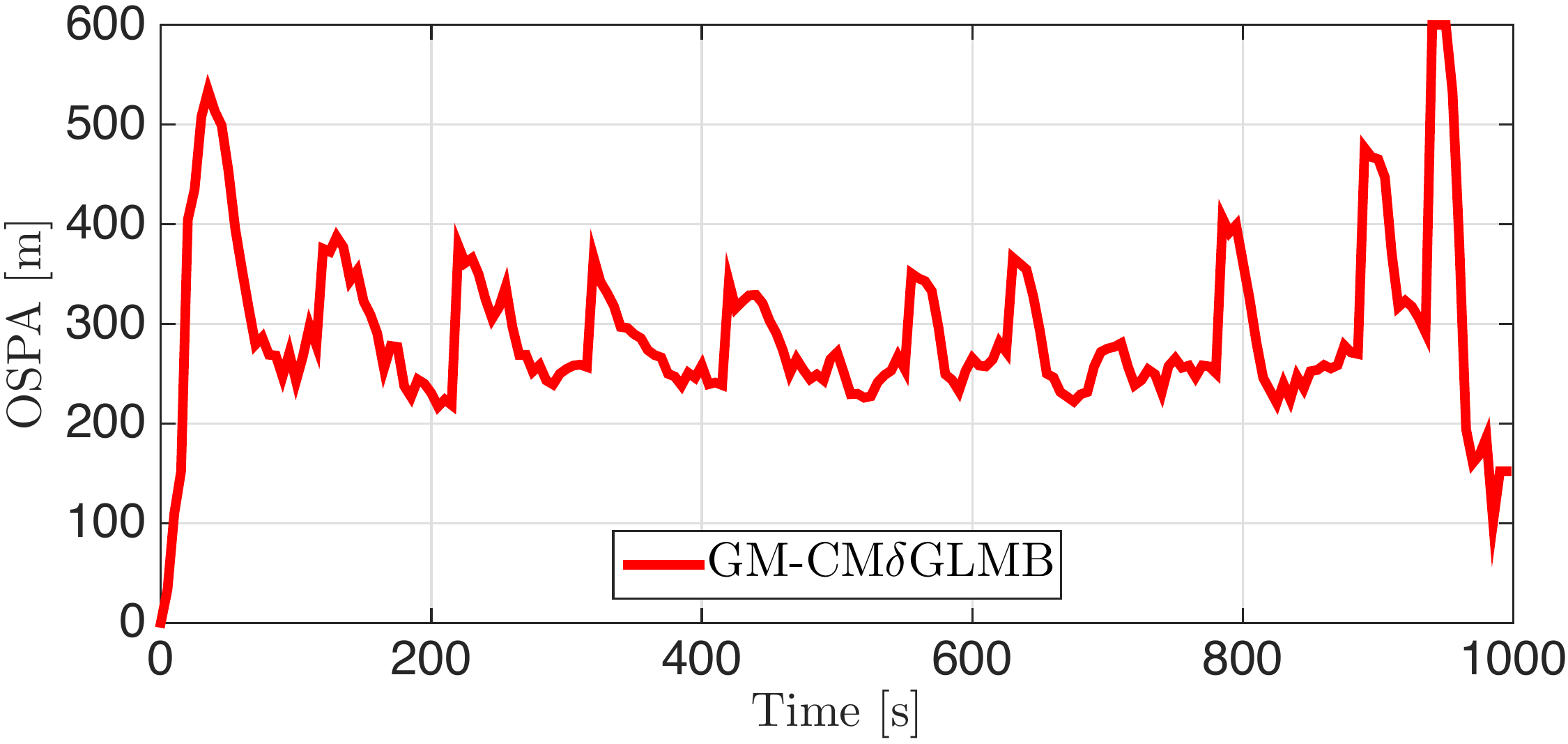}
		\caption{OSPA distance ($c = 600 \, [m]$, $p = 2$) under low $P_{D}$.}
		\label{fig:3:ospa}
        \end{minipage}
\end{figure}

% SUMMORY, CONLCUSION AND FUTURE WORK
\chapter{Conclusions and future work}
\label{chap:conclusion}
The present dissertation has summarized the work carried out during my 3-year Ph.D. research work in the field of distributed filtering and tracking.
Contributions have spanned from single-object to multi-object scenarios, in both centralized and distributed settings.
In particular, distributed tracking over sensor networks requires redesigning the architecture and algorithms to address the following issues:
\begin{enumerate}[label=\arabic*.]
	\item lack of a central fusion node;
	\item scalable processing with respect to the network size;
	\item each node operates without knowledge of the network network topology;
	\item each node operates without knowledge of the dependence between its own information and the information received from other nodes.
\end{enumerate}
Four main ingredients have been considered in order to provide novel research contributions.
\begin{enumerate}[label=\textsc{\roman*}.]
	\item The Bayesian framework for modeling the quantities to be estimated as random variables characterized by Probability Density Functions (PDFs), and for providing an improved estimation of such quantities by conditioning the PDF on the available noisy measurements.
	\item The Random Finite Set (RFS) formulation as it provides the concept of probability density of the multi-object state that allows us to directly generalize single-object estimation to the multi-object case.
	\item The Kullback-Leibler Average (KLA) as a sensible information-theoretic definition for consistently fusing PDFs of the quantity to be estimated provided by different nodes (or agents).
	\item Consensus as a tool for developing scalable and reliable distributed fusion techniques over a network where each agent aims to compute the collective average of a given quantity by iterative regional averages, where the terms ``collective'' and ``regional'' mean ``over all network nodes'' and, respectively, ``over neighboring nodes only''.
\end{enumerate}
\par
\noindent The main contributions of the present thesis are herewith summarized.
\begin{itemize}[label=\scriptsize{$\square$}]
	\item  A novel distributed single-object filter, namely \textit{Parallel Consensus on Likelihoods and Priors} (CLCP). The proposed algorithm is based on the idea of carrying out, in parallel, a separate consensus for the novel information (likelihoods) and one for the prior information (priors). This parallel procedure is conceived as an improvement of the \textit{Consensus on Posteriors} (CP) approach to avoid underweighting the novel information during the fusion steps. The outcomes of the two consensuses are then combined to provide the fused posterior density.
	\item Two novel consensus \textit{Multiple Model} (MM) filters to be used for tracking a maneuvering object, namely \textit{Distributed First Order Generalized Pseudo-Bayesian} (DGPB$_{1}$) and \textit{Distributed Interacting Multiple Model} (DIMM) filters.
	\item A novel result on \textit{Distributed Multi-Object Filtering} (DMOF) over a sensor network using RFSs. Specifically, a generalisation of the single-object \textit{Kullback-Leibler Average} (KLA) to the RFS framework has been devised.
	\item A novel consensus multi-object filter, the \textit{Consensus Cardinalized Probability Hypothesis Density} (CCPHD) filter. Each tracking agent locally updates multi-object CPHD, i.e. the cardinality distribution and the PHD, exploiting the multi-object dynamics and the available local measurements, exchanges such information with communicating agents and then carries out a fusion step to combine the information from all neighboring agents.
	\item A novel approximation of the $\delta$-GLMB filter, namely the \textit{Marginalized $\delta$-Generalised Labeled Multi-Bernoulli} (M$\delta$-GLMB). The result is based on a principled \textit{Generalised Labeled Multi-Bernoulli} (GLMB) approximation of the labeled RFS posterior that matches exactly the posterior PHD and cardinality distribution. The proposed approximation can be interpreted as performing a marginalization with respect to the association histories arising from the $\delta$-GLMB filter. The key advantage of the new filter lies in the reduced growth rate of the number of new components generated at each filtering step. In particular, the approximation (or marginalization) step performed after each update is guaranteed to reduce the number of generated components which normally arise from multiple measurement-to-track association maps. Typically, the proposed M$\delta$-GLMB filter requires much less computation and storage especially in multi-sensor scenarios compared to the $\delta$-GLMB filter. Furthermore the proposed M$\delta$-GLMB filter inherits the same implementation strategies and parallelizability of the $\delta$-GLMB filter.
	\item The generalisation of the KLA to labeled RFSs.
	\item Two novel consensus tracking filters, namely a \textit{Consensus Marginalized $\delta$-Generalized Labeled Multi-Bernoulli} (CM$\delta$GLMB) and \textit{Consensus Labeled Multi-Bernoulli} (CLMB) tracking filter. The proposed algorithms provide fully distributed, scalable and computationally efficient solutions for multi-object tracking.
\end{itemize}
\par
\noindent From the above-mentioned contributions, the following conclusions can be drawn.
\begin{itemize}[label=$\triangleright$]
	\item The KLA turns out to be an effective information-theoretic approach for distributing information over a sensor network.
	\item Consensus algorithms can be exploited so as to fuse, in a fully distributed and scalable fashion, the information collected from multiple sensors which are supposed to be heterogeneous and geographically dispersed.
	\item The CLCP has proven to be a reliable distributed filter. Its performance is crucially dependant on the choice of a suitable weight of the innovation (likelihood) term combined with the prior. In particular, the choice of such a weight equal to the (usually unknown) number of nodes provides the best performance when the number of consensus steps is sufficiently large, but for very few consensus steps can be improved resorting to other simple choices, proposed in this thesis work, which do not require knowledge of the network size.
	\item The DGPB$_{1}$ and DIMM have been proven to be capable of tracking a highly-maneuvering object with a network consisting of range-only and angle-only sensors along with communication nodes. DIMM performs significantly better than DGPB$_{1}$ when a turn is approaching.
	\item The multi-object KLA in the RFS framework turns out to be a significant contribution in the field of distributed multi-object filtering. The proposed consensus algorithm admits an intuitive interpretation in terms of distributed averaging of the local multi-object densities. Its formulation has led to the development of the CCPHD filter which proved its effectiveness in a realistic simulation scenario.
	\item The labeled RFS approach turns out to be a strong mathematical tool for modeling the probability density of a multi-object state. Its key features allow to directly generalise, in a rigorous and elegant way, the multi-object KLA to the broader labeled multi-object case.
	\item The M$\delta$-GLMB provides an excellent tradeoff between computational/storage requirements and performance for multi-object tracking. Moreover, it has an appealing mathematical formulation which $1)$ facilitates an efficient and tractable closed-form solution for distributed multi-object tracking and $2)$ preserves key summary statistics of the full multi-object posterior.
	\item The CLMB fails to track objects in low SNR scenarios. The problem is due to the approximation made to convert a $\delta$-GLMB into an LMB, becoming significant with low SNR. In particular, each local tracking filter fails to properly set the existence probability of the tracks in the case of having: $1)$ no local observability, $2)$ high clutter rate and $3)$ loss of the full posterior cardinality distribution after the probability density conversion. By having low existence probabilities, the extraction of the tracks fails even if the single object densities are correctly propagated in time.
\end{itemize}
\par
\noindent Based on the above achievements, possible future work will concern the following issues.
\begin{itemize}[label=$\circ$]
	\item To consider sensors with different field-of-view to cope with real-world sensor capabilities. So far, it has been assumed that the area subjected to surveillance is exactly the intersection of the sensor fields-of-view, or that each sensor field-of-view is exactly the surveillance region. By a practical viewpoint this issue needs to be addressed. Thus, the proposed algorithms will have to take into account a state-dependent probability of detection instead of a constant one as assumed in the present work. Moreover, the fusion techniques will have to be redesigned to consider the above-mentioned context. 
	\item To investigate distributed measurement-driven object initialization. The object birth process of the RFS filtering approach is capable of accounting for a large variety of scenarios in which object birth locations could also be unknown. However, a measurement-driven approach would be much more effective compared to a static birth process. The information provided by the sensors could be used to dynamically accommodate each sensor birth process. Thus, a more complete picture of newborn objects could be achieved by fusing local information with the one provided by the other nodes.
\end{itemize}

% APPENDIX
\appendix
\chapter{Proofs}
\label{chap:appendix}
\begin{proof}[\textbf{Proof of Theorem \ref{thm:wkla}:}]\label{apx:proofthm1}
\phantomsection{}
\addcontentsline{toc}{section}{Proof of Theorem \ref{thm:wkla}}
~\\
The cost to be minimized by the KLA is
\begin{IEEEeqnarray*}{rCl}
	J\left( \mu \right) & = & \displaystyle{\sum_{i \in \mathcal{N}}} \, \pi^i \, D_{KL} \left( \mu \parallel \mu^i \right)\, = \, \displaystyle{\sum_{i \in \mathcal{N}}} \, \pi^i \, \displaystyle{\sum_{j \in \mathcal{R}}} \, \mu^j \, \log\!\left( \dfrac{\mu^j}{\mu^{i,j}} \right) \\
	& = & \displaystyle{\sum_{j \in \mathcal{R}}} \, \mu^j \, \displaystyle{\sum_{i \in \mathcal{N}}} \, \pi^i \, \log\!\left( \dfrac{\mu^i}{\mu^{i,j}} \right)\\
	& = & \displaystyle{\sum_{j \in \mathcal{R}}} \, \mu^j \, \displaystyle{\sum_{i \in \mathcal{N}}} \, \log\!\left( \dfrac{\mu^j}{\mu^{i,j}} \right)^{\pi^i}  \\
	 & = & \displaystyle{\sum_{j \in \mathcal{R}}} \, \mu^j \, \log\!\left( \displaystyle{\prod_{i \in \mathcal{N}}}\dfrac{(\mu^j)^{\pi^i}}{( \mu^{i,j})^{\pi^i}} \right) \\
	 & = & \displaystyle{\sum_{j \in \mathcal{R}}} \, \mu^j \, \log\!\left( \dfrac{\mu^j}{\displaystyle{\prod_{i \in \mathcal{N}}} \, \left( \mu^{i,j} \right)^{\pi^i} } \right)
\end{IEEEeqnarray*}
where the relationship $\sum_{i \in \mathcal{N}} \, \pi^i = 1$ has been exploited.
Let $\overline{\mu}$ be defined as in (\ref{thm:pmfkla}), then 
\be
	J(\mu) = \displaystyle{\sum_{j \in \mathcal{R}}} \, \mu^j \, \log\!\left( \dfrac{\mu^j}{c \, \overline{\mu}^j} \right)
\ee
where $c \triangleq \displaystyle{\sum_{h \in \mathcal{R}}} \, \displaystyle{\prod_{i \in \mathcal{N}}} \, \left( \mu^{i,h} \right)^{\pi^i}$.
Hence
\be
	J \left( \mu \right) = D_{KL} \left( \mu \parallel \overline{\mu} \right) - \log \, c
\ee
where the relationship $\sum_{j \in \mathcal{R}} \, \mu^j = 1$ has been exploited.
Then the infimum of  $J \left( \mu \right)$, i.e. the weighted KLA of (\ref{KLA-d}), is just provided by $\mu = \overline{\mu}$ as in (\ref{thm:pmfkla}).
\end{proof}

\clearpage
\begin{proof}[\textbf{Proof of Theorem \ref{thm:KLA:GCI}:}]
\phantomsection{}
\addcontentsline{toc}{section}{Proof of Theorem \ref{thm:KLA:GCI}}
~\\
Let
\begin{IEEEeqnarray}{rCl}
	\widetilde{f}(X) & = & \displaystyle{\prod_i} \left[ f^i(X) \right]^{\omega_i} \, \\ 
	c & = & \displaystyle{\int} \widetilde{f}(X) \delta X \, ,
\end{IEEEeqnarray}
so that the fused multi-object density in (\ref{KLA:GCI}) can be expressed as $\overline{f}(X) = \widetilde{f}(X) / c$.
Then, the cost to be minimized by the KLA is
\begin{IEEEeqnarray*}{rCl}
	J(f) & = & \displaystyle{\sum_i} \omega_i D_{KL} \left( f \parallel f^i \right)\\
	& = & \displaystyle{\sum_i} \omega_i \displaystyle{\int}  f \left( X \right) \log\!\left(\dfrac{f(X)}{f^i(X)}\right) \delta X \\
	& = & \displaystyle{\sum_i} \omega_i \displaystyle{\sum_{n=0}^{\infty}} \dfrac{1}{n!} \displaystyle{\int} f \left( \{ x_1, \dots,x_n \} \right) \log\!\left( \dfrac{f \left( \{ x_1, \dots,x_n \} \right)}{f^i \left( \{ x_1, \dots,x_n \} \right)} \right) dx_1 \cdots dx_n \\
	& = & \displaystyle{\sum_{n=0}^{\infty}} \dfrac{1}{n!} \displaystyle{\int} f \left( \{ x_1, \dots,x_n \} \right) \displaystyle{\sum_i}  \omega_i \log\!\left( \dfrac{f \left( \{ x_1, \dots,x_n \} \right)}{f^i \left( \{ x_1, \dots,x_n \} \right)} \right) dx_1 \cdots dx_n \\
	& = & \displaystyle{\sum_{n=0}^{\infty}} \dfrac{1}{n!} \displaystyle{\int} f \left( \{ x_1, \dots,x_n \} \right) \log\!\left( \dfrac{f \left( \{ x_1, \dots,x_n \} \right)}{\displaystyle\prod_i \left[ f^i \left( \{ x_1, \dots,x_n \} \right)\right]^{\omega_i}} \right) dx_1 \cdots dx_n \\
	& = & \displaystyle{\sum_{n=0}^{\infty}} \dfrac{1}{n!} \displaystyle{\int} f \left( \{ x_1, \dots,x_n \} \right) \log\!\left( \dfrac{f \left( \{ x_1, \dots,x_n \} \right)}{ \widetilde{f} \left( \{ x_1, \dots,x_n \} \right) } \right) dx_1 \cdots dx_n \\
	& = & \displaystyle{\sum_{n=0}^{\infty}} \dfrac{1}{n!} \displaystyle{\int} f \left( \{ x_1, \dots,x_n \} \right) \log\!\left( \dfrac{f \left( \{ x_1, \dots,x_n \} \right)}{ c \cdot \overline{f} \left( \{ x_1, \dots,x_n \} \right) } \right) dx_1 \cdots dx_n \\
	& = & \displaystyle{\sum_{n=0}^{\infty}} \dfrac{1}{n!} \displaystyle{\int} f \left( \{ x_1, \dots,x_n \} \right) \log\!\left( \dfrac{f \left( \{ x_1, \dots,x_n \} \right)}{  \overline{f} \left( \{ x_1, \dots,x_n \} \right) }\right) dx_1 \cdots dx_n +\\
	&& - \displaystyle{\sum_{n=0}^{\infty}} \dfrac{1}{n!} \displaystyle{\int} f \left( \{ x_1, \dots,x_n \} \right) \log\!\left( c \right) dx_1 \cdots dx_n \\
	& = & \displaystyle{\int}  f \left( X \right) \, \log\!\left( \dfrac{f(X)}{\overline{f}(X)} \right) \delta X - \log\!\left( c \right) \displaystyle{\int} f(X) \delta X \\
	& = & D_{KL}\!\left( f \parallel \overline{f} \right) - \log\!\left( c \right)
\end{IEEEeqnarray*}
Since the KLD is always nonnegative and is zero if and only if its two arguments coincide almost everywhere, the above $J(f)$ is trivially minimized by $f_{KLA}(X)$ defined as in (\ref{KLA:GCI}).
\end{proof}

\begin{proof}[\textbf{Proof of Proposition \ref{pro:glmbapprox}:}]
\phantomsection{}
\addcontentsline{toc}{section}{Proof of Proposition \ref{pro:glmbapprox}}
~\\
First, the cardinality distributions of $\clmb{}{}{\cdot}$ and $\clmb[\hat]{}{}{\cdot}$ are proven to be the same.
Following (\ref{eq:lrfscardinality}), the cardinality distribution of the GLMB in (\ref{eq:ApproximateGLMB}) is given by:
\bie
	{\hat{\rho}}{(n)} & = & \sum_{L\in{\mathcal{F}}_{n}({\mathbb{L}})}\sum_{I\in{\mathcal{F}}({\mathbb{L}})}{w}^{(I)}{(L)} \\
	& = & \sum_{L\in{\mathcal{F}}_{n}({\mathbb{L}})}{w(L)} \\
	& = & \frac{1}{n!}\sum_{\left(\ell_{1},\ldots\ell_{n}\right)\in{\mathbb{L}}^{n}}{w(\{\ell}_{1}{,\ldots,\ell}_{n}{\})} \\
	& = & \frac{1}{n!}\sum_{\left(\ell_{1},\ldots\ell_{n}\right)\in{\mathbb{L}}^{n}}\int{\pi}\left(\left\{ (x_{1},\ell_{1}),\ldots,(x_{n},\ell_{n})\right\} \right)dx_1 \cdots dx_n\\
	& = & \frac{1}{n!}\int{\pi}\left(\left\{ {\mathbf{x}}_{1},\ldots,{\mathbf{x}}_{n})\right\} \right) d\mathbf{x}_1 \cdots d\mathbf{x}_n\\
	& = & \rho(n)
\eie
Second, the PHD corresponding to $\clmb{}{}{\cdot}$ and $\clmb[\hat]{}{}{\cdot}$ are proven to be the same.
Note from the definition of $p^{(I)}(x,\ell)$ that
\bie
	{p}^{(\{\ell,\ell_{1},\ldots,\ell_{n}\})}{(x,\ell)} & = & 1_{\{\ell,\ell_{1},\ldots,\ell_{n}\}}{(\ell)p}_{\{\ell_{1},\ldots,\ell_{n}\}}{(x,\ell)} \\
	& = & \int{p(}\left\{ (x,\ell),(x_{1},\ell_{1}),\ldots,(x_{n},\ell_{n})\right\} ) dx_1 \cdots dx_n\label{eq:pImarg}
\eie
Substituting the GLMB density (\ref{eq:ApproximateGLMB}) into the PHD of the object with label (or track) $\ell$ yields
\be
	\hat{d}(x, \ell) = \int \clmb{}{}{\left\{ \left( x, \ell \right) \cup \lb{X} \right\}} \delta \lb{X}
\ee
from which
\bie
	\hat{d}(x,\ell) & = & \int\sum_{I\in{\mathcal{F}}({\mathbb{L}})}w^{(I)}\left(\mathcal{L}(\{(x,\ell)\}\cup\mathbf{X})\right)[p^{(I)}]^{\{(x,\ell)\}\cup\mathbf{X}}\delta\mathbf{X} \\
	& = & \sum_{I\in{\mathcal{F}}({\mathbb{L}})}\left[\int w^{(I)}\left(\{\ell\}\cup\mathcal{L}(\mathbf{X})\right)[p^{(I)}]^{\mathbf{X}}\delta\mathbf{X}\right]p^{(I)}(x,\ell)\\
	& = & \sum_{I\in{\mathcal{F}}({\mathbb{L}})}\sum_{L\in{\mathcal{F}}({\mathbb{L}})}w^{(I)}\left(\{\ell\}\cup L\right)\left[p^{(I)}(\cdot,\ell)dx\right]^{L}p^{(I)}(x,\ell)
\eie
where the last step follows from Lemma 3 in \cite[Section III.B]{vovo1}.
Noting that $p^{(I)}(\cdot,\ell)$ is a probability density, and using (\ref{eq:MarginalizeGeneral-1}) gives
\bie
	\hat{d}(x,\ell) & = &\sum_{L\in{\mathcal{F}}({\mathbb{L}})}\sum_{I\in{\mathcal{F}}({\mathbb{L}})}\delta_{I}\left(\{\ell\}\cup L\right)w(I)p^{(I)}(x,\ell) \\
	& = & \sum_{L\in{\mathcal{F}}({\mathbb{L}})}w\left(\{\ell\}\cup L\right)p^{(\{\ell\}\cup L)}(x,\ell) \\
	& = & \sum_{n=0}^{\infty}\frac{1}{n!}\sum_{\left(\ell_{1},\ldots,\ell_{n}\right)\in{\mathbb{L}}^{n}}w\left(\left\{ \ell,\ell_{1},\ldots\ell_{n}\right\} \right)p^{(\left\{ \ell,\ell_{1},\ldots\ell_{n}\right\} )}(x,\ell)
\eie
from which by applying (\ref{eq:pImarg}), we obtain:
\bie
	\hat{d}(x,\ell) & = & \sum_{n=0}^{\infty}\frac{1}{n!}\sum_{\left(\ell_{1},\ldots,\ell_{n}\right)\in{\mathbb{L}}^{n}}w\left(\left\{ \ell,\ell_{1},\ldots\ell_{n}\right\} \right) \int p\left(\left\{ (x,\ell),(x_{1},\ell_{1}),\ldots,(x_{n},\ell_{n})\right\} \right) dx_1 \cdots dx_n\IEEEnonumber\\\\ 
	& = & \int\boldsymbol{\pi}({\left\{ \left(x,\ell\right)\right\} \cup\mathbf{X}})\delta\mathbf{X}\\
	& =  & d(x,\ell) \, .
\eie
\end{proof}

\begin{proof}[\textbf{Proof of Proposition \ref{pro:mdglmb}:}]
\phantomsection{}
\addcontentsline{toc}{section}{Proof of Proposition \ref{pro:mdglmb}}
~\\
	We apply the result in Proposition \ref{pro:glmbapprox} which can be used to calculate the parameters of the marginalized $\delta$-GLMB density.
	Notice that such a result applies to any labeled RFS density and our first step is to rewrite the $\delta$-GLMB density (\ref{eq:dglmbpdf}) in the general form for a labeled RFS density \cite{papi2014}, i.e.
	\be
		\boldsymbol{\pi}(\mathbf{X}) = w(\mathcal{L}(\mathbf{X}))p(\mathbf{X})
	\ee
	where 
	\bie
		w(\left\{ \ell_{1},\ldots,\ell_{n}\right\} ) & \triangleq & \int_{\mathbb{X}^{n}}\boldsymbol{\pi}(\left\{ (x_{1},\ell_{1}),\ldots,(x_{n},\ell_{n})\right\} )d(x_{1},\ldots,x_{n})\\
		& = & \sum_{I\in\mathcal{F}(\mathbb{L})}\delta_{I}(\left\{ \ell_{1},\ldots,\ell_{n}\right\} )\sum_{\xi\in\Xi}w^{(I,\xi)}\int_{\mathbb{X}^{n}}p^{(\xi)}(x_{1},\ell_{1})\cdots p^{(\xi)}(x_{n},\ell_{n})dx_{1}\cdots dx_{n}\IEEEnonumber\\~\\
		& = & \sum_{\xi\in\Xi}w^{(\left\{ \ell_{1},\ldots,\ell_{n}\right\} ,\xi)}\sum_{I\in\mathcal{F}(\mathbb{L})}\delta_{I}(\left\{ \ell_{1},\ldots,\ell_{n}\right\} )\\
		& = & \sum_{\xi\in\Xi}w^{(\left\{ \ell_{1},\ldots,\ell_{n}\right\} ,\xi)}
	\eie
	and
	\bie
		p(\left\{ (x_{1},\ell_{1}),\ldots,(x_{n},\ell_{n})\right\} ) & \triangleq & \frac{\boldsymbol{\pi}(\left\{ (x_{1},\ell_{1}),\ldots,(x_{n},\ell_{n})\right\} )}{w(\left\{ \ell_{1},\ldots,\ell_{n}\right\} )}\\
		& = & \Delta(\left\{ (x_{1},\ell_{1}),\ldots,(x_{n},\ell_{n})\right\}) \dfrac{1}{w(\left\{ \ell_{1},\ldots,\ell_{n}\right\} )} \cdot \IEEEnonumber\\
			&& \cdot \sum_{I\in\mathcal{F}\left(\mathbb{L}\right)}\delta_{I}\left(\left\{ \ell_{1},\ldots,\ell_{n}\right\} \right)\sum_{\xi\in\Xi}w^{\left(I,\xi\right)}\,\left[p^{\left(\xi\right)}\right]^{\left\{ (x_{1},\ell_{1}),\ldots,(x_{n},\ell_{n})\right\}} \, .\label{eq:p}
	\eie
	Applying Proposition \ref{pro:glmbapprox}, the parameters $w^{(I)}$ and $p^{(I)}(\cdot)$ for the M$\delta$-GLMB approximation that match the cardinality and PHD are
	\bie
		w^{(I)}(L) & = & \delta_{I}(L)w(I)=\delta_{I}(L)\sum_{\xi\in\Xi}w^{(I,\xi)}
	\eie
	and
	\bie
		p^{(I)}(x,\ell) & = & 1_{I}(\ell)p_{I-\left\{ \ell\right\} }(x,\ell)\\
		& = & 1_{I}(\ell)\int p(\left\{ (x,\ell),(x_{1},\ell_{1}),\ldots,(x_{j},\ell_{j})\right\} )d(x_{1},\ldots,x_{j})\label{eq:pmarignal}
	\eie
	where $I\backslash\{\ell\}=\left\{ \ell_{1},\ldots,\ell_{j}\right\} $.
	Substituting the expression (\ref{eq:p}) in (\ref{eq:pmarignal}), we have
	\bie
		p^{(I)}(x,\ell) & = & 1_{I}(\ell) \, \Delta(\left\{ (x,\ell),(x_{1},\ell_{1}),\ldots,(x_{j},\ell_{j})\right\}) \, \dfrac{1}{w(\left\{ \ell,\ell_{1},\ldots,\ell_{j}\right\} )}\cdot \IEEEnonumber\\
			&& \phantom{1_{I}(\ell)} \cdot \sum_{J\in\mathcal{F}\left(\mathbb{L}\right)}\delta_{J}\left(\left\{ \ell,\ell_{1},\ldots,\ell_{j}\right\} \right) \sum_{\xi\in\Xi}w^{\left(J,\xi\right)}\,\int\left[p^{\left(\xi\right)}\right]^{\left\{ (x,\ell),(x_{1},\ell_{1}),\ldots,(x_{j},\ell_{j})\right\} }d(x_{1},\ldots,x_{j})\IEEEnonumber\\&&\\
		& = & 1_{I}(\ell) \, \Delta(\left\{ (x,\ell),(x_{1},\ell_{1}),\ldots,(x_{j},\ell_{j})\right\}) \, \frac{1}{w(\left\{ \ell,\ell_{1},\ldots,\ell_{j}\right\} )}\IEEEnonumber\\
			&& \phantom{1_{I}(\ell)} \cdot \sum_{J\in\mathcal{F}\left(\mathbb{L}\right)}\delta_{J}\left(\left\{ \ell,\ell_{1},\ldots,\ell_{j}\right\} \right)\sum_{\xi\in\Xi}w^{\left(J,\xi\right)}p^{(\xi)}(x,\ell)
	\eie
	and, noting that $I=\left\{ \ell,\ell_{1},\ldots,\ell_{j}\right\}$, it follows that only one term in the sum over $J$ is non-zero thus giving
	\be
		p^{(I)}(x,\ell) = 1_{I}(\ell)\Delta(\left\{ (x,\ell),(x_{1},\ell_{1}),\ldots,(x_{j},\ell_{j})\right\} )\frac{1}{\displaystyle\sum_{\xi\in\Xi}w^{(I,\xi)}}\sum_{\xi\in\Xi}w^{\left(I,\xi\right)}p^{(\xi)}(x,\ell)
	\ee
	Consequently, the M$\delta$-GLMB approximation is given by
	\bie
		\hat{\boldsymbol{\pi}}(\mathbf{X}) & = & \sum_{I\in\mathcal{F}(\mathbb{L})}w^{(I)}(\mathcal{L}(\mathbf{X}))\left[p^{(I)}\right]^{\mathbf{X}}\\
		& = & \Delta(\mathbf{X})\sum_{I\in\mathcal{F}(\mathbb{L})}\delta_{I}(\mathcal{L}(\mathbf{X}))\sum_{\xi\in\Xi}w^{(I,\xi)}\left[1_{I}(\cdot)\frac{1}{\displaystyle\sum_{\xi\in\Xi}w^{(I,\xi)}}\sum_{\xi\in\Xi}w^{\left(I,\xi\right)}p^{(\xi)}(\cdot,\cdot)\right]^{\mathbf{X}}\\
		& = & \Delta(\mathbf{X})\sum_{I\in\mathcal{F}(\mathbb{L})}\delta_{I}(\mathcal{L}(\mathbf{X}))w^{(I)}\left[p^{(I)}\right]^{\mathbf{X}}
	\eie
	where
	\bie
		w^{(I)} & = & \sum_{\xi\in\Xi}w^{(I,\xi)}\\
		p^{(I)}(x,\ell) & = &1_{I}(\ell)\frac{1}{\displaystyle\sum_{\xi\in\Xi}w^{(I,\xi)}}\sum_{\xi\in\Xi}w^{\left(I,\xi\right)}p^{(\xi)}(x,\ell)\\
		& = & 1_{I}(\ell)\frac{1}{w^{(I)}}\sum_{\xi\in\Xi}w^{\left(I,\xi\right)}p^{(\xi)}(x,\ell)
	\eie
\end{proof}

\begin{proof}[\textbf{Proof of Theorem \ref{thm:mdglmb:nwgm}:}]
\phantomsection{}
\addcontentsline{toc}{section}{Proof of Theorem \ref{thm:mdglmb:nwgm}}
~\\
	For the sake of simplicity, let us consider only two M$\delta$-GLMB densities $\imath \in \left\{ 1, 2 \right\}$. From (\ref{eq:mdglmb:nwgm}) one gets 
	\be
		\overline{\boldsymbol{\pi}}(\mathbf{X}) = \dfrac{1}{K} \left[ \Delta\!(\mathbf{X})\sum_{L \in \mathcal{F}\!\left( \mathbb{L} \right)} \delta_{L}\!\left( \mathcal{L}\!\left( \mathbf{X} \right) \right) \, w_{1}^{\left( L \right)}\left[ p_{1}^{\left( L \right)}\right]^{\mathbf{X}} \right]^{\omega} \left[ \Delta\!(\mathbf{X})\sum_{L \in \mathcal{F}\!\left( \mathbb{L} \right)} \delta_{L}\!\left( \mathcal{L}\!\left( \mathbf{X} \right) \right) \, w_{2}^{\left( L \right)}\left[ p_{2}^{\left( L \right)}\right]^{\mathbf{X}} \right]^{1-\omega} .\label{eq:proof:mdglmbexp}
	\ee
	Notice that the exponentiation of a sum of delta functions is a sum of the exponentiated delta function terms, i.e
	\begin{numcases}{\overline{\boldsymbol{\pi}}(\mathbf{X})^{\omega} = }
		\Delta\!(\mathbf{X}) \left( w^{\left( L_{1} \right)} \right)^{\omega}\left[ \left( p^{\left( L_{1} \right)} \right)^{\omega}\right]^{\mathbf{X}} & \!\!\!\!\!\!\! if $\mathcal{L}\!\left( \mathbf{X} \right) = L_{1}$\nonumber\\
		\vdots & \!\!\!\!\!\vdots\nonumber\\
		\Delta\!(\mathbf{X}) \left( w^{\left( L_{n} \right)} \right)^{\omega}\left[ \left( p^{\left( L_{n} \right)} \right)^{\omega}\right]^{\mathbf{X}} & \!\!\!\!\!\!\! if $\mathcal{L}\!\left( \mathbf{X} \right) = L_{n}$\nonumber
	\end{numcases}
	Thus, (\ref{eq:proof:mdglmbexp}) yields:
	\begin{IEEEeqnarray}{rCl}
		\overline{\boldsymbol{\pi}}(\mathbf{X}) & = & \dfrac{\Delta\!(\mathbf{X})}{K} \sum_{L \in \mathcal{F}\!\left( \mathbb{L} \right)} \delta_{L}\!\left( \mathcal{L}\!\left( \mathbf{X} \right) \right) \, \left( w_{1}^{\left( L \right)} \right)^{\omega} \left( w_{2}^{\left( L \right)} \right)^{1-\omega} \, \left[ \left( p_{1}^{\left( L \right)} \right)^{\omega} \left( p_{2}^{\left( L \right)} \right)^{1-\omega} \right]^{\mathbf{X}}\\
		& = & \dfrac{\Delta\!(\mathbf{X})}{K} \sum_{L \in \mathcal{F}\!\left( \mathbb{L} \right)} \delta_{L}\!\left( \mathcal{L}\!\left( \mathbf{X} \right) \right) \, \left( w_{1}^{\left( L \right)} \right)^{\omega} \left( w_{2}^{\left( L \right)} \right)^{1-\omega} \left[ \int \left( p_{1}^{\left( L \right)}\!\left( x, \cdot \right) \right)^{\omega} \left( p_{2}^{\left( L \right)}\!\left( x, \cdot \right) \right)^{1-\omega} \right]^{L}  \cdot \IEEEnonumber\\
			&& \phantom{\dfrac{\Delta\!(\mathbf{X})}{K} \sum_{L \in \mathcal{F}\!\left( \mathbb{L} \right)}} \cdot \left[ \left( \omega \odot p_{1}^{\left( L \right)} \right) \oplus \left( \left( 1 - \omega \right) \odot p_{2}^{\left( L \right)} \right) \right]^{\mathbf{X}} .\label{eq:proof:mdglmbexp2}
	\end{IEEEeqnarray}
	The normalization constant $K$ can be easily evaluated exploiting Lemma 3 of \cite[Section III.B]{vovo1}, i.e.
	\begin{IEEEeqnarray}{rCl}
		K & = & \int \Delta\!(\mathbf{X}) \sum_{L \in \mathcal{F}\!\left( \mathbb{L} \right)} \delta_{L}\!\left( \mathcal{L}\!\left( \mathbf{X} \right) \right) \, \left( w_{1}^{\left( L \right)} \right)^{\omega} \left( w_{2}^{\left( L \right)} \right)^{1-\omega} \left[ \left( p_{1}^{\left( L \right)} \right)^{\omega} \left( p_{2}^{\left( L \right)} \right)^{1-\omega} \right]^{\mathbf{X}} \delta X\\
		& = & \sum_{L \subseteq \mathbb{L}} \left( w_{1}^{\left( L \right)} \right)^{\omega} \left( w_{2}^{\left( L \right)} \right)^{1-\omega} \left[ \int \left( p_{1}^{\left( L \right)} \right)^{\omega} \left( p_{2}^{\left( L \right)}\!\left( x, \cdot \right) \right)^{1-\omega} dx \right]^{L} \, .\label{eq:proof:mdglmbnorm}
	\end{IEEEeqnarray}
	Applying (\ref{eq:proof:mdglmbnorm}) in (\ref{eq:proof:mdglmbexp2}) one has
	\begin{IEEEeqnarray}{rCl}
		\overline{p}^{\left( L \right)} & = & \left[ \left( \omega \odot p_{1}^{\left( L \right)} \right) \oplus \left( \left( 1 - \omega \right) \odot p_{2}^{\left( L \right)} \right) \right]^{\mathbf{X}}\\
		\overline{w}^{\left( L \right)} & = & \dfrac{\displaystyle\left( w_{1}^{\left( L \right)} \right)^{\omega} \left( w_{2}^{\left( L \right)} \right)^{1-\omega} \left[ \int \left( p_{1}^{\left( L \right)}\!\left( x, \cdot \right) \right)^{\omega} \left( p_{2}^{\left( L \right)}\!\left( x, \cdot \right) \right)^{1-\omega} \right]^{L}}
			{\displaystyle\sum_{F \subseteq \mathbb{L}}\left( w_{1}^{\left( F \right)} \right)^{\omega} \left( w_{2}^{\left( F \right)} \right)^{1-\omega} \left[ \int \left( p_{1}^{\left( F \right)}\!\left( x, \cdot \right) \right)^{\omega} \left( p_{2}^{\left( F \right)}\!\left( x, \cdot \right) \right)^{1-\omega} \right]^{F}}\IEEEnonumber\\
	\end{IEEEeqnarray}
	It can be proved by induction that Theorem \ref{thm:mdglmb:nwgm} holds considering $I$ M$\delta$-GLMB densities instead of $2$.
\end{proof}

\begin{proof}[\textbf{Proof of Proposition \ref{pro:mdglmb:fusion}:}]
\phantomsection{}
\addcontentsline{toc}{section}{Proof of Proposition \ref{pro:mdglmb:fusion}}
~\\
	The proof readily follows by noting that the KLA (\ref{eq:kla:gci}) can be evaluated using (\ref{eq:mdglmb:nwgm}) of Theorem \ref{thm:mdglmb:nwgm}.
\end{proof}

The following Lemma is useful to prove Theorem \ref{thm:lmb:nwgm}.
\begin{lem}[Normalization constant of LMB KLA]
\label{lem:k}
~\newline
	Let $\boldsymbol{\pi}^{1}(\mathbf{X}) = \left\{ \left(r_{1}^{(\ell)}, p_{1}^{(\ell)} \right) \right\}_{\ell \in \mathbb{L}}$ and $\boldsymbol{\pi}^{2}(\mathbf{X}) = \left\{ \left( r_{2}^{(\ell)},p_{2}^{(\ell)} \right) \right\}_{\ell \in \mathbb{L}}$ be two LMB densities on $\mathbb{X} \times \mathbb{L}$ and $\omega \in \left( 0, 1 \right)$ then
	\begin{IEEEeqnarray}{rCl}
		K & \triangleq & \displaystyle\int\boldsymbol{\pi}^{1}(\mathbf{X})^{\omega}\boldsymbol{\pi}^{2}(\mathbf{X})^{1-\omega}\delta\mathbf{X} \label{eq:defK}\\
		& = & \left< \omega \odot \boldsymbol{\pi}^{1}, \left( 1-\omega \right) \odot \boldsymbol{\pi}^{2} \right> \\
		& = & \left( \widetilde{q}^{(\cdot)} + \widetilde{r}^{(\cdot)} \right)^{\mathbb{L}}.
	\label{eq:k}
	\end{IEEEeqnarray}
\end{lem}

\begin{proof}[\textbf{Proof of Lemma \ref{lem:k}:}]
\phantomsection{}
\addcontentsline{toc}{section}{Proof of Lemma \ref{lem:k}}
~\\
	We make use of the Binomial Theorem \cite{math} which states 
	\begin{equation}
		\sum_{L \subseteq \mathbb{L}} f^{L} = \left( 1 + f \right)^{\mathbb{L}}
		\label{eq:bin}
	\end{equation}
	Applying Lemma 3 of \cite[Section III.B]{vovo1} to the definition (\ref{eq:defK}) of $K$ gives 
	\begin{IEEEeqnarray}{rCl}
		K & = & \left( \widetilde{q}^{(\cdot)} \right)^{\mathbb{L}} \sum_{L \subseteq \mathbb{L}}\left( \dfrac{\widetilde{r}^{(\cdot)}}{\widetilde{q}^{(\cdot)}} \right)^{L} \, .
		\label{eq:int:k}
	\end{IEEEeqnarray}
	Applying (\ref{eq:bin}) to (\ref{eq:int:k}) gives
	\begin{IEEEeqnarray}{rCl}
		K & = & \left( \widetilde{q}^{(\cdot)} \right)^{\mathbb{L}} \left( 1 + \dfrac{\widetilde{r}^{(\cdot)}}{\widetilde{q}^{(\cdot)}} \right)^{\mathbb{L}}\\
		& = & \left( \widetilde{q}^{(\cdot)} + \widetilde{r}^{(\cdot)} \right)^{\mathbb{L}}
	\end{IEEEeqnarray}
	having defined 
	\begin{IEEEeqnarray}{rCl}
		\widetilde{r}^{(\ell)} & \triangleq & \displaystyle\int \left( r_{1}^{(\ell)} p_{1}^{(\ell)}(x) \right)^{\omega} \left( r_{2}^{(\ell)} p_{2}^{(\ell)}(x) \right)^{1-\omega}dx, \label{lem:unnex}\\
		\widetilde{q}^{(\ell)} & \triangleq & \left( 1 - r_{1}^{(\ell)} \right)^{\omega} \left( 1 - r_{2}^{(\ell)} \right)^{1-\omega}. \label{lem:unex}
	\end{IEEEeqnarray}
\end{proof}

\begin{proof}[\textbf{Proof of Theorem \ref{thm:lmb:nwgm}:}]
\phantomsection{}
\addcontentsline{toc}{section}{Proof of Theorem \ref{thm:lmb:nwgm}}
~\\
	For the sake of simplicity, let us consider only two LMB densities $\imath \in \left\{ 1, 2 \right\}$. From (\ref{eq:lmb:nwgm}) one gets 
	\begin{IEEEeqnarray}{rCl}
		\overline{\boldsymbol{\pi}}(\mathbf{X}) & = & \dfrac{1}{K} \left[ \Delta\!\left( \mathbf{X} \right) w_{1}(\mathcal{L}(\mathbf{X})) p_{1}^{\mathbf{X}} \right]^{\omega} \left[ \Delta\!\left( \mathbf{X} \right) w_{2}(\mathcal{L}(\mathbf{X})) p_{2}^{\mathbf{X}} \right]^{1-\omega}\\
		& = & \dfrac{\Delta\!\left( \mathbf{X} \right)}{K}
		\left( \widetilde{q}^{(\cdot)} \right)^{\mathbb{L}}
		\left[ 1_{\mathbb{L}}\!\left( \cdot \right) \left( \dfrac{r_{1}^{(\cdot)}}{1 - r_{1}^{(\cdot)}}\right)^{\omega} \left( \dfrac{r_{2}^{(\cdot)}}{1 - r_{2}^{(\cdot)}} \right)^{1-\omega} \right]^{\mathcal{L}(\mathbf{X})}\left( p_{1}^{\omega}\,p_{2}^{1-\omega} \right)^{\mathbf{X}}\\
		& = & \dfrac{\Delta\!\left( \mathbf{X} \right)}{K}
		\left( \widetilde{q}^{(\cdot)} \right)^{\mathbb{L}} \left[ 1_{\mathbb{L}}\!\left( \cdot \right) \left( \dfrac{r_{1}^{(\cdot)}}{1 - r_{1}^{(\cdot)}}\right)^{\omega} \left( \dfrac{r_{2}^{(\cdot)}}{1 - r_{2}^{(\cdot)}} \right)^{1-\omega} \int p_{1}^{\omega}\,p_{2}^{1-\omega} dx \right]^{\mathcal{L}(\mathbf{X})} \cdot \IEEEnonumber\\
		&& \cdot \left[ \left( \omega \odot p_{1} \right) \oplus \left( \left( 1-\omega \right) \odot p_{2} \right) \right]^{\mathbf{X}}
	\end{IEEEeqnarray}
	Hence, recalling definitions (\ref{lem:unnex}) and (\ref{lem:unex}), one has 
	\begin{IEEEeqnarray}{rCl}
		\overline{w}(L) & = & \dfrac{\left( \widetilde{q}^{(\cdot)} \right)^{\mathbb{L}} \left( 1_{\mathbb{L}}\!\left( \cdot \right) \dfrac{\widetilde{r}^{(\cdot)}}{\widetilde{q}^{(\cdot)}} \right)^{L}}{\left( \widetilde{q}^{(\cdot)} + \widetilde{r}^{(\cdot)}\right)^{\mathbb{L}}}\\
		& = & \left( \dfrac{\widetilde{q}^{(\cdot)}}{\widetilde{q}^{(\cdot)} + \widetilde{r}^{(\cdot)}}\right)^{\mathbb{L}\backslash L} \left( \dfrac{1_{\mathbb{L}}\!\left( \cdot \right) \widetilde{r}^{(\cdot)}}{\widetilde{q}^{(\cdot)} + \widetilde{r}^{(\cdot)}}\right)^{L}\label{eq:thm:weight}\\
		& = & \left( \dfrac{\widetilde{q}^{(\cdot)}}{\widetilde{q}^{(\cdot)} + \widetilde{r}^{(\cdot)}} \right)^{\mathbb{L}} \left( 1_{\mathbb{L}}\!\left( \cdot \right) \dfrac{\widetilde{r}^{(\cdot)}}{\widetilde{q}^{(\cdot)}} \right)^{L}, \\
		\overline{p}^{(\ell)}(x) & = & %\dfrac{\displaystylep_{1}^{(\ell)}(x)^{\omega}p_{2}^{(\ell)}(x)^{1-\omega}}{\displaystyle\intp_{1}^{(\ell)}(\xi)^{\omega}p_{2}^{(\ell)}(\xi)^{1-\omega}d\xi}
		\left[ \left( \omega \odot p_{1}^{(\ell)} \right) \oplus \left( \left( 1-\omega \right) \odot p_{2}^{(\ell)} \right) \right]^{\mathbf{X}}.
	\end{IEEEeqnarray}
	It follows from (\ref{eq:thm:weight}) that $\overline{q}^{(\ell)} + \overline{r}^{(\ell)} = 1$, $\forall\, \ell \in \mathbb{L}$, where
	\begin{IEEEeqnarray}{rCl}
		\overline{r}^{(\ell)} & = & \dfrac{\widetilde{r}^{(\ell)}}{\widetilde{q}^{(\ell)} + \widetilde{r}^{(\ell)}} \,\\
		\overline{q}^{(\ell)} & = & \dfrac{\widetilde{q}^{(\ell)}}{\widetilde{q}^{(\ell)} + \widetilde{r}^{(\ell)}}.
	\end{IEEEeqnarray}
	It can be proved by induction that Theorem \ref{thm:lmb:nwgm} holds considering $I$ LMB densities instead of $2$.
\end{proof}

\begin{proof}[\textbf{Proof of Proposition \ref{pro:lmb:fusion}:}]
\phantomsection{}
\addcontentsline{toc}{section}{Proof of Proposition \ref{pro:lmb:fusion}}
~\\
	The proof readily follows by noting that the KLA (\ref{eq:kla:gci}) can be evaluated using (\ref{eq:lmb:nwgm}) of Theorem \ref{thm:lmb:nwgm}.
\end{proof}

% BIBLIOGRAPHY
\cleardoublepage
\singlespacing
\phantomsection{}
\addcontentsline{toc}{chapter}{Bibliography}
\bibliographystyle{IEEEtranSA}
%\nocite{*}
% Generated by IEEEtranSA.bst, version: 1.13 (2008/09/30)
\providecommand{\etalchar}[1]{$^{#1}$}

\end{document}